\documentclass[11pt]{article}

\usepackage[utf8]{inputenc}
\usepackage[T1]{fontenc}
\IfFileExists{newtxtext.sty}{%
  \usepackage{amsmath,amsthm}%
  \usepackage{newtxtext}%
  \usepackage{newtxmath}%
}{%
  \usepackage{amsmath,amssymb,amsfonts,amsthm}%
  \usepackage{mathptmx}%
}  
\usepackage{bm}   
\usepackage{mathtools}   
\usepackage{graphicx} 
\usepackage{booktabs}
\usepackage{array} 
\usepackage{enumitem}  
\usepackage[most]{tcolorbox}
\usepackage{tikz}
\usetikzlibrary{arrows.meta,calc,positioning,decorations.pathreplacing} 
\definecolor{figblue}{RGB}{25,80,160}
\definecolor{figred}{RGB}{178,45,38} 
\definecolor{figgreen}{RGB}{22,115,72}
\definecolor{figorange}{RGB}{205,105,0}
\usepackage{float} 
\usepackage[margin=2.6cm]{geometry}
\usepackage{natbib}
\usepackage{hyperref}
\hypersetup{colorlinks=true, linkcolor=black!70!blue, citecolor=black!60!blue, urlcolor=black!60!blue,
  pdftitle={Shape-Preserving Treatment Effects for Functional Outcomes},
  pdfauthor={Agathe Chabassier, Joris Olivier, Julie Josse, Erwan Scornet}}
\theoremstyle{plain}
\newtheorem{theorem}{Theorem}
\newtheorem{proposition}[theorem]{Proposition}
\newtheorem{lemma}[theorem]{Lemma}
\newtheorem{corollary}[theorem]{Corollary}
\theoremstyle{definition}
\newtheorem{assumption}{Assumption}

\newtheorem{example}[theorem]{Example}
\newtheoremstyle{remarkit}%
  {\topsep}{\topsep}%
  {\normalfont}%
  {}%
  {\itshape}%
  {}%
  {.5em}%
  {\thmname{#1}\thmnumber{ #2}.\thmnote{ \itshape #3\/}}%
\theoremstyle{remarkit}
\newtheorem{remark}[theorem]{Remark}

\newtcolorbox{defbox}[1]{
    colback=white, colframe=white,
    fonttitle=\bfseries\small, coltitle=black!80, colbacktitle=white,
    enhanced, attach boxed title to top left={yshift=-2mm, xshift=2mm},
    title=#1, sharp corners, boxrule=0pt,
    borderline={0.6pt}{0pt}{black!40!gray, dashed},
    fontupper=\small\itshape}

\newtcolorbox{resultbox}[1]{ 
    colback=white, colframe=black!60!gray,
    fonttitle=\bfseries, coltitle=black, colbacktitle=white,
    enhanced, attach boxed title to top left={yshift=-2mm, xshift=2mm},
    title={#1}, sharp corners, boxrule=0.6pt,
    drop shadow={black!10!white}}

\newtcolorbox{mainbox}[1]{
    colback=white, colframe=black,
    fonttitle=\bfseries, coltitle=black, colbacktitle=white,
    enhanced, attach boxed title to top left={yshift=-2mm, xshift=2mm},
    title=#1, sharp corners, boxrule=0.9pt,
    borderline={0.35pt}{3pt}{black},
    left=8pt, right=8pt, top=6pt, bottom=6pt}

\tcolorboxenvironment{proposition}{
    enhanced, breakable,
    colback=white, colframe=black!45,
    sharp corners, boxrule=0.5pt,
    left=6pt, right=6pt, top=4pt, bottom=4pt,
    before skip=10pt, after skip=10pt}

\newtcolorbox{proofenv}[1][Proof.]{
    enhanced, breakable, colback=gray!2, colframe=white, frame hidden,
    borderline west={1pt}{0pt}{black!55!gray}, boxrule=0pt, sharp corners,
    left=10pt, right=4pt, top=4pt, bottom=4pt,
    before upper={\textit{\textbf{#1}}\ \ignorespaces}}

\newcommand{\shp}{\mathfrak s}   %

\newcommand{\F}{\mathcal{F}}
\newcommand{\X}{\mathcal{X}}
\newcommand{\G}{\Gamma}
\newcommand{\Gt}{\widetilde{\Gamma}}
\newcommand{\Gc}{\mathcal{G}}
\newcommand{\gid}{\gamma_{\mathrm{id}}}
\newcommand{\E}{\mathbb{E}}
\newcommand{\Prob}{\mathbb{P}}
\newcommand{\R}{\mathbb{R}}
\newcommand{\Lt}{L^2}
\newcommand{\norm}[1]{\left\| #1 \right\|}
\newcommand{\inner}[2]{\left\langle #1, #2 \right\rangle}
\newcommand{\ind}{\mathbf{1}}
\newcommand{\del}{d_{elastic}}
\newcommand{\dG}{d_{\G}}
\newcommand{\dFR}{d_{\mathrm{FR}}}
\newcommand{\Pstar}{P^{*}}
\newcommand{\Sw}{\mathcal{S}}
\newcommand{\Sphere}{\mathbb{S}_{\infty}}
\newcommand{\pto}{\overset{p}{\longrightarrow}}
\DeclareMathOperator*{\argmin}{arg\,min}
\DeclareMathOperator{\inj}{inj}
\DeclareMathOperator{\supp}{supp}
\newlength{\whywidth}
\newcommand{\why}[1]{&&\parbox[t]{\whywidth}{\raggedright\footnotesize\color{black!55}#1}}
\newcommand{\tk}[1]{} 
\title{\bf Shape-Preserving Treatment Effects for Functional Outcomes}
\author{Agathe Chabassier\textsuperscript{1,2,3}\thanks{Corresponding author: \texttt{agathe.chabassier@withings.com}} \and Joris Olivier\textsuperscript{1} \and Julie Josse\textsuperscript{2} \and Erwan Scornet\textsuperscript{3}}
\date{\small
\textsuperscript{1}Withings, Issy-les-Moulineaux, France\\
\textsuperscript{2}Inria, PreMeDICaL Team, IDESP (UMR Universit\'e de Montpellier -- Inserm), Montpellier, France\\
\textsuperscript{3}Sorbonne Universit\'e, CNRS, Laboratoire de Probabilit\'es, Statistique et Mod\'elisation (LPSM), Paris, France}

\begin{document}
\maketitle
\vspace{-0.8cm} 

\begin{abstract}
\noindent
By recording health measurements daily rather than at a few clinical visits, connected devices pave the way for more precise real-world evidence that measures how the effect of an intervention unfolds over time: does the response rebound, peak late, or oscillate? Yet causal methods for such functional outcomes remain scarce, and those that exist target the Functional Average Treatment Effect (FATE). Its pointwise averaging flattens the features patients share whenever their curves are misaligned in time, as is common with biosignals, where patients respond with their own timing. We develop a causal inference framework that captures the shape of the effect of a binary intervention, based on the phase-centered elastic mean from elastic shape analysis. We lift this mean from a sample statistic to a functional of the data-generating law and prove that it is well defined under explicit geometric assumptions. Because it is not linear, it yields three causal estimands instead of one: populational, individual and local. We identify the populational estimand by inverse-probability weighting and the local estimand through the conditional outcome laws, estimate them by weighting and by matching, and show that the individual one is not identifiable. On simulations with timing variability or heterogeneous effect shapes, our estimators recover the shape that the FATE loses. On real weight trajectories following GLP-1 receptor agonist initiation, they reveal a rebound that the FATE averages away. This framework thus gives clinicians a treatment-effect curve whose shape and timing reflect genuine patient trajectories rather than a misalignment artifact.
\medskip

\end{abstract}

\section{Introduction} \label{sec:intro}

In medicine, the gold standard to prove a drug's effect is the \emph{Randomized Controlled Trial (RCT)}: as treatment is given at random, the treated and control groups have the same characteristics, so that any difference in outcomes observed in the two groups can be attributed to the causal effect of the treatment. Causal inference provides the formal framework for the analysis of such trials, and also for the far more frequent situation in which randomisation is absent. When only observational data are available, the standard formulation is to emulate a target trial \citep{hernan_target_2016}: one states the trial that would have answered the question, then adjusts for confounding to recover the causal effect under explicit assumptions. 

The classical setting, where one studies the impact of a scalar treatment on a scalar
outcome, has been intensively studied in the literature
\citep{hernan_causal_2020,imbens_rubin_2015}, starting from a settled choice of the
quantity to be identified and estimated, called the \emph{causal estimand}. In that setting, the most common way to measure an average effect on the difference scale is the \emph{Average Treatment Effect (ATE)}. Write $Y^{(1)}$ and $Y^{(0)}$ for the \emph{potential outcomes}
of a unit, that is the outcomes that it would have under treatment and no treatment, respectively. The ATE is then the difference between their expectations
\[ \tau_{\mathrm{ATE}} = \E[Y^{(1)}] - \E[Y^{(0)}]. \]
Note that ratio scales, such as the risk ratio, define an average effect as well. We work throughout on the difference scale. 

However, in many applications, the outcome of interest is not a scalar but a function, more specifically a curve indexed by time. Examples include body weight over the year, hormonal variations, physical-activity intensity across the day. Outcomes of this kind are becoming ordinary
rather than exceptional: connected scales, wearable devices and health applications now
record them day by day, for large populations and over long follow-up. In this paper, we restrict ourselves to the setting of a binary treatment acting on such \emph{functional outcomes}. Note that only the outcome is time-dependent and that the treatment is assigned once at baseline, which places our problem outside the literature on time-varying treatments, dynamic regimes, and marginal structural models \citep{hernan_causal_2020}, where time parametrises the interventions themselves.
Before reviewing the methodologies designed for this setting, we consider two ways of reducing it to a simpler one.

First, one can always reduce the functional outcome to a \emph{scalar}, for instance by taking the average of the curve over time, or by taking its maximum value. This aggregation, however, discards clinical information carried by the individual trajectories, and the resulting quantity may not answer the scientific question at hand. A single number cannot describe all the ways a treatment shapes the short-term and the long-term evolution of a trajectory. Alternatively, the functional outcome can be reduced to the \emph{distributions of its values at each time point}, as done in \citet{chabassier_glp1_2026}, to leverage the growing literature on distributional causal inference \citep{lin_causal_2023,katta_interpretable_2024}. This literature treats distribution-valued outcomes via Wasserstein barycenters, with specific functional causal estimands designed for this setting. However, these approaches discard the temporal structure in the trajectories: an increasing curve and a decreasing one can have the same distribution of values, but usually yield very different clinical interpretations.

In order to fully capture all the clinical information contained in the functional outcomes, one needs to develop suitable \emph{functional causal inference methodologies}, able to estimate the impact of a scalar treatment on entire trajectories. Despite the growing availability of functional data, this field has received relatively little attention, especially regarding the choice of tailored estimands. Designing suitable causal estimands is, however, not as straightforward as in the scalar case. Indeed, given the higher-dimensional nature of functional outcomes, several causal estimands can be defined, each targeting a different aspect of the treatment's impact. Some may capture the typical values the functional outcomes take at each time point, others the typical shape of the response over time: does it rebound, peak late, or oscillate? We review below the few existing proposals and their limitations, before presenting ours.

Recently, the scalar ATE has been directly generalized to the setting of functional outcomes through the \emph{Functional Average Treatment Effect (FATE)}, defined as the pointwise average of the difference between the treated and control potential outcomes:
\begin{equation}
  \label{eq:fate}
\tau_{\mathrm{FATE}}(t) = \E\big[Y^{(1)}(t)\big] - \E\big[Y^{(0)}(t)\big],
\qquad t \in [0,T].
\end{equation}
\noindent \citet{ecker_causal_2024} introduced this new causal estimand, and \citet{testa_doubly_2025} developed doubly robust estimators with simultaneous confidence bands for it. In a missing-at-random setting, \citet{liu_doubly_2024} also developed doubly robust estimators for the mean of a functional outcome using parametric models. 

These works, however, leave aside a property of pointwise averaging that is well established in Functional Data Analysis (FDA): averaging at fixed time points does not preserve shape. This phenomenon is called \emph{amplitude flattening} \citep{marron_functional_2015,srivastava_klassen_2016,srivastava_registration_2011}: when curves are misaligned in time, the pointwise mean fails to capture the shape of the individual curves, spreading for example sharp peaks into low, wide bumps that are shorter than any individual's. This misalignment of individual curves, called \emph{phase variation}, is very common in functional data, especially biosignals. Consider for instance growth curves in children: while all children go through the same growth peaks, they do so at different ages, and the pointwise mean of their growth curves will not capture the typical shape pattern of a child's growth curve \citep{srivastava_klassen_2016}.

The FATE, a difference of two pointwise means, \emph{inherits this distortion} and can entirely miss a shape that is present in every single unit. Consider a drug that induces a sharp weight loss followed by a rebound, at a different time for each patient (Figure~\ref{fig:introfate}(a)). The FATE averages the individual effects at fixed time points. As the patients do not drop at the same time, this average becomes a long, shallow plateau that only slowly returns to zero, although every simulated patient drops sharply and rebounds quickly (Figure~\ref{fig:introfate}(b)). Such a plateau suggests a sustained effect with no rebound, which is clinically misleading. The same pattern appears on real weight trajectories (Section~\ref{sec:realdata}).

\begin{figure}[htbp]
\centering
\includegraphics[width=\textwidth]{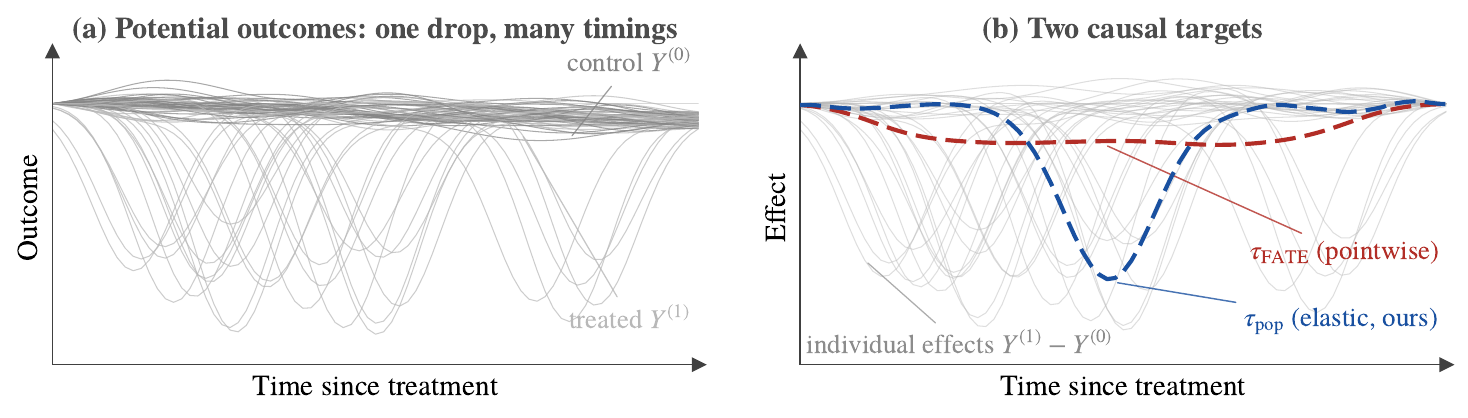}
\caption{\textbf{The FATE flattens a shape that every unit has.} Simulated drop-and-rebound curves, as in Scenario~1 (Section~\ref{sec:simmodel}) but with nadir times spread over most of the follow-up and a narrower dip (width $0.065$ instead of $0.10$): under treatment, every unit drops sharply and then rebounds, each at its own time. \emph{(a)} Potential outcomes of $30$ units under control (dark grey) and under treatment (light grey). \emph{(b)} Individual effects $Y^{(1)}-Y^{(0)}$ of the same units (grey) and two causal targets. The FATE (red, dashed) is the pointwise mean of the individual effects: it is a shallow plateau, although every unit rebounds. The populational elastic effect $\tau_{\mathrm{pop}}$ (blue, dashed) is the difference of the elastic means of the two arms (Section~\ref{sec:three}).}
\label{fig:introfate} 
\end{figure}

The canonical solution from FDA to this amplitude flattening issue is to average curves with the \emph{phase-centered elastic mean}, or \emph{elastic mean} for short, instead of the pointwise mean \citep{srivastava_klassen_2016,srivastava_registration_2011}. The former non-parametrically captures the typical shape of a sample of curves that the cross-sectional mean flattens out, and its features occur at the sample's typical timing (see Figures~\ref{fig:warpclock} and~\ref{fig:karchersteps} in Section~\ref{sec:background}, and Figure~\ref{fig:meansgrid} in Appendix~\ref{app:meansgrid}). Averaging this way means solving a minimisation problem rather than averaging values:
the elastic mean is based on a \emph{Fréchet mean} in a space of shapes. This Fréchet mean is the shape that minimises the average
squared \emph{elastic distance} to the sample, where the elastic distance between two
shapes is measured after aligning representatives of the shapes as well as possible. A second step then
places the features of that shape at the sample's typical timing, which is what
\emph{phase-centered} refers to, and a third step averages the aligned curves. Together, these tools form \emph{elastic shape analysis}, and important analytic properties of the square-root framework that they use, such as the existence of optimal re-timings, were established by \citet{bruveris_optimal_2016}. 

Two recent works instead replace the pointwise mean by Fréchet means in general metric spaces. \citet{kurisu_geodesic_2024} develop this idea for random objects in a general geodesic space, with a doubly robust estimator and convergence guarantees. Conditionally on selecting a suitable metric, this framework can be applied to our functional setting. With the usual $L^2$ metric, however, the Fréchet mean of
curves is again their pointwise mean, so the flattening returns. Using the elastic distance instead is not straightforward, and it violates some of the assumptions of the framework.

The second work, which is closer to ours, is that of \citet{raykov_kernel_2025}. Their causal estimands are Fréchet means of the potential outcomes, identified through outcome regression, with conditional outcome regression models estimated by operator-valued kernel ridge regression. Curve registration to the Fréchet mean shape (the Karcher mean of the SRSF literature) enters separately, as an unadjusted preprocessing step applied before the causal analysis, which then targets the Fréchet mean (under any functional metric) of the already-aligned potential outcomes. Two consequences follow.
First, no clear overall causal estimand is targeted by the whole pipeline (preprocessing followed by Fréchet-mean estimation), which prevents interpreting the output as the typical shape of a treatment effect. Second, the dependence and bias that this pre-processing alignment to a common Fréchet mean shape induces among the aligned curves is left unaddressed. 

Taken together, these elements suggest a simple correction to the amplitude flattening observed in the FATE: build the estimand from the elastic means of the two arms rather than from their pointwise means, as $\tau_{\mathrm{pop}}$ does in Figure~\ref{fig:introfate}(b). Three difficulties stand in the way. The elastic mean is only defined in the functional data analysis literature as a sample-level statistic, whereas a causal estimand must be a functional of the data-generating law, and that lift must be proved well posed in spaces that are non-linear and infinite-dimensional. Secondly, contrary to the pointwise mean, the elastic mean is not separable: writing $\mu(.)$ for the elastic mean, $\mu(Y^{(1)})-\mu(Y^{(0)})$ and $\mu(Y^{(1)}-Y^{(0)})$ are two different quantities, which raises the question of which one to target. Lastly, identification is usually immediate through inverse-probability weighting or outcome regression. The three-step definition of the elastic mean makes it much harder to prove here, and estimation carries part of that difficulty.

\paragraph{Contributions.}
This paper resolves these three difficulties and builds a causal framework
around the elastic mean, from the definition of the target to its estimation
on real data. Our contributions are the following.

\begin{itemize}[leftmargin=1.6em,itemsep=3pt] 

\item \textbf{We lift the elastic mean to the distribution level.} We define the
phase-centered elastic mean of a probability measure on curves (Section~\ref{sec:lift}), and we prove
that it exists and is unique under four geometric assumptions (Section~\ref{sec:wellposed}). We analyse the plausibility of these assumptions in the context of functional data. Notably, the existing Fréchet-mean results do not directly apply to our setting, as the space in which the timings are averaged is infinite-dimensional and positively curved (Appendix~\ref{app:frechetgeo}).

\item \textbf{We define and compare three distinct shape-preserving causal estimands.} Because the
elastic mean is not separable, one target is not enough. We introduce a
populational effect, an individual effect and a local effect, we prove that
all three are well posed, and explain to which aim each one should be used (Section~\ref{sec:three}). 

\item \textbf{We prove the identifiability of our populational and local estimands.} The
populational effect is identified by placing the inverse-probability weights
inside the alignment problems themselves (Section~\ref{sec:identpop}), and the local effect through the conditional outcome laws of the two arms given the covariates, which matching treated and control units on their covariates estimates (Section~\ref{sec:identloc}).

\item \textbf{We build estimators and provide an empirical evaluation on synthetic data.} We propose an
inverse-probability-weighted estimator for the populational
effect (Section~\ref{sec:popest}) and a nearest-neighbour-matched estimator for the local one (Section~\ref{sec:locest}). On
drop-and-rebound simulations (Section~\ref{sec:sims}), we compare them with the doubly robust FATE
estimator of \citet{testa_doubly_2025} and with the registered pipeline of
\citet{raykov_kernel_2025}. Sensitivity analyses vary the sample size, the propensity model and the covariate dimension (Appendix~\ref{app:sens}). 

\item \textbf{We apply the framework to real data.} We estimate the effect of
initiating a GLP-1 receptor agonist on body-weight trajectories, and we
compare the shapes our two estimators recover with those of competitors (Section~\ref{sec:realdata}). 

\end{itemize}

\section{Background: averaging in the shape space} \label{sec:background}

This section recalls the main elements of functional Elastic Shape Analysis (ESA) from \cite{srivastava_klassen_2016,srivastava_registration_2011} required to construct the elastic mean used in our causal targets. The notions introduced are necessary to understand the three-step definitions of our causal estimands. Crucially, the existing literature defines these quantities only as sample-level statistics---we lift these empirical quantities to functionals of data-generating laws that can be used in causal estimands in Section~\ref{sec:estimand}. All the proofs are detailed in the appendices.

\subsection{Functional data: amplitude and phase} \label{sec:fda}

The objective of the elastic mean is to decouple the amplitude (shape) and phase content of the functions $f_i(t)$. The aim is to provide a single representative function that captures the "average shape" of the functions $\{f_i\}_{i=1}^n$. This average shape sits at a central temporal position compared to the functions $f_i$. We therefore start by recalling the notions of phase and amplitude of functions from ESA. This fully non-parametric framework is defined for absolutely continuous functions in
\[
  \F=\{f:[0,T]\to\R\ \text{absolutely continuous}\}\ \ (\text{so that }\dot f\in L^1([0,T])).
\]

\begin{defbox}{Phase content and warping of functions}
    The misalignment of structural features (peaks, valleys, etc.) across a group of functions $\{f_i\}_{i=1}^n \in \mathcal{F}^n$ is modeled as a \emph{reparameterization} of the time domain. Specifically, each function $f_i(t)$ is viewed as being transformed by a warping function $\gamma_i \in \Gamma$, resulting in an observed "warped" version $f_i \circ \gamma_i(t)$.\\
    
    The warping group $\Gamma$ is formally defined as the set of all absolutely continuous, strictly increasing bijections of the time domain $[0, T]$:
    \begin{equation}
        \Gamma = \{ \gamma: [0,T] \to [0,T] \mid \gamma \text{ absolutely continuous},\ \gamma(0)=0, \gamma(T)=T, \dot{\gamma} > 0\text{ almost everywhere} \}.
    \end{equation}
    In this notation, the element $\gamma(t)$ acts as a non-linear "clock" that accelerates or decelerates the progression of time, shifting the original feature $f_i(t)$ from its canonical position $t$ to a new position $\gamma_i^{-1}(t)$ in the observed domain.
\end{defbox}

\begin{defbox}{Amplitude content and shape of functions}
    While phase variation is modeled by the warping group $\Gamma$, the \emph{amplitude} (or shape content) refers to the intrinsic height and structural features of the function that are invariant to such warpings (waves, bell-shapes, straight lines, step shapes, etc.). In the ESA framework, two functions $f_1, f_2 \in \mathcal{F}$ are said to have the same amplitude if there exists  $\gamma \in \Gamma$ such that: $$f_1(t) = f_2 \circ \gamma(t).$$
Section~\ref{sec:shape} slightly enlarges this definition, so that the elastic distance becomes a true distance.
\end{defbox} 

Figure~\ref{fig:warpclock} illustrates the two ingredients of this decomposition: phase is the timing imposed by a warp, amplitude is the shape that every warp leaves unchanged.

\begin{figure}[htbp]
\centering
\begin{tikzpicture}[>=Stealth, scale=1.0]
\begin{scope}
  \node[font=\footnotesize, black!70] at (1.85,4.35) {\textbf{(a) The warp $\gamma$: a nonlinear clock}};
  \draw[->] (0,0) -- (3.6,0) node[right, font=\small] {$t$};
  \draw[->] (0,0) -- (0,3.6) node[above, font=\small] {$\gamma(t)$};
  \draw[black!40, densely dotted] (3,0) -- (3,3) -- (0,3);
  \node[below, font=\small] at (3,0) {$T$};
  \node[left, font=\small] at (0,3) {$T$};
  \draw[dash dot, black!55] (0,0) -- (3,3);
  \node[font=\small, black!55, rotate=45] at (2.42,2.14) {$\gid$};
  \draw[thick, figred] plot[domain=0:3, samples=100] (\x,{\x+0.55*sin(60*\x)});
  \node[font=\small, figred] at (0.55,1.35) {$\gamma$};
  \draw[->, figred!70, densely dotted] (2,0.05) -- (2,2.43);
  \draw[->, figred!70, densely dotted] (1.95,2.476) -- (0.07,2.476);
  \node[below, font=\footnotesize, figred] at (2,0) {$t_0$};
  \node[left, font=\footnotesize, figred] at (0,2.476) {$\gamma(t_0)$};
  \node[font=\footnotesize, black!70] at (2.15,3.36) {$\gamma$ above $\gid$: the clock runs ahead};
\end{scope}
\begin{scope}[xshift=5.6cm]
  \node[font=\footnotesize, black!70] at (1.85,4.35) {\textbf{(b) Phase variation}};
  \node[font=\footnotesize, black!70] at (1.85,3.85) {\textit{Same amplitude, re-timed by $\gamma$}};
  \draw[->] (-0.15,0) -- (3.6,0) node[right, font=\small] {$t$};
  \draw[thick, figblue] plot[domain=0:3, samples=120] (\x,{2.6*exp(-pow(\x-1.8,2)/0.18)});
  \draw[thick, dashed, figred] plot[domain=0:3, samples=160] (\x,{2.6*exp(-pow(\x+0.55*sin(60*\x)-1.8,2)/0.18)});
  \node[font=\small, anchor=west, figblue] at (2.25,1.6) {$f$};
  \node[font=\small, anchor=east, figred] at (0.82,1.6) {$f\circ\gamma$};
  \draw[densely dotted, figblue!70] (1.8,0) -- (1.8,2.6);
  \draw[densely dotted, figred!70] (1.27,0) -- (1.27,2.6);
  \node[below, font=\footnotesize, figblue] at (1.94,0) {$t^\star$};
  \node[below, font=\footnotesize, figred] at (1.05,0) {$\gamma^{-1}(t^\star)$};
  \draw[->, black!70] (1.74,2.85) -- (1.34,2.85);
  \node[above, font=\footnotesize, black!70] at (1.54,2.88) {same height, earlier};
\end{scope}
\begin{scope}[xshift=11.2cm]
  \node[font=\footnotesize, black!70] at (1.85,4.35) {\textbf{(c) Amplitude variation}};
  \node[font=\footnotesize, black!70] at (1.85,3.85) {\textit{Different amplitude, same timing}};
  \draw[->] (-0.15,0) -- (3.6,0) node[right, font=\small] {$t$};
  \draw[thick, figblue] plot[domain=0:3, samples=120] (\x,{2.6*exp(-pow(\x-1.8,2)/0.18)});
  \draw[thick, figgreen] plot[domain=0:3, samples=160] (\x,{1.6*exp(-pow(\x-1.8,2)/0.18)+1.0*exp(-pow(\x-2.45,2)/0.09)});
  \node[font=\small, anchor=east, figblue] at (1.62,2.75) {$f$};
  \node[font=\small, anchor=west, figgreen] at (2.92,0.72) {$g$};
  \draw[densely dotted, black!55] (1.8,0) -- (1.8,2.6);
  \node[below, font=\footnotesize, black!70] at (1.8,0) {$t^\star$};
  \draw[<->, black!70] (1.8,1.68) -- (1.8,2.52);
  \node[font=\footnotesize, black!70, anchor=east] at (1.66,2.1) {lower peak};
  \draw[->, black!70] (2.9,2.05) -- (2.58,1.32);
  \node[font=\footnotesize, black!70, anchor=south] at (2.9,2.08) {extra bump};
\end{scope}
\end{tikzpicture}
\caption{\textbf{Phase is when, amplitude is what.} \emph{(a)} A warp $\gamma\in\G$ (red), an increasing bijection of $[0,T]$ fixing the endpoints, and the identity warp $\gid$ (dash-dotted diagonal): where $\gamma$ lies above the diagonal, the clock runs ahead of $t$. \emph{(b)} Phase variation. The curve $f$ (blue) and its reparametrisation $f\circ\gamma$ (red, dashed) have the same amplitude: the peak of height $f(t^\star)$ is visited at $\gamma^{-1}(t^\star)$ instead of $t^\star$, and nothing else changes. \emph{(c)} Amplitude variation. The curves $f$ and $g$ (green) peak at the same time $t^\star$ but differ in shape (a lower peak and an extra bump): no warp turns one into the other. The elastic mean of Section~\ref{sec:karcher} averages the amplitudes of \emph{(c)} and the timings of \emph{(b)} separately.}
\label{fig:warpclock}
\end{figure}
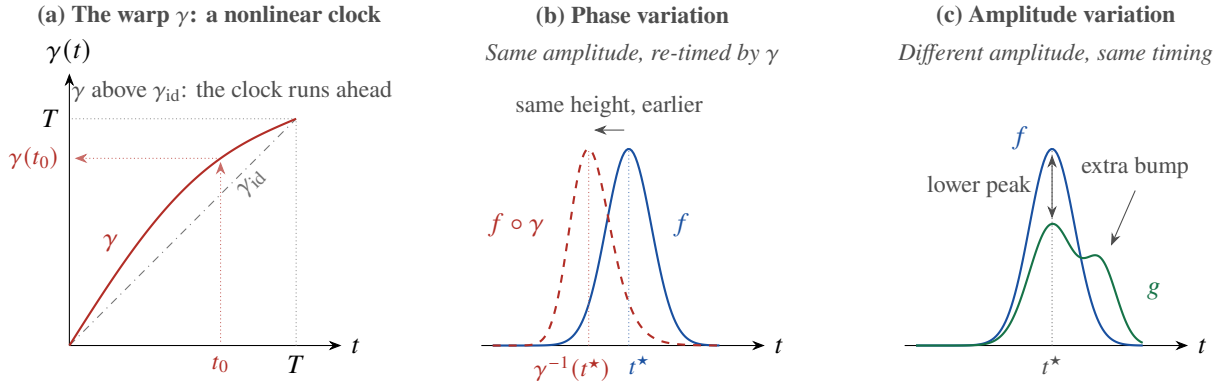

\subsection{The SRSF and the Fisher--Rao isometry} \label{sec:srsf}

To compare functions in $\mathcal{F}$ in a manner that is insensitive to temporal misalignment, we need a notion of distance that is invariant to reparameterization, \textit{i.e.} $d(f_1\circ\gamma,f_2\circ\gamma)=d(f_1,f_2)$ for every $\gamma\in\G$. The canonical reparametrisation-invariant choice in the non-linear manifold $\mathcal{F}$ is the Fisher--Rao metric which induces the distance $\dFR$ \citep{srivastava_klassen_2016}. 

Computing the distance between functions under the Fisher--Rao metric is traditionally very challenging. Because the metric is defined on the tangent space of a non-linear manifold, calculating the geodesic distance $\dFR(f_1, f_2)$ requires solving complex differential equations. The core innovation of the ESA framework is to simplify this geometry through the Square-Root Slope Function (SRSF) transform, defined below.

\begin{defbox}{Square-root slope function (SRSF)}
The square-root slope function of any $f\in\F$ is the transform $q_f$ given by
\[
  q_f(t)=\operatorname{sign}\!\big(\dot f(t)\big)\sqrt{\big|\dot f(t)\big|}.
\]
Based on $q_f$ and the initial height $f(0)$, one can reconstruct $f$ via $f(t)=f(0)+\int_0^t q_f(s)\,|q_f(s)|\,ds$.
\end{defbox}

The transform $f\mapsto q_f$, from $\F$ to $\Lt$, acts as a bijection between the space of functions $\mathcal{F}$ (up to a translation constant $f(0)$) and the standard Hilbert space $\Lt$. It is also continuous, hence measurable. The fundamental advantage of the SRSF is to map the Fisher--Rao distance in $\mathcal{F}$ directly to the standard $\Lt$ distance between the SRSF transforms of the functions.

\begin{proposition}[The Fisher--Rao isometry, {\citealp{srivastava_klassen_2016}}]\label{prop:friso}
For all $f_1,f_2\in\F$,\quad $\dFR(f_1,f_2)=\norm{q_{f_1}-q_{f_2}}_2$.
\end{proposition}

The remaining analytic properties of the SRSF (image, invertibility, continuity) are collected in Proposition~\ref{prop:srsf} and proved in Appendix~\ref{app:lemmas}; from here on we work entirely with SRSFs and the flat $\Lt$ geometry.

Re-timing acts equally cleanly on transforms. A chain-rule computation shows that the SRSF of $f\circ\gamma$ is
\[
  q_{f\circ\gamma}=(q_f\circ\gamma)\sqrt{\dot\gamma}. 
\]
In the following, we use the shorthand $(q,\gamma):= q_{f \circ \gamma} = (q\circ\gamma)\sqrt{\dot\gamma}$ for the right action of a warp $\gamma$ on a transform $q\in\Lt$. The action is linear in $q$ and preserves the $\Lt$ inner product, $\inner{(q_1,\gamma)}{(q_2,\gamma)}=\inner{q_1}{q_2}$, hence every norm, $\norm{(q,\gamma)}_2=\norm{q}_2$ (Lemma~\ref{lem:iso}, Appendix~\ref{app:lemmas}). Reparametrising thus slides a transform along the time axis while leaving its $\Lt$ norm unchanged: warping re-times a curve but cannot alter it otherwise. This norm-preservation is necessary for the existence and well-definedness results of Section~\ref{sec:estimand}. 

\subsection{The shape space and the elastic distance} \label{sec:shape}

We can now define a \emph{shape space}, where every element corresponds to a specific shape (amplitude), irrespective of the time warping used to represent it. One would naturally consider the quotient space $\Lt/\G$, whose \emph{orbits} $\{(q,\gamma):\gamma\in\G\}$ collect all reparameterizations of the same underlying functional morphology. However, as $\G$ is not closed, neither are the orbits $\{(q,\gamma):\gamma\in\G\}$, and the elastic distance defined below is only a pseudo-metric on $\Lt/\G$ (two different orbits can be at zero elastic distance), rendering the resulting quotient space incomplete \citep[App.~B.2]{raykov_kernel_2025}. We provide a concrete example of this failure in Example~\ref{ex:line} (Appendix~\ref{app:assump}). To circumvent this problem, we follow 
\citet{bruveris_optimal_2016}. His construction treats as one shape all the curves that can be obtained by re-timing one common curve, where the re-timing may pause, together with the limits of such re-timings. Formally, besides the group $\G$ of invertible re-timings, we use the larger set $\Gt$ (his $\overline\Gamma$) of non-decreasing, absolutely continuous maps of $[0,T]$ onto itself, whose elements we also call warps or re-timings. Every element of $\Gt$ is a limit of elements of $\G$ \citep[Sect.~4]{bruveris_optimal_2016}. A re-timing in $\Gt$ may \emph{pause}, that is, stay constant on a time interval. A re-timing that pauses cannot be undone, so $\Gt$, unlike $\G$, is not a group (not all its elements can be inverted). Two transforms $q_1,q_2\in\Lt$ have the same \emph{shape}, written $q_1\sim q_2$, when both are re-timings, possibly with pauses, of one common transform: $q_1=(q,\beta)$ and $q_2=(q,\gamma)$ for some $q\in\Lt$ and $\beta,\gamma\in\Gt$. Equivalently, the closures in $\Lt$ of their $\G$-orbits coincide \citep[Prop.~12]{bruveris_optimal_2016}. The shape of $q$ is therefore the closure of its $\G$-orbit,
\begin{equation}
  [q]=\overline{\{(q,\gamma):\gamma\in\G\}},
\end{equation}
and every element of $[q]$ is called a \emph{representative} of this shape. The \emph{shape space} is the set of shapes, $\mathcal S=\Lt/\!\sim$. We say that a curve $f$ is \emph{nowhere flat} when $\dot f\neq0$ almost everywhere, and that it has a \emph{flat piece} when it is constant on some time interval.

\begin{remark}[Why $\mathcal S=\Lt/\!\sim$ and not $\Lt/\Gt$?]\label{rem:notquot}
A re-timing $\gamma \in \Gt$ can create, lengthen or shorten a flat piece, but never remove one. Hence:
\begin{itemize}[leftmargin=1.6em,itemsep=1pt,topsep=2pt]
\item if $f$ is nowhere flat, its shape is exactly the set of its re-timings, $[q_f]=\{(q_f,\gamma):\gamma\in\Gt\}$ \citep[Cor.~9]{bruveris_optimal_2016};
\item if $f$ has a flat piece and is not constant, its shape is larger: it also contains the curve obtained by cutting out the flat piece and stretching the rest back to $[0,T]$, which is a limit of re-timings of $f$ but not a re-timing of $f$. In Example~\ref{ex:line}, cutting out the flat start of $f_1$ gives the line $f_2$.
\end{itemize}
So the sets $\{(q,\gamma):\gamma\in\Gt\}$ are not the shapes in general. The notation $\Lt/\Gt$, common in elastic shape analysis, is exact only for nowhere-flat curves, and \citet[Sect.~5.2]{bruveris_optimal_2016} calls it informal. Appendix~\ref{app:flat} explains what flat pieces change in our construction.
\end{remark}

To compare different shapes, we define the \emph{elastic distance} between two shapes $[q_1]$ and $[q_2]$ in $\mathcal{S}$ as the optimal alignment between any two representatives $q_1$ and $q_2$ in $\Lt([0,T], \mathbb{R})$:  
\begin{equation}
    d_{elastic}([q_1], [q_2]) = \inf_{\gamma \in \Gt} \norm{q_1 - (q_2, \gamma)}_{\Lt}. 
\end{equation}
The infimum is the same over $\G$ and over $\Gt$, and it does not depend on which representatives $q_1$ and $q_2$ of the two shapes are used \citep[Sect.~5.3]{bruveris_optimal_2016}. With this construction,  $d_{elastic}$ is a genuine distance and $(\mathcal{S}, d_{elastic})$ is a complete metric space \citep[Lem.~13]{bruveris_optimal_2016}. The infimum, however, is not always attained. For $C^1$ curves that are nowhere flat, \citet[Prop.~15]{bruveris_optimal_2016} proves that an optimal \emph{pair} of re-timings exists, one for each curve: $d_{elastic}([q_1],[q_2])=\norm{(q_1,\gamma_1)-(q_2,\gamma_2)}$ for some $\gamma_1,\gamma_2\in\Gt$. His result lets both curves be re-timed. Our construction, however, always re-times a data curve towards a fixed curve, which is itself never re-timed; we call this \emph{one-sided registration}. A best one-sided re-timing can fail to exist, for curves with a flat piece (Remark~\ref{rem:onesided}) and even for smooth, nowhere-flat curves (Appendix~\ref{app:flat}), so we assume it where we use it (Assumption~\ref{ass:warp}(iii)).

  \begin{defbox}{Registration (optimal warp) of functions}
        The \emph{registration} between two elements $q_1$ and $q_2$ in $\Lt([0,T], \mathbb{R})$ is the solution, when it exists, $\gamma^\star_{q_2\to q_1}$ of the optimization problem $\inf_{\gamma \in \Gt} \norm{q_1 - (q_2, \gamma)}_{\Lt}$: the \emph{optimal warp} $\gamma^\star_{q_2\to q_1}$ aligns the features of $f_2$ (represented by its SRSF transform $q_2$) to those of $f_1$ (represented by its SRSF transform $q_1$). The subscript always reads from the transform being warped to its target, and both are written explicitly.
\end{defbox}
 
Figure~\ref{fig:orbitdist} (Appendix~\ref{app:lemmas}) pictures the geometry behind these definitions: shapes are closures of orbits, and the elastic distance is the shortest $\Lt$-distance from a representative to the other shape, attained at the registration when it exists.

\subsection{The sample phase-centered elastic mean} \label{sec:karcher}

The previous subsections provide the two ingredients needed to capture the typical shape of a set of curves: a complete metric space $(\mathcal{S}, d_{elastic})$ whose points are shapes, and a procedure to map individual functions to this shape space. However, a mean of shapes cannot be computed as an ordinary average in this space since the complete metric space $(\mathcal{S}, d_{elastic})$ is not flat. But one can still compute a Fréchet mean by finding the point minimising the sum of squared distances to the sample, thus solving an optimization problem. We summarise below the three-step construction of a phase-centered elastic mean of curves $f_1,\dots,f_n$ with SRSF transforms $q_1,\dots,q_n$ from \citet{srivastava_registration_2011,srivastava_klassen_2016}.

The mean shape (Step 1 below) settles what the typical curve looks like but not when its features occur:
it is a shape, a whole family of representatives differing only in their timing. To pick the specific function in that shape whose timing is itself typical (Step 2) we must average not the
curves but the warps that relate them. Warps resist ordinary averaging---they
compose rather than add, and a convex combination of warps may not be a warp---but
the same square-root transform as we used for curves can be used for warps. Each warp $\gamma$ is
represented by its SRSF $\psi_\gamma=\sqrt{\dot\gamma}$, and simple calculations show that the transforms $\psi_\gamma$ all live on a single Hilbert sphere $\Sphere(\sqrt T)$ of radius $\sqrt T$, and that a natural distance between warps is given by the
great-circle arc $\dG(\gamma_1,\gamma_2)=\sqrt T\,\arccos\!\big(T^{-1}\inner{\psi_{\gamma_1}}{\psi_{\gamma_2}}\big)$ (see Appendix~\ref{app:sphere} for details). The SRSF of the optimal warp $\gamma^\star_{q_2\to q_1}$ is written $\psi^\star_{q_2\to q_1}=\psi_{\gamma^\star_{q_2\to q_1}}$. The typical timing of a set of warps can thus be defined as the Fréchet mean of their SRSF representations on the sphere.

\begin{resultbox}{The sample phase-centered elastic mean $\hat\mu_n$}
Given curves $f_1,\dots,f_n$ with SRSF transforms $q_1,\dots,q_n$, the sample phase-centered elastic mean $\hat\mu_n\in\F$ is defined in three steps.

\smallskip
\emph{Step 1 --- mean shape.} The Fréchet mean shape in $(\Sw,\del)$,
\[
  \hat\shp_n\ =\ \argmin_{[q]\in\Sw}\ \frac1n\sum_{i=1}^n\del^2\big([q],[q_i]\big).
\]
 
\smallskip 
\emph{Step 2 --- mean warp.} Fix a nowhere-flat representative $\tilde q\in\hat\shp_n$ and let $\gamma^\star_{q_i\to\tilde q}=\argmin_{\gamma\in\Gt}\norm{\tilde q-(q_i,\gamma)}$ be the warp registering $q_i$ to this fixed $\tilde q$. Take the Fréchet mean of these warps by the Fisher--Rao metric (which is computed using their SRSF representations on the sphere $\Sphere(\sqrt T)$), 
\[
  \bar\gamma_n\ =\ \argmin_{\gamma\in\Gt}\ \frac1n\sum_{i=1}^n\dG^2\big(\gamma,\gamma^\star_{q_i\to\tilde q}\big).
\]
\smallskip
\emph{Step 3 --- phase-centering and averaging.} Re-time every registered curve by the inverse mean warp, and average the resulting curves: the phase-centered elastic mean $\hat\mu_n\in\F$ is
$$
  \forall t\in[0,T],\quad \hat\mu_n(t)\ =\ \frac1n\sum_{i=1}^n \big(f_i\circ\gamma^\star_{q_i\to\tilde q}\circ\bar\gamma_n^{-1}\big)(t).
$$
\end{resultbox}

Figure~\ref{fig:karchersteps} shows the three steps on a small sample. Step~1 returns a central shape (an equivalence class of SRSFs), defined without any specific timing. At the beginning of Step~2, the curves are registered to one representative $\tilde q$ of that central shape, producing aligned versions of the curves and optimal warps. These warps are then averaged on the warp sphere. In Step~3, re-timing the registered curves by the inverse of this mean warp moves them to the sample's typical timing, and their average is $\hat\mu_n$.

\begin{remark}[Vocabulary.]\label{rem:vocab}
The output of Step~1 is the \emph{Fréchet mean shape} $\hat\shp_n$, $\tilde q$ is the \emph{reference representative}, $\bar\gamma_n$ is the \emph{mean warp}, the curves $f_i\circ\gamma^\star_{q_i\to\tilde q}\circ\bar\gamma_n^{-1}$ are the \emph{phase-centered registered curves}, and their average $\hat\mu_n$ is the \emph{phase-centered elastic mean}, or \emph{elastic mean} for short. It is the ``aligned mean'' of \citet{tucker_generative_2013}, the signal estimator of \citet{srivastava_registration_2011} and \citet{kurtek_signal_2011}, and an elastic version of the \emph{structural mean} of \citet{kneip_statistical_1992} and \citet{ramsay_curve_1998}. The word \emph{template} is kept for the template model (Proposition~\ref{prop:propsB}(c)) and the template bias (Section~\ref{sec:noise}).
\end{remark}

\begin{remark}[The elastic mean versus the Karcher mean.]\label{rem:karcher}
In Step~3, $\hat\mu_n$ averages curves in $\F$: each phase-centered registered SRSF is mapped back to its associated curve, which keeps its own starting value, and these curves are averaged in $\F$. Another option is to average the phase-centered registered SRSFs themselves and to map this average back to $\F$, starting it at $\frac1n\sum_i f_i(0)$. This gives the \emph{Karcher mean function} $\hat\mu^{\mathrm K}_n$, whose SRSF is $(\tilde q,\bar\gamma_n^{-1})$. It is the ``Karcher mean'' of \citet{tucker_generative_2013} and of the \texttt{fdasrsf} \citep{tucker_fdasrsf_2020} and \texttt{scikit-fda} \citep{ramos_carreno_scikitfda_2024} packages. The two curves coincide when the phase-centered registered curves differ only by vertical shifts, but not in general (Appendix~\ref{app:twomeans}). Following \citet{srivastava_registration_2011} and \citet{kurtek_signal_2011}, we average the curves rather than their SRSFs: the resulting mean keeps the average amplitude of the curves (Proposition~\ref{prop:propsB}(c)) and reduces to the pointwise mean when every registration is the identity.
\end{remark}

\begin{figure}[htbp]
\centering
\includegraphics[width=\textwidth]{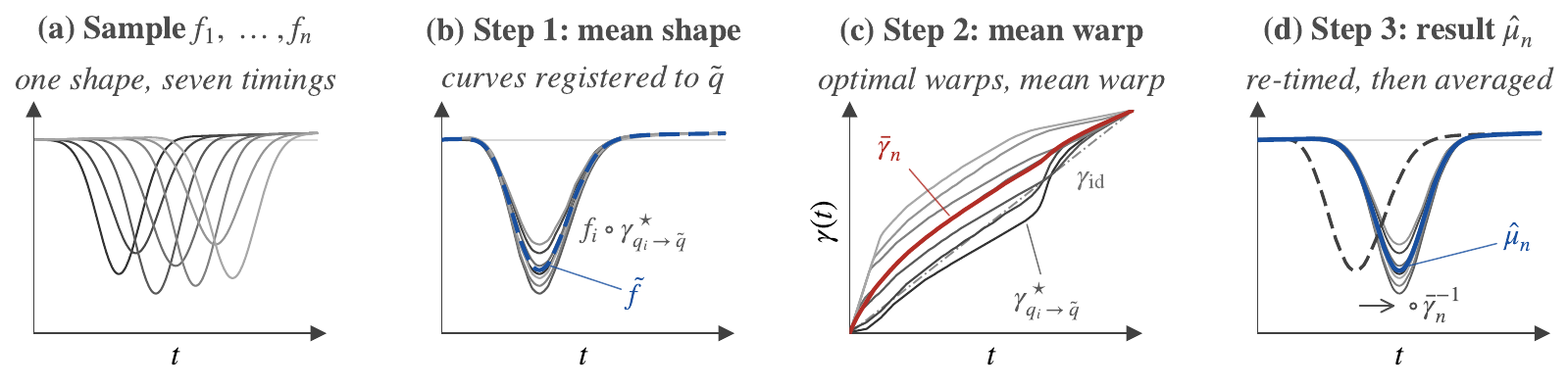}
\caption{\textbf{The phase-centered elastic mean registers the curves to a representative of their mean shape, re-times them by the inverse mean warp, and averages them.} Every panel shows functions of time (curves or warps); the SRSFs are used only in the computations. \emph{(a)} Seven curves of one shape at different timings and depths. \emph{(b)} Step 1: the registered curves $f_i\circ\gamma^\star_{q_i\to\tilde q}$, whose SRSFs $(q_i,\gamma^\star_{q_i\to\tilde q})$ are aligned with the reference representative $\tilde q$ of the mean shape, and the curve $\tilde f$ whose SRSF is $\tilde q$, started at the average initial value (blue, dashed). \emph{(c)} Step 2: the optimal warps $\gamma^\star_{q_i\to\tilde q}$ (grey), their Fréchet mean $\bar\gamma_n$ (red), computed from the SRSFs of the warps on the sphere $\Sphere(\sqrt T)$, and the identity $\gid$ (dash-dotted). \emph{(d)} Step 3: the registered curves re-timed by $\bar\gamma_n^{-1}$, $f_i\circ\gamma^\star_{q_i\to\tilde q}\circ\bar\gamma_n^{-1}$ (grey), and their average $\hat\mu_n$ (blue, solid), at the sample's central timing.}
\label{fig:karchersteps}
\end{figure}

A visual comparison of this elastic mean with the cross-sectional mean and other registration-based averages is deferred to Appendix~\ref{app:meansgrid} (Figure~\ref{fig:meansgrid}). 

\section{Shape-preserving elastic causal estimands} \label{sec:estimand}

We introduce in this section our three shape-preserving causal estimands, leveraging the elastic mean definition of Section~\ref{sec:background}. However, as the latter was only defined for a finite set of curves, we first need to extend it to a functional of the data-generating law and prove it is well posed under specific geometric conditions. This lift is one of our main contributions, and ensures our causal estimands are well-defined, exist and are unique under the same conditions.

We first fix the causal setting and notation (Section~\ref{sec:setting}), then rewrite the phase-centered elastic mean as a functional $P\mapsto\mu(P)$ of a Borel probability measure $P$ on curves and show it is well posed (Section~\ref{sec:lift}), and finally construct three distinct causal estimands using this lifted definition (Section~\ref{sec:three}).

\subsection{Causal setting and notation} \label{sec:setting}

We adopt the Neyman--Rubin potential-outcomes framework with a binary
treatment and a functional outcome. Each unit carries the variables
$(X,A,Y^{(0)},Y^{(1)})\sim\Pstar$ taking values in
$\X\times\{0,1\}\times\F\times\F$: covariates $X$ in a measurable  space $\X$, a treatment indicator $A$,
and two potential-outcome curves, the outcome trajectory $Y^{(1)}$ the unit
would exhibit under treatment and $Y^{(0)}$ without it. Potential outcomes
live in the space $\F\subset\Lt[0,T]$ of Section~\ref{sec:fda}, the absolutely continuous curves on $[0,T]$ (Assumption~\ref{ass:reg} below). We write $q^{(a)}=q_{Y^{(a)}}$ for their SRSFs. The joint law
$\Pstar$ is the \emph{full-data} law: every causal estimand
in this paper is a functional of $\Pstar$.

By the fundamental problem of causal inference, the two potential outcomes of
a unit are never observed together. The data are $n$ i.i.d.\ draws of $O=(X,A,Y)$, whose law $P$ is the \emph{observed-data} law induced by
$\Pstar$ through the consistency relation,
\[
  Y \;=\; Y^{(A)} \;=\; A\,Y^{(1)}+(1-A)\,Y^{(0)} .
\]
When the dependence on the full-data law matters we write $P(\Pstar)$ for this 
induced observed law.
We write $e_a(x)=\Prob(A=a\mid X=x)$ for the propensity of arm $a\in\{0,1\}$
and $q_Y$ for the SRSF of $Y$. An estimand defined on $\Pstar$ is \emph{identified} when
it coincides with a functional of $P$ on every full-data law compatible with
the assumptions made. 

\subsection{From a sample statistic to a functional of the law} \label{sec:lift}

In this subsection, $P$ denotes a generic Borel probability law on curves (not necessarily the observed-data law of Section~\ref{sec:setting}), and $Z\sim P$ a random outcome curve with SRSF $q_Z$. Lifting the sample definition of Section~\ref{sec:background} amounts to replacing the empirical average over $\{f_i\}_{i=1}^n$ by an expectation under $P$. 

\begin{resultbox}{The phase-centered elastic mean $\mu(P)$ of a probability law $P$}
The phase-centered elastic mean $\mu(P)\in\F$ is defined in three steps.
 
\smallskip
\emph{Step 1 --- mean shape.} The Fréchet mean shape in $(\Sw,\del)$,
\[
  \shp(P)\ =\ \argmin_{[q]\in\Sw}\ \E_{P}\!\big[\del^2\big([q],[q_Z]\big)\big].
\]

\smallskip
\emph{Step 2 --- mean warp.} Fix an admissible representative $\tilde q\in\shp(P)$ (any nowhere-flat representative, $\tilde q\neq0$ almost everywhere, at which Assumptions~\ref{ass:warp}(iii-b) and~\ref{ass:conc}(4b) below hold; these assumptions say that such a representative exists, and Theorem~\ref{thm:wellposed}(iv) shows that the result does not depend on which one is used) and let $\gamma^\star_{q_Z\to\tilde q}$ be the warp registering $q_Z$ to this fixed $\tilde q$. Take the Fréchet mean of these warps on the warp sphere $\Sphere(\sqrt T)$,
\[
  \bar\gamma\ =\ \argmin_{\gamma\in\Gt}\ \E_{P}\!\big[\dG^2\big(\gamma,\,\gamma^\star_{q_Z\to\tilde q}\big)\big].
\]
\smallskip
\emph{Step 3 --- phase-centering and averaging.} The phase-centered elastic mean $\mu(P)\in\F$ is the expectation of the registered curves re-timed by the inverse mean warp,
$$
  \forall t\in[0,T],\quad \mu(P)(t)\ =\ \E_{P}\big[\big(Z\circ\gamma^\star_{q_Z\to\tilde q}\circ\bar\gamma^{-1}\big)(t)\big].
$$
\end{resultbox}

\noindent We write $\eta_Z=\gamma^\star_{q_Z\to\tilde q}\circ\bar\gamma^{-1}$ and call $Z\circ\eta_Z$ the \emph{phase-centered registered curve}. By Theorem~\ref{thm:wellposed}, $\mu(P)$ does not depend on the admissible representative $\tilde q$ (point (iv)) and is a curve of $\F$ (point (v)); its SRSF is in general not in $\shp(P)$ (Proposition~\ref{prop:propsB}(c)). Three properties make it a natural average of curves.

\begin{proposition}[Three properties of the phase-centered elastic mean]\label{prop:propsB}
Let Assumptions~\ref{ass:reg}--\ref{ass:conc} hold for the law $P$.
\begin{enumerate}[leftmargin=2em,itemsep=2pt,topsep=2pt,parsep=0pt]
\item[(a)] \emph{End values.} $\mu(P)(0)=\E_P[Z(0)]$ and $\mu(P)(T)=\E_P[Z(T)]$.
\item[(b)] \emph{No phase variation.} If $\gamma^\star_{q_Z\to\tilde q}=\gid$ for $P$-almost every $Z$, then $\mu(P)=\E_P[Z]$, the pointwise mean.
\item[(c)] \emph{Template model.} If $Z=c_Z+s_Z\,(g\circ\gamma_Z)$ with $g\in\F$ nowhere flat, random $c_Z\in\R$, $s_Z>0$ and $\gamma_Z\in\G$, with $\E|c_Z|+\E s_Z<\infty$, then $\mu(P)=\E[c_Z]+\E[s_Z]\,(g\circ\bar\gamma^{-1})$, where $\bar\gamma$ is the Step-2 mean warp computed with the reference $\tilde q=\E[\sqrt{s_Z}]\,q_g$: the template shape, at the mean timing, with the average scale and the average vertical shift. 
\end{enumerate}
\end{proposition}

\noindent The proof is in Appendix~\ref{app:propsB}. Point (c) is the population version of Theorems~1--3 of \citet{srivastava_registration_2011}.

Thus, alignment is part of the quantity being defined, not a preprocessing step performed in advance.  To the best of our knowledge, our definition is the first definition of the phase-centered elastic mean as a functional of a general law. Population versions of template estimation in quotient spaces already exist \citep{bigot_consistency_2011,devilliers_template_2017,miolane_template_2017}, but they define the mean shape only, as our Step~1 does, and not the typical timing of Steps~2 and~3. Population versions of the phase-centering step exist in other frameworks, where the warps are normalised, for instance to have the identity as their mean \citep{kneip_statistical_1992,dupuy_structural_2011}, but not for the elastic mean. The price of our definition is computational: every evaluation of the Step-1 objective requires solving a registration problem per curve, and the Step-1 mean is obtained by alternating registration and averaging until convergence, rather than by a single pass of alignment  followed by an ordinary average. 

\subsection{Geometric assumptions and well-posedness} \label{sec:wellposed}

The following assumptions allow us to write the phase-centered elastic mean $\mu(P)$ as a well-defined functional of a generic law $P$.

\begin{assumption}[Absolutely continuous curves]\label{ass:reg}
The outcome curves are absolutely continuous: $P$ is a Borel probability law on the space $\F$ of absolutely continuous curves of Section~\ref{sec:fda}, so that the transform $q_Z$ of $Z\sim P$ is defined.
\end{assumption} 

\begin{assumption}[Regular, compact warping family]\label{ass:warp}
Registrations and phase-warpings range over a family $\Gc\subseteq\Gt$ that is
\begin{enumerate}[leftmargin=2em,itemsep=3pt,topsep=3pt,parsep=0pt]
\item[(i)] compact: the set of warp-SRSFs $\Psi=\{\sqrt{\dot\gamma}:\gamma\in\mathcal G\}$ is compact in $(L^\infty([0,T]),\norm{\cdot}_\infty)$,
\item[(ii)] uniformly bi-Lipschitz, $0<c\le\dot\gamma\le C<\infty$, and
\item[(iii)] attaining: $\gid\in\Gc$, and the registrations used by the three steps are attained in $\Gc$:
\begin{itemize}[leftmargin=1.5em,itemsep=1pt,topsep=1pt]
\item[(iii-a)] \emph{(Step 1)} searching $\Gc$ instead of $\Gt$ in the registrations does not change the infimum value of the Step-1 objective $\inf_{[q]\in\Sw}M([q])$ over candidate shapes. Precisely, there are shapes $[p_1],[p_2],\dots$ with $M([p_k])\to\inf_{[q]\in\Sw}M([q])$, each with a representative $p_k$ such that $\min_{\gamma\in\Gc}\norm{p_k-(q_Z,\gamma)}=\del([p_k],[q_Z])$ for $P$-almost every $Z$ (if one shape attains the minimum, it is enough to take it);
\item[(iii-b)] \emph{(Steps 2 and 3)} there is a nowhere-flat representative $\tilde q\in\shp(P)$ ($\tilde q\neq0$ almost everywhere) at which the alignment minimiser
$\gamma^\star_{q_Z\to \tilde q}\ :=\ \operatorname*{arg\,inf}_{\gamma\in\Gt}\ \norm{\tilde q-(q_Z,\gamma)}$
exists, is unique and lies in $\Gc$, for $P$-almost every $Z$.
\end{itemize}
Consequently, at these representatives, the computable $\min_{\gamma\in\Gc}\norm{q-(q_Z,\gamma)}$
equals the true elastic distance $\del([q],[q_Z])=\inf_{\gamma\in\Gt}\norm{q-(q_Z,\gamma)}$.
\end{enumerate}
\end{assumption} 

A registration re-times a data curve $Z$ towards a fixed curve: in Step~1, a representative $q$ of a candidate mean shape; in Steps~2 and~3, the reference $\tilde q$. Clauses (i) and (ii) concern the method, not the data: they describe the set $\Gc$ of re-timings that the computer searches, whose speeds lie between $c$ and $C$, so that they never pause and can be undone, and in which a best re-timing always exists. Clause (iii) concerns the data: part (iii-a), for Step~1, asks that the computable objective, which searches $\Gc$, has the same infimum over candidate shapes as the true one, which searches $\Gt$. Part (iii-b), for Steps~2 and~3, asks that almost every data curve has a unique best re-timing lying in $\Gc$ towards a nowhere-flat reference. Appendix~\ref{app:flat} explains what clause (iii) excludes, and Appendix~\ref{app:assump} discusses each clause in detail.

\begin{remark}\label{rem:nogroup}
$\Gc$ is a set of warps, not a group: composing or inverting two of its elements may violate the bi-Lipschitz bounds. It is also the family the computation uses, since the dynamic-programming algorithm of \citet{srivastava_klassen_2016} searches piecewise-linear warps of bounded slope on the time grid, $[1/4,4]$ in our implementation. On a fixed grid, these warps are described by finitely many slopes that range over a closed and bounded set, so their warp-SRSFs form a compact set of $L^\infty$, as clause (i) requires. Searching a smaller set can only increase the residual, $\min_{\gamma\in\Gc}\norm{q_1-(q_2,\gamma)}\ge\del([q_1],[q_2])$, and clause~(iii) is exactly what makes the two equal. 
\end{remark}

\begin{assumption}[Finite second moment]\label{ass:mom}
$\E_{P}\norm{q_Z}^2=\E_{P}\int_0^T|\dot Z|<\infty$ and $\E_{P}|Z(0)|<\infty$.
\end{assumption}

\begin{assumption}[Concentration and uniqueness]\label{ass:conc}
One condition for each space:
\begin{enumerate}[leftmargin=2em,itemsep=3pt,topsep=3pt,parsep=0pt]
\item[(4a)] \emph{amplitude (in $\Sw$):} the Step-1 mean shape $\shp(P)$ is unique, i.e., the minimiser of $\E_P[\del^2([q],[q_Z])]$ over $[q]\in\Sw$ is a single shape.
\item[(4b)] \emph{phase (on $\Sphere(\sqrt T)$):} there is a nowhere-flat representative $\tilde q\in\shp(P)$ at which Assumption~\ref{ass:warp}(iii-b) holds and at which the optimal warps $\gamma^\star_{q_Z\to\tilde q}$ concentrate in a geodesic ball of radius $<\pi\sqrt T/8$ whose centre is itself in the feasible set $\{\psi_\gamma:\gamma\in\Gt\}$ (e.g.\ the identity, $\psi_{\gid}\equiv1$), and the resulting Step-2 mean belongs to $\Gc$, hence is bi-Lipschitz and invertible. 
\end{enumerate}
\end{assumption}

\begin{remark}[Clause (iii) versus Assumption 4b.]\label{rem:twoassump}
Clause~(iii) concerns the \emph{inner alignments that compute elastic distances}: it forces the optimal alignment warp into $\Gc$, so the computed $\min_{\Gc}$ equals the true $\del$. Assumption~4b instead concerns the \emph{solution of the Step-2 Fréchet optimization}: it forces the resulting mean warp to lie in $\Gc$, so it can be inverted in Step~3. 
\end{remark}

\begin{remark}[Are the geometric assumptions plausible?]\label{rem:plaus}
Assumptions~\ref{ass:reg} and~\ref{ass:mom} are mild and standard. The absolute continuity required in Assumption~\ref{ass:reg} is usual for biosignals and physiological trajectories (smoothing is a standard preprocessing for noisy data). Assumption~\ref{ass:mom} is the standard finite-second-moment hypothesis from Fréchet-mean theory: it makes the Fréchet function finite and the existence proof’s coercivity argument work, and it makes the Step-3 average of the phase-centered registered curves finite (Lemma~\ref{lem:alignavg}).

Assumption~\ref{ass:warp} is the strongest one. The first two clauses are choices about the method, not conditions on the data: they restrict the re-timings that the computer searches, and they can be enforced by a proper choice of $\Gc$ (Appendix~\ref{app:assump} details the regularity they assume, and Appendix~\ref{app:aa} shows how to construct such families concretely). The third clause, attainment, is the strongest one and concerns the data. It can be satisfied for instance when the curves share a common shape (they are re-timings of one nowhere-flat curve, up to a vertical shift and a positive scale) and differ by warps whose speeds stay in a known band: a bi-Lipschitz family $\Gc$ with a suitably wide speed band contains the optimal warps. Appendix~\ref{app:flat} gives more details on the kind of data that satisfies this condition. When clause (iii) fails, the computation still runs, since a best re-timing within $\Gc$ always exists, but Theorem~\ref{thm:wellposed} no longer guarantees that the quantity being estimated exists and is unique. The shapes and the elastic distance are defined for all curves, flat or not (Appendix~\ref{app:flat}). 

Regarding Assumption~\ref{ass:conc}, condition (4a) holds for most shape laws, but not for all. It fails when the shape law is symmetric, that is, when a shape $q$ and its mirror image $-q$ are equally likely (Appendix~\ref{app:bias}). Condition (4b) implies that the phase-variations must be moderate: the optimal warps must lie within $\pi\sqrt T/8$ on a sphere whose pole-to-antipode distance is $\pi\sqrt T$, that is, within an eighth of its full extent. This is reasonable when features occur approximately at consistent times, and is checkable from data as one can compute the optimal warps and measure their spread on the sphere. It can fail under extreme timing heterogeneity, where no single typical timing exists.
\end{remark}

\paragraph{Well-posedness.}
We now prove existence and well-definedness of the generic lifted mean, hence of every estimand below.

\begin{mainbox}{Well-posedness of the lifted elastic mean}
\begin{theorem}\label{thm:wellposed}
Let Assumptions~\ref{ass:reg}--\ref{ass:conc} hold for the law $P$. Then the phase-centered elastic mean $\mu(P)$ is well-defined in the following sense:
\begin{enumerate}[leftmargin=2em,itemsep=2pt,topsep=2pt]
\item[(i)] \emph{(Step 1, existence.)} If $P$ has finite second moment, the
Step-1 objective has a non-empty minimiser set over $\Sw$: a Step-1 mean shape
exists. (Only Assumptions~\ref{ass:reg}--\ref{ass:mom} are used for this point, and the
same conclusion holds for every bounded reweighting of $P$ for which Assumption~\ref{ass:warp}(iii-a) holds.)
\item[(ii)] \emph{(Step 1, uniqueness.)} Under Assumption~\ref{ass:conc}(4a), the
minimiser set of point (i) is a single shape, denoted $\shp(P)$: the Step-1 mean shape is
unique.
\item[(iii)] \emph{(Step 2, existence and uniqueness.)} Under
Assumption~\ref{ass:conc}(4b), the Step-2 objective has a unique minimiser: the
Step-2 phase mean exists and is unique.
\item[(iv)] \emph{(Independence of the representative.)} The phase-centered registered curves $Z\circ\gamma^\star_{q_Z\to\tilde q}\circ\bar\gamma^{-1}$, hence the recovered mean $\mu(P)$,
do not depend on the representative $\tilde q\in\shp(P)$ chosen in Step~2, within
the class of admissible representatives related to one another by invertible warps (any two admissible representatives are related in this way, Appendix~\ref{app:inv}).
\item[(v)] \emph{(The mean is a curve.)} The expectation defining $\mu(P)$ in Step~3 is finite for every $t$, and $\mu(P)\in\F$, with $\int_0^T|\dot\mu(P)|\le\E_P\norm{q_Z}^2$.
\end{enumerate}
Consequently $\mu(P)\in\F$ is well-defined, exists and is unique. 
\end{theorem}
\end{mainbox}

Theorem~\ref{thm:wellposed} is what legitimates our estimands: without it, ``the'' population mean shape or ``the'' typical timing may not exist, or may not be a single well-defined object. It does not, however, follow from the classical theory of Fréchet means. Recall that for a Borel probability measure $\nu$ on a metric space $(\mathcal M,d)$ with finite second moment, the Fréchet function is $F(p)=\int d^2(p,y)\,d\nu(y)$ and the Fréchet mean set is $\argmin_p F$ \citep{frechet_elements_1948,ziezold_expected_1977}. The classical guarantees for that set are existence on complete, proper spaces or on complete non-positively curved (CAT(0)) spaces, uniqueness from non-positive curvature, and the local Karcher--Afsari theory. Our three-step construction includes two such Fréchet-mean problems, in two different spaces, and each of these guarantees fails in at least one of the two, essentially because both spaces are infinite-dimensional and neither is non-positively curved. Appendix~\ref{app:frechetgeo} reviews these classical conditions in detail and explains precisely which of them fails in each of our two spaces. Assumptions~\ref{ass:reg}--\ref{ass:conc} above recover existence and uniqueness for each of the two Fréchet problems.

\begin{proofenv}[Outline of the proof.]
Each part of the proof replaces the compactness that fails in infinite dimensions by another property, and each assumption intervenes at a precise point. \emph{Step-1 existence} (point (i)) uses the following reasoning: Assumption~\ref{ass:mom} makes the objective of the Fréchet function finite and coercive, compactness of the warping family (Assumption~\ref{ass:warp}(i)) makes it weakly lower-semicontinuous, and attainment (Assumption~\ref{ass:warp}(iii-a)) ensures that its infimum over representatives equals its infimum over shapes, so the shape of the weak limit of a minimising sequence is a minimiser. \emph{Step-1 uniqueness} (point (ii)) cannot be obtained from geometry and is therefore assumed in Assumption~\ref{ass:conc}(4a). \emph{Step-2 existence and uniqueness} (point (iii)) re-derive the Karcher--Afsari local argument on the infinite-dimensional sphere: under the concentration claim of Assumption~\ref{ass:conc}(4b) the Fréchet function is strongly geodesically convex, so a minimising sequence is Cauchy and converges by completeness of $\Lt$. \emph{Well-definedness} (point (iv)) shows that replacing the Step-2 representative $\tilde q$ by another admissible one re-warps every optimal warp by a fixed $\gamma_0$, which shifts the phase mean by the same $\gamma_0$, and undoing the mean warp in the phase-centered registered curves cancels it. \emph{Regularity} (point (v)) bounds every phase-centered registered curve by its initial value plus its total variation, which is integrable by Assumption~\ref{ass:mom}, and shows that the average of absolutely continuous curves with integrable total variation is absolutely continuous. The full proofs are in Appendix~\ref{app:exist} (points (i)--(iii)) and Appendix~\ref{app:inv} (points (iv), (v) and assembly). 
\end{proofenv}

\subsection{Three shape-preserving causal estimands} \label{sec:three}

We now leverage the lifted elastic mean to define shape-preserving causal estimands, by replacing the usual Euclidean mean with the phase-centered elastic mean in $\tau_{\mathrm{FATE}}(t) = \E[Y^{(1)}(t) - Y^{(0)}(t)]$ \eqref{eq:fate}. However, in the latter, the linearity of expectation makes two natural targets coincide, $\tau_{\mathrm{FATE}} = \E[Y^{(1)}]-\E[Y^{(0)}]=\E[Y^{(1)}-Y^{(0)}]$. With the phase-centered elastic mean, which is non separable, the equivalent two expressions diverge ($\mu(Y^{(1)})-\mu(Y^{(0)}) \neq \mu(Y^{(1)}-Y^{(0)})$), and a third, covariate-local one  can also be defined. We therefore define three different estimands, following \citet{even_rethinking_2025}. 

\begin{resultbox}{{Our three elastic causal estimands}}
\textit{Throughout, $\mu(Z)$ abbreviates $\mu(P)$ for $Z\sim P$, and corresponds to the phase-centered elastic mean of the law of $Z$. The three estimands are:}

\emph{Populational} --- typical treated shape minus typical control shape:
\[
  \tau_{\mathrm{pop}}(\Pstar)=\mu\big(Y^{(1)}\big)-\mu\big(Y^{(0)}\big).
\]
\emph{Individual} --- typical shape of the Individual Treatment Effect (ITE) curves $\Delta_i=Y_i^{(1)}-Y_i^{(0)}$, which is not identifiable from the observed data due to the fundamental problem of causal inference:
\[
  \tau_{\mathrm{ind}}(\Pstar)=\mu\big(Y^{(1)}-Y^{(0)}\big).
\] 
\emph{Local} --- typical shape of the local ITEs, which correct the non-identifiability of the coupling
$(Y^{(1)},Y^{(0)})$ by pairing each unit's observed outcome with an \emph{independent}
counterfactual drawn at the same covariate value. For $(X,A,Y^{(1)},Y^{(0)})\sim\Pstar$,
let $Y^{(X),(1)}\sim P_{Y^{(1)}\mid X}$ and $Y^{(X),(0)}\sim P_{Y^{(0)}\mid X}$ be drawn independently of each other and of $(A,Y^{(1)},Y^{(0)})$ given $X$.
The local ITE of a unit pairs its observed outcome with the counterfactual of the other arm:
\begin{equation}\label{eq:locallaw}
\Delta_{\mathrm{loc}}=
\begin{cases}
Y^{(1)}-Y^{(X),(0)} & \text{if } A=1,\\[0.3ex]
Y^{(X),(1)}-Y^{(0)} & \text{if } A=0.
\end{cases}
\end{equation}
The local estimand is the phase-centered elastic mean of these local ITEs:
\[
\tau_{\mathrm{loc}}(\Pstar)=\mu\big(\Delta_{\mathrm{loc}}\big).
\] 
\end{resultbox}

In these three estimands, registration is part of the causal quantity itself, through the phase-centered elastic mean $\mu$. This is the structural difference from \citet{raykov_kernel_2025}, who register the curves in a separate preprocessing step, before the causal analysis.

As highlighted previously, with a separable measure, the two natural definitions of an average effect coincide. For the pointwise mean, $\E[Y^{(1)}]-\E[Y^{(0)}]=\E[Y^{(1)}-Y^{(0)}]$: the effect on the population equals the average of the individual effects. This equality is often seen as desirable, since it gives the population-level effect an individual-level interpretation \citep{lin_causal_2023}. However, the elastic mean is not linear, so the equality fails, and $\tau_{\mathrm{pop}}$ and $\tau_{\mathrm{ind}}$ are two different ways of defining an average effect. Neither is preferable in general.

\smallskip
This situation is not specific to functional data. It arises for any non-linear effect measure, such as a ratio, where $\E[Y^{(1)}]/\E[Y^{(0)}]$ differs from $\E[Y^{(1)}/Y^{(0)}]$. The populational version is then the standard choice, because it is identifiable \citep[Ch.~1]{hernan_causal_2020}. There is, however, an intermediate choice: a local version, which coincides with the individual effect when the two potential outcomes are independent given the covariates (and treatment is ignorable given the covariates, Assumption~\ref{ass:exch}) and remains identifiable, as proposed for the win ratio by \citet{even_rethinking_2025}. We recommend computing both the populational and the local effects, since they reveal different behaviours and do not suffer from the same limitations (Section~\ref{sec:noise} and Appendix~\ref{app:bias}).

\begin{defbox}{Interpretation of each elastic causal estimand}

\smallskip
\emph{Populational.} The elastic mean of the treated curves minus that of the control curves. The two arm means $\mu(Y^{(1)})$ and $\mu(Y^{(0)})$ can be shown side by side, which helps clinical interpretation.

\smallskip
\emph{Individual.} The elastic mean of the individual effect curves $Y^{(1)}-Y^{(0)}$, closest to what a clinician means by the effect on a patient. It is not identifiable (Theorem~\ref{thm:nonident}) and serves as a reference in simulations only.

\smallskip
\emph{Local.} The elastic mean of the effect curves between treated and control units that share the same covariates. It is an identifiable proxy for the individual effect (Section~\ref{sec:ident}).
\end{defbox}

\begin{remark}[ATE, ATT and ATC forms.]\label{rem:armforms}
The three estimands above average over the whole population: they are the \emph{ATE} (Average Treatment Effect) forms, used throughout the paper. Each estimand also has an \emph{ATT} (Average Treatment effect on the Treated) form, which averages over the treated units ($A=1$), and an \emph{ATC} (Average Treatment effect on the Control) form, which averages over the control units ($A=0$). With $\mathcal A_{\mathrm{ATE}}=\{0,1\}$, $\mathcal A_{\mathrm{ATT}}=\{1\}$, $\mathcal A_{\mathrm{ATC}}=\{0\}$ and $\bullet\in\{\mathrm{ATE},\mathrm{ATT},\mathrm{ATC}\}$, the three forms are obtained by conditioning on the target population:
\[
\begin{aligned}
\tau_{\mathrm{pop}}^{\bullet}(\Pstar) &= \mu\big(Y^{(1)}\mid A\in\mathcal A_\bullet\big)-\mu\big(Y^{(0)}\mid A\in\mathcal A_\bullet\big),\\
\tau_{\mathrm{ind}}^{\bullet}(\Pstar) &= \mu\big(Y^{(1)}-Y^{(0)}\mid A\in\mathcal A_\bullet\big),\\
\tau_{\mathrm{loc}}^{\bullet}(\Pstar) &= \mu\big(\Delta_{\mathrm{loc}}\mid A\in\mathcal A_\bullet\big),
\end{aligned}
\]
so that $\tau^{\mathrm{ATE}}_{\mathrm{pop}}=\tau_{\mathrm{pop}}$, $\tau^{\mathrm{ATE}}_{\mathrm{ind}}=\tau_{\mathrm{ind}}$ and $\tau^{\mathrm{ATE}}_{\mathrm{loc}}=\tau_{\mathrm{loc}}$. The results of the paper hold for each form, with two changes. In the identification and estimation of the populational effect (Proposition~\ref{prop:objident} and Section~\ref{sec:popest}), the weight $1/e_a(x)$ of arm $a$ becomes $e_1(x)/\big(e_a(x)\Prob(A=1)\big)$ for the ATT form and $e_0(x)/\big(e_a(x)\Prob(A=0)\big)$ for the ATC form. The constants $\Prob(A=1)$ and $\Prob(A=0)$ do not change the minimisers of Steps~1 and~2, and the Hájek normalisation of the estimator absorbs them.  In the identification of the local effect (Theorem~\ref{thm:identloc} and its proof), the covariate law $P_X$ is replaced by the law of $X$ given $A\in\mathcal A_\bullet$. The matching estimator of Section~\ref{sec:locest} gives the three forms directly.

\noindent Unless treatment is randomised, effect modifiers are in general distributed differently across arms, so the ATE, ATT and ATC forms generally differ.  
\end{remark}

\begin{mainbox}{Well-posedness of the three estimands}
\begin{corollary}\label{cor:wp3}
Under Assumptions~\ref{ass:reg}--\ref{ass:conc} (holding for the relevant law), each of
$\tau_{\mathrm{pop}}$, $\tau_{\mathrm{ind}}$ and $\tau_{\mathrm{loc}}$ exists, is unique, and is well-defined as a functional of $\Pstar$.
\end{corollary}
\end{mainbox}

\begin{proofenv} 
Each estimand is the
phase-centered elastic mean of a law on $\F$: the arm laws for $\tau_{\mathrm{pop}}$,
the difference law for $\tau_{\mathrm{ind}}$, the law of $\Delta_{\mathrm{loc}}$ for
$\tau_{\mathrm{loc}}$. Theorem~\ref{thm:wellposed} applies to each as the finite-second-moment hypothesis transfers to the difference
laws via $\norm{q_\Delta}^2\le\norm{q^{(1)}}^2+\norm{q^{(0)}}^2$, and its baseline clause via $|\Delta(0)|\le|Y^{(1)}(0)|+|Y^{(0)}(0)|$.
\end{proofenv}

\section{Identification from observational data} \label{sec:ident}

Any causal estimand is a theoretical quantity that cannot be computed from observed data alone, since it contains potential outcomes that are never observed. In this section, we assess identifiability of our three elastic causal estimands: the populational and local ones are identifiable from observational data under standard causal assumptions, while the individual one is not. 

\subsection{Causal assumptions} \label{sec:causalass}

While the geometric assumptions of Section~\ref{sec:wellposed} make the causal estimands well defined, the following three make the populational and local ones identifiable from observational data. We keep the notation of Section~\ref{sec:setting}: $q_Y$ and $q^{(a)}=q_{Y^{(a)}}$ for the SRSFs of $Y$ and $Y^{(a)}$, and $e_a(x)=\Prob(A=a\mid X=x)$.

\begin{assumption}[Consistency (SUTVA)]\label{ass:sutva}
$Y=Y^{(A)}$: each unit's observed curve is the potential outcome of the arm it received. In particular $q_Y=q^{(A)}$ and $Y(0)=Y^{(A)}(0)$.
\end{assumption}

\begin{assumption}[Positivity]\label{ass:pos}
There is $c_e>0$ with $c_e\le e_1(X)\le1-c_e$ almost surely; equivalently, $e_a(X)\ge c_e$ for both arms $a\in\{0,1\}$.
\end{assumption}

\begin{assumption}[Conditional exchangeability at the curve level]\label{ass:exch}
$Y^{(a)}\perp\!\!\!\perp A\mid X$ for each $a\in\{0,1\}$.
\end{assumption}
 
\begin{remark}[Conditional independence transfers through the SRSF.]
Since conditional independence is preserved by measurable maps and $f\mapsto q_f$ is Borel (Proposition~\ref{prop:srsf}(iii)), Assumption~\ref{ass:exch} transfers to every quantity the proofs condition on: $q^{(a)}$, the baseline $Y^{(a)}(0)$, the squared elastic distance $\del^2([q],[q^{(a)}])$ for each fixed $[q]$, and the optimal warp aligning $q^{(a)}$ to a fixed reference representative, and, for a fixed reference representative, invertible warp $\beta$ and time $t$, the value $(Y^{(a)}\circ\gamma^\star_{q^{(a)}\to\tilde q_a}\circ\beta)(t)$.
\end{remark}

\begin{remark}[How demanding are these assumptions?]
Assumption~\ref{ass:exch} is the strongest one, and stronger than its scalar analogue: the entire curve-valued potential outcome, not a scalar summary, must be independent of treatment given $X$. It is untestable and must be backed up by domain knowledge. Assumption~\ref{ass:sutva} encompasses both the assumptions of no interference and of no hidden versions of treatment, and defining a target-trial is considered best practice to check  its validity \citep{hernan_causal_2020}. 
\end{remark}

\subsection{Identification of the populational effect} \label{sec:identpop}

The identification uses standard inverse-probability-weighting, but applied to the elastic distance inside the Fréchet objective. The proofs of this section are in Appendix~\ref{app:ident}. 

\begin{proposition}[Objective identification: weights inside the alignment]\label{prop:objident}
Under the geometric Assumptions~\ref{ass:reg}--\ref{ass:conc} for the laws of $Y^{(0)}$ and $Y^{(1)}$ and the causal Assumptions~\ref{ass:sutva}--\ref{ass:exch}, each $\Pstar$-level objective is equal to an observable ($P$-level) inverse-probability-weighted objective, for each arm $a\in\{0,1\}$:
\begin{enumerate}[leftmargin=2em,itemsep=3pt,topsep=3pt,parsep=0pt]
\item[(i)] \emph{Step-1 (amplitude) objective:}
\[
  \E_{\Pstar}\!\big[\del^2([q],[q^{(a)}])\big]
  =\E_{P}\!\Big[\tfrac{\ind_{A=a}}{e_a(X)}\,\del^2([q],[q_Y])\Big]
  \quad\text{for all }[q]\in\Sw;
\]
\item[(ii)] \emph{Step-2 (phase) objective:} for a fixed representative $\tilde q_a\in\shp^{(a)}:=\shp(Y^{(a)})$ of the Step-1 mean shape of arm $a$,
\[
  \E_{\Pstar}\!\big[\dG^2\big(\gamma,\gamma^\star_{q^{(a)}\to\tilde q_a}\big)\big]
  =\E_{P}\!\Big[\tfrac{\ind_{A=a}}{e_a(X)}\,\dG^2\big(\gamma,\gamma^\star_{q_Y\to\tilde q_a}\big)\Big]
  \quad\text{for all }\gamma\in\Gt;
\]
\item[(iii)] \emph{Step-3 average of the phase-centered registered curves:} with $\tilde q_a$ as in (ii), for every invertible warp $\beta\in\G$ and every $t\in[0,T]$,
\[
  \E_{\Pstar}\big[\big(Y^{(a)}\circ\gamma^\star_{q^{(a)}\to\tilde q_a}\circ\beta\big)(t)\big]
  =\E_{P}\!\Big[\tfrac{\ind_{A=a}}{e_a(X)}\,\big(Y\circ\gamma^\star_{q_Y\to\tilde q_a}\circ\beta\big)(t)\Big].
\]
\end{enumerate}
\end{proposition}

The intuition is the ordinary one for inverse-probability weighting. In an observational study, treated and control units are not exchangeable: a unit with covariates making treatment likely is over-represented among the treated. Re-weighting each observed unit by the inverse of its propensity $e_a(X)$ rebuilds the distribution that random assignment would have given. What is new here is the position where the weights are applied: reweighting inside the Step-1 and Step-2 objectives re-balances both Fréchet means, and the same weights then re-balance the Step-3 average of the phase-centered registered curves. The result is that confounding is corrected within the alignment itself contrary to \cite{raykov_kernel_2025}, where the registration is performed without adjustment        as a pre-processing step.

Writing $\mu^{(a)}_{\mathrm{IPW}}(P)$ for the curve produced by the three-step construction of
Section~\ref{sec:lift} run on the weighted objectives, with the weighted average of Proposition~\ref{prop:objident}(iii), taken at $\beta=\bar\gamma_a^{-1}$ (with $\bar\gamma_a$ the minimiser of the weighted Step-2 objective of Proposition~\ref{prop:objident}(ii)), as Step~3, and writing $\mu^{(a)}(\Pstar)=\mu\big(Y^{(a)}\big)$, we get the following identification result for the populational effect.

\begin{mainbox}{Identification of the populational effect}
\begin{theorem}\label{thm:ident}
Under the geometric Assumptions~\ref{ass:reg}--\ref{ass:conc} and the causal
Assumptions~\ref{ass:sutva}--\ref{ass:exch}, each arm mean is identified by its weighted
observable counterpart,
\[
  \mu^{(a)}(\Pstar)\ =\ \mu^{(a)}_{\mathrm{IPW}}(P),\qquad a\in\{0,1\},
\]
and consequently
\[
  \tau_{\mathrm{pop}}(\Pstar)\ =\ \mu^{(1)}_{\mathrm{IPW}}(P)-\mu^{(0)}_{\mathrm{IPW}}(P)
\]
is a functional of the observed law $P$.
\end{theorem}
\end{mainbox}

\noindent The proof (Appendix~\ref{app:ident}) promotes the equality of the objectives in Proposition~\ref{prop:objident} to an equality of their minimisers, which are unique by Assumption~\ref{ass:conc}.

\subsection{Non-identification of the individual effect } \label{sec:identind}

 The law of the objective defining $\tau_{\mathrm{ind}}$ depends on the joint conditional law $P_{(Y^{(1)},Y^{(0)})\mid X}$, while the observed data identify only the two marginals $P_{Y^{(a)}\mid X}=P_{Y\mid A=a,X}$ under the causal assumptions mentioned above.

\begin{mainbox}{Non-identification of the individual effect}
\begin{theorem}\label{thm:nonident}
Under Assumptions~\ref{ass:sutva}--\ref{ass:exch}, the individual effect is not identified:
there is no functional $\Theta$ of the observed law such that
$\tau_{\mathrm{ind}}(\Pstar)=\Theta\big(P(\Pstar)\big)$ for every full-data law $\Pstar$
satisfying the assumptions.
\end{theorem}
\end{mainbox}

The proof, by counterexample, is in Appendix~\ref{app:nonident}.

\subsection{Identification of the local effect} \label{sec:identloc}

The local estimand $\tau_{\mathrm{loc}}=\mu(\Delta_{\mathrm{loc}})$ was constructed as an identifiable proxy for the individual effect: the law of $\Delta_{\mathrm{loc}}$ depends on $\Pstar$ only through the conditional outcome marginals $P_{Y\mid A=1,X=x}$ and $P_{Y\mid A=0,X=x}$ and the covariate law $P_X$. The three-step construction of the elastic mean $\mu$ from Section~\ref{sec:lift}, applied to the law of $\Delta_{\mathrm{loc}}$, uses three quantities: a Step-1 objective, a Step-2 objective and the Step-3 average of the phase-centered registered curves. We now prove that each of them, and the overall local estimand, can be identified using the observed law. Throughout, $Y^{(X),(1)}$ and $Y^{(X),(0)}$ are, as in \eqref{eq:locallaw}, independent draws from the two observed conditional outcome laws $P_{Y\mid A=1,X=x}$ and $P_{Y\mid A=0,X=x}$ at the common covariate value $X$.
  
\smallskip
\emph{Step-1 objective.} For $[q]\in\Sw$, define
\begin{equation}\label{eq:locobj}
\begin{split}
  M_{\mathrm{loc}}([q])&=\E_{X\sim P_X}\big[m_{\mathrm{loc}}([q],X)\big],\\[0.4ex]
  m_{\mathrm{loc}}([q],x)&=\E_{P_{Y\mid A=1,X=x}\otimes P_{Y\mid A=0,X=x}}
  \!\big[\del^2\big([q],[q_{Y^{(X),(1)}-Y^{(X),(0)}}]\big)\big].
\end{split}
\end{equation}

\smallskip
\emph{Step-2 objective.} Step~2 registers each curve to a fixed representative of the Step-1 mean shape. Let $\tilde q\in\shp_{\mathrm{loc}}$ be such a representative, admissible in the sense of Section~\ref{sec:lift}, $\shp_{\mathrm{loc}}$ denoting the Step-1 mean shape of the law of $\Delta_{\mathrm{loc}}$. Theorem~\ref{thm:identloc}(i) below shows that under our working assumptions, this shape, which depends on the unobserved causal law, is the observed minimiser of $M_{\mathrm{loc}}$, so $\tilde q$ can be chosen from the latter. For $\gamma\in\Gt$, define
\begin{equation}\label{eq:locobjphase}
  N_{\mathrm{loc}}(\gamma)=\E_{X\sim P_X}\Big[\E_{P_{Y\mid A=1,X}\otimes P_{Y\mid A=0,X}}
  \!\big[\dG^2\big(\gamma,\gamma^\star_{q_{Y^{(X),(1)}-Y^{(X),(0)}}\to\tilde q}\big)\big]\Big].
\end{equation}
It has the same structure as $M_{\mathrm{loc}}$: only the integrand changes.

\smallskip
\emph{Step-3 average of the phase-centered registered curves.} For an invertible warp $\beta\in\G$ and $t\in[0,T]$, define
\begin{equation}\label{eq:locbase}
  L_{\mathrm{loc},\beta}(t)=\E_{X\sim P_X}\Big[\E_{P_{Y\mid A=1,X}\otimes P_{Y\mid A=0,X}}
  \!\big[\big(\Delta^{(X)}\circ\gamma^\star_{q_{\Delta^{(X)}}\to\tilde q}\circ\beta\big)(t)\big]\Big],
  \qquad \Delta^{(X)}=Y^{(X),(1)}-Y^{(X),(0)} .
\end{equation}
Unlike the FATE, this functional does not split into one term per arm: the registration warp of $\Delta^{(X)}$ depends on both curves.

\begin{mainbox}{Identification of the local effect}
\begin{theorem}\label{thm:identloc}
Under the geometric Assumptions~\ref{ass:reg}--\ref{ass:conc} for the law of $\Delta_{\mathrm{loc}}$ and the causal Assumptions~\ref{ass:sutva}--\ref{ass:exch}, the three quantities of the three-step construction of $\tau_{\mathrm{loc}}$ are functionals of the observed law:
\begin{enumerate}[leftmargin=2em,itemsep=3pt,topsep=3pt,parsep=0pt]
\item[(i)] \emph{Step-1 objective:} $\E_{\Pstar}\big[\del^2\big([q],[q_{\Delta_{\mathrm{loc}}}]\big)\big]=M_{\mathrm{loc}}([q])$ for all $[q]\in\Sw$, hence $\shp_{\mathrm{loc}}=\argmin_{[q]\in\Sw}M_{\mathrm{loc}}([q])$;
\item[(ii)] \emph{Step-2 objective:} for $\tilde q\in\shp_{\mathrm{loc}}$, $\E_{\Pstar}\big[\dG^2\big(\gamma,\gamma^\star_{q_{\Delta_{\mathrm{loc}}}\to\tilde q}\big)\big]=N_{\mathrm{loc}}(\gamma)$ for all $\gamma\in\Gt$;
\item[(iii)] \emph{Step-3 average of the phase-centered registered curves:} for $\tilde q$ as in (ii), every invertible warp $\beta\in\G$ and every $t$, $\E_{\Pstar}\big[\big(\Delta_{\mathrm{loc}}\circ\gamma^\star_{q_{\Delta_{\mathrm{loc}}}\to\tilde q}\circ\beta\big)(t)\big]=L_{\mathrm{loc},\beta}(t)$.
\end{enumerate}
Consequently, $\tau_{\mathrm{loc}}(\Pstar)$ is identified: it is the curve produced by the three-step construction of Section~\ref{sec:lift} run on $M_{\mathrm{loc}}$, $N_{\mathrm{loc}}$ and $L_{\mathrm{loc},\beta}$ with $\beta=\bar\gamma^{-1}$ ($\bar\gamma$ being the minimiser of $N_{\mathrm{loc}}$).
\end{theorem}
\end{mainbox}
 
\noindent The proof is in Appendix~\ref{app:identloc}.

\begin{defbox}{The gap between the individual and local effects.}
The individual causal estimand carries the richest information as it involves the joint realisation $(Y^{(1)},Y^{(0)})$, and is unidentifiable for that same reason. The local causal estimand replaces the unobserved joint by the product of identified conditional marginals. Following \citet{even_rethinking_2025}, under $Y^{(1)}\perp\!\!\!\perp Y^{(0)}\mid X$, the two estimands coincide, but this independence is extremely strong and rarely holds in practice.
\end{defbox} 

\begin{remark}[Why no exact pairwise reweighting for the local effect.]\label{rem:noipw}
The local effect can be identified (Theorem~\ref{thm:identloc}), but not with IPW-like reweighting as in the populational case. The reason is that the objective
$M_{\mathrm{loc}}$ integrates a function of a \emph{pair} of curves drawn at a
common covariate value, whereas the data supply pairs of distinct units, whose covariates often differ. Appendix~\ref{app:noipw} expands this point.
\end{remark}

\section{Estimation} \label{sec:est}  

After identification comes estimation, or how to estimate a statistical quantity defined at the scale of the observed data distribution from a sample. We develop an IPW estimator for the populational effect and a matching estimator for the local effect. The individual effect is not identifiable, so it cannot be estimated from observed data.

\subsection{Estimating the populational effect by inverse-probability weighting} \label{sec:popest}

\noindent The IPW elastic populational estimator is computed by following the steps below.
\begin{itemize}[leftmargin=2em,itemsep=2pt,topsep=3pt]
\item \emph{Propensity.} Fit the propensity $\hat e_a$ on the whole sample.
\item \emph{For each arm} $a\in\{0,1\}$:
  \begin{itemize}[leftmargin=1.5em,itemsep=1pt,topsep=1pt]
  \item \emph{Mean shape (Step~1).} Solve the weighted Step-1 problem
  $\hat\shp^{\mathrm{IPW}}_{a,n}=\argmin_{[q]\in\Sw}\frac1n\sum_i\frac{\ind_{A_i=a}}{\hat e_a(X_i)}\del^2([q],[q_{Y_i}])$ by the alternating registration/averaging algorithm of \citet{srivastava_registration_2011} with IPW weights.
  \item \emph{Mean warp (Step~2).} Register each observed curve of arm $a$ to a nowhere-flat representative $\tilde q_a$ of the estimated shape, and average the resulting optimal warps on the sphere with the same weights,
  $\hat{\bar\gamma}_a=\argmin_{\gamma\in\Gt}\frac1n\sum_i\frac{\ind_{A_i=a}}{\hat e_a(X_i)}\dG^2\big(\gamma,\gamma^\star_{q_{Y_i}\to\tilde q_a}\big)$.
  \item \emph{Arm curve (Step~3).} Re-time each registered curve of arm $a$ by the inverse mean warp and average them with the same weights,
  \[
  \hat\mu^{\mathrm{IPW}}_{a,n}(t)=\frac{\sum_i w_{a,i}\,\big(Y_i\circ\gamma^\star_{q_{Y_i}\to\tilde q_a}\circ\hat{\bar\gamma}_a^{-1}\big)(t)}{\sum_i w_{a,i}},
  \qquad w_{a,i}=\frac{\ind_{A_i=a}}{\hat e_a(X_i)} .
  \]
  The weights are normalised to sum to one (Hájek form, Remark~\ref{rem:hajek}).
  \end{itemize}
\item \emph{Effect.} Difference the arms, $\hat\tau_n=\hat\mu^{\mathrm{IPW}}_{1,n}-\hat\mu^{\mathrm{IPW}}_{0,n}$.
\end{itemize} 

\begin{remark}[Horvitz--Thompson and Hájek forms.]\label{rem:hajek}
Write $g_{a,i}=Y_i\circ\gamma^\star_{q_{Y_i}\to\tilde q_a}\circ\hat{\bar\gamma}_a^{-1}$ for the registered curves of Step~3. The \emph{Horvitz--Thompson} form \citep{horvitz_generalization_1952} divides by the sample size, $\frac1n\sum_iw_{a,i}\,g_{a,i}$. It is simple and, in the scalar case, unbiased when the propensities are known. But its weights do not sum to one in a finite sample, so it rescales the curve: if every $g_{a,i}$ equals one curve $g$, it returns $\big(\frac1n\sum_iw_{a,i}\big)\,g$ instead of $g$. The \emph{Hájek} form \citep{hajek_comment_1971} divides by the sum of the weights, $\sum_iw_{a,i}\,g_{a,i}/\sum_iw_{a,i}$. It has a small finite-sample bias, but usually a lower variance, especially when some weights are large \citep{lunceford_stratification_2004}. Both forms estimate the same quantity $\mu^{(a)}_{\mathrm{IPW}}(P)$: they differ only by the factor $\frac1n\sum_iw_{a,i}$, which tends to $\E[\ind_{A=a}/e_a(X)]=1$ as $n$ grows, and the normalisation does not change the minimisers of Steps~1 and~2. We use the Hájek form because it keeps the average end values of each arm.
\end{remark}

\begin{remark}[Computation.]\label{rem:optim}
The Step-1 objective is not convex, because of the registration, so the alternating algorithm may stop at a local minimiser. Our shape space does not meet the classical conditions that guarantee a unique Fréchet mean (Appendix~\ref{app:frechetgeo}).
\end{remark}

\subsection{Estimating the local effect by matching} \label{sec:locest}

Let $\mathcal N_a=\{i:A_i=a\}$, $n_a=|\mathcal N_a|$, and for a unit $j$ let $\sigma_{a,k}(j)\subset\mathcal N_a$ index its $k$ nearest neighbours in the opposite arm by covariate distance. We consider the pairwise difference curves $\{Y_j-Y_\ell\}_{\ell\in\sigma_{a,k}(j)}$ (not the difference with a pointwise-averaged neighbour). The matched Step-1 objectives are 
\begin{gather*}
  \hat M^{\mathrm{ATT}}_{\mathrm{loc},n}([q])=\tfrac{1}{n_1k}\!\sum_{j\in\mathcal N_1}\sum_{\ell\in\sigma_{0,k}(j)}\!\del^2\big([q],[q_{Y_j-Y_\ell}]\big),\\
  \hat M^{\mathrm{ATC}}_{\mathrm{loc},n}([q])=\tfrac{1}{n_0k}\!\sum_{i\in\mathcal N_0}\sum_{\ell\in\sigma_{1,k}(i)}\!\del^2\big([q],[q_{Y_\ell-Y_i}]\big),\\
  \hat M_{\mathrm{loc},n}=\tfrac{n_1}{n}\,\hat M^{\mathrm{ATT}}_{\mathrm{loc},n}+\tfrac{n_0}{n}\,\hat M^{\mathrm{ATC}}_{\mathrm{loc},n}.
\end{gather*}
The objective $\hat M_{\mathrm{loc},n}$ pairs every unit with its $k$ nearest neighbours and estimates $M_{\mathrm{loc}}$. Its two parts average over the treated units only and over the control units only: they estimate the objectives of the ATT and ATC forms (Remark~\ref{rem:armforms}).
The mean warp (Step~2) and the final curve (Step~3) then proceed as in the populational case with these matched difference curves, giving $\hat\tau_{\mathrm{loc},n}$, and, with the weights of $\hat M^{\mathrm{ATT}}_{\mathrm{loc},n}$ or $\hat M^{\mathrm{ATC}}_{\mathrm{loc},n}$, the estimators $\hat\tau^{\mathrm{ATT}}_{\mathrm{loc},n}$ and $\hat\tau^{\mathrm{ATC}}_{\mathrm{loc},n}$. 

\begin{remark}[Pairwise contrasts, not averaged neighbours.]\label{rem:pairwise}
Because the SRSF is non-linear, the matched proxy must be the collection of pairwise differences $Y_j-Y_\ell$, not the single difference $Y_j-\bar Y$ with a pointwise-averaged neighbour curve: averaging the neighbours first would re-inject amplitude flattening into the matching step. With exact matches, $X_\ell=X_j$, and since the two units are distinct and independent, the pair $(Y_j,Y_\ell)$ is a draw from the product of the two conditional outcome marginals at that covariate value,
$(Y_j,Y_\ell)\ \sim\ P_{Y\mid A=1,X=X_j}\otimes P_{Y\mid A=0,X=X_j}$. This is the curve analogue of the pairwise comparisons of \citet{even_rethinking_2025}.
\end{remark}

\section{Simulations} \label{sec:sims}

This section has three parts. We first compare the causal estimands themselves,
to see which of them correctly captures the shape carried by the simulated units. We then compare estimators, to see how closely each recovers the target it aims at. Finally, we show how noise affects the elastic estimands (Section~\ref{sec:noise}). The simulation model is given below, with more details in Appendix~\ref{app:sims}. The code implementing the methods and reproducing the simulations will be made available in a later version of this preprint.

\subsection{The simulation model} \label{sec:simmodel}

Time since treatment is rescaled to $[0,1]$ on a uniform grid of $m=81$
points, and each simulated study observes $(X_i,A_i,Y_i)_{i\le n}$ with
$n=250$. A single confounder $X\sim\mathrm{Unif}[0,1]$ drives the treatment,  
$A\mid X\sim\mathrm{Bernoulli}(e(X))$ with
$e(X)=\operatorname{expit}(-1.1+2.2X)\in[0.249,0.751]$, and it also drives the
control trajectory, the depth of the treated response and its timing. Control
potential outcomes decline slowly from their baseline value. Treated potential
outcomes drop to a nadir at a unit-specific time and then follow one of two
shapes: a share $p_{\mathrm{rebound}}$ of units return to their baseline
level, and the others stay near the level they have reached. This mimics body
weight during and after a course of a GLP-1 receptor agonist.

Write $\varphi_{c,w}(t)=\exp\{-(t-c)^2/(2w^2)\}$ for a Gaussian bump and
$\sigma(z)=\operatorname{expit}(z)$. Let $U_1,U_2\sim N(0,1)$ be unit-level
latent variables, let $B\sim\mathrm{Bernoulli}(p_{\mathrm{rebound}})$ indicate
a rebounding unit, and let $\varepsilon^{(1)}$ and $\varepsilon^{(0)}$ be two independent smooth mean-zero noise curves,
independent across units, of scale $\sigma_\varepsilon=0.03$. The potential
outcomes are
\begin{align*} 
Y^{(1)}(t) &= \begin{cases}
  -h\,\varphi_{c,\,0.10}(t)+\varepsilon^{(1)}(t), & B=1 \text{ (rebound)},\\[2pt]
  -h\,\sigma\!\big(\tfrac{t-c}{0.03}\big)\,
      \big(1+0.05\,(t-c)_+\big)+\varepsilon^{(1)}(t), & B=0 \text{ (plateau)},
\end{cases}
& c &= 0.40+\alpha\,\big(0.34\,(X-\tfrac12)+0.08\,U_1\big),\\
&& h &= 0.90+0.55\,X+0.06\,U_2,\\
Y^{(0)}(t) &= -(0.05+0.12X)\,t^{3/2}+\varepsilon^{(0)}(t). &&
\end{align*}
The phase-variation scale $\alpha$ multiplies both timing terms, so it sets
how much the nadir time $c$ varies from one unit to the next. We report five scenarios, all with $\alpha=1$, numbered by decreasing rebound
share: Scenarios~1 to~5 have $p_{\mathrm{rebound}}=100\%$, $75\%$, $50\%$, $25\%$
and $0\%$. In Scenario~1 every treated unit rebounds, and in Scenario~5 every
treated unit plateaus: these two scenarios have homogeneous shapes, while
Scenarios~2 to~4 carry a heterogeneity of shapes. We also sweep $\alpha$ from $0$
to $1$ in Scenario~1. Positivity (Assumption~\ref{ass:pos}) holds
since $e(X)\in[0.249,0.751]$, and conditional exchangeability (Assumption~\ref{ass:exch})
holds at the curve level by construction. In this model $Y^{(0)}$ depends only on $X$ and on its own noise, so the two potential outcomes are independent given $X$, and $\tau_{\mathrm{ind}}=\tau_{\mathrm{loc}}$ (Section~\ref{sec:identloc}): the simulations do not separate these two estimands. Appendix~\ref{app:sims} gives the
remaining details.

\subsection{Comparison of the causal estimands} \label{sec:simestimands}

Figure~\ref{fig:estgrid} shows the five scenarios, from Scenario~1 down to Scenario~5. Across scenarios, the elastic
estimands recover the shape present in the simulation much better than the
FATE. In Scenarios~1 and~5 there is a single shape to recover, and both
$\tau_{\mathrm{pop}}$ and $\tau_{\mathrm{ind}}$ reproduce it in depth and
timing, whereas $\tau_{\mathrm{FATE}}$ flattens it. In Scenarios~2 to~4 the elastic estimands return the corresponding intermediate shape (reflecting the proportions of plateauing and rebounding curves): they drop to the typical depth at the typical time, then return part of the way towards baseline, to the average final level. That level is also the final level of the FATE, since both end values of an elastic mean are averages (Proposition~\ref{prop:propsB}), while the
FATE stays flattened throughout. Finally, $\tau_{\mathrm{pop}}$ and
$\tau_{\mathrm{FATE}}$ are each built from two arm-level means, shown in
column~(c). In clinical applications these two trajectories are often easier
to read than their difference.

\begin{figure}[p]
\centering 
\includegraphics[width=0.93\textwidth]{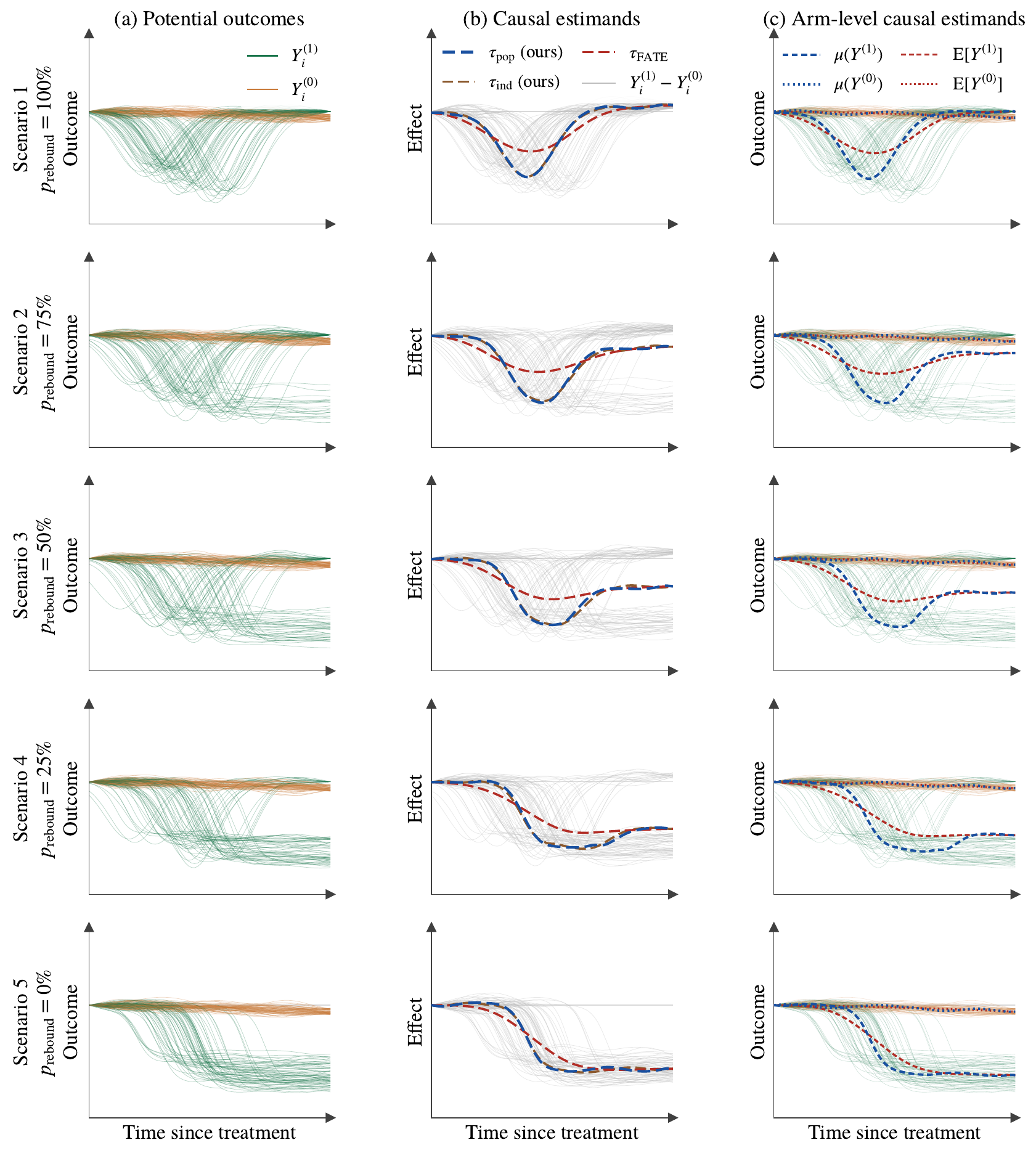}
\caption{\textbf{The causal estimands, across scenarios} (one row per
scenario, from Scenario~1, $p_{\mathrm{rebound}}=100\%$, down to Scenario~5,
$p_{\mathrm{rebound}}=0\%$). \emph{(a)} Potential
outcomes of a subsample of units: treated $Y^{(1)}_i$ (green), control
$Y^{(0)}_i$ (orange). \emph{(b)} The causal estimands (dashed):
$\tau_{\mathrm{pop}}$ (blue), $\tau_{\mathrm{ind}}$ (brown) and
$\tau_{\mathrm{FATE}}$ (red), over the differences $Y^{(1)}_i-Y^{(0)}_i$
(light grey). \emph{(c)} The arm-level estimands: phase-centered elastic means
$\mu(Y^{(a)})$ (blue) and pointwise means $\E[Y^{(a)}]$ (red); short dashes:
treated arm, dots: control arm.}
\label{fig:estgrid}
\end{figure}

\subsection{Comparison of the estimators} \label{sec:simestimators}

The comparison above is between targets, not between performances. What it shows is that the elastic estimands capture the shape of the effect, whereas the FATE flattens it. The question now is whether that shape can be estimated, and how closely
each estimator approaches the target it aims at. We compare four estimators: our IPW estimator of $\tau_{\mathrm{pop}}$
(Section~\ref{sec:popest}), our matched local estimator
(Section~\ref{sec:locest}), the doubly robust FATE estimator DR-FoS of
\citet{testa_doubly_2025}, and the SRVF operator-kernel ATE of
\citet{raykov_kernel_2025}, their best-performing estimator and the closest
competing method, denoted ``SRVF kernel ATE'' below. This last method
registers the curves before the causal analysis and has no well-defined causal
estimand (Section~\ref{sec:intro}), so it enters the estimator comparison
only.

The four estimators do not target the same causal estimand, so we score them against the typical shape of the effect: the $L^2$ distance of the estimated curve to $\tau_{\mathrm{pop}}$, and the absolute error of its depth, onset time and drop speed relative to the median of these metrics over the individual effects (Appendix~\ref{app:sims}, Figure~\ref{fig:metricschema}). Figure~\ref{fig:mcgrid} shows the Monte-Carlo distribution of the estimated effect curves over $50$ replications, in Scenarios~1 and~3. All five scenarios are in Appendix~\ref{app:sims}, Figure~\ref{fig:mcgridfull}, together with the error distributions for the shape evaluation metrics mentioned above (Figure~\ref{fig:boxgrid}).  

Three findings emerge, similar across the five scenarios. First, both elastic estimators recover the shape of the effect, with small errors to the populational causal estimand and to the three shape metrics. Second, DR-FoS estimates its own target accurately, but that target is the flattened FATE under timing heterogeneity: measured against the typical shape it is systematically biased, with errors on the depth and the drop speed several times those of the elastic estimators in every scenario, and a large error on the onset time in Scenarios~1 to~3. This gap comes from the choice of causal estimand, not from the estimator. Third, the SRVF kernel ATE fails to recover the shape in Scenarios~1 to~4, and is less accurate than the elastic estimators in Scenario~5. 

Neither competitor comes closer to the elastic target as the sample grows. In
the sample-size sweep of Appendix~\ref{app:sens} their $L^2$ error stays
between about $0.17$ and $0.21$ from $n=20$ to $n=1000$, while the error of the two
elastic estimators falls to about $0.03$. For DR-FoS, what separates it from the elastic estimators is the estimand and
not the sample size.

\begin{figure}[htbp]
\centering
\includegraphics[width=0.92\textwidth]{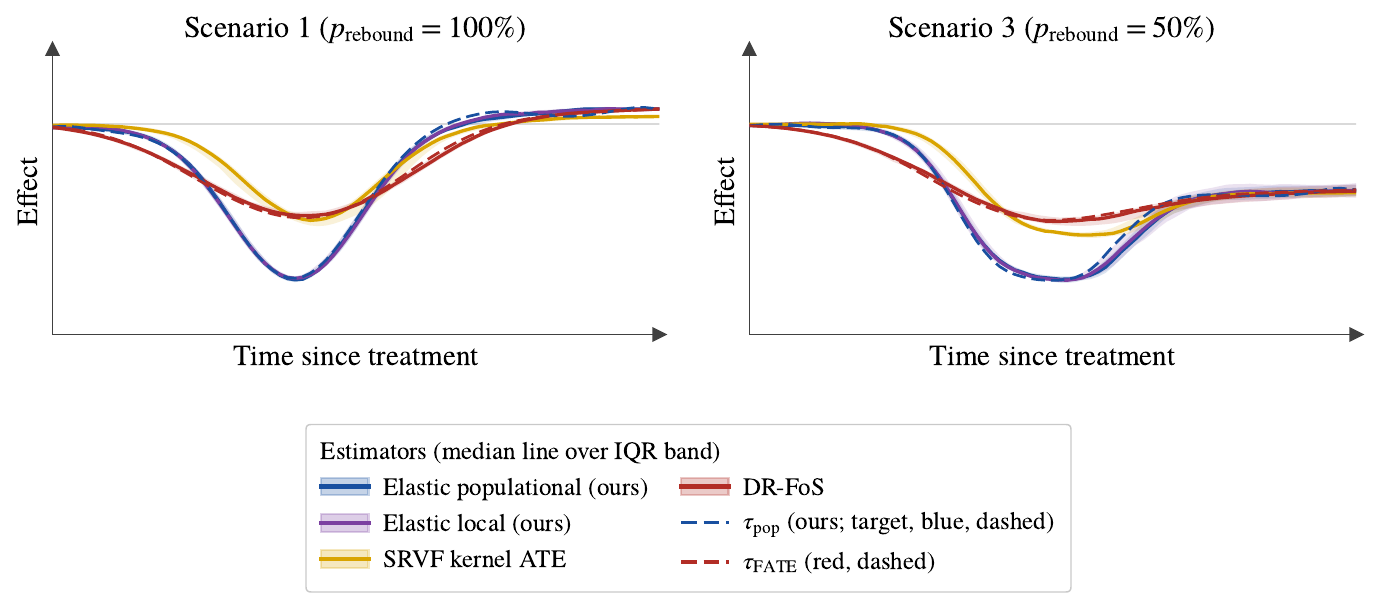}
\caption{\textbf{Monte-Carlo distribution of the four estimators in two
scenarios}, Scenario~1 ($p_{\mathrm{rebound}}=100\%$, left) and Scenario~3
($p_{\mathrm{rebound}}=50\%$, right); $50$ replications per scenario,
$n=250$. The three remaining scenarios are in
Figure~\ref{fig:mcgridfull}. Each panel: pointwise median (solid) and
pointwise interquartile band of the four estimated effect curves, with the
causal estimands $\tau_{\mathrm{pop}}$ (blue, dashed) and
$\tau_{\mathrm{FATE}}$ (red, dashed).}
\label{fig:mcgrid}
\end{figure}

\paragraph{The role of phase variation.}
The phase-variation scale $\alpha$ controls how much timing heterogeneity the
data carry. We sweep it from $0$ to $1$ in Scenario~1
(Figure~\ref{fig:phase}). At
$\alpha=0$ the estimands nearly coincide ($\|\tau_{\mathrm{FATE}}-\tau_{\mathrm{pop}}\|_{L^2}=0.03$) and the
pointwise DR-FoS estimator is as accurate for $\tau_{\mathrm{pop}}$ as our
estimators. As $\alpha$ grows, the estimand gap widens, the FATE flattens, and the errors of both competing estimators rise accordingly, while the errors of the elastic estimators remain
small. In Scenarios~2 to~4 the gap is large already at $\alpha=0$, because the heterogeneity of shapes alone separates the causal estimands (Appendix~\ref{app:mix}). Thus, the elastic estimands better capture the shape of the effect not only when phase variation is present, but also when the data-generating process carries a heterogeneity of shapes.

\begin{figure}[t]
\centering
\includegraphics[width=0.95\textwidth]{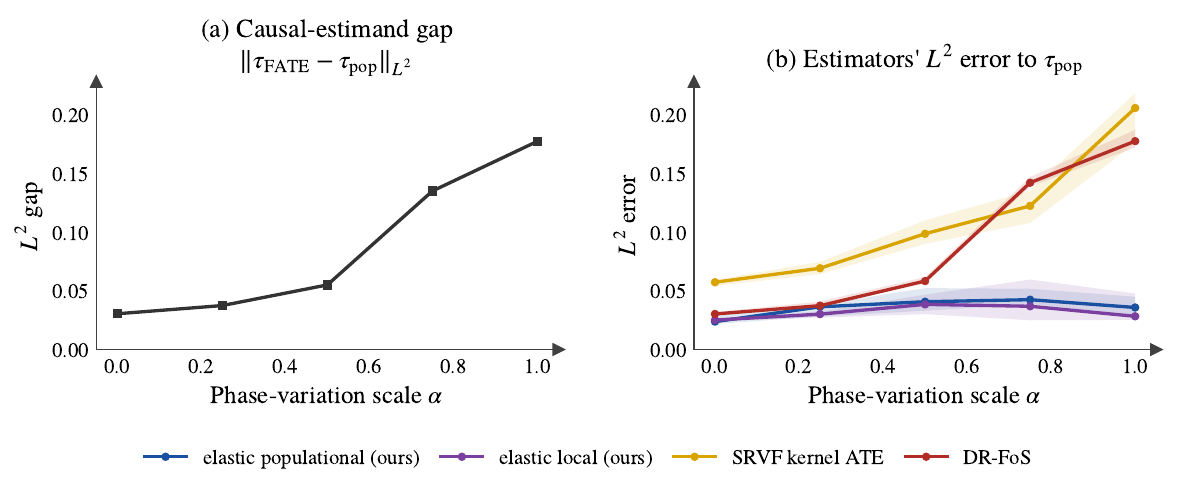}
\caption{\textbf{Effect of the phase-variation scale} (Scenario~1,
$p_{\mathrm{rebound}}=100\%$; $n=250$, $10$ replications per level). \emph{(a)} The causal-estimand gap
$\|\tau_{\mathrm{FATE}}-\tau_{\mathrm{pop}}\|_{L^2}$ as a function of
$\alpha$. \emph{(b)} Estimators' median $L^2$ error to $\tau_{\mathrm{pop}}$
(bands: interquartile range).}
\label{fig:phase}
\end{figure}

\paragraph{Sensitivity analyses.}
Appendix~\ref{app:sens} tests the robustness of these results to the sample size (Figure~\ref{fig:sensn}), the propensity model (Figure~\ref{fig:sensmis}) and the covariate dimension (Figure~\ref{fig:sensdim}). In the baseline design and across sample sizes the two elastic estimators perform very similarly, and both recover the shape of the effect far better than the competitors, whose error mainly comes from an estimand gap. Each elastic estimator has its own weak point: misspecification of the propensity model affects mainly the populational estimator, and a large covariate dimension mainly the local one. Their one shared weakness is structural: a target aimed at the typical shape reads noisy oscillations as signal, so noisy curves must be smoothed first. This is the noise dependence of the elastic mean, discussed in Section~\ref{sec:noise}.

\paragraph{Computational cost.}
The elastic mean alternates two operations: register every curve to the
current candidate shape, then average the registered curves. Each registration is
a dynamic programme on an $m\times m$ grid, so one alternation over $N$ curves
costs of the order of $N\,m^2$ operations, and $N$ determines the cost of a fit.
For the populational estimator and the SRVF kernel ATE, $N=n$. For DR-FoS, which
registers nothing, the cost is negligible. For the local estimator, $N$ is the
number of matched difference curves, $k\,n$ for $\hat\tau_{\mathrm{loc},n}$. A uniform random
subsample of these differences gives an unbiased estimate of the matched
objective, so $N$ may be capped at any $N_{\max}<k\,n$ without changing the
estimand. Three further savings leave the estimand unchanged. The $N$ registrations inside one alternation are independent, since they all align to the same fixed candidate shape,
and they parallelise directly. Besides, the early alternations can run on a coarser grid,
since the cost grows with $m^2$---the number of alternations trades cost for
accuracy. In our simulations, we stop after two alternations and one final registration, which we found sufficient at the noise level of the simulations (Appendix~\ref{app:sims}, paragraph on the estimands).
 
\subsection{When noise dominates the signal: the template bias} \label{sec:noise}

In the case where all curves are one template curve re-timed by random warps, possibly rescaled and shifted by a constant, the elastic mean recovers the template shape \citep{kurtek_signal_2011}.  However, as soon as the curves also carry noise, or any other variability between units, the elastic mean is biased with respect to the template \citep{devilliers_template_2017,devilliers_inconsistency_2017, miolane_template_2017}. This bias is generally assessed in terms of the total variation of the curve $f$, $\int_0^T|\dot f|=\norm{q_f}^2$, the absolute total distance it travels on the ordinate axis. Indeed, this bias inflates the total variation of the elastic mean above that of the template, as Figure~\ref{fig:biasgrid} shows in Scenario~1 (Appendix~\ref{app:bias} gives the numbers). The reason is that aligning a noisy curve to a candidate shape also aligns part of its noise with that shape, and these aligned parts add up in the average instead of cancelling out. The template bias grows linearly with the noise level for large noise, and more slowly for small noise. It vanishes completely only when the noise reduces to random constants, which is the template model. It can however be reduced by pre-smoothing noisy curves before the analysis \citep{bigot_consistency_2011}.

This template bias is the price of correcting phase variation. The pointwise mean has no such template bias, but it does not align the curves, so it flattens their misaligned features, as seen previously. The elastic mean aligns the curves, and the same alignment reads part of the noise as signal. Which of the two is closer to the typical shape depends on how much phase variation there is compared to noise. A trade-off can be found by pre-smoothing the curves before the analysis \citep{bigot_consistency_2011}, which reduces the bias at the cost of leaving some phase variation uncorrected.

\paragraph{When does it matter?}
The bias is a problem when the noise is more important than the signal inside the elastic mean, that is, when the total variation added by the noise is of the order of the total variation of the curve being averaged, or larger. This happens in three situations, detailed in Appendix~\ref{app:bias}, the first two of which are illustrated in Figure~\ref{fig:biasgrid}.

\begin{figure}[t]
\centering
\includegraphics[width=0.93\textwidth]{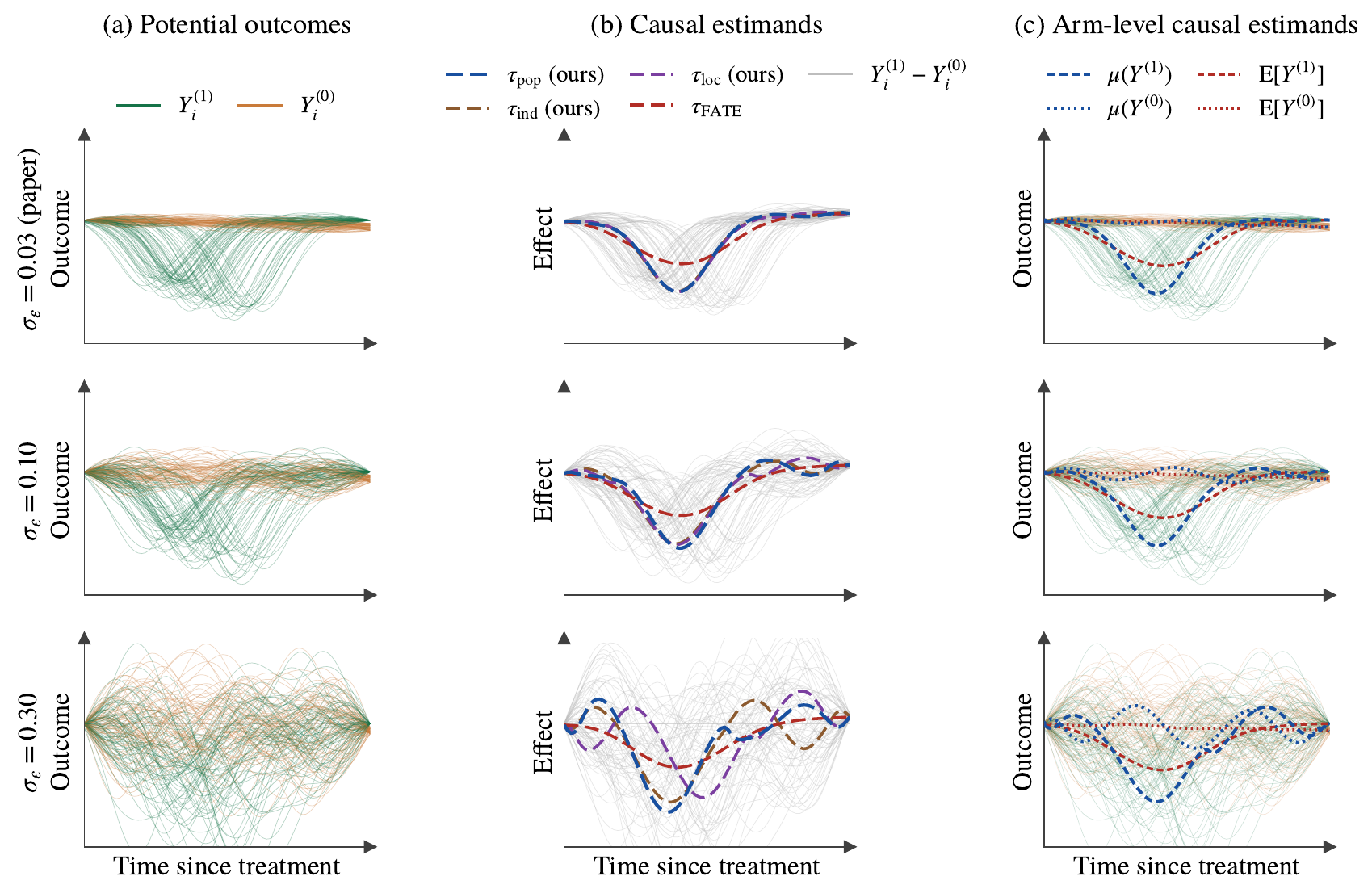}
\caption{\textbf{The causal estimands as the noise level grows} (Scenario~1, the same $800$ units at three noise scales; the first row is the noise level of all the other simulations). Columns as in Figure~\ref{fig:estgrid}: \emph{(a)} potential outcomes of a subsample of units (treated: green; control: orange); \emph{(b)} the causal estimands (dashed) $\tau_{\mathrm{pop}}$ (blue), $\tau_{\mathrm{ind}}$ (brown), $\tau_{\mathrm{loc}}$ (purple) and $\tau_{\mathrm{FATE}}$ (red) over the individual differences (grey); \emph{(c)} the arm-level estimands, elastic means $\mu(Y^{(a)})$ (blue) and pointwise means $\E[Y^{(a)}]$ (red); short dashes: treated arm, dots: control arm.}
\label{fig:biasgrid}
\end{figure}

\paragraph{What to do.}
The template bias is not specific to our estimands: it affects any analysis
that registers noisy curves, and the literature corrects for it by smoothing
the curves first \citep{ramsay_silverman_2005,bigot_consistency_2011}. We
follow this practice. The amount of smoothing should remove noise but keep signal, which requires domain
knowledge to avoid erasing the features we want
to compare, such as weekly patterns or weight cycling. In the Withings data of Section~\ref{sec:realdata},
the knowledge of the type of noise in weight measurements guided the choice of smoothing parameters.

\clearpage

\section{Real data application: longitudinal impact of GLP-1 RA
medications on weight}
\label{sec:realdata}

GLP-1 receptor agonists (GLP-1 RA) produce a large weight loss at first, often
followed by a plateau, a rebound, or several rebounds, the yo-yo effect
\citep{wilding_weight_2022,rubino_effect_2021}. Patients reach these phases at
their own times (some rebounding early, others later), so the response displays strong phase variability, the regime
our elastic estimands target. We apply $\hat\tau_{\mathrm{pop}}$ and
$\hat\tau_{\mathrm{loc}}$ to weight trajectories following initiation of an
anti-obesity GLP-1 RA, with DR-FoS \citep{testa_doubly_2025} and the
SRVF-registered kernel estimator \citep{raykov_kernel_2025} as competitors, as
in Section~\ref{sec:sims}. All four estimators are reported as an effect on
the treated, the clinically relevant population here.
 
\paragraph{Cohort, covariates and outcome curves.}\label{sec:cohort_covariates_outcome} 

We link two private data sources at the patient level from 2020 to 2026:
administrative medical claims give the treatment $A$, Withings connected-scale
measurements give the outcome $Y$, and both give the covariates $X$. The
treated are new users of an anti-obesity GLP-1 RA, the controls are non-users,
and the index date $t_0$ is the first fill for the treated and a random
eligible date for the controls. After the eligibility filters, the cohort has
544 treated and $5\,000$ controls, followed over 24 months. The
covariates $X_i\in\mathbb{R}^{14}$ encode 9 variables measured before
$t_0$: socio-demographics, baseline health status and motivation.

Scale measurements are self-recorded, so they happen at irregular and
patient-specific times. We turn them into curves before any causal
analysis. For each patient, we fit a penalised cubic smoothing
spline to the weights, take its value at $t_0$ as the baseline weight, express the
spline as a relative change from this baseline, and finally evaluate it on a common
weekly grid. Each patient therefore gets one outcome curve $Y_i\in\mathbb{R}^{104}$, displayed in
column~(a) of Figure~\ref{fig:rd_main}. Appendix~\ref{app:realdata} details the
cohort construction, the covariate definitions, the baseline weight assessment and the
choice of the penalty parameter for the cubic splines.

\paragraph{Comparison of the estimators.}
The four
estimators agree on a substantial weight reduction after treatment, but not on its shape.
DR-FoS gives the smoothest effect curve of the four, almost monotone 
and with no rebound, suggesting that the weight lost is maintained.
This is the amplitude flattening of Figure~\ref{fig:introfate}, now observed on
real data: as rebounds happen at patient-specific times, averaging at fixed
time points cancels them. DR-FoS therefore recovers the level of the effect
without recovering its shape, exactly as in the simulations of
Section~\ref{sec:sims}.

This application also suggests that aligning the curves
as an unadjusted preprocessing step, as in the SRVF-registered kernel estimator,
recovers the shape of the effect less faithfully than our elastic estimators.
While it captures the overall decrease, it exhibits a small early peak and reaches
its minimum only after about $19$ months, two features rarely observed in the
individual trajectories of column~(a), which mostly decrease from $t_0$ and reach
their minimum earlier. Its curve also departs substantially from DR-FoS, so it
does not reflect the pointwise average effect either.

Our two estimators, in contrast, both recover a typical trajectory
with a rebound, of the kind visible in the curves of column~(a). They both end near the level 
of DR-FoS at $24$ months: the FATE recovers the net loss,
but not the rebound that led to it. The elastic estimates
therefore describe how the effect of a GLP-1 RA unfolds over time, and not
only its size.
Unlike in the simulations of Section~\ref{sec:sims}, where they
behaved almost identically, the local estimate reaches a slightly deeper nadir and a larger 
rebound than
the populational one, with a small phase shift. This difference is likely due to the
template bias described in Section~\ref{sec:noise}, which affects them unequally (Appendix~\ref{app:stability_estimations}). The local estimand is
the identifiable target closest to the individual effect, which makes it
the more clinically relevant of the two.

\paragraph{Clinical insights enabled by the elastic estimators.}
The maximal effect of the typical trajectory, about $-12.5$ 
percentage points (pp)
after $14$ months for the local estimate, is close in both size and timing to the
pivotal STEP~1 trial of semaglutide, where weight loss reaches its nadir at week 60,
about $14$ months, for a placebo-adjusted effect of $-12.4$ pp at week 68
\citep{wilding_once-weekly_2021}. But this trial effect is a difference of pointwise
means, like the FATE, so it should strictly be compared with DR-FoS, about
$-8.5$ pp after one year, maintained thereafter.
In a trial, however, most patients stay on
treatment and follow the same dose schedule, so their weight losses are more synchronised and the pointwise mean is likely less flattened than with our real-world data. Our cohort keeps every
patient who went on weighing themselves, whether or not they continued the
treatment. It measures the effect of initiating a GLP-1 RA in real life,
where adherence, molecules and dosage, and hence the timing of weight loss, vary significantly  between
patients. The FATE flattens this effect, which is why the elastic
estimate, and not DR-FoS, matches the trial. This DR-FoS value
is in turn in line with the mean weight loss of $8.7\%$ after one year reported in
clinical practice \citep{gasoyan_changes_2025}.

The rebound of the typical trajectory also has a clinical explanation:
treatment discontinuation. Most patients without type~2 diabetes stop GLP-1
RAs within a year \citep{rodriguez_discontinuation_2025}, and weight comes
back once they stop: in the extension of STEP~1, participants regained two
thirds of their weight loss within one year of withdrawal
\citep{wilding_weight_2022}. Only part of our cohort stops, which is
consistent with a smaller rebound, about two fifths of the
maximal loss for the local estimate. 
A longer follow-up, beyond our 24-month horizon, would tell
whether this regain goes on or whether the weight stabilises.
\begin{figure}[h]
\centering
\includegraphics[width=1\textwidth]{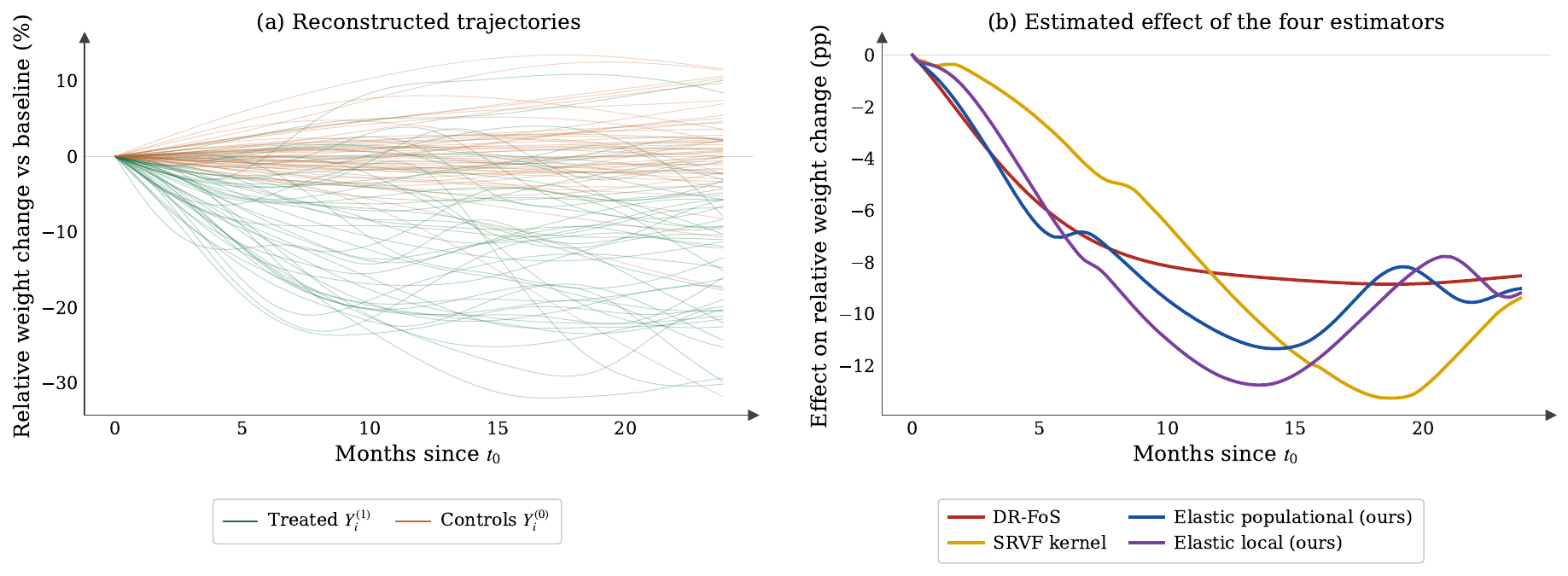}
\caption{\textbf{The elastic estimators recover a loss and a
regain that DR-FoS flattens.} \emph{(a)} Reconstructed relative 
weight change of a sample of $50$ patients per arm. \emph{(b)} Estimated effect on the
treated for each of the four estimators.}
\label{fig:rd_main}
\end{figure} 
\section{Discussion and conclusion} \label{sec:disc}

We defined causal estimands and estimators that capture, non-parametrically,
the shape of a treatment effect on a functional outcome. To do so, we lifted the
phase-centered elastic mean, so far a sample statistic, to a functional of a
probability measure on curves, and proved that it exists and is unique under
explicit geometric assumptions. Built on this mean, our estimands keep the
depth, timing and shape of the individual trajectories, instead of smearing them
out when the same features occur at different times across individuals. Our two identifiable estimands answer complementary
questions: the populational effect describes each arm separately, through the 
typical treated and control trajectories, and the local effect compares treated and control units that share the same covariates. 

In
simulations they recover the shape of the effect where competing methods do
not, and on real weight trajectories after GLP-1 receptor agonist initiation
they recover a rebound pattern that the pointwise estimator smooths away. Aligning curves has a price: the elastic mean can read part of the noise as
shape, which biases it (Section~\ref{sec:noise}). In practice, we therefore
recommend smoothing the curves, guided by domain knowledge, before applying
elastic methods.

These
results are promising for the real-world evidence that high frequency
data from connected devices,
wearables and health applications now make possible. Beyond the average level
of an effect, this evidence can then tell clinicians when the effect starts,
how deep it goes, how fast it settles, and whether it holds or gives way to a
rebound in the long run---paving the way for a richer form of real-world evidence.

This work opens several directions. First, we have not studied the statistical properties of our estimators: their consistency, their convergence rates and the construction of confidence bands remain open. These questions are not standard here, because the empirical Fréchet means may be non-unique and the alternating algorithm may stop at a local minimiser (Remark~\ref{rem:optim}). Our estimands summarise the effect by a
single typical shape, whereas patients may respond with different shapes, so
capturing the heterogeneity of treatment-effect shapes is one possible
extension. Besides, the causal assumptions
(Assumptions~\ref{ass:sutva}--\ref{ass:exch}) are strong in this functional
context as they must hold for whole trajectories rather than for a single
scalar. The development of suitable sensitivity analyses for these assumptions, in particular for unmeasured confounding, would
help use the methods with confidence. 

\section*{Acknowledgments}
This work was funded by Withings and by the Association Nationale de la
Recherche et de la Technologie (ANRT) through a CIFRE convention. We thank the
Withings users who contributed their data to this research. We also thank
Farès Abdelli, Pierre Binon, Aline Criton, Ariane Cwiling, Mathieu Even, Nicolas Nghe and Benjamin Vittrant for insightful discussions.

\subsection*{Author contributions}
A.C. conceived the methodology, developed the theory, implemented the methods, carried out the simulations, supervised the real-data analysis and wrote the manuscript. J.O. prepared the real data and applied the methodology to them. J.J. and E.S. supervised the work and validated the methodology. All authors read and approved the final manuscript.

\subsection*{Competing interests}
A.C.\ is a doctoral candidate employed by Withings under a CIFRE contract funded by Withings and
the ANRT, in a collaboration between Withings, Inria and Sorbonne Universit\'e. Withings manufactures the connected devices whose data are analysed in Section~\ref{sec:realdata}. J.O.\ was an intern at Withings when this work was carried
out. J.J.\ and E.S.\ declare no
competing interests.

\subsection*{Ethics statement}
The real-data application is a retrospective observational analysis of
de-identified data. The study was conducted in adherence with the General Data
Protection Regulation (GDPR) and the Health Insurance Portability and
Accountability Act (HIPAA) security standards. The data linkage process was
the subject of an expert determination under the HIPAA Privacy Rule
(45~CFR~\S164.514(b)(1)), which attested that the risk is very small that the
records could be used, alone or in combination with other reasonably
available information, to identify any individual. The analysis dataset is
therefore de-identified health information and not Protected Health
Information (PHI), so its use requires no individual HIPAA authorization.
Its secondary research use falls under the exemption for secondary research
with non-identifiable information (45~CFR~\S46.104(d)(4)) and did not require
review by an institutional review board (IRB).

\subsection*{Data and code availability}
Individual participant data, including the linked insurance claims and the
high-frequency device measurements, will not be made available owing to GDPR
regulations and commercial confidentiality. The code implementing the methods
and reproducing the simulations will be made available in a later version of this preprint.

\appendix
\bigskip\hrule
\section*{Annex: assumptions in detail, fully detailed proofs, and auxiliary results}
\addcontentsline{toc}{section}{Annex}
\setcounter{theorem}{0}\renewcommand{\thetheorem}{A.\arabic{theorem}}\renewcommand{\theHtheorem}{A.\arabic{theorem}}%

This annex gives additional detail on the geometric assumptions and collects all the proofs of the main text.
The standard results we use are stated in the toolbox of Appendix~\ref{app:toolbox} and are cited as
(T1), (T2), \dots{}, (T14) throughout. 

\section{A toolbox of standard facts} \label{app:toolbox}

\paragraph{Notations.} The following objects appear throughout the annex.

\begin{itemize}[leftmargin=1.6em,itemsep=1pt]
\item $\F=\{f:[0,T]\to\R \text{ absolutely continuous}\}$: outcome curves. Absolute
continuity means exactly that $f(t)=f(0)+\int_0^t \dot f(s)\,ds$ for an integrable
derivative $\dot f\in L^1[0,T]$ (fact~\ref{T:ftc} below). We equip $\F$ with the
\emph{total-variation metric} $d_{\F}(f_1,f_2)=|f_1(0)-f_2(0)|+\int_0^T|\dot f_1-\dot f_2|$.
\item $\Lt=L^2([0,T],\R)$: the Hilbert space of square-integrable functions, with
inner product $\inner{u}{v}=\int_0^T u(t)v(t)\,dt$ and norm $\norm u=\inner{u}{u}^{1/2}$.
\item $q_f=\operatorname{sign}(\dot f)\sqrt{|\dot f|}\in\Lt$ for $f\in\F$: the
square-root slope function of $f$, with $\operatorname{sign}(0):=0$. The SRSF map
is written $f\mapsto q_f$.

\item $\G$: warps, i.e.\ absolutely continuous bijections $\gamma:[0,T]\to[0,T]$ with
$\gamma(0)=0$, $\gamma(T)=T$ and $\dot\gamma>0$ almost everywhere. $\Gt$: the monoid of
\emph{non-decreasing} absolutely continuous surjections of $[0,T]$ ($\overline\Gamma$ in \citealp{bruveris_optimal_2016}, and also the closure of $\G$ among absolutely continuous functions); its elements may pause (stay constant on a time interval), and then cannot be inverted. $\Gc\subseteq\Gt$:
the compact, uniformly bi-Lipschitz working family of Assumption~\ref{ass:warp}
($0<c\le\dot\gamma\le C$).
\item $(q,\gamma)=(q\circ\gamma)\sqrt{\dot\gamma}$: the warp action on $\Lt$;
$[q]$: the shape of $q$, that is, the closure in $\Lt$ of its $\G$-orbit $\{(q,\gamma):\gamma\in\G\}$, whose elements are the \emph{representatives} of the shape; $(\Sw,\del)$, with $\Sw=\Lt/\!\sim$: the shape space with the elastic distance (Section~\ref{sec:shape}).
\item Norms: an unsubscripted $\norm{\cdot}$ always denotes the $\Lt$ norm
$\norm{\cdot}_2$; $\norm{\cdot}_\infty$ is the sup norm on $[0,T]$.
\item $\Sphere(\sqrt T)=\{\psi\in\Lt:\norm{\psi}=\sqrt T\}$: the Hilbert sphere of
radius $\sqrt T$; $\psi_\gamma=\sqrt{\dot\gamma}$: the warp-SRSF of $\gamma$;
$\Psi(\Gt)=\{\psi_\gamma:\gamma\in\Gt\}$: the feasible set of Step~2;
$\dG(\psi_1,\psi_2)=\sqrt T\arccos(\inner{\psi_1}{\psi_2}/T)$: the great-circle
(geodesic) distance.
\item $h(q;y)=\min_{\gamma\in\Gc}\norm{q-(y,\gamma)}^2$: the registration residual of a
data transform $y$ against a candidate $q$; $M([q])=\E[\del^2([q],[q_Z])]$: the Step-1 objective; $G(q)=\E[h(q;q_Z)]$: its lift to $\Lt$
(the lifted Step-1 objective); $N(\psi)=\E[\dG^2(\psi,\psi^\star_{q_Z\to\tilde q})]$: the Step-2 (phase) objective, where
$\psi^\star_{q_Z\to\tilde q}$ is the SRSF of the random optimal warp registering $q_Z$ to
the reference representative $\tilde q$.
\item $(X,A,Y)$: covariates, binary treatment, observed outcome curve $Y=Y^{(A)}$;
$e_a(X)=\Prob(A{=}a\mid X)$: the propensity; $W=\ind_{A=a}/e_a(X)$: the IPW weight;
$\pto$: convergence in probability.
\end{itemize}

\paragraph{Statements used in the proofs.} Throughout the proofs, we use the following known results.

\begin{enumerate}[label=\textup{(T\arabic*)}, ref=T\arabic*, leftmargin=3.4em, itemsep=4pt]

\item \label{T:ftc}\textbf{Absolute continuity and the fundamental theorem of
calculus.} $f:[0,T]\to\R$ is absolutely continuous iff there is $g\in L^1[0,T]$ with
$f(t)=f(0)+\int_0^t g(s)\,ds$ for all $t$; then $f$ is differentiable a.e.\ with
$\dot f=g$ a.e. \citep[Ch.~3]{leoni_first_2017}.

\item \label{T:cov}\textbf{Chain rule and change of variables for monotone AC maps.}
Let $\gamma:[0,T]\to[0,T]$ be non-decreasing, absolutely continuous and surjective
(i.e.\ $\gamma\in\Gt$). \emph{(a)}~For every $g\in L^1([0,T])$,
\[
\int_0^T g(\gamma(t))\,\dot\gamma(t)\,dt=\int_0^T g(s)\,ds ,
\]
with the convention that the integrand on the left is $0$ wherever $\dot\gamma(t)=0$
(regardless of whether $g(\gamma(t))$ is defined there); no injectivity of $\gamma$ is
required. \emph{(b)}~If $f$ is absolutely continuous then so is $f\circ\gamma$, and
$(f\circ\gamma)^{\textstyle\cdot}=(\dot f\circ\gamma)\,\dot\gamma$ a.e., with the same
convention on $\{\dot\gamma=0\}$. The difficulty these results handle is that
$\dot f(\gamma(t))$ may be undefined on a set where $\dot\gamma=0$. The convention
on $\{\dot\gamma=0\}$ makes both sides well defined and the identities true
\citep{serrin_varberg_1969,leoni_first_2017}.

\item \label{T:weier}\textbf{Weierstrass extreme-value theorem.} A continuous
real-valued function on a (sequentially) compact metric space attains its minimum and
maximum \citep[Thm.~4.16]{rudin_principles_1976}.

\item \label{T:aa}\textbf{Arzel\`a--Ascoli theorem.} A family $S\subset C^0([0,T])$ is
\emph{relatively} compact for the sup-norm iff it is uniformly bounded and
equicontinuous; it is compact iff, in addition, it is closed
\citep[Thms.~7.24--7.25]{rudin_principles_1976}.

\item \label{T:heine}\textbf{Heine--Cantor theorem.} A continuous function on a compact metric
space is uniformly continuous \citep[Thm.~4.19]{rudin_principles_1976}.

\item \label{T:dense}\textbf{Density.} $C([0,T])$ is dense in $\Lt$: for every
$y\in\Lt$ and $\varepsilon>0$ there is a continuous $g$ with $\norm{y-g}<\varepsilon$
\citep[Cor.~4.23]{brezis_functional_2011}.

\item \label{T:weakcpt}\textbf{Weak sequential compactness.} Every bounded sequence in
a Hilbert space (more generally, a reflexive Banach space) has a weakly convergent
subsequence \citep[Thm.~3.18]{brezis_functional_2011}. ($q_n\rightharpoonup q$ means
$\inner{q_n}{v}\to\inner{q}{v}$ for every $v$.)

\item \label{T:weaklsc}\textbf{Weak lower semicontinuity of the norm; weak--strong
pairing.} If $q_n\rightharpoonup q$ then $\norm q\le\liminf_n\norm{q_n}$; if moreover
$v_n\to v$ strongly, then $q_n+v_n\rightharpoonup q+v$ and
$\inner{q_n}{v_n}\to\inner{q}{v}$ \citep[Prop.~3.5]{brezis_functional_2011}.

\item \label{T:fatou}\textbf{Fatou's lemma.} If $h_n\ge0$ are measurable then
$\E[\liminf_n h_n]\le\liminf_n\E[h_n]$ \citep[Lem.~4.1]{brezis_functional_2011}.

\item \label{T:dct}\textbf{Dominated convergence theorem.} If $h_n\to h$ pointwise a.e.\ and
$|h_n|\le g$ with $\E g<\infty$, then $\E h_n\to\E h$
\citep[Thm.~4.2]{brezis_functional_2011}.

\item \label{T:tower}\textbf{Conditional expectation: tower rule, taking out what is
known, and the propensity identity.} For integrable $V$:
$\E[V]=\E\big[\E(V\mid X)\big]$ and $\E[g(X)V\mid X]=g(X)\,\E[V\mid X]$ for
$X$-measurable $g$. Consequently, for $A\in\{0,1\}$, integrable $V$, and
$e_a(X)=\Prob(A{=}a\mid X)$,
\[
\E\big[\ind_{A=a}\,V\bigm|X\big]
=e_a(X)\,\E\big[V\bigm|A{=}a,X\big]
\quad\text{a.s.\ on }\{e_a(X)>0\},
\]
which is the identity behind Horvitz--Thompson/IPW reweighting
\citep{horvitz_generalization_1952}; \citep{imbens_rubin_2015}.

\item \label{T:krn}\textbf{Measurable selection (Kuratowski--Ryll-Nardzewski theorem).} If
$\mathcal K$ is a compact metric space, $(\omega,\kappa)\mapsto\Phi(\omega,\kappa)$ is
measurable in $\omega$ and continuous in $\kappa$, and the minimiser
$\kappa^\star(\omega)=\argmin_\kappa\Phi(\omega,\kappa)$ is unique for every $\omega$ outside a null set (on which $\kappa^\star$ is set to a fixed $\kappa_0\in\mathcal K$; measurability is then for the completed $\sigma$-field),
then $\omega\mapsto\kappa^\star(\omega)$ is measurable
(see \citealp[Thm.~18.19]{aliprantis_infinite_2006}).

\item \label{T:midpt}\textbf{Second-derivative criterion and midpoint inequality.} Let
$g:[0,L]\to\R$ be continuous, and twice differentiable with $g''\ge\lambda$ at every
interior point. Then
\[
g\big(\tfrac L2\big)\ \le\ \tfrac12g(0)+\tfrac12g(L)-\tfrac{\lambda}{8}\,L^2 .
\]
(Proof: $h(s)=g(s)-\tfrac\lambda2 s^2$ has $h''\ge0$, hence is convex, hence
mid-point convex: $h(L/2)\le\tfrac12h(0)+\tfrac12h(L)$; rearranging and using
$\tfrac12\cdot0^2+\tfrac12L^2-(L/2)^2=L^2/4$ gives the display, since
$\tfrac{\lambda}{2}\cdot\tfrac{L^2}{4}=\tfrac{\lambda}{8}L^2$.)

\item \label{T:isom}\textbf{Isometry equivariance of Fréchet means.} Let
$(\mathcal M,d)$ be a metric space, $C\subseteq\mathcal M$, $R:\mathcal M\to\mathcal M$
a bijective isometry with $R(C)=C$, and $\nu$ a law on $\mathcal M$ with
$F_\nu(\psi)=\int d^2(\psi,y)\,d\nu(y)$. Then $F_{R\#\nu}(R\psi)=F_\nu(\psi)$ for all
$\psi$, where $R\#\nu$ is the image law of $\nu$ under $R$. Hence $\bar\psi$ is the
unique minimiser of $F_\nu$ over $C$ iff $R\bar\psi$ is the unique minimiser of
$F_{R\#\nu}$ over $C$. (Proof: $F_{R\#\nu}(R\psi)=\int d^2(R\psi,Ry)\,d\nu(y)
=\int d^2(\psi,y)\,d\nu(y)$ by the isometry property and the image-measure
substitution; minimising over $R(C)=C$ transports minimisers by $R$. See
\citealp{sturm_probability_2003}.)

\end{enumerate}
\section{Two points of rigour behind Assumptions~\ref{ass:warp} and~\ref{ass:conc}} \label{app:assump}

This section gives the material behind the two most technical assumptions. For
Assumption~\ref{ass:warp}, we show how to construct a family satisfying the
compactness required by clause~(i) (Appendix~\ref{app:aa}), and we give an example
of what the bi-Lipschitz restriction of clause~(ii) excludes
(Appendix~\ref{app:exline}). For Assumption~\ref{ass:conc}(4b), we describe the
geometry of the warp sphere, where the Step-2 mean lives
(Appendix~\ref{app:sphere}), and we prove the two convexity lemmas used in
Appendix~\ref{app:exist} (Appendix~\ref{app:convexity}). We also explain in simple terms what curves with flat pieces change, and what clause~(iii) of Assumption~\ref{ass:warp} requires of the data (Appendix~\ref{app:flat}).

\subsection{How to get compactness of $\Psi$ concretely} \label{app:aa}

Assumption~\ref{ass:warp}(i) requires the set of warp-SRSFs
$\Psi=\{\sqrt{\dot\gamma}:\gamma\in\Gc\}$ to be compact in
$(L^\infty([0,T]),\norm{\cdot}_\infty)$. For a family of continuous warp-SRSFs, as below, this is compactness in $(C^0([0,T]),\norm{\cdot}_\infty)$. Compactness of a family of continuous
functions is checked with the Arzel\`a--Ascoli theorem (\ref{T:aa}): $\Psi$ is
compact if and only if it is (a) uniformly bounded, (b) equicontinuous and
(c) closed. We go through these three conditions and check what the bi-Lipschitz
bounds of Assumption~\ref{ass:warp}(ii) give.

\paragraph{(a) Uniform boundedness.} If $\dot\gamma\le C$ then
$\norm{\sqrt{\dot\gamma}}_\infty\le\sqrt C$.

\paragraph{(b) Equicontinuity can fail.} Bounded derivatives do not rule out
faster and faster oscillations. Take $T$ arbitrary and define, for
$n\in\mathbb N$,
\[
\dot\gamma_n(t)\ =\ 1+\tfrac12\sin\!\Big(\tfrac{2\pi n t}{T}\Big),
\qquad
\gamma_n(t)=\int_0^t\dot\gamma_n(s)\,ds .
\]
Each $\gamma_n$ is a valid warp: $\dot\gamma_n\in[\tfrac12,\tfrac32]$, $\gamma_n(0)=0$, and
$\gamma_n(T)=\int_0^T\dot\gamma_n=T$.
So every $\gamma_n$ satisfies the bi-Lipschitz bounds of Assumption~\ref{ass:warp}(ii) with the same
constants $c=\tfrac12,\ C=\tfrac32$, and $\{\psi_n\}$ ($\psi_n=\sqrt{\dot\gamma_n}$) is uniformly
bounded. It is however not relatively compact. In
$(C^0([0,T]),\norm{\cdot}_\infty)$, relative compactness means that every sequence has a
uniformly convergent subsequence, so it suffices to show that $(\psi_n)$ has none. A uniformly convergent subsequence $\psi_{n_k}$ would make $\dot\gamma_{n_k}=\psi_{n_k}^2$ converge uniformly, hence equicontinuous. But equicontinuity
fails: at $t_n=\tfrac{T}{4n}$ and $s_n=\tfrac{3T}{4n}$ we have $|t_n-s_n|=\tfrac{T}{2n}\to0$ while
$|\dot\gamma_n(t_n)-\dot\gamma_n(s_n)|=1$ for every $n$. Hence
no uniformly convergent subsequence exists. Bounded derivatives give uniform boundedness but not
equicontinuity, and Arzel\`a--Ascoli needs both.

Figure~\ref{fig:aacounter} illustrates the counterexample.

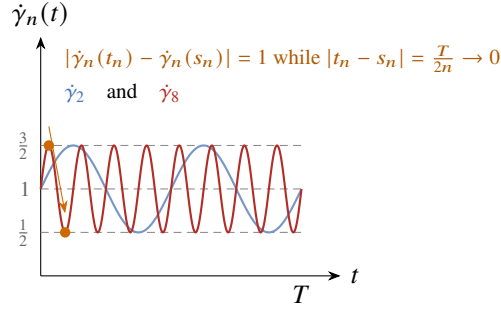
\begin{figure}[htbp]
\centering
\begin{tikzpicture}[>=Stealth, scale=1.15]
  \draw[->] (0,0) -- (3.45,0) node[right, font=\small] {$t$};
  \draw[->] (0,0) -- (0,2.75) node[above, font=\small] {$\dot\gamma_n(t)$};
  \node[below, font=\small] at (3,0) {$T$};
  \draw[black!45, densely dashed] (0,0.5) -- (3,0.5); \node[left, font=\scriptsize, black!60] at (0,0.5) {$\tfrac12$};
  \draw[black!45, densely dashed] (0,1.5) -- (3,1.5); \node[left, font=\scriptsize, black!60] at (0,1.5) {$\tfrac32$};
  \draw[black!45, densely dashed] (0,1) -- (3,1);
  \node[left, font=\scriptsize, black!60] at (0,1) {$1$};
  \draw[figblue!60, thick] plot[domain=0:3, samples=200] (\x,{1+0.5*sin(240*\x)});
  \draw[thick, figred] plot[domain=0:3, samples=500] (\x,{1+0.5*sin(960*\x)});
  \fill[figorange] (0.09375,1.5) circle (1.7pt);
  \fill[figorange] (0.28125,0.5) circle (1.7pt);
  \draw[->, figorange] (0.09375,1.72) -- (0.28125,0.72);
  \node[font=\scriptsize, anchor=west, figorange!85!black] at (0.15,2.5) {$|\dot\gamma_n(t_n)-\dot\gamma_n(s_n)|=1$ while $|t_n-s_n|=\tfrac{T}{2n}\to0$};
  \node[font=\scriptsize, anchor=west] at (0.15,2.1) {{\color{figblue!80}$\dot\gamma_2$}\quad and\quad {\color{figred}$\dot\gamma_8$}};
\end{tikzpicture}
\caption{\textbf{Bounded derivative bands do not give equicontinuity.} The family $\dot\gamma_n(t)=1+\tfrac12\sin(2\pi nt/T)$, drawn for $n=2$ (blue) and $n=8$ (red), stays in the band $[\tfrac12,\tfrac32]$; within every window of width $T/(2n)$ it travels the full gap $1$ (orange dots).}
\label{fig:aacounter}
\end{figure}

\paragraph{A sufficient condition for equicontinuity.} The simplest sufficient
condition is a uniform bound on the second derivative. If $|\ddot\gamma|\le C_2$
and $\dot\gamma\ge c>0$, then wherever the derivative exists,
\[
\Big|\tfrac{d}{dt}\sqrt{\dot\gamma(t)}\Big|
=\frac{|\ddot\gamma(t)|}{2\sqrt{\dot\gamma(t)}}
\le\frac{C_2}{2\sqrt c},
\]
so all the $\psi=\sqrt{\dot\gamma}$ share the Lipschitz constant
$C_2/(2\sqrt c)$, and a family with a common Lipschitz constant is equicontinuous.

\paragraph{(c) Closedness must be imposed.} Uniform boundedness and equicontinuity
only give relative compactness (\ref{T:aa}): the family must also contain its
uniform limits. This holds when $\Gc$ is defined through closed constraints. For
instance, the family
\[
\Psi_{c,C,C_2}
=\Big\{\psi\in C^0([0,T])\ :\ \sqrt c\le\psi\le\sqrt C,\
\operatorname{Lip}(\psi)\le\tfrac{C_2}{2\sqrt c},\ \int_0^T\psi^2=T\Big\}
\]
is uniformly bounded and equicontinuous by construction. It is also closed under
uniform convergence: the two-sided bounds and the Lipschitz constant pass to
pointwise limits, and the integral constraint passes to uniform limits. It is
therefore compact by Arzel\`a--Ascoli (\ref{T:aa}). Every $\psi\in\Psi_{c,C,C_2}$ is the
warp-SRSF of the warp $\gamma(t)=\int_0^t\psi^2$, which satisfies
$c\le\dot\gamma\le C$, so the corresponding $\Gc$ is a concrete family satisfying
Assumption~\ref{ass:warp}(i)--(ii). Note that $\sqrt{\dot\gamma}=\psi\in C^0$
forces $\dot\gamma=\psi^2$ to be continuous, so every $\gamma\in\Gc$ is
automatically $C^1$.

\subsection{What the bi-Lipschitz restriction excludes: an example} \label{app:exline}

The following example shows why the group quotient $\Lt/\G$ only carries a
pseudo-metric, why the shape space treats curves that differ only by a flat piece as one shape, and which warps Assumption~\ref{ass:warp} rules out.

\begin{example}[A line and a flat-then-line]\label{ex:line}
Take $T=1$ and consider the two curves
\[
f_2(t)=t
\qquad\text{and}\qquad
f_1(t)=\begin{cases}0,&t\in[0,\tfrac12],\\ 2t-1,&t\in[\tfrac12,1].\end{cases}
\]
Their SRSFs (positive slopes, so
$q=\operatorname{sign}(\dot f)\sqrt{|\dot f|}=\sqrt{\dot f}$) are
\[
q_2\equiv1,
\qquad
q_1=\begin{cases}0,&t\in[0,\tfrac12),\\ \sqrt2,&t\in(\tfrac12,1].\end{cases}
\]

\smallskip
\noindent\emph{(a) The two curves lie in distinct $\G$-orbits.}
The $\G$-orbit of $q_2$ is
$\{(q_2,\gamma)=1\cdot\sqrt{\dot\gamma}=\sqrt{\dot\gamma}:\gamma\in\G\}$. Every
$\gamma\in\G$ has $\dot\gamma>0$ almost everywhere, so every member of this orbit is
$>0$ almost everywhere. But $q_1$ vanishes on $[0,\tfrac12)$, a set of positive
measure, so $q_1$ is not a member: $q_1$ and $q_2$ lie in different $\G$-orbits, so they are different points of $\Lt/\G$.

\smallskip
\noindent\emph{(b) Yet their elastic distance is zero.}
For $\varepsilon\in(0,2)$, define the piecewise-linear warp
\[
\dot\gamma_\varepsilon=\begin{cases}\varepsilon,&t\in[0,\tfrac12),\\
2-\varepsilon,&t\in(\tfrac12,1],\end{cases}
\qquad
\gamma_\varepsilon(t)=\int_0^t\dot\gamma_\varepsilon .
\]
Each $\gamma_\varepsilon$ is a legitimate element of $\G$ as defined in
Section~\ref{sec:fda}. Now warp $q_2$ by $\gamma_\varepsilon$:
\[
(q_2,\gamma_\varepsilon)=(q_2\circ\gamma_\varepsilon)\sqrt{\dot\gamma_\varepsilon}
=\sqrt{\dot\gamma_\varepsilon}
=\begin{cases}\sqrt\varepsilon,&t\in[0,\tfrac12),\\
\sqrt{2-\varepsilon},&t\in(\tfrac12,1],\end{cases}
\]
because $q_2\equiv1$ is unaffected by composition. Hence
\begin{flalign*}
\norm{q_1-(q_2,\gamma_\varepsilon)}^2
&=\int_0^{1/2}\big|0-\sqrt\varepsilon\big|^2
+\int_{1/2}^1\big|\sqrt2-\sqrt{2-\varepsilon}\big|^2
\why{split the integral at $\tfrac12$}\\
&=\tfrac12\,\varepsilon+\tfrac12\big(\sqrt2-\sqrt{2-\varepsilon}\big)^2 .
\why{integrands are constant on each half}
\end{flalign*}
The second term is $O(\varepsilon^2)$: by the conjugate trick,
$\sqrt2-\sqrt{2-\varepsilon}
=\varepsilon/(\sqrt2+\sqrt{2-\varepsilon})\le\varepsilon/\sqrt2$, so its square is
$\le\varepsilon^2/2$. Hence
\[
\norm{q_1-(q_2,\gamma_\varepsilon)}^2\le\tfrac12\varepsilon+\tfrac12\varepsilon^2
\ \xrightarrow[\varepsilon\downarrow0]{}\ 0 ,
\]
and therefore $\del([q_1],[q_2])=\inf_{\gamma}\norm{q_1-(q_2,\gamma)}=0$ although
the two $\G$-orbits differ: on $\Lt/\G$, the elastic ``distance'' does not separate these two
orbits, so it is only a pseudo-metric.

\smallskip
\noindent\emph{(c) The infimum is attained, but only in the monoid $\Gt$.}
The limit of the $\gamma_\varepsilon$ is the map $\gamma^\star_{q_2\to q_1}=f_1$ itself (it is the optimal warp registering $q_2$ to $q_1$, as shown below):
$\dot\gamma^\star_{q_2\to q_1}=0$ on $[0,\tfrac12)$ and $=2$ on $(\tfrac12,1]$, which is
non-decreasing, absolutely continuous, and surjective onto $[0,1]$
($\gamma^\star_{q_2\to q_1}(1)=1$), hence $\gamma^\star_{q_2\to q_1}\in\Gt$. But it is not invertible
(constant on $[0,\tfrac12]$), hence $\gamma^\star_{q_2\to q_1}\notin\G$. For this warp,
\[
(q_2,\gamma^\star_{q_2\to q_1})=\psi^\star_{q_2\to q_1}
=\begin{cases}0,&t\in[0,\tfrac12),\\ \sqrt2,&t\in(\tfrac12,1],\end{cases}
\ =\ q_1\quad\text{exactly},
\]
so the residual is $0$ and the infimum is attained in $\Gt$. 
In the shape space $\Sw=\Lt/\!\sim$ the two curves have the same shape, since both are re-timings of the line $f_2$: $f_1=f_2\circ\gamma^\star_{q_2\to q_1}$ and $f_2=f_2\circ\gid$. The reverse registration is not attained: for every $\gamma\in\Gt$, $\gamma(t)<\tfrac12$ near $t=0$, so $(q_1,\gamma)$ vanishes there and differs from $q_2$, which is only a limit of the $(q_1,\gamma)$ with $\gamma\in\G$. So $\{(q_1,\gamma):\gamma\in\Gt\}$ is strictly smaller than the orbit closure $[q_1]$. 

\smallskip
\noindent\emph{(d) What the bi-Lipschitz lower bound changes.}
Under Assumption~\ref{ass:warp}(ii), $\dot\gamma\ge c>0$, the degenerate warps are
excluded: every $\gamma_\varepsilon$ with $\varepsilon<c$ is ruled out. The
alignment can then no longer erase the flat piece: for any admissible $\gamma$ 
(with $\dot\gamma\ge c$),
\[
\norm{q_1-(q_2,\gamma)}^2
\ \ge\ \int_0^{1/2}\big|0-\sqrt{\dot\gamma}\big|^2
=\int_0^{1/2}\dot\gamma\ \ge\ \frac c2\ >\ 0 ,
\]
so the residual computed over the restricted family is at least $\sqrt{c/2}$, while the elastic distance is $0$. So, with $q_1$ as the fixed curve, the best re-timing within $\Gc$ does not attain the elastic distance: restricting the search to $\Gc$ overestimates it. Clause (iii) of Assumption~\ref{ass:warp} asks precisely that the search within $\Gc$ attains the elastic distance at the representatives that the construction uses, and therefore excludes laws whose registrations need such pausing warps. 
\end{example}

\begin{remark}[One-sided registration can fail.]\label{rem:onesided}
In Example~\ref{ex:line}, re-timing the line $f_2$ towards $f_1$ using a warp $\gamma \in \Gt$ works exactly: pausing at the start reproduces the flat piece of $f_1$. The reverse fails. Re-timing $f_1$ towards the line $f_2$ can shorten its flat start as much as we like, but never remove it, so the mismatch gets as close to $0$ as we want without reaching it: the best re-timing of $f_1$ towards $f_2$ does not exist, although the elastic distance is $0$. Smoothness does not help, the flat piece is the cause: the same happens with the $C^1$ curves $f_1(t)=\max(0,t-\tfrac12)^2$ and $f_2(t)=t^2/4$. 
\end{remark}

Figure~\ref{fig:exampleA1} illustrates the example.

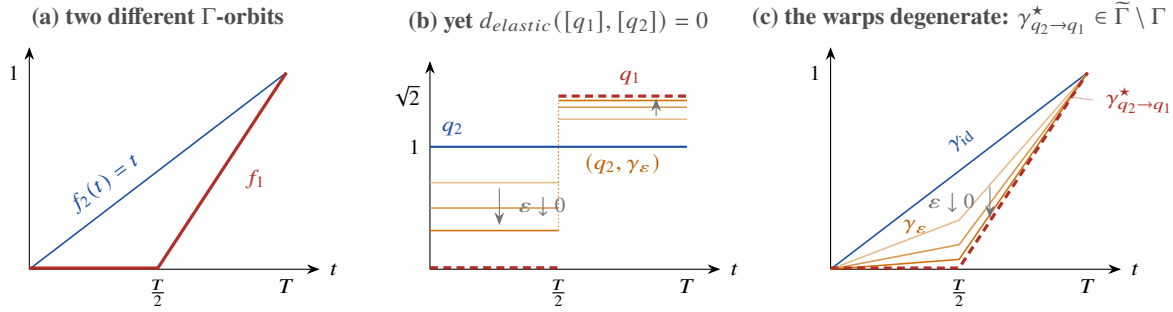
\begin{figure}[htbp]
\centering
\begin{tikzpicture}[>=Stealth]
\begin{scope}[xscale=3.4, yscale=2.6]
  \node[font=\footnotesize, black!70] at (0.5,1.28) {\textbf{(a) two different $\G$-orbits}};
  \draw[->] (0,0) -- (1.13,0) node[right, font=\scriptsize] {$t$};
  \draw[->] (0,0) -- (0,1.13);
  \draw[black!40, densely dotted] (0.5,0) -- (0.5,0.02);
  \node[below, font=\scriptsize] at (0.5,0) {$\tfrac T2$};
  \node[below, font=\scriptsize] at (1,0) {$T$};
  \node[left, font=\scriptsize] at (0,1) {$1$};
  \draw[semithick, figblue] (0,0) -- (1,1);
  \node[font=\scriptsize, figblue, rotate=37, anchor=south] at (0.34,0.36) {$f_2(t)=t$};
  \draw[very thick, figred] (0,0.006) -- (0.5,0.006) -- (1,1);
  \node[font=\scriptsize, figred, anchor=north west] at (0.80,0.56) {$f_1$};
\end{scope}
\begin{scope}[xshift=5.3cm, xscale=3.4, yscale=1.62]
  \node[font=\footnotesize, black!70] at (0.5,2.02) {\textbf{(b) yet $\del([q_1],[q_2])=0$}};
  \draw[->] (0,0) -- (1.13,0) node[right, font=\scriptsize] {$t$};
  \draw[->] (0,0) -- (0,1.70);
  \node[below, font=\scriptsize] at (0.5,0) {$\tfrac T2$};
  \node[below, font=\scriptsize] at (1,0) {$T$};
  \node[left, font=\scriptsize] at (0,1) {$1$};
  \node[left, font=\scriptsize] at (0,1.414) {$\sqrt2$};
  \foreach \lo/\hi/\op in {0.707/1.225/45, 0.5/1.323/70, 0.316/1.378/95}{
    \draw[figorange!\op, semithick] (0,\lo) -- (0.5,\lo);
    \draw[figorange!\op, semithick] (0.5,\hi) -- (1,\hi);
    \draw[figorange!\op, densely dotted, thin] (0.5,\lo) -- (0.5,\hi);}
  \draw[thick, figblue] (0,1) -- (1,1);
  \draw[very thick, densely dashed, figred] (0,0.012) -- (0.5,0.012);
  \draw[very thick, densely dashed, figred] (0.5,1.414) -- (1,1.414);
  \node[font=\scriptsize, figblue, anchor=south east] at (0.17,1.01) {$q_2$};
  \node[font=\scriptsize, figred, anchor=south] at (0.78,1.43) {$q_1$};
  \node[font=\scriptsize, figorange!90!black] at (0.75,0.86) {$(q_2,\gamma_\varepsilon)$};
  \draw[->, black!55] (0.27,0.66) -- (0.27,0.36);
  \node[font=\scriptsize, black!55, anchor=west] at (0.30,0.51) {$\varepsilon\downarrow0$};
  \draw[->, black!55] (0.88,1.25) -- (0.88,1.395);
\end{scope}
\begin{scope}[xshift=10.6cm, xscale=3.4, yscale=2.6]
  \node[font=\footnotesize, black!70] at (0.5,1.28) {\textbf{(c) the warps degenerate: $\gamma^\star_{q_2\to q_1}\in\Gt\setminus\G$}};
  \draw[->] (0,0) -- (1.13,0) node[right, font=\scriptsize] {$t$};
  \draw[->] (0,0) -- (0,1.13);
  \node[below, font=\scriptsize] at (0.5,0) {$\tfrac T2$};
  \node[below, font=\scriptsize] at (1,0) {$T$};
  \node[left, font=\scriptsize] at (0,1) {$1$};
  \draw[semithick, figblue] (0,0) -- (1,1);
  \node[font=\scriptsize, figblue, rotate=37, anchor=south] at (0.55,0.585) {$\gid$};
  \foreach \m/\op in {0.25/45, 0.125/70, 0.05/95}{
    \draw[figorange!\op, semithick] (0,0) -- (0.5,\m) -- (1,1);}
  \node[font=\scriptsize, figorange!90!black] at (0.33,0.215) {$\gamma_\varepsilon$};
  \draw[very thick, densely dashed, figred] (0,0.006) -- (0.5,0.006) -- (1,1);
  \node[font=\scriptsize, figred, anchor=west] at (1.03,0.86) {$\gamma^\star_{q_2\to q_1}$};
  \draw[figred!40] (1.02,0.86) -- (0.94,0.876);
  \draw[->, black!55] (0.62,0.42) -- (0.62,0.26);
  \node[font=\scriptsize, black!55, anchor=east] at (0.60,0.34) {$\varepsilon\downarrow0$};
\end{scope}
\end{tikzpicture}
\caption{\textbf{Example~\ref{ex:line}: zero elastic distance between two different $\G$-orbits.}
\emph{(a)} The curves $f_2(t)=t$ and the flat-then-line $f_1$.
\emph{(b)} Their SRSFs: $q_1$ (dashed), jumping from $0$ to $\sqrt2$ at $T/2$, and $q_2\equiv1$, with the warped $(q_2,\gamma_\varepsilon)$ for small $\varepsilon$.
\emph{(c)} The aligning warps $\gamma_\varepsilon$ and their non-invertible limit $\gamma^\star_{q_2\to q_1}$ (dashed).}
\label{fig:exampleA1}
\end{figure}

\subsection{\texorpdfstring{Flat pieces and what clause (iii) of Assumption~\ref{ass:warp} requires}{Flat pieces and what clause (iii) of Assumption 2 requires}} \label{app:flat}
Shapes and the elastic distance are defined for all curves (Remark~\ref{rem:notquot}). Flat pieces, intervals on which a curve is constant, only cause trouble in the registrations.

\emph{Non-uniqueness.} If a data curve is constant on an interval, a re-timing can cross that interval at any speed without changing the mismatch, so the best re-timing, which Step~2 averages, is not unique. Clause (iii-b) therefore excludes laws that give positive probability to curves with a flat piece. For the same reason, the reference is taken nowhere flat.

\emph{Moving against the reference.} If the reference goes up while the re-timed data curve goes down, crossing that stretch faster and then pausing reduces the mismatch. For $T=1$, $f_1(t)=t$ and $f_2(t)=-t$, $\norm{q_1-(q_2,\gamma)}^2=2+2\int_0^1\sqrt{\dot\gamma}$ tends to $2$ without reaching it, so registering $f_2$ to $f_1$ has no solution. As the re-timings of $\Gc$ never pause, (iii-b) requires the data curves to rise and fall in the same order as the reference.

\emph{Consequence.} Clause (iii-b) holds when the data curves are re-timings $f\circ\gamma_Z$ of the nowhere-flat reference curve $f$, up to a vertical shift and a positive scale, by invertible warps with $\gamma_Z^{-1}\in\Gc$, but fails for mixtures of patterns or noisy curves, as in our simulations. The computation still runs, but Theorem~\ref{thm:wellposed} no longer guarantees its target.

\subsection{The warp sphere and its feasible set} \label{app:sphere}

We now set up the space on which the Step-2 mean lives.

\paragraph{Why warps live on a sphere.} A warp $\gamma\in\Gt$ is encoded by its
warp-SRSF $\psi_\gamma=\sqrt{\dot\gamma}$. Because $\gamma$ is absolutely continuous
with $\gamma(0)=0$ and $\gamma(T)=T$ (\ref{T:ftc}),
\[
\norm{\psi_\gamma}^2=\int_0^T\big(\sqrt{\dot\gamma}\big)^2
=\int_0^T\dot\gamma=\gamma(T)-\gamma(0)=T:
\]
\emph{every} warp-SRSF has the same $\Lt$-norm $\sqrt T$. The warp-SRSFs therefore all
lie on the sphere $\Sphere(\sqrt T)=\{\psi\in\Lt:\norm\psi=\sqrt T\}$. Distances on a
sphere are measured along geodesics (great circles). For $\psi_1,\psi_2\in\Sphere(\sqrt T)$ the
geodesic distance is
\[
\dG(\psi_1,\psi_2)=\sqrt T\,\arccos\!\Big(\frac{\inner{\psi_1}{\psi_2}}{T}\Big),
\]
which is well defined by the Cauchy--Schwarz inequality. The sphere has
constant sectional curvature $1/T$ (the curvature of a sphere of radius $r$ is
$1/r^2$, here $r=\sqrt T$) and injectivity radius $\pi\sqrt T$, the largest arc
length below which geodesics remain minimizing: a great-circle arc is minimizing
until it reaches the antipode of its starting point, at arc length
$\pi\times(\text{radius})=\pi\sqrt T$. Minimizing geodesics between non-antipodal points
(angle $\vartheta=\dG(\psi_1,\psi_2)/\sqrt T$, $\vartheta<\pi$) are unique and given by the great-circle arc
\begin{equation}\label{eq:slerp}
\psi(s)=\frac{\sin\big((1-s)\vartheta\big)}{\sin\vartheta}\,\psi_1
+\frac{\sin(s\vartheta)}{\sin\vartheta}\,\psi_2,
\qquad s\in[0,1],
\end{equation}
which stays on the sphere and joins $\psi_1$ to $\psi_2$ (at $s=0,1$ the
coefficients are $(1,0)$ and $(0,1)$). This is the classical geometry of the Hilbert
sphere, used for warp functions in \citet{srivastava_klassen_2016}. The two facts
we need about geodesics, convexity of the squared distance along them and convexity
of small balls, are proved below
(Lemmas~\ref{lem:sphereconv} and~\ref{lem:ballconv}).

\paragraph{The feasible set is a quarter of the sphere.} Step~2 involves solving a minimization problem over warps,
that is over the set $\Psi(\Gt)$ of warp-SRSFs (the feasible set), not over the
whole sphere. This set has an explicit description.

\begin{lemma}[The feasible set]\label{lem:feasible}
$\Psi(\Gt)=\{\psi\in\Lt:\ \psi\ge0\ \text{a.e.},\ \norm\psi^2=T\}$, and this set is
closed in $\Lt$. Moreover the map $\gamma\mapsto\psi_\gamma$ is a bijection from
$\Gt$ onto $\Psi(\Gt)$, with inverse $\psi\mapsto\gamma(t)=\int_0^t\psi^2$.
\end{lemma}

\begin{proofenv}[Proof.]
The representation is known (references in the proof). Our only addition is the
closedness check at the end, which the existence argument of
Appendix~\ref{app:exist} uses. We must show three facts: $\Psi(\Gt)$ equals the
non-negative part of the sphere (two inclusions), this set is closed, and
$\gamma\mapsto\psi_\gamma$ is a bijection.

\smallskip\noindent
\emph{($\subseteq$).} If $\gamma\in\Gt$ then $\psi_\gamma=\sqrt{\dot\gamma}\ge0$
pointwise (a square root), and $\norm{\psi_\gamma}^2=T$ was computed above.

\smallskip\noindent
\emph{($\supseteq$).} Let $\psi\ge0$ a.e.\ with $\int_0^T\psi^2=T$, and set
$\gamma(t)=\int_0^t\psi(s)^2\,ds$. Since $\psi^2\in L^1$ (its integral is $T$),
$\gamma$ is absolutely continuous with $\dot\gamma=\psi^2\ge0$ a.e.\ (\ref{T:ftc}).
It is non-decreasing (its derivative is non-negative), with $\gamma(0)=0$ and
$\gamma(T)=T$ by construction. A continuous non-decreasing function from $0$ to
$T$ takes every intermediate value, so $\gamma$ is onto $[0,T]$. Hence $\gamma\in\Gt$
and $\psi_\gamma=\sqrt{\psi^2}=|\psi|=\psi$ a.e.

\smallskip\noindent
\emph{(Closedness.)} Suppose $\psi_n\to\psi$ in $\Lt$ with $\psi_n\ge0$ a.e.\ and
$\norm{\psi_n}^2=T$. Norm convergence gives $\norm\psi^2=T$ (the norm is continuous).
$\Lt$-convergence yields a subsequence converging a.e., and a.e.\ limits of
non-negative functions are non-negative a.e. So $\psi$ is in the set.

\smallskip\noindent
\emph{(Bijection.)} The maps $\Phi:\gamma\mapsto\sqrt{\dot\gamma}$ and
$\Xi:\psi\mapsto\int_0^{\cdot}\psi^2$ are mutually inverse:
\[
  \Xi\big(\Phi(\gamma)\big)=\int_0^{\cdot}\dot\gamma{=}\gamma 
  \quad(\gamma(0)=0),
  \qquad
  \Phi\big(\Xi(\psi)\big)=\sqrt{\psi^2}=\psi\ \text{a.e.}\ (\psi\ge0).
\]
Hence $\Phi$ is a bijection with inverse $\Xi$. \hfill$\square$
\end{proofenv}

\begin{remark}[The feasible set has diameter $\pi\sqrt T/2$.]\label{rem:quarter}
For feasible $\psi_1,\psi_2\ge0$, the inner product
$\inner{\psi_1}{\psi_2}=\int\psi_1\psi_2\ge0$, so the angle between any two
warp-SRSFs is at most $\pi/2$. The feasible set therefore sits inside a quarter of
the sphere, of diameter $\pi\sqrt T/2$, and does not include antipodal warps. In
particular, minimizing geodesics between distinct feasible points exist and are
unique. Assumption~\ref{ass:conc}(4b), which confines the optimal warps to a ball
of radius $\pi\sqrt T/8$, is thus a concentration requirement on a set of total
extent $\pi\sqrt T/2$.
\end{remark}

\begin{remark}[How the bi-Lipschitz family sits inside the feasible set.]\label{rem:feasvsG}
Lemma~\ref{lem:feasible} describes the feasible set of the full monoid $\Gt$, over
which Step~2 is posed. The compact family $\Gc$ of Assumption~\ref{ass:warp}
(constructed in Appendix~\ref{app:aa}) is the smaller set of warps the method
actually uses. For $\gamma\in\Gc$ the bi-Lipschitz bounds give
$\sqrt c\le\psi_\gamma\le\sqrt C$ pointwise, so $\Psi(\Gc)$ is a compact subset of
$\Psi(\Gt)$, bounded away from the warp-SRSFs that vanish on a set of positive
measure (see Example~\ref{ex:line}). The pointwise and compactness statements proved below for $\Psi(\Gt)$ therefore also apply on $\Psi(\Gc)$. Geodesic convexity does not: the great-circle arc between two warp-SRSFs that both equal $\sqrt C$ at some time exceeds $\sqrt C$ there.  The bounds also force concentration: for $\gamma\in\Gc$,
$\inner{\psi_\gamma}{\psi_{\gid}}=\int_0^T\sqrt{\dot\gamma}\ge\sqrt c\,T$, so
$\dG(\gamma,\gid)\le\sqrt T\arccos(\sqrt c\,)$. A tight enough speed band
($\sqrt c>\cos(\pi/8)$, that is $c\gtrsim0.854$) therefore confines all warps in
$\Gc$ to a geodesic ball of radius $\pi\sqrt T/8$ around the identity, which is one
way to obtain the concentration required by Assumption~\ref{ass:conc}(4b).
\end{remark}

\subsection{Two convexity lemmas on the sphere} \label{app:convexity}

The key ingredient of the Step-2 existence and uniqueness theory is one fact: the
squared distance $\psi\mapsto\dG^2(\psi,y)$ is convex along geodesics, as long as
one stays within a quarter-turn ($\pi\sqrt T/2$) of the reference point $y$. On an
abstract manifold, proving this requires Riemannian tools. The warp sphere is explicit, with
great circles in $\Lt$ as geodesics, so we can prove the two statements we need
directly from inner products and one derivative: convexity of $\dG^2(\cdot,y)$
along geodesics within a quarter-turn of $y$ (Lemma~\ref{lem:sphereconv}), and
geodesic convexity of small geodesic balls (Lemma~\ref{lem:ballconv}). The radius
$\pi\sqrt T/8$ of Assumption~\ref{ass:conc}(4b) is chosen so that every distance
arising in the existence proof stays inside this quarter-turn: with
$r_0<\pi\sqrt T/8$ the radius of the ball on which the optimal-warp law
concentrates, the distances in play are bounded by
$3r_0<3\pi\sqrt T/8<\pi\sqrt T/2$ (see Appendix~\ref{app:exist}).

Figure~\ref{fig:sphereconv} shows the objects of the lemma.

\begin{figure}[htbp]
\centering
\begin{tikzpicture}[>=Stealth, scale=1.0]
  \draw[black!45] (0,0) circle (2.6);
  \node[font=\small, black!70, anchor=east] at (-1.95,2.0) {$S_R\subset H$};
  \fill (0,0) circle (1.4pt); \node[font=\scriptsize, anchor=north east] at (-0.03,-0.06) {$0$};
  \draw[->, black!45] (0,0) -- (215:2.6) node[midway, above, sloped, font=\footnotesize] {$R$};
  \draw[black!45] plot[domain=180:360, samples=60] ({2.6*cos(\x)},{0.75*sin(\x)});
  \draw[black!30, densely dashed] plot[domain=0:180, samples=60] ({2.6*cos(\x)},{0.75*sin(\x)});
  \draw[black!35, densely dashed] (0,-2.6) -- (0,2.6);
  \fill[figblue] (0,2.6) circle (1.7pt); \node[above, font=\small, figblue] at (0,2.62) {$y$};
  \draw[figblue!60, densely dashed] plot[domain=180:360, samples=60] ({1.992*cos(\x)},{1.671+0.574*sin(\x)});
  \draw[figblue!30, densely dashed] plot[domain=0:180, samples=60] ({1.992*cos(\x)},{1.671+0.574*sin(\x)});
  \node[font=\scriptsize, figblue!85, anchor=west] at (2.8,2.02) {$\{\dG(\cdot,y)=R\,\rho(s)\}$};
  \draw[figblue!40] (2.77,2.00) -- (1.97,1.70);
  \draw[figblue!35, dashed] plot[domain=180:360, samples=60] ({2.472*cos(\x)},{0.803+0.713*sin(\x)});
  \node[font=\scriptsize, figblue!70, anchor=east] at (-2.56,0.85) {$\{\rho=\rho_{\max}\}$};
  \draw[black!40] plot[domain=0:180, samples=80] ({1.116*sin(\x)},{2.6*cos(\x)});
  \draw[very thick, figblue] plot[domain=0:63.16, samples=40] ({1.116*sin(\x)},{2.6*cos(\x)});
  \node[font=\footnotesize, figblue, anchor=west] at (0.90,2.20) {$R\,\rho(s)$};
  \node[font=\scriptsize, figblue!85, anchor=west] at (2.8,2.52) {$g(s)=\dG^{2}(\psi(s),y)$};
  \draw[figblue!40] (2.77,2.50) -- (0.74,2.34);
  \draw[black!40] (0,0) -- (0.996,1.174);
  \draw[figblue] (0,0.52) arc[start angle=90, end angle=49.7, radius=0.52];
  \node[font=\footnotesize, figblue, fill=white, inner sep=1pt] at (68:0.58) {$\rho(s)$};
  \draw[figorange!60, densely dashed] plot[domain=60:103, samples=40]
    ({0.899*2.6*cos(\x)+0.438*1.51*sin(\x)},{-0.438*2.6*cos(\x)+0.899*1.51*sin(\x)});
  \fill[figorange] (0.996,1.174) circle (1.8pt);
  \node[font=\small, figorange, anchor=north east, fill=white, inner sep=1pt] at (0.97,1.06) {$\psi(s)$};
  \draw[->, very thick, figorange] (0.996,1.174) -- (1.775,0.787);
  \node[font=\footnotesize, figorange, anchor=west] at (2.8,0.95) {$\psi'(s)$: unit speed};
  \draw[figorange!40] (2.77,0.94) -- (1.82,0.80);
  \draw[->, densely dashed, semithick, figred] (0.996,1.174) -- (1.176,0.343);
  \node[font=\scriptsize, figred, anchor=west] at (2.8,0.42) {radial part: signed speed $R\,\rho'(s)$};
  \draw[figred!40] (2.77,0.42) -- (1.22,0.36);
  \draw[->, densely dashed, semithick, figgreen] (0.996,1.174) -- (1.835,1.314);
  \node[font=\scriptsize, figgreen, anchor=west] at (2.8,1.48) {transverse part: speed $\sqrt{1-(R\rho'(s))^{2}}$};
  \draw[figgreen!40] (2.77,1.47) -- (1.87,1.33);
\end{tikzpicture}
\caption{\textbf{The geometry of Lemma~\ref{lem:sphereconv}}, in the lemma's
notation. The point $\psi(s)$ travels at unit speed along a great circle
(orange, dashed) on the sphere $S_R\subset H$; the geodesic to the fixed
target $y$ (bold blue) runs along a meridian, of length
$\dG(\psi(s),y)=R\,\rho(s)$ with $\rho(s)$ the angle at the centre (blue
arc). The velocity $\psi'(s)$ (orange arrow) splits into a radial part of
signed speed $R\rho'(s)$ (red) and a transverse part of speed
$\sqrt{1-(R\rho'(s))^{2}}$ (green); the outer dashed circle is the cap
boundary $\{\rho=\rho_{\max}\}$.}
\label{fig:sphereconv}
\end{figure}
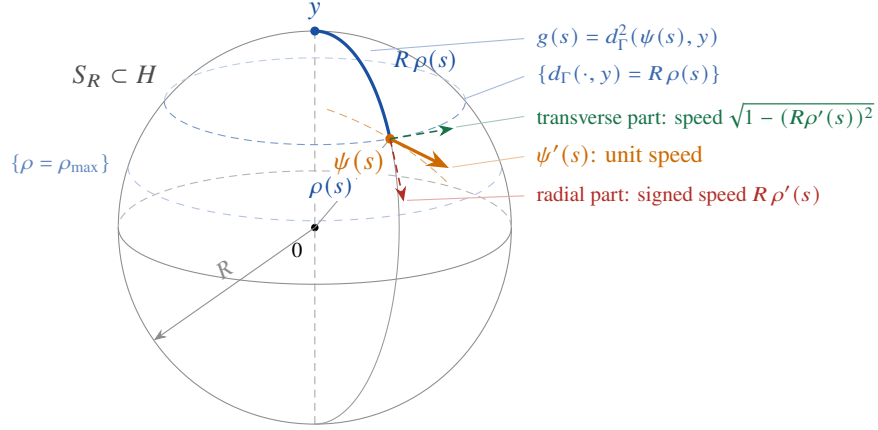

\begin{lemma}[Convexity of the squared distance along great circles]\label{lem:sphereconv}
Let $H$ be a Hilbert space, $R>0$, and $S_R=\{\psi\in H:\norm\psi=R\}$. Fix
$y\in S_R$ and let $s\mapsto\psi(s)$ be a unit-speed great-circle arc on $S_R$
(arc-length parametrisation). Write
\[
\rho(s)=\frac{\dG(\psi(s),y)}{R}\in[0,\pi]
\]
for the \emph{angular} distance from $\psi(s)$ to $y$, and $g(s)=\dG^2(\psi(s),y)$
for the squared geodesic distance. Then, at every $s$ where $\rho(s)\in(0,\pi)$, $g$
is twice differentiable with
\begin{equation}\label{eq:gsecond}
g''(s)\;=\;2\Big[\underbrace{\big(R\,\rho'(s)\big)^2}_{\text{radial part}}
\;+\;\Big(1-\big(R\,\rho'(s)\big)^2\Big)\,\rho(s)\cot\rho(s)\Big], 
\qquad \big(R\,\rho'(s)\big)^2\le1 ,
\end{equation}
a convex combination (weights $(R\rho')^2$ and $1-(R\rho')^2$) of the two values $1$
and $\rho\cot\rho$. Consequently, if the whole arc keeps angular distance
$\rho(s)\le\rho_{\max}$ for some $\rho_{\max}<\tfrac\pi2$, then
\[
g''(s)\ \ge\ 2\,\rho_{\max}\cot\rho_{\max}\ >\ 0
\qquad\text{along the arc,}
\]
and the same lower bound holds (in the classical sense, $g$ being smooth there) at
isolated points where $\rho=0$.
\end{lemma}

\begin{proofenv}[Proof.]
On finite-dimensional spheres the statement is classical and follows from
Jacobi-field comparison \citep{karcher_riemannian_1977,afsari_riemannian_2011}. The
proof below is dimension-free, so the conclusion also holds on the
infinite-dimensional warp sphere, where the classical references do not apply.

\emph{What we must show.} The formula \eqref{eq:gsecond} for $g''$, the speed bound
$(R\rho')^2\le1$, and the lower bound $g''\ge2\rho_{\max}\cot\rho_{\max}$. The whole
proof takes place on $S_R$, with the lemma's own objects: the arc $\psi(s)$, the
target $y$, the angular distance $\rho(s)$, and $g(s)=\dG^{2}(\psi(s),y)=(R\rho(s))^2$.

\emph{Notation for the proof:}
\begin{itemize}[leftmargin=1.6em,itemsep=2pt,topsep=2pt]
\item $u(s):=\inner{\psi(s)}{y}/R^{2}$ is the cosine of the angle between the
moving point and the target: by the definition of the spherical distance,
$\dG(\psi(s),y)=R\arccos\!\big(\inner{\psi(s)}{y}/R^{2}\big)$, so
\[
u(s)=\cos\rho(s):
\]
$u$ carries the same information as $\rho$, in a form that is easy to differentiate.
\item $\alpha:=\inner{\psi(0)}{y}/R^{2}$ and $\beta:=\inner{\psi'(0)}{y}/R$ are two
constants fixed by the initial data: $\alpha=u(0)=\cos\rho(0)$ is the
starting cosine, and $\beta$ measures how much the initial velocity points towards
the target.
\end{itemize}
Primes always denote $d/ds$ (arc length on $S_R$). We work at parameters where
$\rho(s)\in(0,\pi)$, that is $u(s)\in(-1,1)$. The boundary case $\rho=0$ is
treated at the very end. The proof has four steps. (A)~Write the great circle explicitly,
which turns $u$ into a single-frequency oscillation with coefficients
$\alpha,\beta$. (B)~Prove $\alpha^2+\beta^2\le1$ and deduce the speed bound.
(C)~Differentiate $\rho=\arccos u$ twice. (D)~Assemble.

\smallskip\noindent
\emph{Step A (the great circle, explicitly).} On the sphere the position and the
velocity are always orthogonal, $\inner{\psi}{\psi'}=\tfrac12(\norm{\psi}^2)'\equiv0$.
So $\psi(0)$ (norm $R$) and $\psi'(0)$ (norm $1$, unit speed) are orthogonal, and the
arc stays in the plane they span:
\[
\psi(s)=\cos\!\big(\tfrac sR\big)\,\psi(0)+R\,\sin\!\big(\tfrac sR\big)\,\psi'(0).
\]
Two checks confirm this is the right parametrisation. Expanding with
$\inner{\psi(0)}{\psi'(0)}=0$,
\[
\norm{\psi(s)}^{2}
=R^{2}\cos^{2}\!\big(\tfrac sR\big)+R^{2}\sin^{2}\!\big(\tfrac sR\big)=R^{2}
\qquad\text{(the arc stays on $S_R$)},
\]
and differentiating,
$\psi'(s)=-\sin(\tfrac sR)\,\tfrac{\psi(0)}{R}+\cos(\tfrac sR)\,\psi'(0)$, whose norm
is $1$ by the Pythagorean theorem, $\psi(0)$ and $\psi'(0)$ being orthonormal: $\norm{\psi'(s)}^2=\sin^2(\tfrac sR)+\cos^2(\tfrac sR)=1$. Now take the inner product of $\psi(s)$
with $y/R^{2}$: by the definitions of $u,\alpha,\beta$,
\[
u(s)=\alpha\cos\!\big(\tfrac sR\big)+\beta\sin\!\big(\tfrac sR\big),
\qquad\text{hence}\qquad
u''=-\tfrac1{R^{2}}\,u .
\]

\smallskip\noindent
\emph{Step B (the speed bound $(R\rho')^2\le1$).}
\emph{First, $\alpha^2+\beta^2\le1$.} The vectors $\psi(0)/R$ and $\psi'(0)$ are
orthonormal, and $\alpha,\beta$ are the coordinates of the unit vector $y/R$ in
the orthonormal system $\{\psi(0)/R,\psi'(0)\}$:
\[
\alpha=\inner{\psi(0)/R}{\,y/R},\qquad \beta=\inner{\psi'(0)}{\,y/R}.
\]
Bessel's inequality then gives $\alpha^{2}+\beta^{2}\le\norm{y/R}^{2}=1$.
Now, differentiate $u$ and multiply by $R$:
$Ru'(s)=-\alpha\sin(\tfrac sR)+\beta\cos(\tfrac sR)$, so
\begin{flalign}\label{eq:bessel}
u(s)^{2}+\big(Ru'(s)\big)^{2}
&=\alpha^{2}+\beta^{2}\ \le\ 1
\why{expand; the cross terms cancel; Bessel's inequality}
\end{flalign}
\emph{Finally, pass from $u$ to $\rho$.} From $u=\cos\rho$ with $\rho\in(0,\pi)$,
differentiating gives 
\begin{equation}\label{eq:rhoidentity}
u'=-\rho'\sin\rho.
\end{equation}
Therefore,
\begin{equation}\label{eq:radialspeed} 
R\rho'(s)=-\frac{Ru'(s)}{\sin\rho(s)}
\qquad\text{satisfies}\qquad
\big(R\rho'(s)\big)^{2}=\frac{(Ru')^{2}}{\sin^{2}\rho}
\le\frac{1-u^{2}}{\sin^{2}\rho}=1 ,
\end{equation}
using $(Ru')^{2}\le1-u^{2}$ \eqref{eq:bessel} and $1-u^{2}=\sin^{2}\rho$. The
quantity $R\rho'$ is the radial speed, the part of the unit velocity that moves
towards or away from $y$. By \eqref{eq:radialspeed} it cannot exceed the total
speed $1$.

\smallskip\noindent
\emph{Step C (second derivative of the angle).} Differentiate the identity
$\rho'\sin\rho=-u'$ \eqref{eq:rhoidentity} once more:
\begin{flalign*}
\rho''\sin\rho+(\rho')^{2}\cos\rho&=-u''=\tfrac1{R^{2}}\,u=\tfrac1{R^{2}}\cos\rho
\why{product rule; $u''=-u/R^{2}$; $u=\cos\rho$}
\end{flalign*}
\vspace{-2em}
\begin{flalign}\label{eq:step3}
\text{which implies}\qquad
\rho''&=\cot\rho\,\Big(\tfrac1{R^{2}}-(\rho')^{2}\Big).
\why{divide by $\sin\rho>0$}
\end{flalign}

\smallskip\noindent
\emph{Step D (assemble).} With $g=R^{2}\rho^{2}$:
\begin{flalign*}
g''=\big(2R^{2}\rho\rho'\big)'
&=2R^{2}(\rho')^{2}+2R^{2}\rho\rho''
\why{product rule}\\
&=2\big(R\rho'\big)^{2}+2\,\rho\cot\rho\,\Big(1-\big(R\rho'\big)^{2}\Big),
\why{\eqref{eq:step3}; distribute $R^{2}$}
\end{flalign*}
which is exactly \eqref{eq:gsecond}: a convex combination of the two values $1$ and
$\rho\cot\rho$, with weights $(R\rho')^{2}\in[0,1]$ and $1-(R\rho')^{2}$. For the
lower bound we need two properties of the function
$\rho\mapsto\rho\cot\rho$ on $[0,\pi)$: it extends continuously with value $1$ at
$\rho=0$, and it is strictly decreasing, since
\[
\frac{d}{d\rho}\big(\rho\cot\rho\big)=\cot\rho-\frac{\rho}{\sin^{2}\rho}
=\frac{\sin\rho\cos\rho-\rho}{\sin^{2}\rho}
=\frac{\tfrac12\sin(2\rho)-\rho}{\sin^{2}\rho}<0
\qquad(0<\rho<\pi),
\] 
because $\sin x<x$ for $x>0$. In particular $\rho\cot\rho\le1$ everywhere, and
$\rho\cot\rho\ge\rho_{\max}\cot\rho_{\max}>0$ for $\rho\le\rho_{\max}<\pi/2$. A
convex combination of $1$ and $\rho\cot\rho$ is then at least
$\min(1,\rho\cot\rho)=\rho\cot\rho\ge\rho_{\max}\cot\rho_{\max}$, which proves the
bound $g''\ge2\rho_{\max}\cot\rho_{\max}$ wherever $\rho\in(0,\rho_{\max}]$.

\smallskip\noindent
\emph{The exceptional case $\rho(s_0)=0$ (the arc passes through $y$).} Here the
formula for $g''$ does not apply (it divides by $\sin\rho=0$), so we check the
bound directly.
Write $u(s)=\alpha\cos(\tfrac sR)+\beta\sin(\tfrac sR)=A\cos\!\big(\tfrac sR-\varphi\big)$
with amplitude $A=\sqrt{\alpha^2+\beta^2}\le1$ (Step B). Since $u(s_0)=\cos\rho(s_0)=1$
and a cosine of amplitude $\le1$ can reach $1$ only at full amplitude and at its
peak, we get $A=1$ and $\varphi=s_0/R$, hence $u(s)=\cos\!\big(\tfrac{s-s_0}{R}\big)$ and
$\rho(s)=|s-s_0|/R$. The distance $\rho$ has a corner at $s_0$, but the squared
distance is smooth there:
\[
  g(s)=\big(R\rho(s)\big)^2=(s-s_0)^2,\qquad\text{so}\qquad g''(s_0)=2 .
\]
As $\rho\cot\rho\le1$ gives $2\rho_{\max}\cot\rho_{\max}\le2$, the bound
$g''\ge2\rho_{\max}\cot\rho_{\max}$ holds across $s_0$ as well. \hfill$\square$
\end{proofenv}

Figure~\ref{fig:ballconv} illustrates the statement.

\begin{figure}[htbp]
\centering
\begin{tikzpicture}[>=Stealth, scale=1.0]
  \draw[black!45] (0,0) circle (2.6);
  \node[font=\small, black!70, anchor=east] at (-1.95,2.0) {$S_R\subset H$};
  \fill (0,0) circle (1.4pt); \node[font=\scriptsize, anchor=north east] at (-0.03,-0.06) {$0$};
  \draw[->, black!45] (0,0) -- (215:2.6) node[midway, above, sloped, font=\footnotesize] {$R$};
  \draw[black!35, densely dashed] (0,-2.6) -- (0,2.6);
  \draw[figred!50] plot[domain=180:360, samples=60] ({2.6*cos(\x)},{0.75*sin(\x)});
  \draw[figred!30, densely dashed] plot[domain=0:180, samples=60] ({2.6*cos(\x)},{0.75*sin(\x)});
  \node[font=\scriptsize, figred!85, anchor=west] at (2.7,-0.30) {$\dG(\cdot,\psi_0)=\tfrac{\pi R}{2}$};
  \draw[figred!40] (2.67,-0.30) -- (2.36,-0.32);
  \fill[figblue] (0,2.6) circle (1.7pt); \node[above, font=\small, figblue] at (0,2.62) {$\psi_0$};
  \draw[figblue!70] plot[domain=180:360, samples=60] ({1.992*cos(\x)},{1.671+0.574*sin(\x)});
  \draw[figblue!35, densely dashed] plot[domain=0:180, samples=60] ({1.992*cos(\x)},{1.671+0.574*sin(\x)});
  \draw[figblue, semithick] plot[domain=0:61.74, samples=40] ({-1.4544*sin(\x)},{2.6*cos(\x)});
  \node[font=\footnotesize, figblue, anchor=east] at (-0.86,2.18) {$b$};
  \node[font=\scriptsize, figblue!85, anchor=west] at (2.5,2.45) {$\bar B(\psi_0,b)$};
  \draw[figblue!40] (2.47,2.43) -- (0.42,2.36);
  \fill[figorange] (-0.95,1.75) circle (1.7pt);
  \node[font=\small, figorange, anchor=east] at (-1.06,1.78) {$\psi_1$};
  \fill[figorange] (1.15,1.62) circle (1.7pt);
  \node[font=\small, figorange, anchor=west] at (1.26,1.64) {$\psi_2$};
  \draw[figorange, very thick] (-0.95,1.75) .. controls (-0.10,1.52) and (0.50,1.50) .. (1.15,1.62);
  \node[font=\small, figorange, anchor=north] at (0.18,1.47) {$\psi(s)$};
  \draw[black!30] (0,0) -- (-0.95,1.75);
  \draw[black!30] (0,0) -- (1.15,1.62);
  \draw[black!50] (54.6:0.5) arc[start angle=54.6, end angle=118.5, radius=0.5];
  \node[font=\footnotesize, black!60, fill=white, inner sep=1pt] at (81:0.76) {$\vartheta$};
\end{tikzpicture}
\caption{\textbf{The geometry of Lemma~\ref{lem:ballconv}}, in the lemma's
notation. The closed ball $\bar B(\psi_0,b)\subset S_R$ is the spherical cap
of geodesic radius $b$ (blue meridian arc) around $\psi_0$; two points
$\psi_1,\psi_2$ of the ball, at angle $\vartheta$ from the centre $0$, are
joined by the minimizing great-circle arc of \eqref{eq:slerp}; the red
equator lies at distance $\pi R/2$ from $\psi_0$.}
\label{fig:ballconv}
\end{figure}
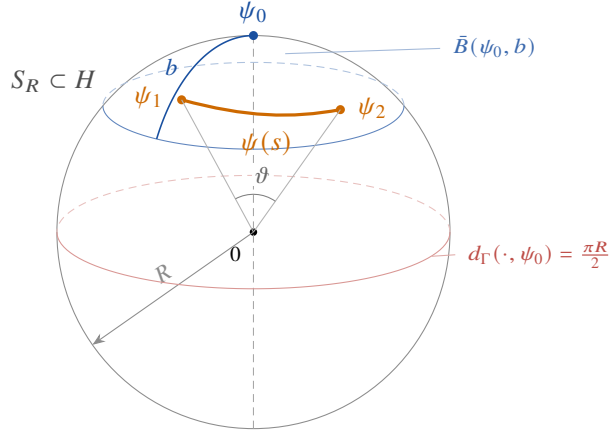

\begin{lemma}[Small balls on the sphere are geodesically convex]\label{lem:ballconv}
Let $b<\pi R/2$ and $\psi_0\in S_R$. If $\psi_1,\psi_2\in\bar B(\psi_0,b) \subset S_R$ (a closed
geodesic ball), then the minimizing great-circle arc \eqref{eq:slerp} joining them
stays inside $\bar B(\psi_0,b)$.
\end{lemma}

\begin{proofenv}[Proof.]
The statement is classical in finite dimensions (balls below the convexity radius
are geodesically convex, see \citealp{afsari_riemannian_2011}). Our contribution is
the dimension-free proof below, which is valid on the
infinite-dimensional sphere. We must show that $\dG(\psi(s),\psi_0)\le b$ for every
$s\in[0,1]$. If
$\psi_1=\psi_2$ the arc is constant and there is nothing to prove, and two points of
$\bar B(\psi_0,b)$ are never antipodal (their distance is at most $2b<\pi R$), so
we may assume that the angle $\vartheta$ between $\psi_1$ and $\psi_2$ lies in
$(0,\pi)$, which is the range where formula \eqref{eq:slerp} is valid.

\emph{Notation.} As in the proof of Lemma~\ref{lem:sphereconv}, let
\[
u(s):=\frac{\inner{\psi(s)}{\psi_0}}{R^{2}}
\]
be the cosine of the angular distance from the moving point to the centre of
the ball: $\dG(\psi(s),\psi_0)=R\arccos u(s)$ (recall from \eqref{eq:slerp} that the
arc stays on $S_R$). Since $\arccos$ is decreasing, membership in the ball
translates into a lower bound on the cosine:
\[
\psi\in\bar B(\psi_0,b)
\qquad\iff\qquad
\frac{\inner{\psi}{\psi_0}}{R^{2}}\ \ge\ \cos\tfrac bR
\qquad\big(\psi\in S_R\big).
\]
So the hypothesis reads $u(0)\ge\cos\tfrac bR$ and $u(1)\ge\cos\tfrac bR$, and the
conclusion to be proved reads $u(s)\ge\cos\tfrac bR$ for every $s\in[0,1]$. The
hypothesis $b<\pi R/2$ enters through one sign only: it makes $\cos\tfrac bR>0$.

\smallskip\noindent
\emph{Step A ($u$ along the arc).} Take the inner product of \eqref{eq:slerp} with
$\psi_0/R^{2}$: by linearity of the inner product,
\[
u(s)=\frac{\sin\big((1-s)\vartheta\big)}{\sin\vartheta}\,u(0)
+\frac{\sin(s\vartheta)}{\sin\vartheta}\,u(1),
\]
and both coefficients are $\ge0$: the angles $s\vartheta$ and
$(1-s)\vartheta$ lie in $[0,\vartheta]\subset[0,\pi)$, where $\sin\ge0$, and
$\sin\vartheta>0$.

\smallskip\noindent
\emph{Step B (the coefficients add up to at least $1$).} By the sum-to-product
identity,
\[
\sin\big((1-s)\vartheta\big)+\sin(s\vartheta)
=2\,\sin\tfrac\vartheta2\,\cos\!\Big(\tfrac{(1-2s)\vartheta}{2}\Big)
\ \ge\ 2\,\sin\tfrac\vartheta2\,\cos\tfrac\vartheta2
=\sin\vartheta ,
\]
because $|1-2s|\le1$ for $s\in[0,1]$, so
$\big|\tfrac{(1-2s)\vartheta}{2}\big|\le\tfrac\vartheta2<\tfrac\pi2$, and on
$[0,\tfrac\pi2)$ the cosine decreases in the absolute value of its argument.
Dividing by $\sin\vartheta>0$: the two coefficients of Step A sum to at least $1$.

\smallskip\noindent
\emph{Step C (conclude).} Bound each endpoint cosine below by $\cos\tfrac bR$
(allowed because the coefficients are $\ge0$), then apply Step B:
\begin{flalign*}
u(s)
&\ge\Big(\frac{\sin\big((1-s)\vartheta\big)}{\sin\vartheta}
+\frac{\sin(s\vartheta)}{\sin\vartheta}\Big)\,\cos\tfrac bR
\why{Step A; $u(0),u(1)\ge\cos\tfrac bR$; coefficients $\ge0$}\\
&\ge\ \cos\tfrac bR .
\why{Step B; $\cos\tfrac bR>0$}
\end{flalign*}
Therefore, $\psi(s)\in\bar B(\psi_0,b)$ for every
$s\in[0,1]$: the arc stays in the ball. \hfill$\square$
\end{proofenv}

\begin{remark}
The same non-negative-coefficients observation proves that the feasible set
$\Psi(\Gt)$ of Lemma~\ref{lem:feasible} is geodesically convex.
\end{remark}

\section{Auxiliary lemmas} \label{app:lemmas} 

\noindent In this section, we first record the analytic properties of the SRSF transform, then three
facts about the warp action used throughout. Recall (toolbox, Appendix~\ref{app:toolbox})
that $\F$ is the space of absolutely continuous functions on $[0,T]$, equipped with
the total-variation metric
$d_{\F}(f_1,f_2)=|f_1(0)-f_2(0)|+\int_0^T|\dot f_1-\dot f_2|$, and that, by convention,
$\operatorname{sign}(0):=0$.

Figure~\ref{fig:orbitdist} illustrates the geometry behind the shape space of Section~\ref{sec:shape}: shapes are closures of orbits, and the elastic distance is the shortest $\Lt$-distance from a representative to the other shape, attained at the registration when it exists.

\begin{figure}[htbp]
\centering
\begin{tikzpicture}[>=Stealth, scale=1.0]
  \draw[black!35, fill=black!2, rounded corners=14pt] (-2.15,-0.6) rectangle (4.3,4.32);
  \node[font=\small, black!70, anchor=south west] at (-2.05,4.40) {the SRSF space $(\Lt,\norm{\cdot})$};
  \coordinate (O) at (0,0);
  \fill (O) circle (1.3pt); \node[below left=-2pt, font=\small] at (O) {$0$};
  \draw[figblue!35, dashed] (14:2.0) arc[start angle=14, end angle=80, radius=2.0];
  \draw[figred!35, dashed] (8:3.3) arc[start angle=8, end angle=80, radius=3.3];
  \draw[->, figblue!70] (O) -- (16:2.0) node[midway, above, sloped, font=\footnotesize] {$\norm{q_1}$};
  \draw[->, figred!70] (O) -- (8:3.3) node[pos=0.68, below, sloped, font=\footnotesize] {$\norm{q_2}$};
  \draw[line width=1.4pt, figblue] (28:2.0) arc[start angle=28, end angle=58, radius=2.0];
  \draw[line width=1.4pt, figred] (20:3.3) arc[start angle=20, end angle=70, radius=3.3];
  \node[font=\footnotesize, figblue] at (36:1.62) {$[q_1]$};
  \node[font=\footnotesize, figred, anchor=south east] at (74:3.55) {$[q_2]=\overline{\{(q_2,\gamma):\gamma\in\G\}}$};
  \draw[figred!40] (73:3.50) -- (67:3.34);
  \fill[figblue] (52:2.0) circle (1.5pt);
  \node[font=\footnotesize, figblue, anchor=north east] at (50.5:1.94) {$q_1$};
  \fill[figred] (24:3.3) circle (1.5pt);
  \node[font=\footnotesize, figred, anchor=west] at (23:3.42) {$q_2$};
  \fill[figred] (52:3.3) circle (1.5pt);
  \node[font=\footnotesize, figred, anchor=west] at (51:3.42) {$(q_2,\gamma^\star_{q_2\to q_1})$};
  \draw[->, figred!70] (28:3.14) arc[start angle=28, end angle=46, radius=3.14];
  \draw[densely dashed, thick, figorange] (52:2.0) -- (52:3.3);
  \node[font=\scriptsize, figorange, anchor=east] at (1.60,2.22) {$\del([q_1],[q_2])$};
  \draw[black!70] (0.3,-1.9) -- (4.9,-1.9);
  \node[font=\small, anchor=west] at (5.0,-1.9) {the shape space $(\Sw,\del)$};
  \fill[figblue] (1.7,-1.9) circle (1.4pt); \node[below, font=\footnotesize, figblue] at (1.7,-1.95) {$[q_1]$};
  \fill[figred] (3.5,-1.9) circle (1.4pt); \node[below, font=\footnotesize, figred] at (3.5,-1.95) {$[q_2]$};
  \draw[<->, figorange] (1.7,-1.58) -- (3.5,-1.58) node[midway, above, font=\scriptsize, figorange] {$\del([q_1],[q_2])$};
  \draw[->, black!40, dashed] (44:2.0) to[out=-75,in=95] (1.7,-1.80);
  \draw[->, black!40, dashed] (40:3.3) to[out=-55,in=80] (3.58,-1.80);
  \node[font=\scriptsize, black!50, anchor=west] at (4.0,-1.05) {$q\mapsto[q]$};
  \node[font=\footnotesize, align=left, text width=4.6cm, anchor=west] at (5.7,2.2)
    {the action is isometric, $\norm{(q,\gamma)}=\norm q$ (Lemma~\ref{lem:iso}), so each orbit stays on its own norm sphere (dashed arcs)};
\end{tikzpicture}
\caption{\textbf{Shapes are closures of orbits, and the elastic distance is a shape-to-shape distance.} \emph{Top: inside $\Lt$} (schematic $2$-D cross-section through $0$): the shapes $[q_1]$ and $[q_2]$, each on the sphere of radius $\norm{q_i}$. Registration slides $q_2$ within its shape $[q_2]$ to the point $(q_2,\gamma^\star_{q_2\to q_1})$ nearest $q_1$ (when such a point exists), and the dashed connector has length $\del([q_1],[q_2])$. \emph{Bottom:} the shape space $\Sw=\Lt/$$\!\sim$, one point per shape.}
\label{fig:orbitdist}
\end{figure}
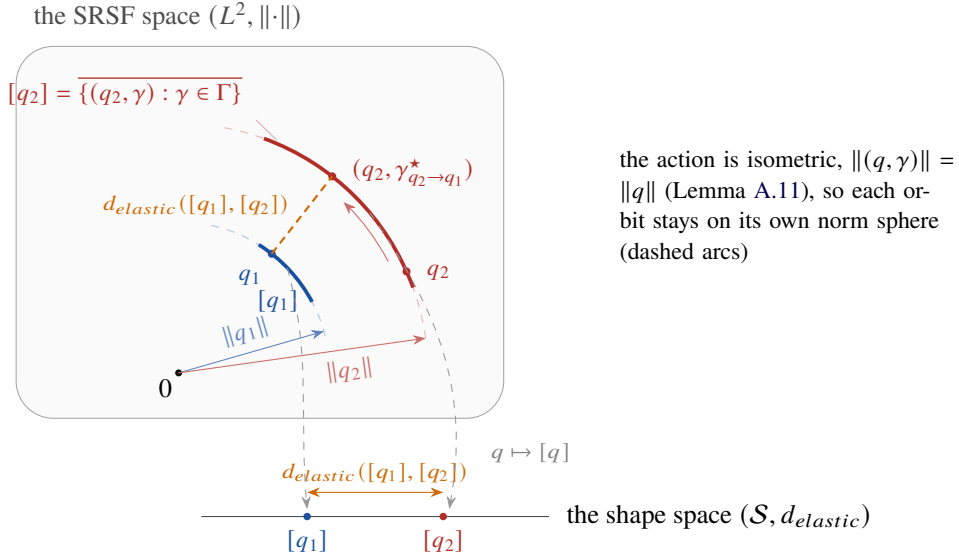

The next proposition collects the properties of the SRSF the paper relies on.
Item (i) says the SRSF belongs to the Hilbert space $\Lt$. Item (ii) says that
passing from $f$ to $(f(0),q_f)$ loses no information, so that curves and transforms can be used interchangeably. Item (iii) says that the
SRSF is continuous, hence measurable, which makes the random transform $q_Z$, and
therefore the random shape $[q_Z]$, a well-defined random element
(Section~\ref{sec:estimand}).

\begin{proposition}[SRSF properties]\label{prop:srsf}
\leavevmode
\begin{enumerate}[leftmargin=2em,itemsep=1pt,topsep=2pt]
\item[(i)] For every $f\in\F$, $q_f\in\Lt$, with
$\norm{q_f}_2^2=\int_0^T|\dot f|$ (the total variation of $f$).
\item[(ii)] The map $f\mapsto(f(0),q_f)$ is a bijection from $\F$ onto
$\R\times\Lt$, with inverse $f(t)=f(0)+\int_0^t q(s)\,|q(s)|\,ds$.
\item[(iii)] $f\mapsto q_f$ is $\tfrac12$-H\"older from $(\F,d_\F)$ to $\Lt$:
$\norm{q_{f_1}-q_{f_2}}^2\le2\int_0^T|\dot f_1-\dot f_2|\le 2\,d_\F(f_1,f_2)$,
with the constant $2$ sharp. In particular $f\mapsto q_f$ is continuous, hence Borel
measurable.
\end{enumerate}
\end{proposition}

The fourth classical property of the SRSF, the Fisher--Rao isometry
$\dFR(f_1,f_2)=\norm{q_{f_1}-q_{f_2}}_2$, is a known result stated as
Proposition~\ref{prop:friso} in the main text
(\citealp[Ch.~4]{srivastava_klassen_2016}; \citealp{bruveris_optimal_2016}); we do
not re-prove it.

\begin{proofenv}[Proof.]
Nothing here is new: every item is established in the literature, as referenced
below. We only adapt them to our notations and establish the remark that follows.
Throughout,
$\phi:\R\to\R$ denotes the scalar square-root-slope map
\[
\phi(x)=\operatorname{sign}(x)\sqrt{|x|},
\qquad\text{so that}\qquad q_f=\phi\circ\dot f .
\]

\smallskip
\noindent\emph{Proof of (i).} For every real $x$,
$\phi(x)^2=\operatorname{sign}(x)^2\,|x|=|x|$. Integrating this identity along $\dot f$,
\begin{flalign*}
\norm{q_f}^2
&=\int_0^T\phi\big(\dot f(t)\big)^2\,dt
\why{definition of $q_f$}\\
&=\int_0^T\big|\dot f(t)\big|\,dt\ <\ \infty .
\why{$\phi(x)^2=|x|$; $\dot f\in L^1$ since $f$ is AC (\ref{T:ftc})}
\end{flalign*}

\smallskip
\noindent\emph{Proof of (ii).} Since, for all $x\in\R$, $\phi(x)\,|\phi(x)|=x$, we have
$q_f|q_f|=\dot f$ almost everywhere, so by the fundamental
theorem of calculus for absolutely continuous functions (\ref{T:ftc}), for all $t$,
\[
f(0)+\int_0^t q_f(s)|q_f(s)|\,ds
=f(0)+\int_0^t\dot f(s)\,ds
=f(t) ,
\]
so the stated inverse recovers $f$ from $(f(0),q_f)$. Conversely, take any
$(b,q)\in\R\times\Lt$ and set $f(t)=b+\int_0^tq|q|$. The integrand satisfies
$\big|q|q|\big|=q^2\in L^1$ (its integral is $\norm q^2<\infty$), so $f$ is absolutely
continuous (\ref{T:ftc}) with $\dot f=q|q|$ a.e., and then
$\phi(\dot f)=\phi(q|q|)=\operatorname{sign}(q)\sqrt{q^2}
=\operatorname{sign}(q)\,|q|=q$ a.e., so $q_f=q$ and
$f(0)=b$. The two maps are mutually inverse.

\smallskip
\noindent\emph{Proof of (iii): the sharp pointwise bound.} We claim that, for all $u,v\in\R$,
\begin{equation}\label{eq:phibound}
|\phi(u)-\phi(v)|^2\ \le\ 2\,|u-v|,
\end{equation}
and that the constant $2$ cannot be improved. By symmetry, we analyze two different
cases: when $u$ and $v$ are of the same sign and when their signs differ.

\emph{Case A: same sign, say $u\ge v\ge0$.} We have
\begin{flalign*}
|\phi(u)-\phi(v)|^2
&=(\sqrt u-\sqrt v)^2
=u+v-2\sqrt{uv}
\why{expand the square}\\
&\le u+v-2v
\why{$\sqrt{uv}\ge\sqrt{v\cdot v}=v$ since $u\ge v\ge0$}\\
&=u-v=|u-v| \ \le\ 2|u-v| .
\end{flalign*}

\emph{Case B: opposite signs, say $u>0>v$, and write $a=u\ge0$, $b=-v\ge0$.} Then
$\phi(u)-\phi(v)=\sqrt a+\sqrt b$ and $|u-v|=a+b$, so
\begin{flalign*}
|\phi(u)-\phi(v)|^2
&=a+b+2\sqrt{ab}
\why{expand}\\
&\le 2(a+b)=2|u-v| .
\why{AM--GM inequality: $2\sqrt{ab}\le a+b$}
\end{flalign*}
\emph{Sharpness.} We have, at $u=1$, $v=-1$, $|\phi(1)-\phi(-1)|^2=(1-(-1))^2=4=2|1-(-1)|$,
so equality holds in \eqref{eq:phibound}.

\smallskip
\noindent Integrate \eqref{eq:phibound} along the two
derivatives:
\begin{flalign*}
\norm{q_{f_1}-q_{f_2}}^2
&=\int_0^T\big|\phi(\dot f_1)-\phi(\dot f_2)\big|^2
\why{definition}\\
&\le 2\int_0^T\big|\dot f_1-\dot f_2\big|
\ \le\ 2\,d_\F(f_1,f_2).
\why{\eqref{eq:phibound}; definition of $d_\F$}
\end{flalign*}
So $f\mapsto q_f$ is $\tfrac12$-H\"older, in particular continuous from $(\F,d_\F)$ to
$\Lt$. A continuous map between metric spaces is Borel measurable.
\end{proofenv}

\begin{remark}[Why $\F$ must be all of the absolutely continuous functions.]\label{rem:whyac}
If one restricted $\F$ to curves with $\dot f\in\Lt$, the image of $f\mapsto q_f$ would
be $\{q:\int q^4<\infty\}=L^4([0,T])\subsetneq\Lt$ (since $q_f^4=\dot f^2$, and the inclusion $L^4\subset L^2$ on the finite interval $[0,T]$ follows from the Cauchy--Schwarz inequality), and item
(ii) would fail: e.g.\ $q(t)=t^{-1/4}$ lies in $\Lt$ ($\int_0^1t^{-1/2}=2$) but its
preimage has $\dot f=q|q|=t^{-1/2}\notin\Lt$. A transform of $\Lt$, such as the phase-centered representative $(\tilde q,\bar\gamma^{-1})$, could then fail to correspond to a curve of the model
space. With $\F$ the set of all absolutely continuous functions, the setting of
\citet{bruveris_optimal_2016}, the correspondence is exact, and the Step-3 average of the phase-centered registered curves that defines $\mu(P)$ stays in $\F$ (Lemma~\ref{lem:alignavg}).
\end{remark}

The next lemma establishes three claims: composing with a warp keeps us inside
$\F$ and acts on transforms by an explicit formula, the action is linear, and it
preserves the $\Lt$ inner product, so a common warp changes neither norms nor
distances between transforms. Norm preservation is the key ingredient behind
coercivity (Proposition~\ref{thm:exist}) and the equivariance step
(Proposition~\ref{thm:inv}).

\begin{lemma}[Isometric action]\label{lem:iso}
For $\gamma\in\Gt$ and $f\in\F$, we have $f\circ\gamma\in\F$, and the SRSF of $f\circ\gamma$
is the warp action
\[
(q_f,\gamma):=(q_f\circ\gamma)\sqrt{\dot\gamma}
\qquad(\text{with the convention that the right side is }0\text{ on }\{\dot\gamma=0\}).
\]
For each fixed $\gamma\in\Gt$, the map $q\mapsto(q,\gamma)$ on $\Lt$ is linear and
preserves the inner product:
$\inner{(q_1,\gamma)}{(q_2,\gamma)}=\inner{q_1}{q_2}$ for all $q_1,q_2\in\Lt$; in
particular $\norm{(q,\gamma)}=\norm q$.
\end{lemma}

One point needs care: elements of $\Gt$ may not be invertible and may have
$\dot\gamma=0$ on a set of positive measure (Example~\ref{ex:line}(c)), so the
elementary chain rule does not apply. The proof uses the chain rule for
compositions with monotone absolutely continuous maps (\ref{T:cov}(b)), whose
convention on $\{\dot\gamma=0\}$ exists precisely because $\dot f(\gamma(t))$ may be
undefined on that set.

\begin{proofenv}[Proof.]
All of this is established in the literature, as referenced below. We state it in
the exact form our proofs use.

\smallskip
\noindent\emph{Step A: the action formula is the monotone chain rule.} Let
$f\in\F$ and $\gamma\in\Gt$. By \ref{T:cov}(b), $f\circ\gamma$ is absolutely
continuous, so $f\circ\gamma\in\F$ and its SRSF is defined, with
$(f\circ\gamma)^{\textstyle\cdot}=(\dot f\circ\gamma)\,\dot\gamma$ a.e. Now evaluate the SRSF pointwise, splitting on
$\dot\gamma$:
\begin{itemize}[leftmargin=1.6em,itemsep=1pt]
\item On $\{\dot\gamma(t)>0\}$: the slope of $f\circ\gamma$ is $\dot f(\gamma(t))\dot\gamma(t)$, whose sign
is that of $\dot f(\gamma(t))$ (as $\dot\gamma(t)>0$), so
\begin{flalign*}
q_{f\circ\gamma}(t)
&=\operatorname{sign}\!\big(\dot f(\gamma(t))\,\dot\gamma(t)\big)
\sqrt{\big|\dot f(\gamma(t))\,\dot\gamma(t)\big|}
\why{definition of the SRSF}\\
&=\operatorname{sign}\!\big(\dot f(\gamma(t))\big)
\sqrt{\big|\dot f(\gamma(t))\big|}\ \sqrt{\dot\gamma(t)}
\why{$\dot\gamma(t)>0$: its sign is $1$; $\sqrt{ab}=\sqrt a\sqrt b$}\\
&=q_f(\gamma(t))\,\sqrt{\dot\gamma(t)} .
\why{definition of $q_f$}
\end{flalign*}
\item On $\{\dot\gamma(t)=0\}$: the slope of $f\circ\gamma$ is $0$ (see above), so
$q_{f\circ\gamma}(t)=\phi(0)=0=(q_f\circ\gamma)(t)\,\sqrt{\dot\gamma(t)}$.
\end{itemize}
Hence $q_{f\circ\gamma}=(q_f,\gamma)$ a.e. \emph{Linearity} of $q\mapsto(q,\gamma)$
is immediate from the formula: $(q\circ\gamma)\sqrt{\dot\gamma}$ is linear in $q$.

\smallskip
\noindent\emph{Step B: inner-product preservation.} Fix $q_1,q_2\in\Lt$. The product
$q_1q_2$ is integrable by the Cauchy--Schwarz inequality, so the change-of-variables
formula for monotone AC surjections (\ref{T:cov}(a)) applies to $g=q_1q_2$:
\begin{flalign*}
\inner{(q_1,\gamma)}{(q_2,\gamma)}
&=\int_0^T q_1(\gamma(t))\sqrt{\dot\gamma(t)}\;q_2(\gamma(t))\sqrt{\dot\gamma(t)}\,dt
\why{definition of the action}\\
&=\int_0^T q_1(\gamma(t))\,q_2(\gamma(t))\,\dot\gamma(t)\,dt
\why{$\sqrt{\dot\gamma}\,\sqrt{\dot\gamma}=\dot\gamma$}\\
&=\int_0^T q_1(s)\,q_2(s)\,ds=\inner{q_1}{q_2}.
\why{change of variables \ref{T:cov}(a); no injectivity needed}
\end{flalign*}
Taking $q_1=q_2=q$ gives $\norm{(q,\gamma)}^2=\norm q^2$. \hfill$\square$
\end{proofenv}

The existence proof needs the registration residual
$\gamma\mapsto\norm{q-(y,\gamma)}_2$ to be continuous in the warp. That reduces to
the following lemma.

\begin{lemma}[Composition continuity]\label{lem:compos}
Let $\gamma_k,\gamma_\star\in\Gc$ with $\sqrt{\dot\gamma_k}\to\sqrt{\dot\gamma_\star}$
uniformly on $[0,T]$. Then $(y,\gamma_k)\to(y,\gamma_\star)$ in $\Lt$, for every
$y\in\Lt$.
\end{lemma}

\begin{proofenv}[Proof.]
We did not find this statement, for this mode of convergence, in the literature
(joint-continuity statements for the action in other topologies appear in
\citealp{bruveris_optimal_2016}), so we prove it from the toolbox facts. The plan:
first upgrade the hypothesis to uniform convergence of the warps themselves, then
split the difference $(y,\gamma_k)-(y,\gamma_\star)$ into two terms, $A_k$ (the
square-root factors differ) and $B_k$ (the compositions differ), and show that
each goes to $0$. All unsubscripted norms below are $\Lt$ norms.

\smallskip
\noindent\emph{Step A: $\gamma_k\to\gamma_\star$ uniformly.} We have
\begingroup\setlength{\whywidth}{0.27\textwidth}%
\begin{flalign*}
\big|\dot\gamma_k-\dot\gamma_\star\big|
&=\big|\sqrt{\dot\gamma_k}-\sqrt{\dot\gamma_\star}\big|
\cdot\big(\sqrt{\dot\gamma_k}+\sqrt{\dot\gamma_\star}\big)
\le2\sqrt C\,\norm{\sqrt{\dot\gamma_k}-\sqrt{\dot\gamma_\star}}_\infty
\why{$a^2-b^2=(a-b)(a+b)$; $\dot\gamma\le C$}\\
\big|\gamma_k(t)-\gamma_\star(t)\big|
&=\Big|\int_0^t(\dot\gamma_k-\dot\gamma_\star)\Big|
\le T\,\norm{\dot\gamma_k-\dot\gamma_\star}_\infty\ \to\ 0
\why{$\gamma(t)=\int_0^t\dot\gamma$ (\ref{T:ftc}); bound the integrand}
\end{flalign*}\endgroup
uniformly in $t$.

\smallskip
\noindent\emph{Step B: the decomposition.} Simple calculations show that
\begin{align*}
(y,\gamma_k)-(y,\gamma_\star)
&=(y\circ\gamma_k)\sqrt{\dot\gamma_k}-(y\circ\gamma_\star)\sqrt{\dot\gamma_\star}\\
&=\underbrace{(y\circ\gamma_k)\big(\sqrt{\dot\gamma_k}-\sqrt{\dot\gamma_\star}\big)}_{A_k}
+\underbrace{\big(y\circ\gamma_k-y\circ\gamma_\star\big)\sqrt{\dot\gamma_\star}}_{B_k}.
\end{align*}
By the triangle inequality, to complete the proof, we just have to show that
$\norm{A_k}\to0$ and $\norm{B_k}\to0$.

\smallskip
\noindent\emph{Step C: $\norm{A_k}\to0$.} Recall from Assumption~\ref{ass:warp}(ii)
that members of $\Gc$ satisfy $0<c\le\dot\gamma\le C<\infty$. In particular each
$\gamma\in\Gc$ is a bi-Lipschitz bijection, and its inverse $\gamma^{-1}$ is
absolutely continuous with
$(\gamma^{-1})^{\textstyle\cdot}=1/\dot\gamma(\gamma^{-1})\in[1/C,1/c]$. Pull out the
uniform factor and control the composed norm by a change of variables:
\begingroup\setlength{\whywidth}{0.29\textwidth}%
\begin{flalign*}
\norm{A_k}
&\le\norm{\sqrt{\dot\gamma_k}-\sqrt{\dot\gamma_\star}}_\infty\,\norm{y\circ\gamma_k}
\why{$\norm{gh}\le\norm g_\infty\norm h$}\\
\norm{y\circ\gamma_k}^2
&=\int_0^T|y(\gamma_k(t))|^2\,dt
=\int_0^T|y(s)|^2\,\frac{ds}{\dot\gamma_k(\gamma_k^{-1}(s))}
\le\frac1c\,\norm y^2
\why{substitute $s=\gamma_k(t)$ (\ref{T:cov}(a) applied to $\gamma_k^{-1}$); $\dot\gamma_k\ge c$}\\
\text{Thus,}\quad
\norm{A_k}&\le c^{-1/2}\,\norm y\,
\norm{\sqrt{\dot\gamma_k}-\sqrt{\dot\gamma_\star}}_\infty\ \to\ 0 .
\why{hypothesis}
\end{flalign*}\endgroup

\smallskip
\noindent\emph{Step D: $\norm{B_k}\to0$.} Since
$\norm{\sqrt{\dot\gamma_\star}}_\infty\le\sqrt C$, it suffices that
$\norm{y\circ\gamma_k-y\circ\gamma_\star}\to0$. Fix $\varepsilon>0$.
\begin{itemize}[leftmargin=1.6em,itemsep=1pt]
\item[\emph{(i)}] Pick a continuous $g$ with $\norm{y-g}<\varepsilon$ (density
\ref{T:dense}).
\item[\emph{(ii)}] By the same change-of-variables bound as in Step C (valid for
$\gamma_k$ and for $\gamma_\star$, both in $\Gc$),
$\norm{y\circ\gamma_k-g\circ\gamma_k}=\norm{(y-g)\circ\gamma_k}\le
c^{-1/2}\norm{y-g}<c^{-1/2}\varepsilon$, and likewise with $\gamma_\star$.
\item[\emph{(iii)}] $g$ is uniformly continuous on the compact $[0,T]$
(\ref{T:heine}). Since $\gamma_k\to\gamma_\star$ uniformly (Step A), the moduli give
$\norm{g\circ\gamma_k-g\circ\gamma_\star}_\infty\to0$, hence
$\norm{g\circ\gamma_k-g\circ\gamma_\star}\le\sqrt T\,
\norm{g\circ\gamma_k-g\circ\gamma_\star}_\infty\to0$ (Cauchy--Schwarz inequality with $v\equiv1$,
applied to the constant bound).
\end{itemize}
Combining the three by the triangle inequality, for all $\varepsilon>0$, we have
\[
\limsup_k\ \norm{y\circ\gamma_k-y\circ\gamma_\star}
\ \le\ c^{-1/2}\varepsilon+0+c^{-1/2}\varepsilon
\ =\ 2c^{-1/2}\varepsilon ,
\]
which concludes the proof. \hfill$\square$
\end{proofenv}

For $q,y\in\Lt$, recall the \emph{registration residual} and the \emph{lifted
Step-1 objective}
\[
h(q;y)=\min_{\gamma\in\Gc}\norm{q-(y,\gamma)}_2^2,
\qquad
G(q)=\E\big[h(q;q_Z)\big]. 
\]
The next lemma shows that $h$ is well defined (the minimum is attained) and that
it is never smaller than the squared elastic distance, with equality when the optimal warp lies in $\Gc$. Together with the attainment clause of Assumption~\ref{ass:warp}(iii-a), this makes the two infima equal and turns any minimiser of $G$ over $\Lt$ into a minimiser of the Step-1 objective over $\Sw$.

\begin{lemma}[The registration residual: attainment and comparison with the elastic distance]\label{lem:attain}
Let Assumptions~\ref{ass:reg}--\ref{ass:warp} hold.
\begin{enumerate}[leftmargin=2em,itemsep=2pt,topsep=2pt]
\item[(a)] For every $q,y\in\Lt$, the map $\gamma\mapsto\norm{q-(y,\gamma)}_2^2$
attains its minimum over $\Gc$; hence $h(q;y)$ is well defined.
\item[(b)] For all $q,y\in\Lt$,
$h(q;y)\ge\inf_{\gamma\in\Gt}\norm{q-(y,\gamma)}_2^2=\del^2([q],[y])$, with equality whenever a minimiser of $\gamma\mapsto\norm{q-(y,\gamma)}_2$ over $\Gt$ lies in $\Gc$. In particular, $G(q)\ge M([q])$ for every $q\in\Lt$, with equality at the representatives given by Assumption~\ref{ass:warp}(iii-a).
\end{enumerate}
\end{lemma}

\begin{proofenv}[Proof.]
\emph{(a) The inner minimum is attained.} Fix $q,y\in\Lt$ and write
$f(\gamma)=\norm{q-(y,\gamma)}_2^2$. We show that $f$ attains its minimum over
$\Gc$ by the Weierstrass theorem (\ref{T:weier}). We verify its two hypotheses in
the topology of Assumption~\ref{ass:warp}(i), namely uniform convergence of the
warp-SRSFs $\sqrt{\dot\gamma}$.

\emph{$f$ is continuous.} If $\sqrt{\dot\gamma_k}\to\sqrt{\dot\gamma_\star}$ uniformly
(all warps in $\Gc$), then $(y,\gamma_k)\to(y,\gamma_\star)$ in $\Lt$
(Lemma~\ref{lem:compos}). The reverse triangle inequality turns this
convergence of the arguments into convergence of the distances,
\[
\big|\,\norm{q-(y,\gamma_k)}_2-\norm{q-(y,\gamma_\star)}_2\,\big|
\ \le\ \norm{(y,\gamma_k)-(y,\gamma_\star)}_2\ \to\ 0,
\]
so $f(\gamma_k)\to f(\gamma_\star)$.

\emph{$\Gc$ is compact.} By Assumption~\ref{ass:warp}(i), $\Psi(\Gc)=\{\sqrt{\dot\gamma}:
\gamma\in\Gc\}$ is compact in $L^\infty$. In a metric space this means every sequence of
warps in $\Gc$ has a subsequence whose SRSFs converge uniformly, the limit again lying
in $\Gc$.

\emph{Conclusion (the two combine).} Let $m=\inf_{\gamma\in\Gc}f(\gamma)$ and take a
minimising sequence $f(\gamma_k)\to m$. Compactness gives a subsequence
$\sqrt{\dot\gamma_{k_j}}\to\sqrt{\dot\gamma_\star}$ uniformly with $\gamma_\star\in\Gc$.
Continuity gives $f(\gamma_\star)=\lim_j f(\gamma_{k_j})=m$. Thus the infimum $m$
is attained at $\gamma_\star$.

\smallskip\noindent
\emph{(b) Comparison with the elastic distance.} Since $\Gc\subseteq\Gt$, a minimum over $\Gc$ is at least the infimum over $\Gt$, which is the squared elastic distance:
\[
h(q;y)=\min_{\gamma\in\Gc}\norm{q-(y,\gamma)}_2^2
\ \ge\ \inf_{\gamma\in\Gt}\norm{q-(y,\gamma)}_2^2=\del^2([q],[y]).
\]
If a minimiser over $\Gt$ lies in $\Gc$, the inequality is an equality. The right-hand side depends on $(q,y)$ only through the pair of shapes
$([q],[y])$: this is the statement that $\del$ is a well-defined function on
$\Sw\times\Sw$, guaranteed by the construction of the shape space (Section~\ref{sec:shape})
\citep{bruveris_optimal_2016}. Taking expectations gives $G(q)\ge M([q])$, with equality when $h(q;q_Z)=\del^2([q],[q_Z])$ for $P$-almost every $Z$, which is what Assumption~\ref{ass:warp}(iii-a) provides at the representatives it names.  \hfill$\square$
\end{proofenv}

\begin{lemma}[The registration residual is $1$-Lipschitz in each argument]\label{lem:hcont}
Let $h$ be the registration residual (the minimum is attained by
Lemma~\ref{lem:attain}(a)). Then for all $q,q',y,y'\in\Lt$:
\[
\big|\sqrt{h(q;y)}-\sqrt{h(q;y')}\big|\le\norm{y-y'}
\qquad\text{and}\qquad
\big|\sqrt{h(q;y)}-\sqrt{h(q';y)}\big|\le\norm{q-q'} .
\]
In particular $h$ is continuous on $\Lt\times\Lt$, and for fixed $q$ the map
$\omega\mapsto h(q;q_Z(\omega))$ is measurable whenever $q_Z$
is an $\Lt$-valued random element.
\end{lemma}

\begin{proofenv}[Proof.]
We prove this result because we found no reference stating
it for the constrained minimum over $\Gc$, and both the measurability and the
continuity arguments of the paper rest on it.

\emph{Step A: $1$-Lipschitz in the second argument.} Fix any $\gamma\in\Gc$. The warp
action is linear and norm-preserving (Lemma~\ref{lem:iso}), so
$(y',\gamma)-(y,\gamma)=(y'-y,\gamma)$ and $\norm{(y'-y,\gamma)}=\norm{y'-y}$. Hence, by the
triangle inequality,
\[
\norm{q-(y,\gamma)}\ \le\ \norm{q-(y',\gamma)}+\norm{(y'-y,\gamma)}
\ =\ \norm{q-(y',\gamma)}+\norm{y-y'} .
\]
This holds for every $\gamma$. Taking the minimum over $\gamma$ on both sides yields
\[
\sqrt{h(q;y)}
\ \le\ \sqrt{h(q;y')}+\norm{y-y'} .
\]
Swapping $y\leftrightarrow y'$ gives the reverse bound, hence
$\big|\sqrt{h(q;y)}-\sqrt{h(q;y')}\big|\le\norm{y-y'}$.

\emph{Step B: $1$-Lipschitz in the first argument.} The same argument applies by
exchanging $q$ and $y$. In this case the triangle inequality suffices and leads to
$\big|\sqrt{h(q;y)}-\sqrt{h(q';y)}\big|\le\norm{q-q'}$.

\emph{Step C: continuity and measurability.} By Steps~A--B, $\sqrt h$ is $1$-Lipschitz in
each argument, hence continuous. So $h=(\sqrt h)^2$ is continuous, and in particular
$y\mapsto h(q;y)$ is Borel. Therefore $\omega\mapsto h(q;q_Z(\omega))$, the
composition of this Borel map with the random element $q_Z$
(Proposition~\ref{prop:srsf}(iii)), is measurable. \hfill$\square$
\end{proofenv}
\section{Fréchet means in metric spaces: classical guarantees} \label{app:frechetgeo}

This section reviews the classical existence and uniqueness guarantees for Fréchet means (referred to in the discussion following Theorem~\ref{thm:wellposed}), which of them fail in the two spaces our construction uses. Recall that for a Borel probability measure $\nu$ on a metric space $(\mathcal M,d)$ with finite second moment, the Fréchet function is $F(p)=\int d^2(p,y)\,d\nu(y)$ and the Fréchet mean set is $\argmin_p F$.

\begin{itemize}[leftmargin=1.4em,itemsep=1pt]
  \item Existence holds in either of two cases: $(\mathcal M,d)$
  complete and proper (closed bounded sets are compact, or equivalently, for complete
  length spaces, locally compact, by Hopf--Rinow\footnote{Here proper means that closed and bounded sets are compact, the property that guarantees a continuous function attains its minimum. The Hopf--Rinow theorem ensures it for any complete, locally compact length space. Local compactness is automatic in finite dimensions but fails on infinite-dimensional normed spaces, whose closed balls are not compact. Properness is therefore a finite-dimensional property.} \citealp{burago_2001,bridson_haefliger_1999};
  \citealp{ziezold_expected_1977}, manifolds \citealp{bhattacharya_large_2003}), or
  $(\mathcal M,d)$ complete and non-positively curved
  \citep[][Prop.~4.3]{sturm_probability_2003}, a method that, unlike properness, can work in infinite-dimensional spaces, since infinite-dimensional CAT(0) spaces (e.g.\ $\Lt$ itself) qualify.
  \item Uniqueness, by contrast, comes from non-positive curvature
  (CAT(0)) alone: in a
  CAT(0) space the metric is convex along geodesics \citep[][Prop.~II.2.2]{bridson_haefliger_1999}
  and $p\mapsto d^2(p,y)$ is in fact strictly convex, so $F$ is strictly convex with at most
  one minimiser \citep[][Prop.~4.3]{sturm_probability_2003}. When uniqueness fails, one works instead with mean
  sets \citep{evans_jaffe_2024}.
  
  \item Locally, both properties hold together: if $\supp\nu$ lies in a geodesic ball of
  radius $<r^\star=\tfrac12\min(\inj,\pi/\sqrt{\kappa_{\max}})$, where $\inj$ is the
  injectivity radius and $\kappa_{\max}$ an upper curvature bound, the Fréchet mean
  exists and is unique in that ball
  \citep{karcher_riemannian_1977,afsari_riemannian_2011}, which on a sphere is an open
  hemisphere. However, these Karcher--Afsari results are proved on
  finite-dimensional complete manifolds, where the closed ball is compact by
  Hopf--Rinow. The existence half relies on this compactness, which fails in
  infinite dimensions.
\end{itemize}
Our construction poses two distinct Fréchet-mean problems, in two different
spaces. We state which of existence and uniqueness each space gives us directly.
\begin{itemize}[leftmargin=1.4em,itemsep=3pt]
  \item \emph{Step~1, the amplitude mean in the shape space $\Sw$: directly, we have
  neither.} $\Sw$ fails the two existence conditions: it is infinite-dimensional,
  hence not proper (no local compactness), and it is not non-positively curved
  either. For uniqueness the situation is worse: we know of no curvature bound for the SRSF shape space, so neither the
  CAT(0) convexity argument nor a finite Karcher--Afsari radius applies, and no
  geometric criterion for uniqueness is available. We prove existence by a direct
  variational argument on $\Sw$ that replaces compactness
  (Theorem~\ref{thm:wellposed}(i), proved in Appendix~\ref{app:exist}), and
  uniqueness is imposed as a regularity condition on the law of shapes.
  \item \emph{Step~2, the phase mean on the warp sphere $\Sphere(\sqrt T)$: neither holds unconditionally, but both hold locally.} The sphere is infinite-dimensional (not proper) and positively curved (not CAT(0)), so neither global route applies, and, being infinite-dimensional, it lies outside the Karcher--Afsari theorem. We therefore prove local existence and uniqueness directly, under a concentration assumption, by reproducing the Karcher--Afsari local argument in our infinite-dimensional setting (Appendix~\ref{app:exist}).
\end{itemize}

\section{Proof of Theorem~\ref{thm:wellposed}, points (i)--(iii): the two Fréchet means} \label{app:exist}

Points (i)--(iii) of Theorem~\ref{thm:wellposed} are the content of the following
proposition, stated at the level at which it is used (a generic law $P$). Its
clause (i) is the theorem's point (i) and, since uniqueness is exactly
Assumption~\ref{ass:conc}(4a), also settles point (ii). Its clause (ii) is the
theorem's point (iii).

\begin{proposition}[Step-1 existence; Step-2 existence and uniqueness]\label{thm:exist}
Let Assumptions~\ref{ass:reg}--\ref{ass:mom} hold.
\begin{enumerate}[leftmargin=2em,itemsep=2pt,topsep=2pt]
\item[(i)] \emph{(Step 1, existence.)} For every law $P$ of finite second moment for which Assumption~\ref{ass:warp}(iii-a) holds,
the Step-1 objective has a non-empty minimiser set over
$\Sw$; hence a Step-1 mean shape exists. Uniqueness is not asserted; it is imposed
separately as Assumption~\ref{ass:conc}(4a). (Corollary~\ref{cor:reweight} extends
existence to every bounded reweighting of $P$, in particular the IPW objective of
Section~\ref{sec:ident}.)
\item[(ii)] \emph{(Step 2, existence and uniqueness.)} If in addition
Assumption~\ref{ass:conc}(4b) holds, the Step-2 objective has a unique minimiser;
hence the Step-2 (phase) mean exists and is unique.
\end{enumerate}
\end{proposition}

The two clauses of the proposition are proved separately, each in its own
subsection. Clause~(i), existence of the Step-1 mean shape, is proved by the
direct method of the calculus of variations (Appendix~\ref{app:existi}).
Clause~(ii), existence and uniqueness of the Step-2 phase mean, is proved by
re-deriving the Karcher--Afsari theorem on the infinite-dimensional warp sphere
(Appendix~\ref{app:ka}). That argument uses the geometric preliminaries of
Appendix~\ref{app:assump}: the feasible set (Lemma~\ref{lem:feasible}), convexity
of the squared distance (Lemma~\ref{lem:sphereconv}), and convexity of small
balls (Lemma~\ref{lem:ballconv}).

\subsection{Clause (i) of Proposition~\ref{thm:exist}: the Step-1 mean shape exists (the direct method)} \label{app:existi}

The task is to produce a shape minimising
$M([q])=\E[\del^2([q],[q_Z])]$ over $\Sw$, for any law of a curve $Z$ satisfying Assumption~\ref{ass:warp}(iii-a) and
satisfying Assumption~\ref{ass:mom} (an arm law, a difference law, or a local
difference law: Corollary~\ref{cor:wp3}). The difficulty is that in the
infinite-dimensional $\Sw$, closed bounded sets
are not compact, so the elementary route ``a continuous function on a compact set
attains its minimum'' (\ref{T:weier}) is unavailable. We instead run the
\emph{direct method of the calculus of variations} \citep{dacorogna_direct_2008} on
the Hilbert space $\Lt$ in which shapes are represented:
\begin{quote}\small
A functional that is (a) \emph{coercive} (its sublevel sets are bounded) and (b)
\emph{weakly lower-semicontinuous} (it cannot jump down along weakly
convergent sequences) attains its minimum on a Hilbert space, because a minimising
sequence is bounded (by (a)), hence has a weak limit (\ref{T:weakcpt}), at which the
functional is no larger than the limiting values (by (b)).
\end{quote}
The functional in question is the lifted objective $G(q)=\E[h(q;q_Z)]$, where
$h$ is the registration residual and $q_Z\in\Lt$ is the random transform
(Proposition~\ref{prop:srsf}(iii)). The groundwork is already done in
Appendix~\ref{app:lemmas}: $h$ is well defined and bounded below by the squared elastic distance
(Lemma~\ref{lem:attain}), and for fixed $q$ the map
$\omega\mapsto h(q;q_Z(\omega))$ is a non-negative random variable
(Lemma~\ref{lem:hcont}), so $G(q)\in[0,\infty]$ is well defined,
$\inf_{\Lt}G=\inf_{\Sw}M$, and if $q^\star$ minimises $G$ then $[q^\star]$
minimises $M$. Indeed, $G(q)\ge M([q])$ for every $q$ (Lemma~\ref{lem:attain}(b)), so $\inf_{\Lt}G\ge\inf_{\Sw}M$. Conversely, Assumption~\ref{ass:warp}(iii-a) gives representatives $q_k$ with $G(q_k)=M([q_k])\to\inf_{\Sw}M$, so $\inf_{\Lt}G\le\inf_{\Sw}M$. Finally, if $q^\star$ minimises $G$, then $M([q^\star])\le G(q^\star)=\inf_{\Lt}G=\inf_{\Sw}M$.

\begin{proofenv}[Proof of Proposition~\ref{thm:exist}(i).]
The result is new: existence of the Step-1 mean for the lifted, population-level
elastic objective is one of this paper's contributions
(Theorem~\ref{thm:wellposed}(i)). What is ours is the check that the
elastic objective survives the inner registration: coercivity via the isometric
action (Lemma~\ref{lem:iso}), and weak lower semicontinuity via
Assumption~\ref{ass:warp}(i) and Lemma~\ref{lem:compos}. All unsubscripted norms
are $\Lt$ norms.

\smallskip
\noindent\emph{Step A (coercivity).} For every
$\gamma\in\Gc\subset\Gt$, norm preservation (Lemma~\ref{lem:iso}) gives
$\norm{(y,\gamma)}=\norm y$, so by the reverse triangle inequality
\[
\norm{q-(y,\gamma)}\ \ge\ \big|\norm q-\norm{(y,\gamma)}\big|=\big|\norm q-\norm y\big| ,
\]
which implies
\[
h(q;y)\ \ge\ \big(\norm q-\norm y\big)^2 .
\]
Take $y=q_Z$ and expectations:
\begin{flalign*}
G(q)
&\ge\E\big[(\norm q-\norm{q_Z})^2\big]
\why{monotonicity of $\E$}\\
&=\norm q^2-2\norm q\,\E\norm{q_Z}+\E\norm{q_Z}^2 ,
\why{expand; $\norm q$ is a constant}
\end{flalign*}
where $\E\norm{q_Z}\le(\E\norm{q_Z}^2)^{1/2}<\infty$ by the Cauchy--Schwarz inequality and
Assumption~\ref{ass:mom}. The right-hand side is a quadratic in $\norm q$ with
leading coefficient $1$: it tends to $+\infty$ as $\norm q\to\infty$. Hence every
sublevel set $\{q:G(q)\le c\}$ is bounded in $\Lt$.

\smallskip
\noindent\emph{Step B (weak lower-semicontinuity).} We must show: if
$q_n\rightharpoonup q$ weakly in $\Lt$, then $\liminf_nG(q_n)\ge G(q)$. We first
prove that, for each fixed $y\in\Lt$,
\begin{equation}\label{eq:ptwiselsc}
L\ :=\ \liminf_n\,h(q_n;y)\ \ge\ h(q;y),
\end{equation}
where the liminf is along the full sequence. Choose a
subsequence $(n_j)$ along which $h(q_{n_j};y)\to L$, and let $\gamma_{n_j}\in\Gc$ be
corresponding minimisers (Lemma~\ref{lem:attain}(a)):
\[
h(q_{n_j};y)=\norm{q_{n_j}-(y,\gamma_{n_j})}^2 .
\]
By compactness of $\Psi(\Gc)$
(Assumption~\ref{ass:warp}(i)), pass to a further subsequence (not relabelled) with
$\sqrt{\dot\gamma_{n_j}}\to\sqrt{\dot\gamma_\star}$ uniformly, $\gamma_\star\in\Gc$.
Then:
\begin{flalign*}
(y,\gamma_{n_j})&\to(y,\gamma_\star)\ \text{strongly in }\Lt
\why{Lemma~\ref{lem:compos}}\\
q_{n_j}-(y,\gamma_{n_j})&\rightharpoonup q-(y,\gamma_\star)
\why{weak $+$ strong $\Rightarrow$ weak (\ref{T:weaklsc})}\\
L=\lim_j\norm{q_{n_j}-(y,\gamma_{n_j})}^2
&\ge\norm{q-(y,\gamma_\star)}^2
\why{weak lower-semicontinuity of $\norm\cdot^2$ (\ref{T:weaklsc})}\\
&\ge\min_{\gamma\in\Gc}\norm{q-(y,\gamma)}^2=h(q;y).
\why{$\gamma_\star\in\Gc$ is one admissible competitor}
\end{flalign*}
This proves \eqref{eq:ptwiselsc}. Now integrate: 
\begin{flalign*}
\liminf_n G(q_n)
&=\liminf_n\E\big[h(q_n;q_Z)\big]
\why{definition of $G$}\\
&\ge\E\big[\liminf_n h(q_n;q_Z)\big]
\why{Fatou's lemma (\ref{T:fatou}); integrands $h(q_n;q_Z)\ge0$}\\
&\ge\E\big[h(q;q_Z)\big]
\why{\eqref{eq:ptwiselsc} at $y=q_Z(\omega)$ for a.e.\ $\omega$, then monotonicity of $\E$}\\
&=G(q).
\why{definition of $G$}
\end{flalign*}

\smallskip
\noindent\emph{Step C (assemble the direct method).} First, the infimum is finite:
taking the candidate $q=0$,
\[
G(0)=\E\Big[\min_{\gamma\in\Gc}\norm{(q_Z,\gamma)}^2\Big]
=\E\norm{q_Z}^2<\infty ,
\]
using norm preservation (Lemma~\ref{lem:iso}) and Assumption~\ref{ass:mom}. Set
$m:=\inf_{\Lt}G\in[0,\infty)$ and take a minimising sequence $G(q_n)\to m$.
Eventually $G(q_n)\le m+1$, so the sequence sits in the bounded sublevel set
$\{G\le m+1\}$ (Step A). By weak sequential compactness (\ref{T:weakcpt}) extract
$q_{n_k}\rightharpoonup q^\star$. By Step B,
\[
G(q^\star)\ \le\ \liminf_kG(q_{n_k})\ =\ m,
\]
and $G(q^\star)\ge m$ by definition of the infimum: $G(q^\star)=m$, the minimum is
attained. By the argument above the proof, $[q^\star]$ minimises $M$ over $\Sw$: the Step-1 mean
shape exists. (Uniqueness is not asserted here: it is imposed as
Assumption~\ref{ass:conc}(4a).)

\hfill$\square$
\end{proofenv}

The causal estimands of Section~\ref{sec:ident} minimise a reweighted objective.
The following corollary checks that a bounded reweighting changes nothing.

\begin{corollary}[Bounded reweightings, in particular IPW]\label{cor:reweight}
Let Assumptions~\ref{ass:reg}--\ref{ass:mom} hold and let $W\ge0$ be a bounded random
weight with $\E W=1$. Then the reweighted Step-1 objective
$G_W(q)=\E[W\,h(q;q_Z)]$ attains its minimum over $\Lt$. If moreover Assumption~\ref{ass:warp}(iii-a) holds for the reweighted law $W\,dP$, the shape of any minimiser of $G_W$ minimises the reweighted Step-1 objective $M_W([q])=\E[W\,\del^2([q],[q_Z])]$ over $\Sw$.
\end{corollary}

\begin{proofenv}[Proof.]
\emph{Why the IPW weight qualifies.} It is $W=\ind_{A=a}/e_a(X)$, with
$e_a(X)=\Prob(A=a\mid X)$ the propensity score. Assumption~\ref{ass:pos} supplies $c_e>0$ with
$e_a(X)\ge c_e$ a.s., so $W\in[0,1/c_e]$ is non-negative and bounded. Its mean is $1$ by the
Horvitz--Thompson identity,
\[
  \E W=\E_{X}\!\Big[\tfrac{1}{e_a(X)}\,\E_{A\mid X}[\ind_{A=a}]\Big]
  =\E_{X}\!\Big[\tfrac{e_a(X)}{e_a(X)}\Big]=1 
  \qquad\big(\text{tower rule};\ \E_{A\mid X}[\ind_{A=a}]=e_a(X)\big),
\]
so $W\,dP$ is a probability measure (the reweighting to the target population) and $G_W$
is a genuine Fréchet function.

\emph{The direct method survives.} Non-negativity and boundedness of $W$ preserve every
ingredient of Steps~A--C of the proof above, together with Lemmas~\ref{lem:attain} and~\ref{lem:hcont}. The integrand $W\,h(q;q_Z)$ is measurable (a product of random
variables). Multiplying $h(q;q_Z)\ge(\norm q-\norm{q_Z})^2$ by $W$ and using $\E W=1$,
\[
  G_W(q)\ \ge\ \norm q^2-2\norm q\,\E[W\norm{q_Z}]+\E[W\norm{q_Z}^2],
\]
whose weighted moments are finite since $W\le\norm{W}_\infty$ bounds them by $\norm{W}_\infty$ times the unweighted ones
(Assumption~\ref{ass:mom}), so $G_W$ is coercive. Multiplying inequality
\eqref{eq:ptwiselsc} (at $y=q_Z(\omega)$) by $W(\omega)\ge0$ and applying Fatou (\ref{T:fatou}) gives weak
lower-semicontinuity, and $G_W(0)=\E[W\norm{q_Z}^2]<\infty$. The direct method
(Steps~B--C of the proof above) then runs directly, so the weighted Step-1 objective
attains its minimum. The transfer from $G_W$ to $M_W$ is the argument above the proof of Proposition~\ref{thm:exist}(i), run on the law $W\,dP$; Lemma~\ref{lem:attain}(b), on which it relies, is unaffected by the weight. This covers the IPW objective of Section~\ref{sec:ident}.
\hfill$\square$
\end{proofenv}

\subsection{Clause (ii) of Proposition~\ref{thm:exist}: the Step-2 phase mean exists and is unique (Karcher--Afsari re-derived)} \label{app:ka}

\paragraph{Why a re-derivation is needed.} The Karcher--Afsari theory
\citep{karcher_riemannian_1977,afsari_riemannian_2011} gives existence and uniqueness
of Fréchet means of concentrated laws on finite-dimensional complete manifolds. Its
uniqueness half rests on convexity of the squared distance and is dimension-free.
Its existence half extracts a convergent subsequence of a minimising sequence
through Hopf--Rinow compactness of closed balls, and this fails on the
infinite-dimensional sphere $\Sphere(\sqrt T)$, whose closed balls are not compact.
We therefore prove both halves from first principles, replacing compactness by
completeness combined with strong convexity, using only the geometry of
Appendix~\ref{app:assump} (Lemmas~\ref{lem:feasible}, \ref{lem:sphereconv}
and~\ref{lem:ballconv}).

\paragraph{Setup and notation.} The Step-2 average is the Fréchet mean,
on $\Sphere(\sqrt T)$, of the law of the optimal aligning warp:
\begin{itemize}[leftmargin=1.6em,itemsep=1pt]
\item $\psi^\star_{q_Z\to\tilde q}$ is a random element of the feasible set $\Psi(\Gt)$: for each
sample outcome $\omega$,
\[
  \psi^\star_{q_Z(\omega)\to\tilde q}\ =\ \sqrt{\dot\gamma^\star_{q_Z(\omega)\to\tilde q}},
\]
the warp-SRSF of the optimal warp $\gamma^\star_{q_Z(\omega)\to\tilde q}$ registering the
random data transform $q_Z(\omega)$ to the reference representative $\tilde q$.
The reference representative $\tilde q$ is a fixed representative of the population Step-1 mean
shape $\shp(P)$, a deterministic functional of the law $P$, the same for every
$\omega$, so the only randomness in $\psi^\star_{q_Z\to\tilde q}$ enters through $q_Z(\omega)$:
writing $\Phi_{\tilde q}:y\mapsto\psi^\star_{y\to\tilde q}$ for the
(fixed) registration map, $\psi^\star_{q_Z\to\tilde q}=\Phi_{\tilde q}\circ q_Z$. This map is well
defined and measurable: for the fixed $\tilde q$ the optimal warp is the unique
minimiser (Assumption~\ref{ass:warp}(iii)) of a jointly continuous objective over
the compact $\Gc$, hence a measurable function of the data by the
Kuratowski--Ryll-Nardzewski selection theorem (\ref{T:krn}; see
Appendix~\ref{app:ident}). So $\psi^\star_{q_Z\to\tilde q}$ is a $\Psi(\Gt)$-valued random
variable.
\item The objective is the \emph{Fréchet function}, that is, the Step-2 objective $N$ of Section~\ref{sec:lift} (notation of Appendix~\ref{app:toolbox}) written as a function of the warp-SRSF $\psi=\psi_\gamma$, for a fixed reference representative and without weights (weights are treated in step (5e)):
\[
N(\psi)=\E\big[\dG^2(\psi,\psi^\star_{q_Z\to\tilde q})\big],\qquad\psi\in\Psi(\Gt),
\]
to be minimised over the feasible set $\Psi(\Gt)$ (not over the whole sphere: Step~2
must return a warp).
\item Assumption~\ref{ass:conc}(4b) supplies a point $\psi_0\in\Psi(\Gt)$ (the centre
is itself in the feasible set $\Psi(\Gt)$; see Remark~\ref{rem:centre} for what happens without this) and
a radius $r_0<\pi\sqrt T/8$ such that $\psi^\star_{q_Z\to\tilde q}\in\bar B(\psi_0,r_0)$ almost surely,
where $\bar B(\psi_0,r)=\{\psi\in\Sphere(\sqrt T):\dG(\psi,\psi_0)\le r\}$. We call
any point of $\bar B(\psi_0,r_0)\cap\Psi(\Gt)$ a \emph{support point}.
\item Throughout, $r^\star=\pi\sqrt T/2$ denotes the quarter-turn radius: within
angular distance $<\pi/2$ of a point, the squared distance to that point is strongly
convex (Lemma~\ref{lem:sphereconv}), and balls of radius $<r^\star$ are geodesically
convex (Lemma~\ref{lem:ballconv}).
\end{itemize}
We show that $N$ has exactly one minimiser over $\Psi(\Gt)$, in four steps:
localisation, convexity of the reduced problem, uniqueness, existence.

Figure~\ref{fig:kalocal} illustrates the argument on a cross-section of the sphere.
Figure~\ref{fig:kalocal3d} shows the same localisations on the warp sphere itself,
together with the convexity modulus they give.

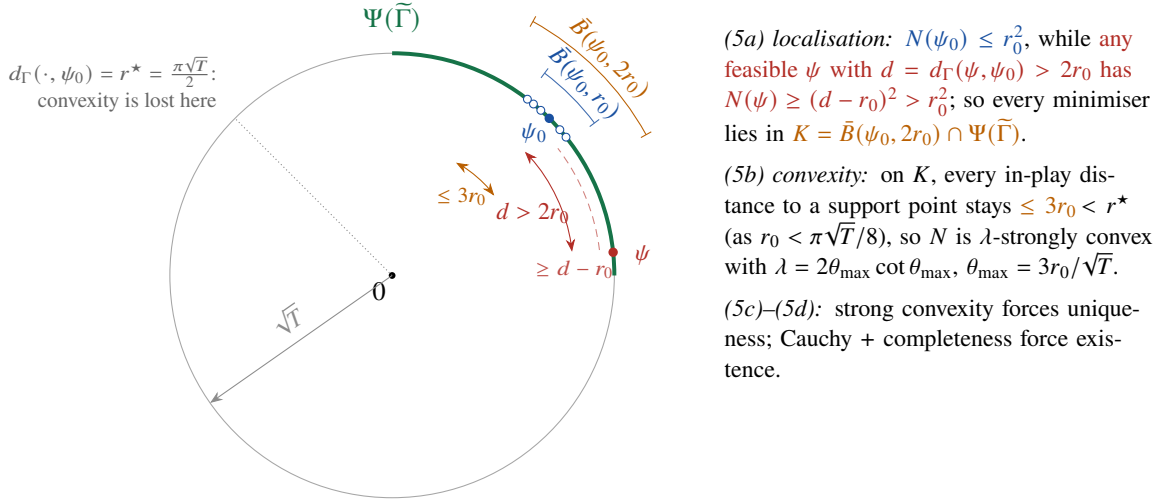
\begin{figure}[htbp]
\centering
\begin{tikzpicture}[>=Stealth, scale=0.98]
  \coordinate (O) at (0,0);
  \draw[black!35] (O) circle (3);
  \fill (O) circle (1.3pt); \node[below left=-2pt, font=\small] at (O) {$0$};
  \draw[->, black!45] (O) -- (215:3) node[midway, above, sloped, font=\footnotesize] {$\sqrt T$};
  \draw[line width=1.5pt, figgreen] (3,0) arc[start angle=0, end angle=90, radius=3];
  \node[font=\small, anchor=south, figgreen] at (90:3.14) {$\Psi(\Gt)$};
  \draw[black!50, densely dotted] (O) -- (135:3);
  \node[font=\scriptsize, black!60, anchor=south east, align=right] at (137:3.05)
    {$\dG(\cdot,\psi_0)=r^\star=\tfrac{\pi\sqrt T}{2}$:\\ convexity is lost here};
  \fill[figblue] (45:3) circle (1.8pt);
  \node[font=\footnotesize, figblue] at (45:2.72) {$\psi_0$};
  \filldraw[fill=white, draw=figblue, thin] (38.5:3) circle (1.5pt); \filldraw[fill=white, draw=figblue, thin] (41:3) circle (1.5pt); \filldraw[fill=white, draw=figblue, thin] (48:3) circle (1.5pt);
  \filldraw[fill=white, draw=figblue, thin] (50.5:3) circle (1.5pt); \filldraw[fill=white, draw=figblue, thin] (52.5:3) circle (1.5pt);
  \draw[figblue!85] (37:3.42) arc[start angle=37, end angle=53, radius=3.42];
  \draw[figblue!85] (37:3.34) -- (37:3.50); \draw[figblue!85] (53:3.34) -- (53:3.50);
  \node[font=\footnotesize, rotate=-45, figblue] at (45:3.68) {$\bar B(\psi_0,r_0)$};
  \draw[figorange!90!black] (29:3.92) arc[start angle=29, end angle=61, radius=3.92];
  \draw[figorange!90!black] (29:3.84) -- (29:4.00); \draw[figorange!90!black] (61:3.84) -- (61:4.00);
  \node[font=\footnotesize, rotate=-45, figorange!90!black] at (45:4.16) {$\bar B(\psi_0,2r_0)$};
  \fill[figred] (6:3) circle (1.8pt);
  \node[font=\footnotesize, anchor=west, figred] at (5:3.16) {$\psi$};
  \draw[<->, figred] (7.5:2.45) arc[start angle=7.5, end angle=43, radius=2.45];
  \node[font=\footnotesize, figred] at (24:2.08) {$d>2r_0$};
  \draw[figred!60, densely dashed] (7:2.82) arc[start angle=7, end angle=37.5, radius=2.82];
  \node[font=\scriptsize, figred!85] at (2:2.46) {$\ge d-r_0$};
  \draw[<->, figorange!90!black] (38:1.75) arc[start angle=38, end angle=60.5, radius=1.75];
  \node[font=\scriptsize, figorange!90!black] at (49:1.42) {$\le 3r_0$};
  \node[anchor=west, align=left, font=\footnotesize, text width=5.7cm] at (4.35,1.0) {%
    \emph{(5a) localisation:} {\color{figblue}$N(\psi_0)\le r_0^2$}, while {\color{figred}any feasible $\psi$ with $d=\dG(\psi,\psi_0)>2r_0$ has $N(\psi)\ge(d-r_0)^2>r_0^2$}; so every minimiser lies in {\color{figorange!90!black}$K=\bar B(\psi_0,2r_0)\cap\Psi(\Gt)$}.\\[4pt]
    \emph{(5b) convexity:} on $K$, every in-play distance to a support point stays {\color{figorange!90!black}$\le3r_0$}$\,<r^\star$ (as $r_0<\pi\sqrt T/8$), so $N$ is $\lambda$-strongly convex with $\lambda=2\theta_{\max}\cot\theta_{\max}$, $\theta_{\max}=3r_0/\sqrt T$.\\[4pt]
    \emph{(5c)--(5d):} strong convexity forces uniqueness; Cauchy $+$ completeness force existence.};
\end{tikzpicture}
\caption{\textbf{The localisation step of the Karcher--Afsari re-derivation}, on a cross-section of the warp sphere. The optimal-warp law is supported in the ball $\bar B(\psi_0,r_0)$ around the feasible centre $\psi_0$ (dots: support points of $\psi^\star_{q_Z\to\tilde q}$); a candidate $\psi$ at distance $d>2r_0$ from $\psi_0$ lies at distance $\ge d-r_0$ from every support point (dashed arc); inside $K$ all distances in play stay below $3r_0$, within the quarter-turn $r^\star=\pi\sqrt T/2$.}
\label{fig:kalocal}
\end{figure}

\begin{figure}[htbp]
\centering
\input{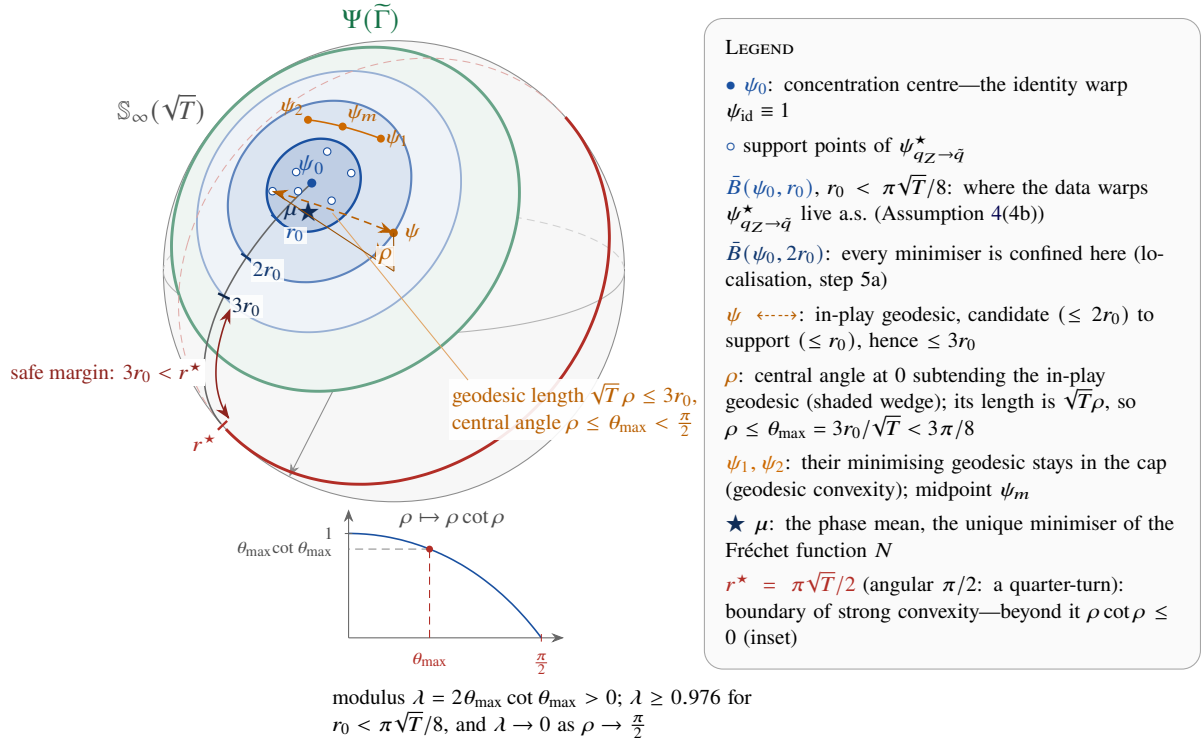}
\caption{\textbf{The localisations of the Step-2 re-derivation, on the warp
sphere itself} (3D companion to Figure~\ref{fig:kalocal}). Inside the
feasible set $\Psi(\Gt)$ (green): the support ball $\bar B(\psi_0,r_0)$
around a feasible centre $\psi_0$ (e.g.\ the identity), the coercivity ball $\bar B(\psi_0,2r_0)$
containing every minimiser, two candidates $\psi_1,\psi_2$ with midpoint
$\psi_m$, and the quarter-turn $r^\star=\pi\sqrt T/2$ (red), where
$\rho\cot\rho$ (inset) vanishes.}
\label{fig:kalocal3d}
\end{figure}

\begin{proofenv}[Re-derivation of Karcher--Afsari on the warp sphere (Step-2 existence and uniqueness).]
Steps (5a) to (5d) are the four steps announced above; step (5e) extends them to bounded reweightings.

\emph{(5a) Localisation: the problem happens inside a small ball.}
\emph{Intuition.} A candidate far from the data cloud is far from every support
point at once, so its average squared distance is large and it cannot beat the
cloud's own centre $\psi_0$, which is available as a competitor because
Assumption~\ref{ass:conc}(4b) makes it feasible. First bound the centre's own
objective value: since $\dG(\psi_0,\psi^\star_{q_Z\to\tilde q})\le r_0$ almost surely,
\begin{flalign*}
N(\psi_0)&=\E\big[\dG^2(\psi_0,\psi^\star_{q_Z\to\tilde q})\big]\ \le\ r_0^2 .
\why{monotonicity of $\E$; $\dG\le r_0$ a.s.}
\end{flalign*}
Next, let $\psi\in\Psi(\Gt)$ be any feasible candidate with
$d:=\dG(\psi,\psi_0)>2r_0$. For almost every realisation of $\psi^\star_{q_Z\to\tilde q}$,
\begin{flalign*}
\dG(\psi,\psi^\star_{q_Z\to\tilde q})&\ge \dG(\psi,\psi_0)-\dG(\psi_0,\psi^\star_{q_Z\to\tilde q})
\why{reverse triangle inequality}\\
&\ge d-r_0\ >\ r_0\ >\ 0 .
\why{$\psi^\star_{q_Z\to\tilde q}$ in the $r_0$-ball; $d>2r_0$}
\end{flalign*}
Both sides being non-negative, we may square and take expectations:
\begin{flalign*}
N(\psi)&\ge (d-r_0)^2\ >\ r_0^2\ \ge\ N(\psi_0).
\why{monotonicity of $\E$; then the two bounds above}
\end{flalign*}
The inequality $N(\psi)>N(\psi_0)$ is strict: no feasible point outside
$\bar B(\psi_0,2r_0)$ can minimise $N$ over $\Psi(\Gt)$, because the feasible
competitor $\psi_0$ already does strictly better. Therefore
\[
\inf_{\Psi(\Gt)}N=\inf_{K}N,
\qquad
K:=\bar B(\psi_0,2r_0)\cap\Psi(\Gt)\ \ni\ \psi_0 ,
\]
and every minimiser of $N$ over $\Psi(\Gt)$ (if one exists) lies in $K$. On $K$
all the distances that appear are small: for $\psi\in K$ and a.e.\ realisation of
$\psi^\star_{q_Z\to\tilde q}$,
\[
\dG(\psi,\psi^\star_{q_Z\to\tilde q})\le\dG(\psi,\psi_0)+\dG(\psi_0,\psi^\star_{q_Z\to\tilde q})\le2r_0+r_0=3r_0,
\]
by the triangle inequality. In particular $N\le9r_0^2<\infty$ on $K$.

\smallskip
\noindent\emph{(5b) The reduced problem is convex.} The set $K$ is the intersection
of two geodesically convex sets: the ball $\bar B(\psi_0,2r_0)$, convex by
Lemma~\ref{lem:ballconv} since its radius satisfies
$2r_0<2\cdot\pi\sqrt T/8=\pi\sqrt T/4<r^\star$, and the feasible set $\Psi(\Gt)$,
convex by the remark after Lemma~\ref{lem:ballconv}. Hence $K$ is geodesically
convex: the minimizing arc between any two of its points stays in $K$.
$K$ is also closed ($\Psi(\Gt)$ is closed by Lemma~\ref{lem:feasible}) and non-empty
($\psi_0\in K$).

Now fix two points $\psi_1,\psi_2\in K := \bar B(\psi_0,2r_0)\cap\Psi(\Gt)$ and let $\{\psi(s)\}_{s\in[0,L]}$ be the
unit-speed minimizing arc between them, $L=\dG(\psi_1,\psi_2)\le4r_0<\pi\sqrt T$. The
arc stays in $K$, so along it every in-play angular distance obeys
\[
\rho(s)=\frac{\dG(\psi(s),y)}{\sqrt T}\ \le\ \theta_{\max}:=\frac{3r_0}{\sqrt T}
\ <\ \frac{3\pi}{8}\ <\ \frac\pi2
\qquad\text{for every support point }y,
\]
using $r_0<\pi\sqrt T/8$. Lemma~\ref{lem:sphereconv} then gives, for each fixed
support point $y$, writing $g_y(s)=\dG^2(\psi(s),y)$:
\[
g_y''(s)\ \ge\ 2\,\theta_{\max}\cot\theta_{\max}\ =:\ \lambda\ >\ 0
\qquad\text{along the whole arc,}
\]
and the midpoint inequality (\ref{T:midpt}) turns this second-derivative bound into
\[
g_y\big(\tfrac L2\big)\ \le\ \tfrac12g_y(0)+\tfrac12g_y(L)-\tfrac{\lambda}{8}L^2 .
\]
This holds for every support point $y$, that is for almost every realisation
$y=\psi^\star_{q_Z(\omega)\to\tilde q}$. Taking expectations (allowed since all the $g_y$ values are
bounded by $9r_0^2$) yields the midpoint strong-convexity inequality for $N$ on
$K$:
\begin{equation}\label{eq:strongconv}
N(\psi_m)\ \le\ \tfrac12N(\psi_1)+\tfrac12N(\psi_2)-\tfrac{\lambda}{8}\,\dG^2(\psi_1,\psi_2),
\qquad \psi_m:=\text{midpoint of }\psi_1,\psi_2 .
\end{equation}
With $\theta_{\max}<3\pi/8$ and
$\rho\cot\rho$ decreasing (Lemma~\ref{lem:sphereconv}),
\[
\lambda\ \ge\ 2\cdot\frac{3\pi}{8}\cot\frac{3\pi}{8}
=\frac{3\pi}{4}(\sqrt2-1)\ \approx\ 0.976 ,
\]
a uniform, explicitly computable modulus. (We used
$\cot(3\pi/8)=\tan(\pi/8)=\sqrt2-1$.)

\smallskip
\noindent\emph{(5c) Uniqueness.} Suppose $\psi_1\ne\psi_2$ both minimised $N$ over
$\Psi(\Gt)$. By (5a) both lie in $K$, and their midpoint $\psi_m$ lies in $K$ (as
$K$ is geodesically convex), hence is feasible. Then \eqref{eq:strongconv} gives
\[
N(\psi_m)\le\min_{\Psi(\Gt)}N-\tfrac{\lambda}{8}\dG^2(\psi_1,\psi_2)
<\min_{\Psi(\Gt)}N ,
\]
a contradiction ($\psi_m$ would be a feasible point strictly below the minimum). So
the minimiser, if it exists, is unique.

\smallskip
\noindent\emph{(5d) Existence: completeness replaces compactness.}
\emph{Intuition.} In finite dimensions one extracts a convergent subsequence of a
minimising sequence by compactness, which is not available here. But strong
convexity forces any two near-minimisers to be close to each other, so a minimising
sequence must be Cauchy, and completeness then provides the limit. Let $N^\star=\inf_KN\ (\le N(\psi_0)\le r_0^2<\infty$, and $\ge0)$ and pick
$\psi_n\in K$ with $N(\psi_n)\to N^\star$. For indices $m,n$, let $p_{mn}\in K$ be
the midpoint of $\psi_m,\psi_n$ (in $K$ by geodesic convexity, hence feasible, hence
$N(p_{mn})\ge N^\star$). By \eqref{eq:strongconv},
\begin{flalign*}
\tfrac{\lambda}{8}\,\dG^2(\psi_m,\psi_n)
&\le\tfrac12N(\psi_m)+\tfrac12N(\psi_n)-N(p_{mn})
\why{midpoint strong-convexity \eqref{eq:strongconv}, rearranged}\\
&\le\tfrac12N(\psi_m)+\tfrac12N(\psi_n)-N^\star
\ \xrightarrow[m,n\to\infty]{}\ 0 ,
\why{$N(p_{mn})\ge N^\star$; $N(\psi_n)\to N^\star$}
\end{flalign*}
so $(\psi_n)$ is a Cauchy sequence for $\dG$. Since the angle between sphere points
is comparable to their chordal ($\Lt$) distance
($\norm{\psi-\varphi}\le\dG(\psi,\varphi)$, the chord being shorter than the arc),
$(\psi_n)$ is Cauchy in $\Lt$, hence converges in $\Lt$ to some $\psi^\dagger$.
Since $K$ is a closed subset of the complete space $\Lt$ and $\psi^\dagger$ is a
limit of points of $K$, we have $\psi^\dagger\in K$. Convergence also holds in $\dG$
($\dG(\psi_n,\psi^\dagger)=\sqrt T\arccos(\inner{\psi_n}{\psi^\dagger}/T)\to0$ by
continuity of the inner product and of $\arccos$). Finally $N$ is continuous on
$K$, in fact Lipschitz: for $\psi,\varphi\in K$,
\begingroup\setlength{\whywidth}{0.25\textwidth}%
\begin{flalign*}
|N(\psi)-N(\varphi)|
&\le\E\Big[\big|\dG^2(\psi,\psi^\star_{q_Z\to\tilde q})-\dG^2(\varphi,\psi^\star_{q_Z\to\tilde q})\big|\Big]
\why{Jensen/triangle for $\E$}\\
&=\E\Big[\big(\dG(\psi,\psi^\star_{q_Z\to\tilde q})+\dG(\varphi,\psi^\star_{q_Z\to\tilde q})\big)\\
&\qquad\times\big|\dG(\psi,\psi^\star_{q_Z\to\tilde q})-\dG(\varphi,\psi^\star_{q_Z\to\tilde q})\big|\Big]
\why{$|a^2-b^2|=(a+b)|a-b|$, $a,b\ge0$}\\
&\le 6r_0\,\dG(\psi,\varphi),
\why{in-play distances $\le3r_0$ (5a); reverse triangle inequality}
\end{flalign*}\endgroup
so as $\psi_n\to\psi^\dagger$, $N(\psi_n)\to N(\psi^\dagger)$. And as $\psi_n$ was defined so that $N(\psi_n)\to N^\star$, we have $N(\psi^\dagger)=N^\star$: the infimum is attained at
$\psi^\dagger\in K$. Combined with (5a) ($K$-minimisers are global feasible
minimisers, and no feasible minimiser exists outside $K$) and (5c) (uniqueness), $N$
has exactly one minimiser over $\Psi(\Gt)$. This is the Step-2 mean. By
Assumption~\ref{ass:conc}(4b) it moreover lies in $\Gc$, hence is bi-Lipschitz and invertible, which
is what the phase-centering step (undoing the mean warp)  uses.

\smallskip
\noindent\emph{(5e) Bounded reweightings (the IPW objective).} Let $W\ge0$ be a
bounded weight with $\E W=1$ (the IPW weight $W=\ind_{A=a}/e_a(X)\in[0,1/c_e]$
qualifies: $\E W=1$ by \ref{T:tower}), and consider
$N_W(\psi)=\E[W\,\dG^2(\psi,\psi^\star_{q_Z\to\tilde q})]$. All four moves go through directly:
\begin{itemize}[leftmargin=1.6em,itemsep=1pt]
\item the reweighted law of $\psi^\star_{q_Z\to\tilde q}$ (density $W$) is supported inside the
original support: if a set carries no mass, it carries no $W$-weighted mass; so the
concentration ball is unchanged;
\item localisation: $N_W(\psi_0)\le r_0^2\,\E W=r_0^2$ and
$N_W(\psi)\ge(d-r_0)^2\,\E W=(d-r_0)^2$ for $d>2r_0$, by the same pointwise bounds
multiplied by $W\ge0$ before taking expectations;
\item strong convexity: multiply the per-realisation midpoint inequality by
$W(\omega)\ge0$ and integrate; $\E W=1$ reproduces \eqref{eq:strongconv} with the
same modulus $\lambda$;
\item uniqueness and existence then repeat word for word.
\end{itemize}
\hfill$\square$
\end{proofenv}

\begin{remark}[If the ball centre is infeasible.]\label{rem:centre}
Assumption~\ref{ass:conc}(4b) centres the concentration ball at a feasible point
$\psi_0\in\Psi(\Gt)$ (naturally the identity, $\psi_0\equiv1$). This can be dropped.
Suppose only $\psi^\star_{q_Z\to\tilde q}\in\bar B(\psi_0,r_0)$ a.s.\ for an arbitrary
$\psi_0\in\Sphere(\sqrt T)$, still with $r_0<\pi\sqrt T/8$. Support points are feasible and
exist ($\psi^\star_{q_Z\to\tilde q}\in\Psi(\Gt)$ a.s., $\Psi(\Gt)$ closed by Lemma~\ref{lem:feasible}), so
pick one, $y_0$, as competitor in place of $\psi_0$. Being $r_0$ off-centre, $y_0$ inflates
every bound by one $r_0$: localisation lands in $\bar B(\psi_0,3r_0)$ and the largest
 distance in play rises to $4r_0$, so the convexity requirement becomes
$4r_0<r^\star=\pi\sqrt T/2$, that is $r_0<\pi\sqrt T/8$, which is exactly the
radius of (4b). Steps (5b)--(5e) then run verbatim with
$\theta_{\max}=4r_0/\sqrt T<\pi/2$. The price is only quantitative:
$\lambda=2\theta_{\max}\cot\theta_{\max}$ stays positive but decays to $0$ as
$r_0\uparrow\pi\sqrt T/8$, whereas a feasible centre keeps $\lambda\ge0.976$
uniformly.
\end{remark}

\begin{remark}[The radius constants: $\pi\sqrt T/6$ versus $\pi\sqrt T/8$.]\label{rem:radii}
With the feasible centre of (4b), the proof only needs the largest distance in
play to stay inside the quarter-turn: $3r_0<r^\star=\pi\sqrt T/2$, that is
$r_0<\pi\sqrt T/6$. This is sharp for this argument: beyond it the distance in
play reaches the hemisphere, where $\rho\cot\rho$ (hence $\lambda$) vanishes.
Assumption~(4b) takes the stronger $\pi\sqrt T/8$ for two reasons, both visible
here: it gives a uniform modulus
$\lambda\ge\tfrac{3\pi}{4}(\sqrt2-1)\approx0.976$, bounded away from $0$ (at
$r_0\uparrow\pi\sqrt T/6$ the modulus decays to $0$), and it is exactly the
threshold under which the result survives an infeasible ball centre
(Remark~\ref{rem:centre}).
\end{remark}

\section{Proofs of Theorem~\ref{thm:wellposed}, points (iv)--(v) and assembly, and of Proposition~\ref{prop:propsB}} \label{app:inv}

Point (iv) of Theorem~\ref{thm:wellposed} is the content of the following
proposition, and point (v) that of Lemma~\ref{lem:alignavg} below. The section closes with the assembly of the five points into the
theorem, and with the proof of Proposition~\ref{prop:propsB}.

\begin{proposition}\label{thm:inv}
Under Assumptions~\ref{ass:reg},~\ref{ass:warp} and~\ref{ass:conc}(4b), the phase-centered registered curve $Z\circ\gamma^\star_{q_Z\to\tilde q}\circ\bar\gamma^{-1}$ (for every realisation of $Z$), hence the recovered mean $\mu(P)$, does not depend on the admissible representative $\tilde q$ chosen in Step~2, within the class of representatives related to one another by invertible warps, which contains all the admissible representatives (see below).
\end{proposition}

\noindent\emph{Scope of the statement.} Step~2 of the construction picks a
representative $\tilde q$ of the Step-1 mean shape $\shp(P)$, and the phase-centered registered curves, which Step~3 averages, must not depend on this arbitrary choice. The proposition
asserts this independence within the class of representatives related to one
another by invertible warps: we compare $\tilde q$ with
$\tilde q'=(\tilde q,\gamma_0)$ where $\gamma_0$ is an invertible warp
(an element of $\G$, so $\gamma_0^{-1}$ exists and is again a warp).  This class contains all the admissible representatives. Indeed, an admissible representative is nowhere flat ($\tilde q\neq0$ almost everywhere), and the nowhere-flat representatives of one shape form exactly one $\G$-orbit: by \citet[Cor.~9 and Prop.~12]{bruveris_optimal_2016}, the shape is $\{(r,\gamma):\gamma\in\Gt\}$ for any nowhere-flat representative $r$, and $(r,\gamma)\neq0$ almost everywhere exactly when $\dot\gamma>0$ almost everywhere, that is, when $\gamma\in\G$. A representative with a flat piece is never admissible, since Assumption~\ref{ass:warp}(iii-b) requires the reference to be nowhere flat. So the scope of the statement matches the
scope of the construction. We first isolate two facts about the action.

\begin{lemma}[Right action; warp-SRSF equivariance]\label{lem:rightaction}
For $q\in\Lt$ and $\alpha,\beta\in\Gt$:
\emph{(a)}~$\alpha\circ\beta\in\Gt$ and $\big((q,\alpha),\beta\big)=(q,\alpha\circ\beta)$;
\emph{(b)}~the warp-SRSF of a composition is the warped warp-SRSF:
$\psi_{\alpha\circ\beta}=(\psi_\alpha,\beta)$;
\emph{(c)}~for fixed $\beta\in\Gt$ and any $u,v\in\Lt$,
$\norm{(u,\beta)-(v,\beta)}=\norm{u-v}$.
\end{lemma}

\begin{proofenv}[Proof.]
The group case is standard in the SRSF literature
\citep[Ch.~4]{srivastava_klassen_2016}, and the monoid case only needs the chain
rule and change of variables of \ref{T:cov}. We prove it for completeness.

\emph{(a)} $\alpha\circ\beta$ is non-decreasing (composition of non-decreasing maps),
surjective onto $[0,T]$ (composition of surjections), and absolutely continuous:
$\alpha$ is absolutely continuous and $\beta$ is monotone absolutely continuous, so
$\alpha\circ\beta$ is absolutely continuous with
$(\alpha\circ\beta)^{\textstyle\cdot}=(\dot\alpha\circ\beta)\,\dot\beta$ a.e.\
(\ref{T:cov}(b)). Hence $\alpha\circ\beta\in\Gt$. For the action identity, compute
pointwise:
\begingroup\setlength{\whywidth}{0.37\textwidth}%
\begin{flalign*}
\big((q,\alpha),\beta\big)(t)
&=\big[(q\circ\alpha)\sqrt{\dot\alpha}\big](\beta(t))\cdot\sqrt{\dot\beta(t)}
\why{definition of the action, applied twice}\\
&=q\big(\alpha(\beta(t))\big)\,\sqrt{\dot\alpha(\beta(t))\,\dot\beta(t)}
\why{$\sqrt a\sqrt b=\sqrt{ab}$, $a,b\ge0$}\\
&=q\big((\alpha\circ\beta)(t)\big)\,\sqrt{(\alpha\circ\beta)^{\textstyle\cdot}(t)}
=(q,\alpha\circ\beta)(t)
\why{chain rule above (\ref{T:cov}(b))}
\end{flalign*}\endgroup
almost everywhere (on $\{\dot\beta=0\}$ both sides vanish, by the conventions of
\ref{T:cov}).

\emph{(b)} Take $q\equiv1$ in the computation above, or directly:
$\psi_{\alpha\circ\beta}=\sqrt{(\alpha\circ\beta)^{\textstyle\cdot}}
=\sqrt{\dot\alpha(\beta)}\sqrt{\dot\beta}
=(\psi_\alpha\circ\beta)\sqrt{\dot\beta}=(\psi_\alpha,\beta)$.

\emph{(c)} By linearity of the action and norm preservation (Lemma~\ref{lem:iso}):
$\norm{(u,\beta)-(v,\beta)}=\norm{(u-v,\beta)}=\norm{u-v}$. \hfill$\square$
\end{proofenv}

Figure~\ref{fig:invcomm} summarises the proof: the change of reference representative propagates down the construction and cancels in Step~3, when the registered curves are re-timed.

\begin{figure}[htbp]
\centering
\begin{tikzpicture}[>=Stealth, every node/.style={font=\small}]
  \node[figblue] (T1) at (0,0)     {$\tilde q$};
  \node[figred] (T2) at (8.8,0)   {$\tilde q'=(\tilde q,\gamma_0)$};
  \node[figblue] (R1) at (0,-2.0)  {$\gamma^\star_{q_Z\to\tilde q}$};
  \node[figred] (R2) at (8.8,-2.0){$\gamma^\star_{q_Z\to\tilde q'}=\gamma^\star_{q_Z\to\tilde q}\circ\gamma_0$};
  \node[figblue] (M1) at (0,-4.0)  {$\bar\gamma$};
  \node[figred] (M2) at (8.8,-4.0){$\bar\gamma'=\bar\gamma\circ\gamma_0$};
  \node (Q)  at (4.4,-5.9){$Z\circ\gamma^\star_{q_Z\to\tilde q}\circ\bar\gamma^{-1}=Z\circ\gamma^\star_{q_Z\to\tilde q'}\circ(\bar\gamma')^{-1}$};
  \draw[->] (T1) -- (T2) node[midway, above=2pt, font=\footnotesize, figorange] {$(\,\cdot\,,\gamma_0)$}
                         node[midway, below=2pt, font=\scriptsize\itshape] {change of reference representative ($\gamma_0$ invertible)};
  \draw[->] (R1) -- (R2) node[midway, above=2pt, font=\footnotesize, figorange] {$\circ\,\gamma_0$}
                         node[midway, below=2pt, font=\scriptsize\itshape] {step (1): re-index the registration over $\Gt$};
  \draw[->] (M1) -- (M2) node[midway, above=2pt, font=\footnotesize, figorange] {$\circ\,\gamma_0$}
                         node[midway, below=2pt, font=\scriptsize\itshape] {step (2): equivariance under the isometry $R_{\gamma_0}$};
  \draw[->] (T1) -- (R1) node[midway, left=2pt, font=\scriptsize\itshape] {register $q_Z$};
  \draw[->] (T2) -- (R2) node[midway, right=2pt, font=\scriptsize\itshape] {register $q_Z$};
  \draw[->] (R1) -- (M1) node[midway, left=2pt, font=\scriptsize\itshape, align=right] {Fréchet mean\\ on $\Sphere(\sqrt T)$};
  \draw[->] (R2) -- (M2) node[midway, right=2pt, font=\scriptsize\itshape, align=left] {Fréchet mean\\ on $\Sphere(\sqrt T)$};
  \draw[->] (M1) -- (Q) node[pos=0.68, sloped, above=1pt, font=\scriptsize\itshape] {re-time: $\circ\,\bar\gamma^{-1}$};
  \draw[->] (M2) -- (Q) node[pos=0.68, sloped, above=1pt, font=\scriptsize\itshape] {re-time: $\circ\,(\bar\gamma')^{-1}$};
  \node[font=\footnotesize, black!70, align=center] at (4.4,-6.9)
    {step (3), the cancellation:\quad
     $Z\circ\gamma^\star_{q_Z\to\tilde q'}\circ(\bar\gamma')^{-1}
      =Z\circ\gamma^\star_{q_Z\to\tilde q}\circ{\color{figorange}\gamma_0\circ\gamma_0^{-1}}\circ\bar\gamma^{-1}
      =Z\circ\gamma^\star_{q_Z\to\tilde q}\circ\bar\gamma^{-1}$};
\end{tikzpicture}
\caption{\textbf{Independence of the recovered mean from the reference representative} (Proposition~\ref{thm:inv}). Replacing the reference representative $\tilde q$ by $\tilde q'=(\tilde q,\gamma_0)$ shifts every optimal warp, and their Fréchet mean, by $\gamma_0$ (horizontal arrows); re-timing the registered curve by the inverse mean warp cancels the two copies of $\gamma_0$ (bottom), and both columns end at the same phase-centered registered curve, for every realisation of $Z$.}
\label{fig:invcomm}
\end{figure}
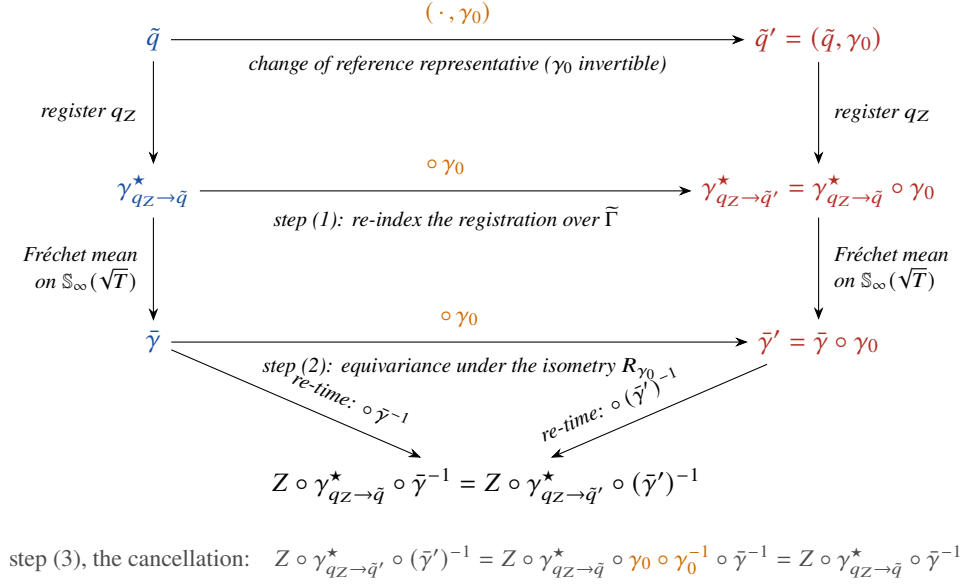

\begin{proofenv}[Proof of Proposition~\ref{thm:inv}.]
The result is new: representative-invariance of the lifted three-step mean is one
of this paper's contributions (Theorem~\ref{thm:wellposed}(iv)). The only outside
input is the equivariance fact \ref{T:isom}; everything else comes from
Lemma~\ref{lem:rightaction} and the results of this annex.

\emph{What we must show.} Let $\tilde q$ be the representative used in
Step~2, and $\tilde q'=(\tilde q,\gamma_0)$ another admissible representative,
$\gamma_0$ an invertible warp of $\G$, with $\gamma_0^{-1}\in\G\subset\Gt$. Write
$\bar\gamma$ (resp.\ $\bar\gamma'$) for the Step-2 warp mean computed from the
reference representative $\tilde q$ (resp.\ $\tilde q'$). We must show that, for every realisation of $Z$,
$Z\circ\gamma^\star_{q_Z\to\tilde q'}\circ(\bar\gamma')^{-1}=Z\circ\gamma^\star_{q_Z\to\tilde q}\circ\bar\gamma^{-1}$. The plan is to track
the change of reference representative through (1)~the registrations and (2)~the warp mean, and to
see the two copies of $\gamma_0$ cancel in (3).

\smallskip
\noindent\emph{(1) Registrations transform by right-composition.} Fix a realisation
of the data transform $q_Z$. For any $\gamma\in\Gt$,
\begin{flalign*}
\norm{\tilde q'-(q_Z,\gamma)}
&=\norm{(\tilde q,\gamma_0)-(q_Z,\gamma)}
\why{definition of $\tilde q'$}\\
&=\norm{\big((\tilde q,\gamma_0),\gamma_0^{-1}\big)
-\big((q_Z,\gamma),\gamma_0^{-1}\big)}
\why{Lemma~\ref{lem:rightaction}(c) with $\beta=\gamma_0^{-1}$}\\
&=\norm{\tilde q-(q_Z,\gamma\circ\gamma_0^{-1})}.
\why{Lemma~\ref{lem:rightaction}(a); $\gamma_0\circ\gamma_0^{-1}=\gid$}
\end{flalign*}
Now minimise both sides over $\gamma\in\Gt$. On the right, substitute
$\eta=\gamma\circ\gamma_0^{-1}$. Since $\gamma_0$ is invertible, both $\gamma_0$ and
$\gamma_0^{-1}$ are in $\Gt$, so $\gamma\mapsto\gamma\circ\gamma_0^{-1}$ is a bijection of
$\Gt$ onto itself (composition of warps stays in $\Gt$, Lemma~\ref{lem:rightaction}(a)).
The substitution only relabels the index set, so the minimum is unchanged:
\[
\min_{\gamma\in\Gt}\norm{\tilde q'-(q_Z,\gamma)}=\min_{\eta\in\Gt}\norm{\tilde q-(q_Z,\eta)} .
\]
The right-hand minimum is attained at the unique $\eta^\star=\gamma^\star_{q_Z\to\tilde q}$
(Assumption~\ref{ass:warp}(iii)). Translating back via $\gamma=\eta\circ\gamma_0$ gives the
unique left-hand optimum
\[
\gamma^\star_{q_Z\to\tilde q'}=\gamma^\star_{q_Z\to\tilde q}\circ\gamma_0 .
\]
(This is why Step~2 is posed over the monoid $\Gt$, not $\Gc$: the composition
$\gamma^\star_{q_Z\to\tilde q}\circ\gamma_0$ may leave the band of $\Gc$, which is not a
group (Remark~\ref{rem:nogroup}). Only uniqueness of the $\Gt$-optimum is used here, and
that is what clause~(iii) provides.)

\smallskip
\noindent\emph{(2) The warp mean is equivariant.} Pass to warp-SRSFs: by
Lemma~\ref{lem:rightaction}(b), the optimal warp-SRSF for the new reference representative is
\[
\psi^\star_{q_Z\to\tilde q'}
=\psi_{\gamma^\star_{q_Z\to\tilde q}\circ\gamma_0}
=\big(\psi^\star_{q_Z\to\tilde q},\gamma_0\big)
=R_{\gamma_0}\,\psi^\star_{q_Z\to\tilde q},
\qquad
R_{\gamma_0}:\psi\mapsto(\psi,\gamma_0).
\]
So the law of the new optimal warp-SRSF is the image (push-forward) under
$R_{\gamma_0}$ of the old one. The map $R_{\gamma_0}$ is a bijective isometry of the
warp sphere that preserves the feasible set:
\begin{itemize}[leftmargin=1.6em,itemsep=1pt]
\item \emph{isometry:} it preserves inner products (Lemma~\ref{lem:iso}), hence
norms, hence the spherical distance
$\dG(\psi_1,\psi_2)=\sqrt T\arccos(\inner{\psi_1}{\psi_2}/T)$;
\item \emph{bijective on $\Psi(\Gt)$:} $R_{\gamma_0}\psi_\gamma=\psi_{\gamma\circ\gamma_0}$
maps $\Psi(\Gt)$ into itself and has inverse $R_{\gamma_0^{-1}}$
($R_{\gamma_0^{-1}}R_{\gamma_0}\psi=(\psi,\gamma_0\circ\gamma_0^{-1})=\psi$ by
Lemma~\ref{lem:rightaction}(a)); note that $\Psi(\Gc)$ would not be preserved
(the band constraint above), which is another reason Step~2 is posed on $\Gt$;
\item \emph{the concentration hypothesis is preserved:} the image of the ball
$\bar B(\psi_0,r_0)$ is the ball $\bar B(R_{\gamma_0}\psi_0,r_0)$ of the same radius
(isometry), and its centre $R_{\gamma_0}\psi_0=(\psi_0,\gamma_0)$ is again in the feasible set $\Psi(\Gt)$. So Assumption~\ref{ass:conc}(4b) holds for the transformed law, and the
transformed Step-2 mean is also unique (Appendix~\ref{app:ka}).
\end{itemize}
By isometry-equivariance of unique Fréchet means (\ref{T:isom}, applied with
$\mathcal M=\Sphere(\sqrt T)$, $C=\Psi(\Gt)$, $R=R_{\gamma_0}$), the new Fréchet mean is the
image of the old one:
\[
\psi_{\bar\gamma'}=R_{\gamma_0}\psi_{\bar\gamma}=\psi_{\bar\gamma\circ\gamma_0}
\qquad\Longrightarrow\qquad
\bar\gamma'=\bar\gamma\circ\gamma_0 ,
\]
where the last implication uses warps that are recovered from their SRSF
($\gamma=\int_0^\cdot\psi_\gamma^2$, Lemma~\ref{lem:feasible}: the correspondence is
bijective).

\smallskip
\noindent\emph{(3) Cancellation.} The mean warp $\bar\gamma$ is invertible
(Assumption~\ref{ass:conc}(4b)), and so is $\bar\gamma'=\bar\gamma\circ\gamma_0$
(composition of invertibles), with
$(\bar\gamma\circ\gamma_0)^{-1}=\gamma_0^{-1}\circ\bar\gamma^{-1}$. Now re-time the curve registered to the new reference representative, for a fixed realisation of $Z$:
\begin{flalign*}
Z\circ\gamma^\star_{q_Z\to\tilde q'}\circ(\bar\gamma')^{-1}
&=Z\circ\big(\gamma^\star_{q_Z\to\tilde q}\circ\gamma_0\big)\circ\big(\gamma_0^{-1}\circ\bar\gamma^{-1}\big)
\why{Steps (1)--(2)}\\
&=Z\circ\gamma^\star_{q_Z\to\tilde q}\circ\big(\gamma_0\circ\gamma_0^{-1}\big)\circ\bar\gamma^{-1}
\why{composition of maps is associative}\\
&=Z\circ\gamma^\star_{q_Z\to\tilde q}\circ\bar\gamma^{-1},
\why{$\gamma_0\circ\gamma_0^{-1}=\gid$}
\end{flalign*}
which is the phase-centered registered curve obtained from the original reference representative. The change of reference representative has cancelled out
exactly, for every realisation of $Z$, so the law of the phase-centered registered curve, and its expectation $\mu(P)$, are unchanged. The same computation on transforms, with Lemma~\ref{lem:rightaction}(a) in place of associativity, shows that the phase-centered representative $(\tilde q,\bar\gamma^{-1})$ is unchanged as well. \hfill$\square$
\end{proofenv}

\begin{lemma}[The average of the phase-centered registered curves is a curve]\label{lem:alignavg}
Let Assumptions~\ref{ass:reg}--\ref{ass:mom} and~\ref{ass:conc}(4b) hold for the law $P$ of $Z$, let $\tilde q$ be an admissible representative of the Step-1 mean shape, $\bar\gamma$ the Step-2 mean warp, and $\eta_Z=\gamma^\star_{q_Z\to\tilde q}\circ\bar\gamma^{-1}$. Then:
\begin{enumerate}[leftmargin=2em,itemsep=2pt,topsep=2pt]
\item[(a)] for every realisation, $\eta_Z\in\Gt$, $Z\circ\eta_Z\in\F$, $(Z\circ\eta_Z)(0)=Z(0)$, $(Z\circ\eta_Z)(T)=Z(T)$, $\int_0^T|(Z\circ\eta_Z)^{\textstyle\cdot}|=\norm{q_Z}^2$, and
$\sup_{t}|(Z\circ\eta_Z)(t)|\le|Z(0)|+\norm{q_Z}^2$;
\item[(b)] for every $t$, $(Z\circ\eta_Z)(t)$ is a measurable function of $Z$;
\item[(c)] $\mu(P)(t)=\E_P[(Z\circ\eta_Z)(t)]$ is finite for every $t$, $\mu(P)\in\F$, $\mu(P)(0)=\E_P[Z(0)]$, $\mu(P)(T)=\E_P[Z(T)]$, and $\int_0^T|\dot\mu(P)|\le\E_P\norm{q_Z}^2$.
\end{enumerate}
\end{lemma}

\begin{proofenv}[Proof.]
\emph{(a)} The registration $\gamma^\star_{q_Z\to\tilde q}$ lies in $\Gt$ (in $\Gc$ for the representative to which Assumption~\ref{ass:warp}(iii) applies; only membership in $\Gt$ is used here) and $\bar\gamma^{-1}$ is a bi-Lipschitz warp (Assumption~\ref{ass:conc}(4b)), so $\eta_Z\in\Gt$ by Lemma~\ref{lem:rightaction}(a). By \ref{T:cov}(b), $Z\circ\eta_Z$ is absolutely continuous with derivative $(\dot Z\circ\eta_Z)\,\dot\eta_Z$, and by \ref{T:cov}(a) with $g=|\dot Z|$,
\[
\int_0^T\big|(Z\circ\eta_Z)^{\textstyle\cdot}\big|=\int_0^T|\dot Z\circ\eta_Z|\,\dot\eta_Z=\int_0^T|\dot Z|=\norm{q_Z}^2 ,
\]
the last equality by Proposition~\ref{prop:srsf}(i). Since $\eta_Z(0)=0$ and $\eta_Z(T)=T$, the end values are those of $Z$, and for every $t$,
$|(Z\circ\eta_Z)(t)|\le|Z(0)|+\int_0^t|(Z\circ\eta_Z)^{\textstyle\cdot}|\le|Z(0)|+\norm{q_Z}^2$ by \ref{T:ftc}.

\emph{(b)} The map $Z\mapsto\gamma^\star_{q_Z\to\tilde q}$ is measurable into the compact $\Gc$, equipped with the sup norm (the measurable selection argument of Appendix~\ref{app:ident}, \ref{T:krn}). The map $(f,\gamma)\mapsto f\big(\gamma(\bar\gamma^{-1}(t))\big)$ is continuous from $(C^0([0,T]),\norm{\cdot}_\infty)\times(\Gc,\norm{\cdot}_\infty)$ to $\R$, since $f$ is uniformly continuous (\ref{T:heine}), and $Z\mapsto Z$ is continuous from $(\F,d_\F)$ to $C^0([0,T])$, since $\norm{f_1-f_2}_\infty\le d_\F(f_1,f_2)$. The composition is measurable.

\emph{(c)} By (a), $|(Z\circ\eta_Z)(t)|\le|Z(0)|+\norm{q_Z}^2$, which is integrable by Assumption~\ref{ass:mom}, so $\mu(P)(t)$ is finite, and the end values follow from (a) by taking expectations. For absolute continuity, define, for Borel sets $B\subseteq[0,T]$,
\[
\nu(B)=\E_P\Big[\int_B\big|(Z\circ\eta_Z)^{\textstyle\cdot}(u)\big|\,du\Big].
\]
For an interval $B=[s,t]$, the inner integral is the total variation of the continuous curve $Z\circ\eta_Z$ on $[s,t]$, a supremum of sums of increments over partitions with rational interior points, hence a measurable function of $Z$ by (b); a monotone-class argument extends measurability to every Borel $B$, and monotone convergence makes $\nu$ a measure. It is finite, $\nu([0,T])=\E_P\norm{q_Z}^2$ by (a), and absolutely continuous with respect to the Lebesgue measure, since the inner integral vanishes on every Lebesgue-null set. Hence, for every $\varepsilon>0$ there is $\delta>0$ with $\nu(U)<\varepsilon$ whenever $U$ has Lebesgue measure less than $\delta$: otherwise there would be Borel sets $U_k$ of Lebesgue measure less than $2^{-k}$ with $\nu(U_k)\ge\varepsilon$, and their upper limit $\bigcap_j\bigcup_{k\ge j}U_k$ would be Lebesgue-null (Borel--Cantelli) with $\nu$-measure at least $\varepsilon$ (continuity from above of the finite measure $\nu$), a contradiction. For pairwise disjoint intervals $[s_k,t_k]$ of total length less than $\delta$,
\[
\sum_k\big|\mu(P)(t_k)-\mu(P)(s_k)\big|
\le\sum_k\E_P\big|(Z\circ\eta_Z)(t_k)-(Z\circ\eta_Z)(s_k)\big|
\le\nu\Big(\textstyle\bigcup_k[s_k,t_k]\Big)<\varepsilon ,
\]
so $\mu(P)$ is absolutely continuous, i.e.\ $\mu(P)\in\F$. The same inequality over a partition of $[0,T]$ bounds its total variation, $\int_0^T|\dot\mu(P)|\le\nu([0,T])=\E_P\norm{q_Z}^2$. \hfill$\square$
\end{proofenv}

\begin{proofenv}[Proof of Theorem~\ref{thm:wellposed} (the five points together).]
We put together points (i)--(iii) from Appendix~\ref{app:exist}, point (iv)
from the proposition above, and point (v) from Lemma~\ref{lem:alignavg} above. No outside input is needed beyond the assumptions.
``Well-posed'' means three things: the lifted mean $\mu(P)$ exists, is unique, and
does not depend on the arbitrary choice made in its construction (the reference representative). 

\smallskip
\emph{Existence and uniqueness of the Step-1 mean shape.} Proposition~\ref{thm:exist}(i)
(proved in Appendix~\ref{app:exist}) gives a non-empty set of minimising shapes,
and Assumption~\ref{ass:conc}(4a) says this set is a single shape $\shp(P)$.

\emph{Existence and uniqueness of the Step-2 mean.} Fix any admissible representative
$\tilde q\in\shp(P)$. The optimal warps $\gamma^\star_{q_Z\to\tilde q}$ exist, are
unique (Assumption~\ref{ass:warp}(iii)) and are measurable in the data
(\ref{T:krn}; see Appendix~\ref{app:ident}), so their law on the warp sphere is
well defined. Under Assumption~\ref{ass:conc}(4b) the Fréchet mean $\bar\gamma$ of
that law exists, is unique, and is invertible (Proposition~\ref{thm:exist}(ii),
proved in Appendix~\ref{app:ka}).

\emph{Determination of the curve.} The phase-centered registered curve $Z\circ\gamma^\star_{q_Z\to\tilde q}\circ\bar\gamma^{-1}$ is thereby determined for every realisation of $Z$. By Lemma~\ref{lem:alignavg} above, it is a random element with an integrable bound, its expectation $\mu(P)(t)$ is finite for every $t$, and $\mu(P)\in\F$.

\emph{Independence of the reference representative.} By Proposition~\ref{thm:inv} (just proved),
replacing $\tilde q$ by any other admissible representative leaves every phase-centered registered curve, hence
$\mu(P)$, unchanged.

Assembling: $\mu(P)$ exists, is unique, and is well-defined as a functional of the
law $P$. \hfill$\square$
\end{proofenv}

\subsection*{Proof of Proposition~\ref{prop:propsB}}\label{app:propsB}

\emph{(a)} is Lemma~\ref{lem:alignavg}(c).

\emph{(b)} If $\gamma^\star_{q_Z\to\tilde q}=\gid$ almost surely, the law of the optimal warp-SRSF is the point mass at $\psi_{\gid}\equiv1$, whose Fréchet function $\psi\mapsto\dG^2(\psi,1)$ has the unique minimiser $1$, so $\bar\gamma=\gid$. The phase-centered registered curve is then $Z$ itself, and $\mu(P)=\E_P[Z]$.

\emph{(c)} This is the population version of Theorems~1--3 of \citet{srivastava_registration_2011}, whose model $c_i\,(g\circ\gamma_i)+e_i$ is ours with $s_Z$ and $c_Z$ in place of $c_i$ and $e_i$, and whose proofs carry over with expectations in place of sample averages. Write $q_g$ for the transform of $g$ and $\kappa=\E[\sqrt{s_Z}]$. Their Theorem~1 gives the Step-1 mean shape $[\kappa\,q_g]$. For the representative $\tilde q=\kappa\,q_g$, the proof of their Theorem~2 gives $\gamma^\star_{q_Z\to\tilde q}=\gamma_Z^{-1}$, and the proof of their Theorem~3 gives the phase-centered registered curves $c_Z+s_Z\,(g\circ\bar\gamma^{-1})$, with $\bar\gamma$ the Fréchet mean of the warps $\gamma_Z^{-1}$. Taking expectations gives $\mu(P)$. They assume that this mean warp is the identity, which we do not, so the template keeps the mean timing $\bar\gamma^{-1}$. Finally, the Karcher mean function has transform $(\kappa\,q_g,\bar\gamma^{-1})$, hence the scale $\kappa^2=(\E\sqrt{s_Z})^2\le\E[s_Z]$ by Jensen's inequality. \hfill$\square$

\section{Proofs of Proposition~\ref{prop:objident} and Theorem~\ref{thm:ident} (identification of $\tau_{\mathrm{pop}}$)} \label{app:ident}

The identification has two layers. First, the $\Pstar$-level and observed-data
objectives are equal (Proposition~\ref{prop:objident}). Then the argmin principle
below promotes that equality to the minimisers. Throughout, fix an arm
$a\in\{0,1\}$ and a candidate shape $[q]\in\Sw$, and write
\[
M^{(a)}([q])=\E_{\Pstar}\!\big[\del^2([q],[q^{(a)}])\big], 
\qquad
M^{\mathrm{IPW}}_a([q])
=\E_{P}\!\Big[\tfrac{\ind_{A=a}}{e_a(X)}\,\del^2([q],[q_Y])\Big].
\]
Assumptions~\ref{ass:sutva}--\ref{ass:exch} enter in the following forms. Consistency
(Assumption~\ref{ass:sutva}): on the event $\{A=a\}$ the observed curve $Y$ is the
potential curve $Y^{(a)}$. Positivity (Assumption~\ref{ass:pos}): $e_a(X)\ge c_e>0$
a.s., so conditioning on $\{A=a,X\}$ is well defined and the weight is bounded by
$1/c_e$. Conditional exchangeability at the curve level (Assumption~\ref{ass:exch}):
$Y^{(a)}\perp\!\!\!\perp A\mid X$. The last one extends to every measurable function
of $Y^{(a)}$: the SRSF $q^{(a)}$, the squared distance $\del^2([q],[q^{(a)}])$, the
optimal warp of $q^{(a)}$ to a fixed reference representative and the baseline $Y^{(a)}(0)$ and, for a fixed reference representative, a fixed invertible warp $\beta$ and a fixed $t$, the value $(Y^{(a)}\circ\gamma^\star_{q^{(a)}\to\tilde q_a}\circ\beta)(t)$ are all
conditionally independent of $A$ given $X$. 

\smallskip
\noindent\emph{Preliminaries.} Every expectation written below is finite.
Since $\gid\in\Gt$, the elastic distance is dominated by the unwarped $\Lt$ distance,
$\del([q],[y])\le\norm{q-y}\le\norm q+\norm y$ by the triangle inequality, so
\[
D^{(a)}:=\del^2\big([q],[q^{(a)}]\big)\ \le\ \big(\norm q+\norm{q^{(a)}}\big)^2
\ \le\ 2\norm q^2+2\norm{q^{(a)}}^2 ,
\]
which is integrable by Assumption~\ref{ass:mom} for the law of $Y^{(a)}$. The
weighted version is integrable because the weight is bounded by $1/c_e$. The Step-2
integrand $\dG^2$ is bounded by $(\pi\sqrt T)^2$, the diameter of the sphere, and the
baseline $|Y^{(a)}(0)|$ is integrable by the baseline clause of
Assumption~\ref{ass:mom}. Finally $D^{(a)}$ is a measurable function of $q^{(a)}$,
because $y\mapsto\del^2([q],[y])$ is continuous ($y\mapsto\del([q],[y])$ is $1$-Lipschitz, by the argument of Lemma~\ref{lem:hcont} with $\Gt$ in place of $\Gc$) and $q^{(a)}$ is a random element of
$\Lt$ (Proposition~\ref{prop:srsf}(iii)).

\begin{proofenv}[Proof of Proposition~\ref{prop:objident} (objective identification).]
The conditioning argument is the classical inverse-probability identity of \citet{horvitz_generalization_1952} (\ref{T:tower}, see also \citealp{imbens_rubin_2015}). What is new is where we apply it: inside the three components of the phase-centered elastic mean, the Step-1 amplitude objective, the Step-2 phase objective and the Step-3 average of the phase-centered registered curves.

\emph{The Step-1 (amplitude) objective.}
The objective $M^{(a)}$ is written with the unobservable potential outcome
$q^{(a)}$, and we must re-express it with observable quantities only: $(X,A,q_Y)$ and
the propensity $e_a$. The computation is the classical Horvitz--Thompson one, with the integrand $\del^2([q],\cdot)$ in place of a scalar
outcome. Write $D^{(a)}=\del^2([q],[q^{(a)}])$ for the potential transform and
$D_Y=\del^2([q],[q_Y])$ for the observed one. Three facts let the chain below run.
\emph{(Measurability.)} $D^{(a)}$ is a measurable function of $q^{(a)}$
(Preliminaries), so Assumption~\ref{ass:exch} applies to it.
\emph{(Integrability.)} $D^{(a)}$ and its weighted version are integrable
(Preliminaries). \emph{(Consistency.)} On $\{A=a\}$ we have $q_Y=q^{(a)}$, so
$D_Y=D^{(a)}$. Then:
\begin{flalign*}
\E_{\Pstar}\big[D^{(a)}\big]
&=\E_X\Big[\E\big(D^{(a)}\bigm|X\big)\Big]
\why{tower rule \ref{T:tower}}\\
&=\E_X\Big[\E\big(D^{(a)}\bigm|A{=}a,X\big)\Big]
\why{$D^{(a)}\perp\!\!\!\perp A\mid X$ (function of $Y^{(a)}$; Assumption~\ref{ass:exch}); conditioning legitimate as $e_a(X)>0$}\\
&=\E_X\Big[\E\big(D_Y\bigm|A{=}a,X\big)\Big]
\why{consistency: on $\{A{=}a\}$, $Y=Y^{(a)}$, hence $q_Y=q^{(a)}$ and $D_Y=D^{(a)}$}\\
&=\E_X\Big[e_a(X)^{-1}\,\E\big(\ind_{A=a}\,D_Y\bigm|X\big)\Big]
\why{propensity identity \ref{T:tower}: $\E[\ind_{A=a}g\mid X]=e_a(X)\,\E[g\mid A{=}a,X]$}\\
&=\E\Big[\tfrac{\ind_{A=a}}{e_a(X)}\,D_Y\Big]
=M^{\mathrm{IPW}}_a([q]).
\why{$e_a(X)^{-1}$ is $X$-measurable (take it inside); tower rule \ref{T:tower}}
\end{flalign*}

\emph{The Step-2 (phase) objective.} Fix the reference representative $\tilde q_a$. The
integrand is now $\dG^2(\gamma,\gamma^\star_{q^{(a)}\to\tilde q_a})$ for a candidate
warp $\gamma$, where $\gamma^\star_{y\to\tilde q_a}$ is the optimal warp registering a
transform $y$ to $\tilde q_a$. The same three facts let the chain above run with this
integrand. \emph{(Measurability.)} The map $y\mapsto\gamma^\star_{y\to\tilde q_a}$ is
a measurable selection: the objective
$(\gamma,y)\mapsto\norm{\tilde q_a-(y,\gamma)}^2$ is continuous in $\gamma$ on the
compact $\Gc$ (Lemma~\ref{lem:compos}) and continuous in $y$ (by the argument of
Lemma~\ref{lem:hcont}), and the minimiser is unique
(Assumption~\ref{ass:warp}(iii)), so \ref{T:krn} applies. Hence
$\dG^2(\gamma,\gamma^\star_{q^{(a)}\to\tilde q_a})$ is a measurable function of
$q^{(a)}$ and Assumption~\ref{ass:exch} applies to it. \emph{(Integrability.)}
$\dG^2$ is bounded by $\pi^2T$ (Preliminaries). \emph{(Consistency.)} On $\{A=a\}$ we
have $q_Y=q^{(a)}$, so the two optimal warps coincide. The same chain then gives
\[
\E_{\Pstar}\big[\dG^2\big(\gamma,\gamma^\star_{q^{(a)}\to\tilde q_a}\big)\big]
=\E_P\Big[\tfrac{\ind_{A=a}}{e_a(X)}\,
\dG^2\big(\gamma,\gamma^\star_{q_Y\to\tilde q_a}\big)\Big]
\qquad\text{for every candidate }\gamma .
\]

\emph{The average of the phase-centered registered curves.} Fix the reference representative $\tilde q_a$, an invertible warp $\beta\in\G$ and $t\in[0,T]$, and apply the same chain to the scalar integrand $V=\big(Y^{(a)}\circ\gamma^\star_{q^{(a)}\to\tilde q_a}\circ\beta\big)(t)$.
\emph{(Measurability.)} $V$ is a measurable function of the curve $Y^{(a)}$, by the argument of Lemma~\ref{lem:alignavg}(b) (the measurable selection of the Step-2 paragraph above, then a continuous evaluation), so
Assumption~\ref{ass:exch} applies to it. This is where the curve-level form of
Assumption~\ref{ass:exch} is needed, and not merely its SRSF-level form: $V$ depends on the curve, not only on its transform.
\emph{(Integrability.)} By Lemma~\ref{lem:alignavg}(a), whose proof uses only that $\gamma^\star_{q^{(a)}\to\tilde q_a}\circ\beta$ is a non-decreasing absolutely continuous surjection, $|V|\le|Y^{(a)}(0)|+\norm{q^{(a)}}^2$, which is integrable by Assumption~\ref{ass:mom}; the weighted version is integrable because the weight is bounded by $1/c_e$.
\emph{(Consistency.)} On $\{A=a\}$ we have $Y=Y^{(a)}$, hence $q_Y=q^{(a)}$, the two registrations coincide and the two values of $V$ coincide. Then
\[
\E_{\Pstar}\big[\big(Y^{(a)}\circ\gamma^\star_{q^{(a)}\to\tilde q_a}\circ\beta\big)(t)\big]=\E_P\Big[\tfrac{\ind_{A=a}}{e_a(X)}\,\big(Y\circ\gamma^\star_{q_Y\to\tilde q_a}\circ\beta\big)(t)\Big].
\]
At $t=0$ it identifies the mean initial value $\E_{\Pstar}[Y^{(a)}(0)]$.
\hfill$\square$
\end{proofenv}

\paragraph{The argmin principle.}\label{app:argmin}
Two functions that agree on a set have the same infimum and the same minimisers,
even when neither is attained. Proposition~\ref{prop:objident} gives
$M^{(a)}=M^{\mathrm{IPW}}_a$ on all of $\Sw$, so their argmin
sets coincide. Under Assumption~\ref{ass:conc}(4a) this common set is the single
shape $\{\shp^{(a)}\}$, the Step-1 mean shape of the law of $Y^{(a)}$, and it is a functional of the observed law $P$,
since $M^{\mathrm{IPW}}_a$ involves $P$ only. Both objectives are
deterministic functions of the candidate shape, which is what makes the identity
promote from objectives to minimisers. It would not for random sample objectives.

\begin{proofenv}[Proof of Theorem~\ref{thm:ident}.]
The proof assembles Proposition~\ref{prop:objident}, the argmin principle above and the well-posedness results of Appendices~\ref{app:exist}--\ref{app:inv}.

\emph{Step A (shape).} By the argmin principle, the Step-1 mean shape is
$\shp^{(a)}=\argmin_{[q]\in\Sw}M^{\mathrm{IPW}}_a([q])$, a
functional of $P$. Existence of the minimiser is Proposition~\ref{thm:exist}(i), in
its weighted variant Corollary~\ref{cor:reweight}.

\emph{Step B (phase).} Fix an admissible representative $\tilde q_a\in\shp^{(a)}$. By
the second part of Proposition~\ref{prop:objident}, the Step-2 objective equals its
IPW-observable counterpart for every candidate warp. The common objective has a
unique minimiser $\bar\gamma_a$ by Assumption~\ref{ass:conc}(4b) and
Proposition~\ref{thm:exist}(ii), whose proof covers the weighted objective
(Appendix~\ref{app:ka}, step (5e)). By the argmin principle again, $\bar\gamma_a$ is a
functional of $P$.

\emph{Assembly, average, difference.} The mean warp $\bar\gamma_a$ is a functional of $P$ (Step~B) and is invertible (Assumption~\ref{ass:conc}(4b)). The third part of Proposition~\ref{prop:objident}, applied with $\beta=\bar\gamma_a^{-1}$, gives for every $t$
\[
\mu^{(a)}(\Pstar)(t)=\E_{\Pstar}\big[\big(Y^{(a)}\circ\gamma^\star_{q^{(a)}\to\tilde q_a}\circ\bar\gamma_a^{-1}\big)(t)\big]
=\E_P\Big[\tfrac{\ind_{A=a}}{e_a(X)}\,\big(Y\circ\gamma^\star_{q_Y\to\tilde q_a}\circ\bar\gamma_a^{-1}\big)(t)\Big]
=\mu^{(a)}_{\mathrm{IPW}}(P)(t).
\] Each arm mean
$\mu^{(a)}_{\mathrm{IPW}}(P)$ is therefore a functional of $P$, and by
Proposition~\ref{thm:inv} the choice of $\tilde q_a$ does not matter, so no
unidentified choice hides in the construction. Differencing the two
arms,
\[
\tau_{\mathrm{pop}}(\Pstar)
=\mu^{(1)}(\Pstar)-\mu^{(0)}(\Pstar)
=\mu^{(1)}_{\mathrm{IPW}}(P)-\mu^{(0)}_{\mathrm{IPW}}(P),
\]
a functional of the observed law $P$. \hfill$\square$
\end{proofenv}
\section{Proof of Theorem~\ref{thm:nonident} (the individual effect is not identified)} \label{app:nonident}

\begin{proofenv}[Proof.]
The construction and the computations below are ours. The
logic of the counterexample follows the Bernoulli counterexample of
\citet{even_rethinking_2025}. Everything specific to shape (the SRSF
non-separability and the closed-form elastic distances) is new.

\emph{What we must show.} Non-identifiability means that no functional of the
observed law can equal $\tau_{\mathrm{ind}}$ on all data-generating laws satisfying
Assumptions~\ref{ass:sutva}--\ref{ass:exch}. The standard way to prove this is to
construct two full-data laws $P^{*}_1,P^{*}_2$ that (a)~induce the same observed
law, $P(P^{*}_1)=P(P^{*}_2)$, (b)~both satisfy all the assumptions, and (c)~assign
different values to the target. 

\smallskip
\noindent\emph{Step A (the construction).} Take $T=1$. Let $X$ be degenerate (a
constant, so conditioning on $X$ has no effect), let $A$ be a fair coin,
$\Prob(A{=}1)=\tfrac12$, drawn independently of everything else, and give the
potential-outcome curves the two-point marginal laws
\[
Y^{(1)}=\begin{cases}3t&\text{w.\,p. }\tfrac12\\ 3t+t^2&\text{w.\,p. }\tfrac12\end{cases}
\qquad\qquad
Y^{(0)}=\begin{cases}0&\text{w.\,p. }\tfrac12\\ t^2&\text{w.\,p. }\tfrac12 .\end{cases}
\]
The two full-data laws differ only in how the two arms are coupled within a unit:
\begin{itemize}[leftmargin=1.6em,itemsep=2pt]
\item \emph{$P_1^{*}$ (matched coupling): } the pair $(Y^{(1)},Y^{(0)})$ equals
$(3t,\,0)$ or $(3t+t^2,\,t^2)$, each with probability $\tfrac12$. Then
$\Delta=Y^{(1)}-Y^{(0)}\equiv3t$ almost surely.
\item \emph{$P_2^{*}$ (independent coupling): } $Y^{(1)}$ and $Y^{(0)}$ are drawn
independently, the four pairs each have probability $\tfrac14$, and
\[
\Delta=\begin{cases}
\Delta_3=3t&\text{w.\,p. }\tfrac12\quad(\text{pairs }(3t,0)\text{ and }(3t+t^2,t^2)),\\
\Delta_2=3t-t^2&\text{w.\,p. }\tfrac14,\\
\Delta_4=3t+t^2&\text{w.\,p. }\tfrac14 .
\end{cases}
\]
\end{itemize}
Every individual effect curve is increasing, starts at $0$, and has the increment
$a=\Delta(1)\in\{3,2,4\}$, the index of $\Delta_a$. Write
$\gamma_a=\Delta_a/a$: $\gamma_3=\gid$, $\gamma_2(t)=(3t-t^2)/2$ and
$\gamma_4(t)=(3t+t^2)/4$ are warps of $\G$, with speeds
$\dot\gamma_2=(3-2t)/2\in[\tfrac12,\tfrac32]$ and
$\dot\gamma_4=(3+2t)/4\in[\tfrac34,\tfrac54]$, so their inverses have speeds in
$[\tfrac23,2]$ and $[\tfrac45,\tfrac43]$.

\emph{(a) Same observed law.} The observed data are $(X,A,Y)$ with $Y=Y^{(A)}$. $X$
is constant and $A$ is a fair coin under both laws. Given $A=a$, the observed curve
is $Y^{(a)}$, whose marginal law is the same under both couplings. The coupling,
which is the only difference between $P_1^{*}$ and $P_2^{*}$, is invisible. Hence
both full-data laws induce the same observed law $P$.

\emph{(b) The assumptions hold under both.} Consistency holds by construction
($Y=Y^{(A)}$). Positivity: $e_1(X)=\tfrac12$. Exchangeability: $A$ is independent of
the pair $(Y^{(1)},Y^{(0)})$. The moment Assumption~\ref{ass:mom} holds for every law
in play (every curve has speed at most $5$ and starts at $0$). The registrations that
the construction uses (Steps C and D below) are unique and are the inverses of
$\gid,\gamma_2,\gamma_4$. Their warp-SRSFs take values in $[\sqrt{2/3},\sqrt2]$ and have Lipschitz constants at most $2\sqrt2$ (for $C^2$ warps $\gamma\in\G$, $\tfrac{d}{ds}\psi_{\gamma^{-1}}=-\tfrac12\big(\ddot\gamma\,\dot\gamma^{-5/2}\big)\circ\gamma^{-1}$, with $|\ddot\gamma_2|=1$, $\dot\gamma_2\ge\tfrac12$, $|\ddot\gamma_4|=\tfrac12$, $\dot\gamma_4\ge\tfrac34$), so they lie in the compact family $\Gc$ of Appendix~\ref{app:aa} with $c=\tfrac14$, $C=4$ and $C_2=4$ (warp-SRSFs in $[\tfrac12,2]$ with Lipschitz constant at most $4$), which satisfies Assumption~\ref{ass:warp}(i)--(ii). Assumption~\ref{ass:warp}(iii) is used, as in the rest of the paper, for the registrations that the construction performs: here the registrations to the reference representative, which are unique (Step~B); clause (iii-a) holds with the single shape $[\sqrt3]$ (resp.\ $[c^\star]$), whose nowhere-flat representative $\sqrt3$ (resp.\ $c^\star$) is admissible: by the box of Step~B, the best re-timings over $\Gt$ towards it are the inverses of $\gid,\gamma_2,\gamma_4$, which lie in $\Gc$. (The law of $Y^{(0)}$ gives probability $\tfrac12$ to the zero curve, so Theorem~\ref{thm:wellposed} does not cover $\mu(Y^{(0)})$; only the laws of the individual effects are used here.) Assumption~\ref{ass:conc} is checked in Steps C and D. So $\tau_{\mathrm{ind}}$ is well defined under both laws.

\smallskip
\noindent\emph{Step B (closed forms on the shapes of constants).} For $c\ge0$,
$(c,\gamma)=c\,\psi_\gamma$, and the warp-SRSFs $\{\psi_\gamma:\gamma\in\Gt\}$ are
exactly the non-negative functions of unit norm (the inverse map is
$p\mapsto\int_0^{\cdot}p^2$, Lemma~\ref{lem:feasible}). Hence the shape $[c]$, for $c>0$, is
the set of non-negative functions of norm $c$, and every increasing curve $\Delta$ with
$\Delta(0)=0$, $\Delta(1)=a>0$ has $q_\Delta=\sqrt a\,\psi_{\Delta/a}\in[\sqrt a]$.

\begin{resultbox}{Distance to the shape of a non-negative constant}
For $q\in\Lt[0,1]$ and $c\ge0$, with $q^+=\max(q,0)$,
\[
\del^2\big([q],[c]\big)=\norm q^2+c^2-2c\,\norm{q^+}\ \ge\ \big(\norm q-c\big)^2 ,
\]
with equality, when $c>0$, if and only if $q\ge0$ almost everywhere.
If $\Delta=a\,\gamma_\Delta$ with $a>0$ and $\gamma_\Delta\in\G$, the registration of
$q_\Delta$ to the representative $c>0$ is unique: $\gamma^\star_{q_\Delta\to c}=\gamma_\Delta^{-1}$.
\end{resultbox}

\emph{Proof of the box.} For $\gamma\in\Gt$,
$\norm{q-(c,\gamma)}^2=\norm q^2+c^2-2c\inner{q}{\psi_\gamma}$, and over the
non-negative unit-norm $\psi$, $\inner q\psi\le\inner{q^+}{\psi}\le\norm{q^+}$, with
equality at $\psi=q^+/\norm{q^+}$ when $q^+\ne0$ (if $q^+=0$, the supremum $0$ is approached by $\psi_n=\sqrt n\,\ind_{[0,1/n]}$, since $|\inner q{\psi_n}|\le\norm{q\,\ind_{[0,1/n]}}\to0$). This gives the formula, and
the inequality follows from $\norm{q^+}\le\norm q$, with equality if and only if
$q^-=0$. For the registration,
$\norm{c-(q_\Delta,\gamma)}^2=c^2+a-2c\sqrt a\inner{1}{\psi_{\gamma_\Delta\circ\gamma}}$
(Lemma~\ref{lem:rightaction}), and by the Cauchy--Schwarz inequality
$\inner1\psi\le1$ for unit-norm $\psi$, with equality if and only if $\psi\equiv1$, that
is $\gamma_\Delta\circ\gamma=\gid$, that is $\gamma=\gamma_\Delta^{-1}$. \hfill$\lozenge$

\smallskip
\noindent\emph{Step C (the individual effect under $P_1^{*}$).} Here $\Delta\equiv3t$,
so $q_\Delta\equiv\sqrt3$ and $M_1([q])=\del^2([q],[\sqrt3])$, which vanishes if and only if
$[q]=[\sqrt3]$: the Step-1 mean shape is unique. With the representative
$\tilde q\equiv\sqrt3$, the registration is $\gid$ (box, with $\gamma_\Delta=\gid$), the
warp law is the point mass at $\psi_{\gid}$, whose Fréchet mean is $\gid$, and the
phase-centered registered curve is $\Delta$ itself. Hence
\[
\tau_{\mathrm{ind}}(P_1^{*})(t)=3t .
\]

\smallskip
\noindent\emph{Step D (the individual effect under $P_2^{*}$).}
\emph{Step 1.} By the box, for every $q$,
\[
M_2([q])=\E\big[\del^2([q],[\sqrt a])\big]\ \ge\ \E\big[(\norm q-\sqrt a)^2\big]\ \ge\ \operatorname{Var}(\sqrt a),
\]
with equality in both if and only if $q\ge0$ a.e.\ and $\norm q=c^\star:=\E\sqrt a=\tfrac{\sqrt3}2+\tfrac{\sqrt2}4+\tfrac12$,
that is, if and only if $[q]=[c^\star]$. The Step-1 mean shape is the single shape
$[c^\star]$ (Assumption~\ref{ass:conc}(4a) holds).

\emph{Step 2.} Take the representative $\tilde q\equiv c^\star$. By the box, the
registrations are $\gamma_a^{-1}$, with warp-SRSFs $\psi_3\equiv1$,
$\psi_2=\psi_{\gamma_2^{-1}}$ and $\psi_4=\psi_{\gamma_4^{-1}}$, of probabilities
$\tfrac12,\tfrac14,\tfrac14$, where $\psi_{\gamma^{-1}}(s)=\dot\gamma(\gamma^{-1}(s))^{-1/2}$.
By the change of variables \ref{T:cov}, $\inner{1}{\psi_{\gamma^{-1}}}=\int_0^1\sqrt{\dot\gamma}$,
so $\dG(\psi_2,1)=\arccos\int_0^1\sqrt{\dot\gamma_2}\approx0.148$ and
$\dG(\psi_4,1)=\arccos\int_0^1\sqrt{\dot\gamma_4}\approx0.073$: the warp law is
concentrated in the ball of radius $0.15<\pi/8$ about the feasible centre
$\psi_{\gid}\equiv1$. Write $p_3=\tfrac12$, $p_2=p_4=\tfrac14$ and
$N(\psi)=\sum_kp_k\dG^2(\psi,\psi_k)$ for the Fréchet function. The localisation,
convexity, uniqueness and existence steps (5a)--(5d) of Appendix~\ref{app:ka}, run on the
whole sphere instead of the feasible set (they only use that the ball is geodesically
convex and that $N$ is strongly convex on it), give a unique minimiser $m$ of $N$ over
the whole sphere, at distance at most $2r_0=0.30$ from $1$. At $m$ every distance
$\dG(m,\psi_k)$ is below $\pi/2$, so $N$ is differentiable there and $m$ is stationary:
$\sum_kp_k\operatorname{Log}_m\psi_k=0$, with
$\operatorname{Log}_m\psi=\tfrac{\theta}{\sin\theta}(\psi-\cos\theta\,m)$ and
$\theta=\dG(m,\psi)$. Solving for $m$ shows that $m$ is a combination of $1,\psi_2,\psi_4$
with positive coefficients; these functions take values in $[\sqrt{2/3},\sqrt2]$, so $m$
is bounded above and below by positive constants. More precisely, $m=\sum_k\alpha_k\psi_k/\sum_k\alpha_k\cos\theta_k$ with
$\alpha_k=p_k\theta_k/\sin\theta_k>0$ (with $\theta/\sin\theta:=1$ at $\theta=0$) and $\theta_k=\dG(m,\psi_k)\le0.30+0.15=0.45$, so
$\sum_k\alpha_k/\sum_k\alpha_k\cos\theta_k\le1/\cos0.45<1.12$: $m$ takes values in
$[\sqrt{2/3},\,1.12\sqrt2]\subset(\tfrac12,2)$ and has Lipschitz constant at most
$1.12\cdot2\sqrt2<4$. Hence $m$ is feasible and satisfies the constraints defining $\Psi(\Gc)$ strictly (values in $(\tfrac12,2)$, Lipschitz constant below $4$), it is the
unique Step-2 mean over $\Psi(\Gt)$, and $\bar\gamma=\int_0^{\cdot}m^2$ is a bi-Lipschitz,
invertible warp (Assumption~\ref{ass:conc}(4b) holds).

\emph{$\bar\gamma$ is not the identity.} For a unit tangent direction $v\perp1$ at the point
$1$ of the sphere, the derivative of $N$ at $\varepsilon=0$ along the great circle
$\varepsilon\mapsto\cos\varepsilon+v\sin\varepsilon$ is $-2\sum_kp_k\inner{\operatorname{Log}_1\psi_k}{v}=-\tfrac12\inner{w}{v}$ (the gradient of $\dG^2(\cdot,\psi_k)$ is $-2\operatorname{Log}\psi_k$, and $\operatorname{Log}_1\psi_3=0$),
where
\[
w=\operatorname{Log}_1\psi_2+\operatorname{Log}_1\psi_4
=\tfrac{\theta_2}{\sin\theta_2}\big(\psi_2-\cos\theta_2\big)+\tfrac{\theta_4}{\sin\theta_4}\big(\psi_4-\cos\theta_4\big),
\qquad\theta_k=\dG(1,\psi_k)>0 .
\]
The functions $\psi_2,\psi_4$ are continuous on $[0,1]$. If $w\equiv0$, then
$\psi_4(s)-\psi_4(0)$ would be a fixed multiple of $\psi_2(s)-\psi_2(0)$ for all $s$. The
exact values
\[
\psi_2(0)=\sqrt{2/3},\quad\psi_2(\tfrac12)=\sqrt2\,5^{-1/4},\quad\psi_2(1)=\sqrt2;\qquad
\psi_4(0)=2/\sqrt3,\quad\psi_4(\tfrac12)=2\cdot17^{-1/4},\quad\psi_4(1)=2/\sqrt5
\]
(from $\gamma_2^{-1}(\tfrac12)=\tfrac{3-\sqrt5}2$ and $\gamma_4^{-1}(\tfrac12)=\tfrac{\sqrt{17}-3}2$)
give the ratios $\frac{\psi_4(1)-\psi_4(0)}{\psi_2(1)-\psi_2(0)}\approx-0.435$ and
$\frac{\psi_4(1/2)-\psi_4(0)}{\psi_2(1/2)-\psi_2(0)}\approx-1.313$, which differ. So
$w\ne0$. Take $v=w/\norm w$, which is bounded and orthogonal to $1$ (each
$\operatorname{Log}_1\psi_k$ is). For small $\varepsilon>0$, the point
$\cos\varepsilon+v\sin\varepsilon$ is a positive bounded function of unit norm, the SRSF
of a bi-Lipschitz warp, and it has $N(\cos\varepsilon+v\sin\varepsilon)<N(1)$, since the
derivative there is $-\tfrac12\norm w<0$. The point $1=\psi_{\gid}$ is therefore not the
minimiser: $\bar\gamma\ne\gid$.

\emph{Step 3: the phase-centered registered curves.} For each value of $a$,
\[
\Delta_a\circ\gamma^\star_{q_{\Delta_a}\to\tilde q}\circ\bar\gamma^{-1}
=a\,\gamma_a\circ\gamma_a^{-1}\circ\bar\gamma^{-1}=a\,\bar\gamma^{-1},
\]
so
\[
\tau_{\mathrm{ind}}(P_2^{*})(t)=\E[a]\,\bar\gamma^{-1}(t)=3\,\bar\gamma^{-1}(t),
\qquad \E[a]=\tfrac12\cdot3+\tfrac14\cdot2+\tfrac14\cdot4=3 .
\]

\smallskip
\noindent\emph{Step E (conclusion).} Since $\bar\gamma\ne\gid$, also
$\bar\gamma^{-1}\ne\gid$, and $\tau_{\mathrm{ind}}(P_2^{*})=3\,\bar\gamma^{-1}$ differs from
$\tau_{\mathrm{ind}}(P_1^{*})(t)=3t$. Since $P_1^{*}$ and $P_2^{*}$ are observationally
indistinguishable (Step~A(a)) and both satisfy the assumptions (Step~A(b)), no
functional $\Theta$ of the observed law can satisfy $\Theta(P)$$=\tau_{\mathrm{ind}}(\cdot)$ on both:
$\tau_{\mathrm{ind}}$ is not identified. Note that the two couplings share the same
pointwise mean of the individual effects, $\E[\Delta](t)=3t$: they differ only through
the timing that the registrations reveal.

\end{proofenv}

\section{Proof of Theorem~\ref{thm:identloc} (identification of the local effect)} \label{app:identloc}

\begin{proofenv}[Proof.]
Identification of the local estimand is one of this paper's contributions. The proof uses only the classical conditioning facts of \ref{T:tower} and the well-posedness theory of Appendices~\ref{app:exist}--\ref{app:inv}.

\emph{Step A (the law of $\Delta_{\mathrm{loc}}$).} The local estimand $\tau_{\mathrm{loc}}$ is the elastic mean of the law of $\Delta_{\mathrm{loc}}$ \eqref{eq:locallaw}. We first describe this law. Given $X$ and $A=1$, $\Delta_{\mathrm{loc}}=Y^{(1)}-Y^{(X),(0)}$, where $Y^{(1)}$ has law $P_{Y^{(1)}\mid A=1,X}=P_{Y^{(1)}\mid X}$ by Assumption~\ref{ass:exch}, and $Y^{(X),(0)}$, independent of $(A,Y^{(1)})$ given $X$, has law $P_{Y^{(0)}\mid X}$. The case $A=0$ is symmetric. So, given $X$, $\Delta_{\mathrm{loc}}$ is the difference of independent draws from $P_{Y^{(1)}\mid X}$ and $P_{Y^{(0)}\mid X}$, whatever the value of $A$. In other words, draw covariates
$X\sim P_X$, then, given $X$, draw $Y^{(X),(1)}$ and $Y^{(X),(0)}$
independently from the two conditional potential-outcome laws, and form the difference
curve $Y^{(X),(1)}-Y^{(X),(0)}$: it has the law of $\Delta_{\mathrm{loc}}$, and we compute the expectations below with $\Delta_{\mathrm{loc}}=Y^{(X),(1)}-Y^{(X),(0)}$. In this proof, we re-express its
Step-1 objective as the observed-data functional
$M_{\mathrm{loc}}$ of \eqref{eq:locobj}, and do the same for the
phase objective and the average of the phase-centered registered curves in less detail. 

\smallskip
\noindent\emph{Step B (peel the construction).} Fix a candidate $[q]\in\Sw$ and
write $D=\del^2\big([q],[q_{Y^{(X),(1)}-Y^{(X),(0)}}]\big)$, a measurable
functional of the pair $(Y^{(X),(1)},Y^{(X),(0)})$ with finite expectation
(integrability and measurability as in Appendix~\ref{app:ident}, using
Corollary~\ref{cor:wp3} for the moment bound on difference curves). By
the description above, conditionally on $X$ the pair
$(Y^{(X),(1)},Y^{(X),(0)})$ has law $P_{Y^{(1)}\mid X}\otimes P_{Y^{(0)}\mid X}$. The
tower rule (\ref{T:tower}) then gives:
\begin{flalign*}
\E\big[\del^2([q],[q_{\Delta_{\mathrm{loc}}}])\big]
&=\E_{X\sim P_X}\Big[\ \E_{(Y^{(X),(1)},Y^{(X),(0)})\mid X}\big[D\big]\ \Big]
\why{tower rule \ref{T:tower}}\\
&=\E_{X\sim P_X}\Big[\ \E_{P_{Y^{(1)}\mid X}\otimes P_{Y^{(0)}\mid X}}\big[D\big]\ \Big].
\why{the conditional law is the product, as shown above}
\end{flalign*}

\smallskip
\noindent\emph{Step C (identify the conditional marginals).} We claim that for each
arm $a$ and $P_X$-a.e.\ $x$,
\[
P_{Y^{(a)}\mid X=x}
\ \overset{\text{(i)}}{=}\ P_{Y^{(a)}\mid A=a,X=x}
\ \overset{\text{(ii)}}{=}\ P_{Y\mid A=a,X=x}.
\] 
\emph{(i)} is Assumption~\ref{ass:exch}, conditional exchangeability at the curve
level, $Y^{(a)}\perp\!\!\!\perp A\mid X$. Conditioning on $\{A=a,X=x\}$ is legitimate because
$e_a(x)\ge c_e>0$ (Assumption~\ref{ass:pos}). \emph{(ii)} is
Assumption~\ref{ass:sutva}, consistency.

Substituting the identified marginals into Step~B, the inner expectation becomes,
for $P_X$-a.e.\ $x$, 
\[
m_{\mathrm{loc}}([q],x)
=\E_{P_{Y\mid A=1,X=x}\otimes P_{Y\mid A=0,X=x}}
\Big[\del^2\big([q],[q_{Y^{(X),(1)}-Y^{(X),(0)}}]\big)\Big],
\]
which is exactly the inner functional of \eqref{eq:locobj}. Therefore
\[
\E\big[\del^2([q],[q_{\Delta_{\mathrm{loc}}}])\big]
=\E_{X\sim P_X}\big[m_{\mathrm{loc}}([q],X)\big]
=M_{\mathrm{loc}}([q])
\qquad\text{for every }[q]\in\Sw ,
\]
The Step-1 objective is identified.

\smallskip
\noindent\emph{Step D (phase objective, Step-3 average of the phase-centered registered curves, and promotion to the estimand).}
The Step-2 (phase) objective is handled by the same Steps~B and~C, with the integrand
$\dG^2(\gamma,\gamma^\star_{q_{\Delta_{\mathrm{loc}}}\to\tilde q})$, a bounded,
measurable functional of the pair $(Y^{(X),(1)},Y^{(X),(0)})$ (measurable selection as
in Appendix~\ref{app:ident}). The Step-3 average of the phase-centered registered curves is handled by Steps~B and~C as well: for a fixed invertible warp $\beta$ and a fixed $t$, the integrand
$\big(\Delta_{\mathrm{loc}}\circ\gamma^\star_{q_{\Delta_{\mathrm{loc}}}\to\tilde q}\circ\beta\big)(t)$ is a measurable functional of the pair $(Y^{(X),(1)},Y^{(X),(0)})$ (Lemma~\ref{lem:alignavg}(b), the difference map being continuous on $(\F,d_\F)$), bounded by $|Y^{(X),(1)}(0)|+|Y^{(X),(0)}(0)|+\norm{q_{Y^{(X),(1)}}}^2+\norm{q_{Y^{(X),(0)}}}^2$ (Lemma~\ref{lem:alignavg}(a) and Corollary~\ref{cor:wp3}), which is integrable by Assumption~\ref{ass:mom}. Steps~B and~C then give
\[
\E\big[\big(\Delta_{\mathrm{loc}}\circ\gamma^\star_{q_{\Delta_{\mathrm{loc}}}\to\tilde q}\circ\beta\big)(t)\big]=L_{\mathrm{loc},\beta}(t)
\]
of \eqref{eq:locbase}. This functional is not linear in the pair: it is identified because the whole product law $P_{Y\mid A=1,X}\otimes P_{Y\mid A=0,X}$ is, not merely its two means. Finally, Assumption~\ref{ass:conc} applied to the law of
$\Delta_{\mathrm{loc}}$ makes both minimisers unique, so the argmin principle
(Appendix~\ref{app:ident}) promotes each objective identity to an identity of
minimisers. Propositions~\ref{thm:exist} and~\ref{thm:inv} and Lemma~\ref{lem:alignavg} (via
Corollary~\ref{cor:wp3}) make the three-step construction well posed, and taking $\beta=\bar\gamma^{-1}$, with $\bar\gamma$ the identified minimiser of $N_{\mathrm{loc}}$, identifies the Step-3 average. Hence
$\tau_{\mathrm{loc}}(\Pstar)$ is a functional of the observed law $P$. \hfill$\square$
\end{proofenv}

\subsection{Why no exact pairwise reweighting for the local objective} \label{app:noipw}

This section expands Remark~\ref{rem:noipw}. The local effect is identified (Theorem~\ref{thm:identloc}); the question is only whether, like the populational effect, it can be written as a weighted average over observed units, which is what makes the estimator of Section~\ref{sec:popest} a plug-in. It cannot. Weighting turns an observed law into a target law only if the target gives no mass to events of observed probability zero.

\smallskip
\noindent\emph{(i) Populational effect: one unit.} The Step-1 integrand depends on one potential-outcome curve. The target law is that of $(X,Y^{(a)})$ in the whole population, and the observed law is that of $(X,Y)$ restricted to $\{A=a\}$:
\[
\Pi_{\mathrm{tar}}(dx,dy)=P_X(dx)\,\Pstar_{Y^{(a)}\mid X=x}(dy),\qquad
\Pi_{\mathrm{obs}}(dx,dy)=e_a(x)\,P_X(dx)\,P_{Y\mid A=a,X=x}(dy).
\]
The two curve parts coincide,
\[
\Pstar_{Y^{(a)}\mid X=x}=\Pstar_{Y^{(a)}\mid A=a,X=x}=P_{Y\mid A=a,X=x},
\]
the first equality by exchangeability (Assumption~\ref{ass:exch}) and the second by consistency (Assumption~\ref{ass:sutva}), since $Y=Y^{(a)}$ on $\{A=a\}$. The target law is therefore a reweighting of the observed law:
\[
\Pi_{\mathrm{tar}}(dx,dy)=\frac{1}{e_a(x)}\,\Pi_{\mathrm{obs}}(dx,dy).
\]
The weight $1/e_a(x)$ is bounded by positivity (Assumption~\ref{ass:pos}) and is exactly the weight of Proposition~\ref{prop:objident}. It exists even for continuous $X$, since $e_a(x)=\Prob(A=a\mid X=x)$ is a conditional expectation, defined for almost every $x$.

\smallskip
\noindent\emph{(ii) Local effect: a pair of units.} The integrand of $M_{\mathrm{loc}}$ in \eqref{eq:locobj} depends on a treated and a control potential-outcome curve drawn at the same covariate value $x$, where $x$ follows the covariate law $P_X$. The target law is that of the pair $(X_1,X_0,Y^{(1)},Y^{(0)})$ with $X_0=X_1$, and the observed law is that of a treated unit $(X_1,Y_1)$ and an independent control unit $(X_0,Y_0)$:
\[
\Pi_{\mathrm{tar}}=P_X(dx_1)\,\delta_{x_1}(dx_0)\,\Pstar_{Y^{(1)}\mid X=x_1}(dy_1)\,\Pstar_{Y^{(0)}\mid X=x_0}(dy_0),
\]
\[
\Pi_{\mathrm{obs}}=P_{X\mid A=1}(dx_1)\,P_{X\mid A=0}(dx_0)\,P_{Y\mid A=1,X=x_1}(dy_1)\,P_{Y\mid A=0,X=x_0}(dy_0),
\]
with $\delta_{x_1}$ the point mass at $x_1$. As in (i), the curve parts coincide, since $\Pstar_{Y^{(a)}\mid X=x}=P_{Y\mid A=a,X=x}$ for each arm. The two laws differ only in their covariate parts. In the target, the control covariate is forced to equal the treated one, so $\Pi_{\mathrm{tar}}$ puts all its mass on the diagonal $\{X_1=X_0\}$. In the observed law, the two covariates are drawn independently, so for continuous $X$ the diagonal has probability zero. A reweighting $W\,\Pi_{\mathrm{obs}}$ only rescales mass that $\Pi_{\mathrm{obs}}$ already has, so it also gives the diagonal zero mass. Hence no weight $W$ gives $\Pi_{\mathrm{tar}}=W\,\Pi_{\mathrm{obs}}$: the difficulty is not that the right weight is unknown, but that no weight exists.

\smallskip

\begin{center}\small
\begin{tabular}{p{0.34\textwidth}p{0.27\textwidth}p{0.27\textwidth}}
\hline
Event & (i) populational & (ii) local \\
\hline
A single value $\{X=x\}$ & probability $0$ under both laws: harmless & same \\
An interval of values of $X$ & positive under both laws (positivity) & not the issue \\
The diagonal $\{X_1=X_0\}$ & not involved & probability $1$ under $\Pi_{\mathrm{tar}}$, $0$ under $\Pi_{\mathrm{obs}}$ \\
\hline
\end{tabular}
\end{center}

\section{The phase-centered elastic mean and the Karcher mean function} \label{app:twomeans}

Steps~1 and~2 of Section~\ref{sec:karcher} give a Fréchet mean shape, a reference representative $\tilde q$, registrations $\gamma^\star_{q_i\to\tilde q}$ and a mean warp $\bar\gamma_n$. Two curves can be built from them in Step~3 (Remark~\ref{rem:karcher}): the phase-centered elastic mean $\hat\mu_n$, used in the paper, and the Karcher mean function $\hat\mu^{\mathrm K}_n$. This appendix characterises them, summarises the associated literature and compares them. Lemma~\ref{lem:step1avg} and Proposition~\ref{prop:twomeans}(i) and~(iii) also hold for a law $P$, with expectations in place of averages.

\subsection{Two averages of the same phase-centered registered curves} \label{app:twomeans:char}

Write $g_i=f_i\circ\gamma^\star_{q_i\to\tilde q}\circ\bar\gamma_n^{-1}$ for the phase-centered registered curves, so that $\hat\mu_n=\frac1n\sum_ig_i$, and assume, as in the paper, that the registrations to $\tilde q$ are attained.

\begin{lemma}[Stationarity of Step~1]\label{lem:step1avg}
For every representative $\tilde q$ of a Step-1 mean shape, $\tilde q=\frac1n\sum_i\big(q_i,\gamma^\star_{q_i\to\tilde q}\big)$.
\end{lemma}

\begin{proofenv}[Proof.]
Let $M([q])=\frac1n\sum_i\del^2([q],[q_i])$ and, for $q\in\Lt$, $\Psi(q)=\frac1n\sum_i\norm{q-(q_i,\gamma^\star_{q_i\to\tilde q})}^2$. For every $q$, $\del([q],[q_i])\le\norm{q-(q_i,\gamma^\star_{q_i\to\tilde q})}$, so $M([q])\le\Psi(q)$, with equality at $q=\tilde q$ since by definition $\del([\tilde q],[q_i])=\inf_{\gamma}\norm{\tilde q-(q_i,\gamma)}$ attains its infimum at $\gamma^\star_{q_i\to\tilde q}$. Hence $\Psi(\tilde q)=M([\tilde q])\le M([q])\le\Psi(q)$ for every $q\in\Lt$: $\tilde q$ minimises the quadratic $\Psi$, whose unique minimiser is the average of the $(q_i,\gamma^\star_{q_i\to\tilde q})$. \hfill$\square$
\end{proofenv}

\begin{proposition}[The two curves]\label{prop:twomeans}
\leavevmode
\begin{enumerate}[leftmargin=2em,itemsep=2pt,topsep=2pt]
\item[(i)] The transform of the Karcher mean function is the average of the transforms of the phase-centered registered curves: $q_{\hat\mu^{\mathrm K}_n}=(\tilde q,\bar\gamma_n^{-1})=\frac1n\sum_iq_{g_i}$.
\item[(ii)] With $\F_{a}=\{f\in\F:f(0)=a\}$ and $\bar a=\frac1n\sum_if_i(0)$,
\[
  \hat\mu^{\mathrm K}_n=\argmin_{f\in\F_{\bar a}}\ \frac1n\sum_i\norm{q_f-q_{g_i}}^2,
  \qquad
  \hat\mu_n=\argmin_{f\in\F}\ \frac1n\sum_i\norm{f-g_i}^2 .
\]
\item[(iii)] Both curves start at $\bar a$, and for every $t$,
\[
  \hat\mu^{\mathrm K}_n(t)-\hat\mu_n(t)=\int_0^t\Big(\bar q\,|\bar q|-\frac1n\sum_iq_{g_i}|q_{g_i}|\Big)(s)\,ds,
  \qquad \bar q=\frac1n\sum_iq_{g_i}.
\]
\end{enumerate}
\end{proposition}

\begin{proofenv}[Proof.]
(i)\begin{flalign*}
q_{\hat\mu^{\mathrm K}_n}
&=(\tilde q,\bar\gamma_n^{-1})\\
&=\Big(\frac1n\sum_{i=1}^n\big(q_i,\gamma^\star_{q_i\to\tilde q}\big),\ \bar\gamma_n^{-1}\Big)
\why{stationarity of Step~1, Lemma~\ref{lem:step1avg}}\\
&=\frac1n\sum_{i=1}^n\Big(\big(q_i,\gamma^\star_{q_i\to\tilde q}\big),\ \bar\gamma_n^{-1}\Big)
\why{the action is linear in $q$: $(au+bv,\gamma)=a\,(u,\gamma)+b\,(v,\gamma)$}\\
&=\frac1n\sum_{i=1}^n\big(q_i,\ \gamma^\star_{q_i\to\tilde q}\circ\bar\gamma_n^{-1}\big)
\why{right action, Lemma~\ref{lem:rightaction}(a)}\\
&=\frac1n\sum_{i=1}^n q_{g_i},
\qquad g_i=f_i\circ\gamma^\star_{q_i\to\tilde q}\circ\bar\gamma_n^{-1}
\why{SRSF of a warped curve, Lemma~\ref{lem:iso}: $(q_f,\gamma)=q_{f\circ\gamma}$}
\end{flalign*} 

\noindent (ii) On $\F_{\bar a}$ the map $f\mapsto q_f$ is a bijection onto $\Lt$ (Proposition~\ref{prop:srsf}(ii)), so the first convex problem is the minimisation of $\frac1n\sum_i\norm{q-q_{g_i}}^2$ over $q\in\Lt$, solved by the average $\bar q=\frac1n\sum_iq_{g_i}$, which is the transform of $\hat\mu^{\mathrm K}_n$ by (i). The second is a least-squares problem in $\Lt$, solved by the pointwise average.

(iii) Every warp fixes $0$, so $g_i(0)=f_i(0)$ and $\hat\mu_n(0)=\bar a=\hat\mu^{\mathrm K}_n(0)$. By Proposition~\ref{prop:srsf}(ii) and (i),
\[
\hat\mu_n(t)=\bar a+\int_0^t\frac1n\sum_iq_{g_i}|q_{g_i}|,
\qquad
\hat\mu^{\mathrm K}_n(t)=\bar a+\int_0^t\bar q\,|\bar q|,
\]
and subtracting gives the result. \hfill$\square$
\end{proofenv}

The two curves thus share the Fréchet mean shape and the mean warp, and differ only in what is averaged: the SRSFs of the $g_i$ for $\hat\mu^{\mathrm K}_n$, the values of the $g_i$ for $\hat\mu_n$. By (iii), they coincide when the $g_i$ differ only by vertical shifts: then all the $q_{g_i}$ are equal, since the SRSF does not see constants, so the integrand in (iii) vanishes. They differ otherwise (different amplitudes, noise, features that no warp can align). The final value of $\hat\mu^{\mathrm K}_n$ is then in general not the average final value.

\subsection{Position in the literature} \label{app:twomeans:lit}

Table~\ref{tab:twomeanslit} lists which curve each reference computes.  We found no reference that compares the two curves.

\begin{table}[htbp]
\centering
\small
\caption{\textbf{Which curve each reference computes} (elastic mean: the phase-centered elastic mean).}
\label{tab:twomeanslit}
\begin{tabular}{@{}p{0.27\textwidth}p{0.30\textwidth}p{0.36\textwidth}@{}}
\toprule
Reference & Curve & Name used there\\
\midrule
\citet{srivastava_registration_2011} & both; the elastic mean as signal estimator & centre of the shape (transform); mean of the aligned functions\\
\citet{kurtek_signal_2011} & elastic mean & estimate of the signal\\
\citet{tucker_generative_2013} & Karcher mean function; elastic mean in figures & Karcher mean $\mu_f$; aligned mean\\
\texttt{fdasrsf} & Karcher mean function (\texttt{fmean}); aligned functions (\texttt{fn}) & Karcher mean\\
\texttt{scikit-fda} & Karcher mean function (\texttt{fisher\_rao\_karcher\_mean}) & Karcher mean\\
\citet{kneip_statistical_1992}, \citet{ramsay_curve_1998} & mean of the registered curves (other registrations) & structural mean (structural averaging)\\
\citet{chakraborty_functional_2021} & mean of the registered curves & mean of the synchronised process\\
This paper & elastic mean & phase-centered elastic mean\\
\bottomrule
\end{tabular}
\end{table} 

\subsection{The two means compared} \label{app:twomeans:sims}

Figure~\ref{fig:comparemeans} compares the two means on three samples of Figure~\ref{fig:meansgrid}, computed from the same run of Steps~1 and~2. The phase-centered elastic mean captures the typical shape better. It averages the registered curves themselves, so it ends at the average final value of the sample and keeps the average amplitude of the curves (Proposition~\ref{prop:propsB}(a) and~(c)). The Karcher mean function averages their SRSFs. Its final value is not controlled: it drifts away from the average final value as soon as the registered curves differ by more than a vertical shift (Proposition~\ref{prop:twomeans}(iii)). It also shrinks amplitudes: curves of one shape with heights $s_i$ give it the height $\big(\frac1n\sum_i\sqrt{s_i}\big)^2$, below the average height $\frac1n\sum_is_i$ (Proposition~\ref{prop:propsB}(c)).

\begin{figure}[htbp]
\centering
\includegraphics[width=\textwidth]{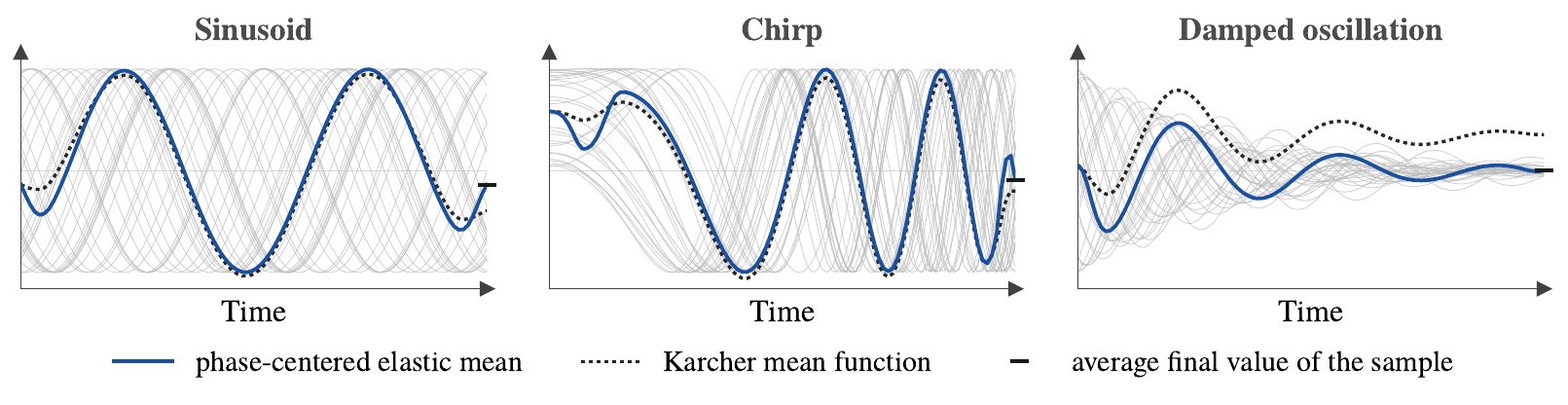}
\caption{\textbf{The two means from the same registrations}, on three samples of Figure~\ref{fig:meansgrid} (grey): the phase-centered elastic mean (blue) and the Karcher mean function (dotted). Black tick: average final value of the sample.}
\label{fig:comparemeans}
\end{figure}

\section{Simulation details} \label{app:sims}

This appendix completes Section~\ref{sec:sims}: the details of the data-generating process, the computation of the estimands and of the shape metrics, the implementation of the estimators, then three complements (a comparison of functional means, the causal-estimand gap over the $(\alpha,p_{\mathrm{rebound}})$ grid, the noise dependence of the estimands) and the sensitivity analyses.

\paragraph{Data-generating process.}
The model is given in Section~\ref{sec:simmodel}, where time is rescaled to $[0,1]$. Each noise curve $\varepsilon^{(a)}$ of that section is a random trigonometric polynomial on the three lowest frequencies,
\[
  \varepsilon(t)\ =\ \sum_{f=1}^{3}\Big(a_f\,\sin(2\pi f t)\ +\ b_f\,\big(\cos(2\pi f t)-1\big)\Big),
  \qquad a_f,b_f\ \overset{\text{iid}}{\sim}\ \mathcal N\!\big(0,\sigma_\varepsilon^2/3\big),
\]
drawn independently for each unit. The cosine terms are written as $\cos-1$ so that $\varepsilon(0)=0$. Figure~\ref{fig:biasgrid} shows the causal estimands as the noise level $\sigma_{\varepsilon}$ grows in Scenario~1.

\paragraph{Estimands and reference values.}
Our causal estimands, which contain elastic means defined at the scale of the population distribution, are computed on $800$ potential-outcome draws with the same elastic-mean code as the estimators. The code runs two alternations of registration and averaging and a final registration (Section~\ref{sec:simestimators}). The reference value of each shape evaluation metric (depth, drop speed and onset time) is its median over the $800$ curves $\Delta_i=Y^{(1)}_i-Y^{(0)}_i$. In Scenario~1 these medians are close to the metrics of $\tau_{\mathrm{pop}}$ (depth $-1.13$ against $-1.14$, onset time $0.28$ against $0.27$, drop speed $7.0$ against $6.4$), which is why the estimators are scored against $\tau_{\mathrm{pop}}$. Being computed on finitely many draws, with two alternations in Step~1, the estimands carry a Monte-Carlo and algorithmic error. Recomputed on independent sets of draws (eight sets of $800$ draws per scenario, four sets of $3200$ draws in Scenario~1), $\tau_{\mathrm{pop}}$ varies by about $0.03$ to $0.04$ in $L^2$ (root-mean-square distance to the average of the replicates; $0.027$ for $3200$ draws). This is small compared with the errors of the competitors. 

\paragraph{Shape evaluation metrics.}
The \emph{depth} of a curve is its minimum, the \emph{onset time} is the time of the half-depth crossing on the way down, and the \emph{drop speed} is the magnitude of its steepest downward slope (Figure~\ref{fig:metricschema}). The three are defined for rebounding and plateauing curves alike, hence in all five scenarios. Curve errors are $L^2$ distances; scalar errors are absolute differences to the reference, summarised by their median and interquartile range over replications.

\begin{figure}[htbp]
\centering
\includegraphics[width=0.95\textwidth]{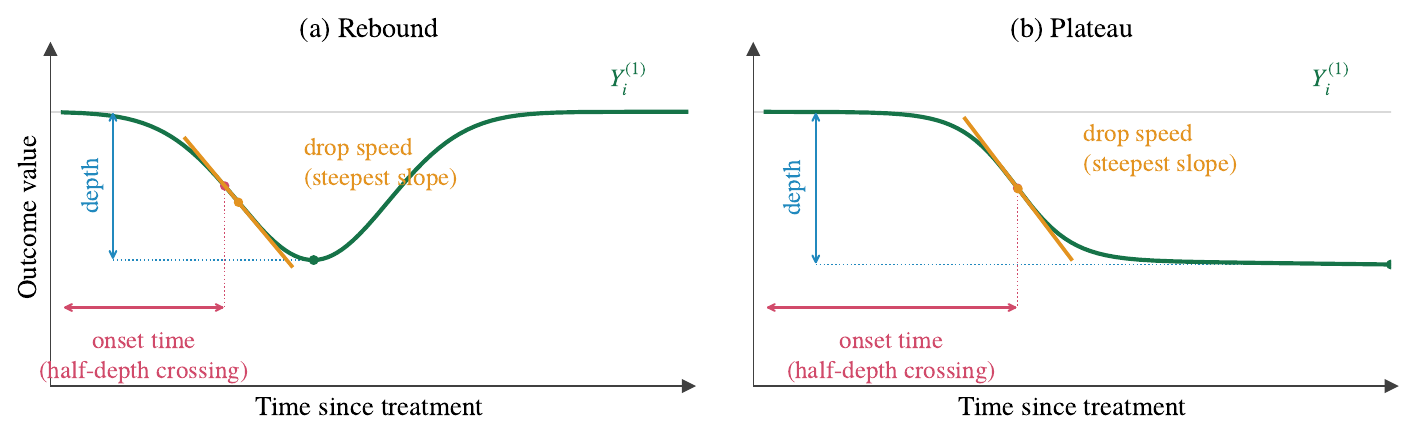}
\caption{\textbf{The three shape evaluation metrics}, on a rebounding
\emph{(a)} and a plateauing \emph{(b)} treated potential outcome (green):
the depth (blue), the onset time (pink, half-depth crossing on the
way down), and the drop speed (yellow, slope of the tangent at the
steepest-descending point).}
\label{fig:metricschema}
\end{figure}

\paragraph{Estimator and competitor implementations.}
All estimators share the registration code (dynamic programming over piecewise-linear warps with slopes in $[1/4,4]$, a family $\Gc$ satisfying clauses (i)--(ii) of Assumption~\ref{ass:warp} with $c=1/4$, $C=4$) and the same time grid. The simulated laws add noise and, in Scenarios~2 to~4, mix rebounding and plateauing curves, so they do not satisfy clause (iii) exactly (Remark~\ref{rem:plaus}): the computed estimands are the outputs of this computation over $\Gc$, for which Theorem~\ref{thm:wellposed} gives no formal guarantee. Our populational estimator places the inverse-propensity weights inside Steps~1 and~2 of the elastic mean and in its Step~3 (Hájek form, Remark~\ref{rem:hajek}), with a logistic propensity fitted on the whole sample and clipped to $[0.02,0.98]$. In Step~1, the aligned transforms $(q_i,\gamma_i)$ are computed as the SRSFs of the re-timed curves $f_i\circ\gamma_i$, which keeps the total variation of features narrower than the grid spacing. The local estimator takes the elastic mean of the matched difference curves ($k=3$ neighbours, we kept at most $300$ pooled differences in our simulations). DR-FoS \citep{testa_doubly_2025} is the pointwise AIPW formula with the same logistic propensity model, cross-fitted over $K=5$ folds, and per-arm pointwise linear outcome regressions. The SRVF kernel ATE follows the authors' released implementation \citep{raykov_kernel_2025}: five-point moving-average smoothing, separate registration of each arm to its Fréchet mean shape, with centered warps, then operator-valued kernel ridge regression with their kernels, their ridge grid $\{10^{-3},1,10^{3}\}$ and their hold-out selection. We kept two features of the released code. Its treatment kernel is the dot product $k(a,a')=aa'$, so the prediction at $A=0$ is identically zero and the estimate reduces to the predicted mean of the registered treated curves. Its ridge parameter is chosen by the prediction error on raw, unregistered held-out curves. We fixed one ordering inconsistency of the released code (time-major outcomes against a sample-major Kronecker system); the fix is what their Algorithm~1 specifies and is favourable to the competitor.

The Monte-Carlo comparisons of Section~\ref{sec:simestimators} run $50$ replications of $n=250$ per scenario, each drawing a new sample and refitting the four estimators. The phase sweep runs $10$ replications per value of $\alpha$, against estimands cached on $800$ draws per $\alpha$.

\paragraph{Estimator distributions across scenarios.}
Figure~\ref{fig:mcgridfull} shows the Monte-Carlo distribution of the estimated effect curves in the five scenarios, two of which are in the main text (Figure~\ref{fig:mcgrid}), and Figure~\ref{fig:boxgrid} gives the error distributions behind them. Both elastic estimators have small errors on every metric, DR-FoS carries the estimand gap of the FATE on every metric, and the SRVF kernel ATE combines a large $L^2$ error with wide depth errors.

\begin{figure}[p]
\centering
\includegraphics[width=0.92\textwidth]{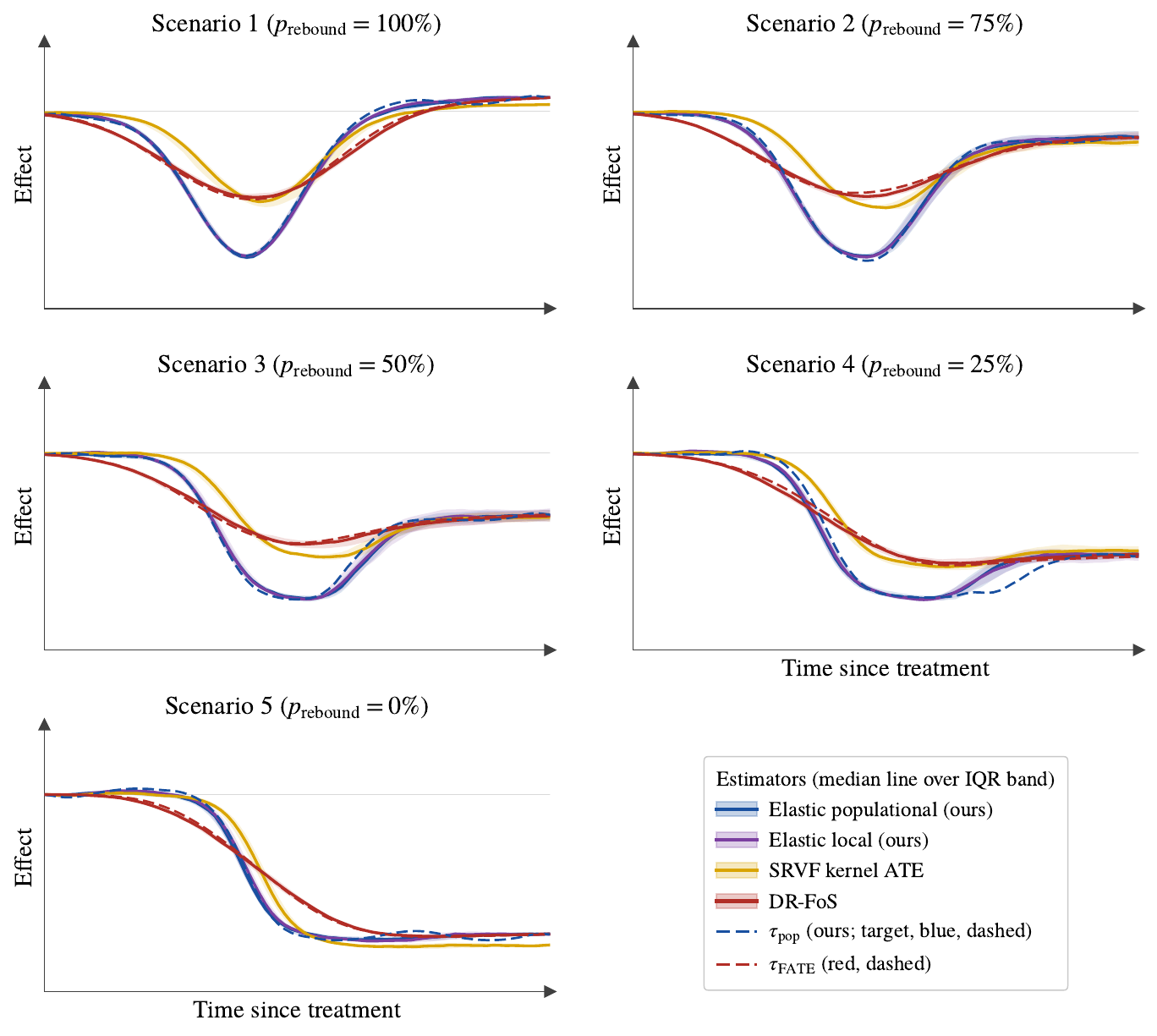}
\caption{\textbf{Monte-Carlo distribution of the four estimators, across
scenarios} (one panel per scenario, from Scenario~1, $p_{\mathrm{rebound}}=100\%$,
down to Scenario~5, $p_{\mathrm{rebound}}=0\%$; $50$ replications per scenario,
$n=250$). Each panel: pointwise
median (solid) and pointwise interquartile band of the four estimated effect
curves, with the causal estimands $\tau_{\mathrm{pop}}$ (blue, dashed) and
$\tau_{\mathrm{FATE}}$ (red, dashed).}
\label{fig:mcgridfull}
\end{figure}  

\begin{figure}[htbp]
\centering
\includegraphics[width=0.98\textwidth]{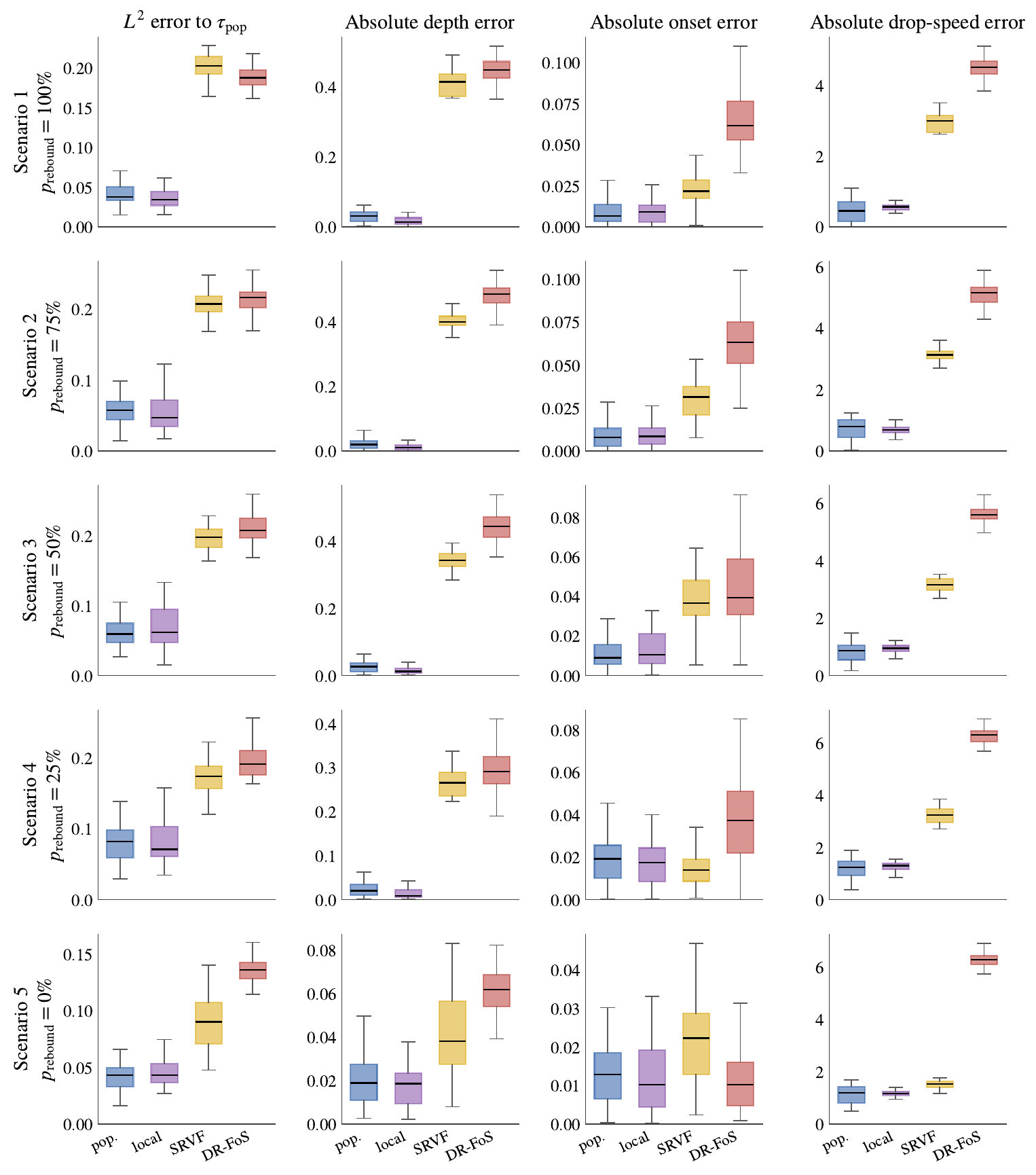}
\caption{\textbf{Estimator absolute errors, across scenarios} (one row
per scenario, from Scenario~1 down to Scenario~5; $50$ replications, $n=250$). Boxes: median and interquartile
range of the absolute errors. Onset timing in units of rescaled time; drop speed in outcome
units per unit of rescaled time.}
\label{fig:boxgrid} 
\end{figure}

\subsection{Functional means across data-generating processes} \label{app:meansgrid}

Figure~\ref{fig:meansgrid} compares the phase-centered elastic mean of Section~\ref{sec:karcher} with the cross-sectional mean and four other registration-based averages on seven data-generating processes. The cross-sectional mean flattens the shared features and matches no curve of the sample. The elastic mean recovers the common shape in every row.

\begin{figure}[htbp]
\centering
\includegraphics[width=\textwidth]{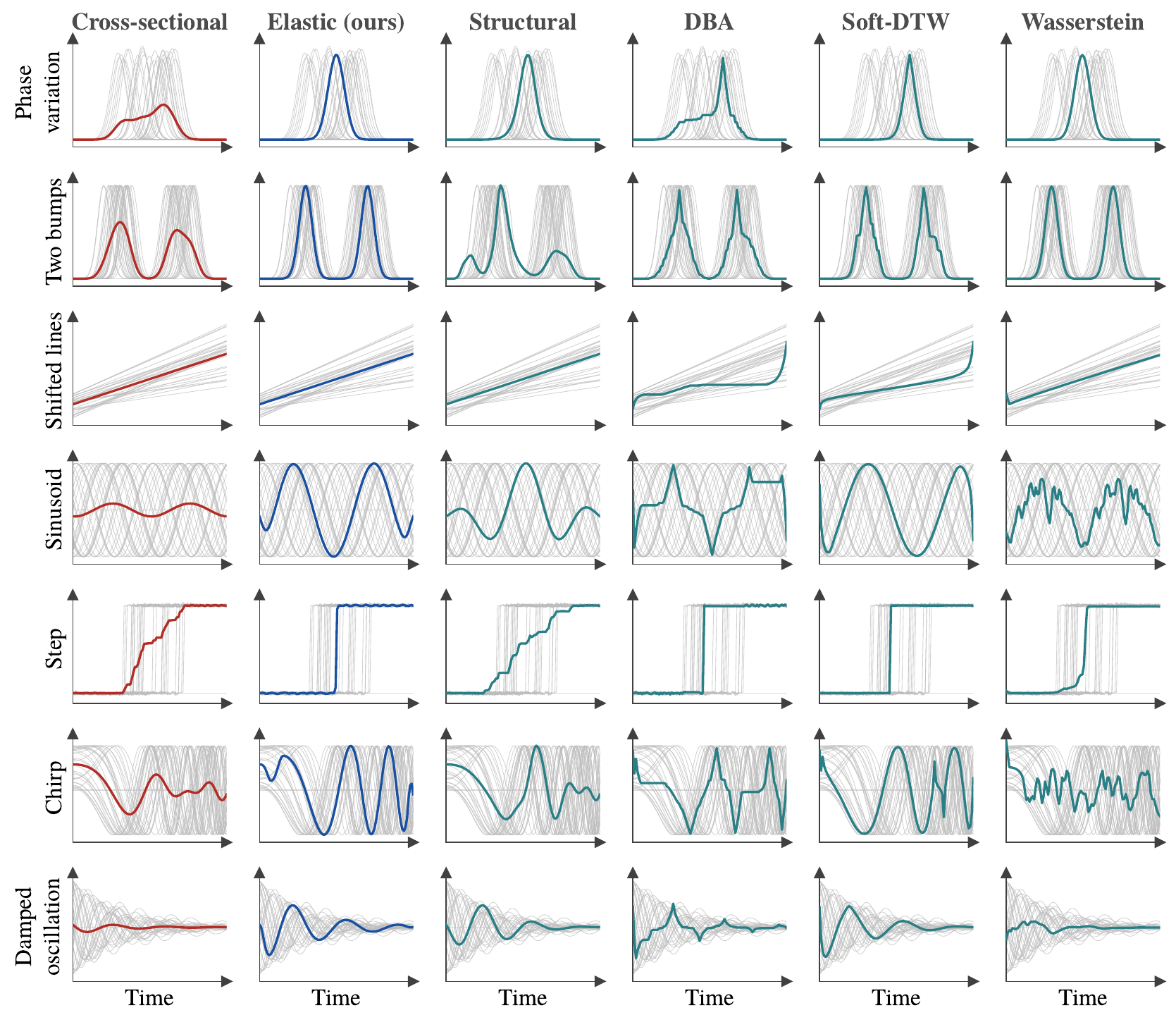}
\caption{\textbf{Functional means across data-generating processes.} One row per data-generating process (grey: a sample of curves sharing a common shape but differing in timing, in some rows also in amplitude). One column per averaging operator: the Euclidean cross-sectional mean (red) and five registration-based means (blue for the elastic mean, teal for the others): structural with landmark registration \citep{kneip_statistical_1992}, elastic (Fisher--Rao; the phase-centered elastic mean), DBA (DTW barycenter averaging) \citep{petitjean_global_2011}, soft-DTW \citep{cuturi_soft-dtw_2017} and Wasserstein \citep{panaretos_amplitude_2016}.}
\label{fig:meansgrid}
\end{figure}

\subsection{The causal-estimand gap over the
$(\alpha,p_{\mathrm{rebound}})$ grid} \label{app:mix}

Phase variation and shape heterogeneity are two distinct mechanisms separating the FATE from the typical shape. The table gives the gap $\|\tau_{\mathrm{FATE}}-\tau_{\mathrm{pop}}\|_{L^2}$ over a grid of $(\alpha,p_{\mathrm{rebound}})$ values: one row per scenario, whose phase-variation scale $\alpha$ is varied, each cell on $800$ potential-outcome draws (Monte-Carlo error about $0.03$ to $0.04$; the draws differ from those of Figure~\ref{fig:phase}, hence the small differences with the gaps of that figure in Scenario~1):
\begin{center}
\begin{tabular}{@{}lccc@{}}
\toprule
 & $\alpha=0$ & $\alpha=0.5$ & $\alpha=1$ \\
\midrule
Scenario~1 ($p_{\mathrm{rebound}}=100\%$) & $0.02$ & $0.07$ & $0.19$ \\
Scenario~2 ($p_{\mathrm{rebound}}=75\%$) & $0.12$ & $0.13$ & $0.21$ \\
Scenario~3 ($p_{\mathrm{rebound}}=50\%$) & $0.13$ & $0.14$ & $0.20$ \\
Scenario~4 ($p_{\mathrm{rebound}}=25\%$) & $0.09$ & $0.12$ & $0.18$ \\
Scenario~5 ($p_{\mathrm{rebound}}=0\%$) & $0.04$ & $0.05$ & $0.12$ \\
\bottomrule
\end{tabular}
\end{center}
In Scenarios~1 and~5 the gap starts small and grows with $\alpha$: timing variability is the only mechanism. In Scenarios~2 to~4 it is large already at $\alpha=0$, it is the largest in Scenarios~2 and~3, and phase variation adds to it.

\subsection{The noise level and the causal estimands} \label{app:bias}

Section~\ref{sec:noise} explains that the template bias is a problem when the noise is more important than the signal inside the elastic mean, that is, when the total variation added by the noise is of the order of the total variation of the curve being averaged, or larger. This happens in three situations.
\begin{enumerate}[leftmargin=2em,itemsep=2pt,topsep=2pt]
\item \emph{Very noisy outcomes, for any elastic estimand}: the noise carries much more total variation than the shape of the effect. This is the regime of the highest noise level in Figure~\ref{fig:biasgrid} (third row), where the estimands themselves move with the noise.
\item \emph{One arm is nearly constant and we use the populational estimand.} The signal of that arm has a total variation close to zero, so any noise dominates it: the elastic mean of the arm is an oscillation whose size grows with the noise (Figure~\ref{fig:biasgrid}, control arm in column~(c), and Figure~\ref{fig:nullgrid}, first row), and $\tau_{\mathrm{pop}}$ carries it, even when the other arm is recovered well.
\item \emph{Treatment has no effect and we compute the local or individual estimands.} The contrast curves inside the elastic mean then carry no signal at all, which is situation~2 again. Besides, in that case the elastic mean is not unique: the law of the contrasts $Y^{(1)}-Y^{(0)}$ (or $Y^{(1)}-Y^{(0)} \mid X=x$) is symmetric, so a shape $q$ and its mirror image $-q$ are equally good minimisers of the objective $\E_{P}\!\big[\del^2([q],[q_{Y^{(1)}-Y^{(0)}}])\big]$, and, because of the template bias, the flat curve $q(t) \equiv 0$ is not a Fréchet minimiser of this objective. Therefore, the estimator returns one shape or its mirror with a sign that depends on the sample. Assumption~\ref{ass:conc}(4a) fails, and the estimate may be a non-zero curve although the effect is zero. Note that, still under the \emph{null}, the populational estimand does not share this limitation when each arm carries a shape whose total variation is large compared with the noise (otherwise it is situation~2 again). Then only the local and individual estimands are template biased (Figure~\ref{fig:nullgrid}, second row). 
\end{enumerate}

Figures~\ref{fig:biasgrid} and~\ref{fig:nullgrid} illustrate the three situations in Scenario~1, each on $800$ potential-outcome draws (the same draws for the three noise levels, new draws for the two rows without effect) and the same layout as Figure~\ref{fig:estgrid}.  In column~(b), $\tau_{\mathrm{loc}}$ (purple) is added to the estimands of Figure~\ref{fig:estgrid}.

\begin{figure}[htbp]
\centering
\includegraphics[width=0.93\textwidth]{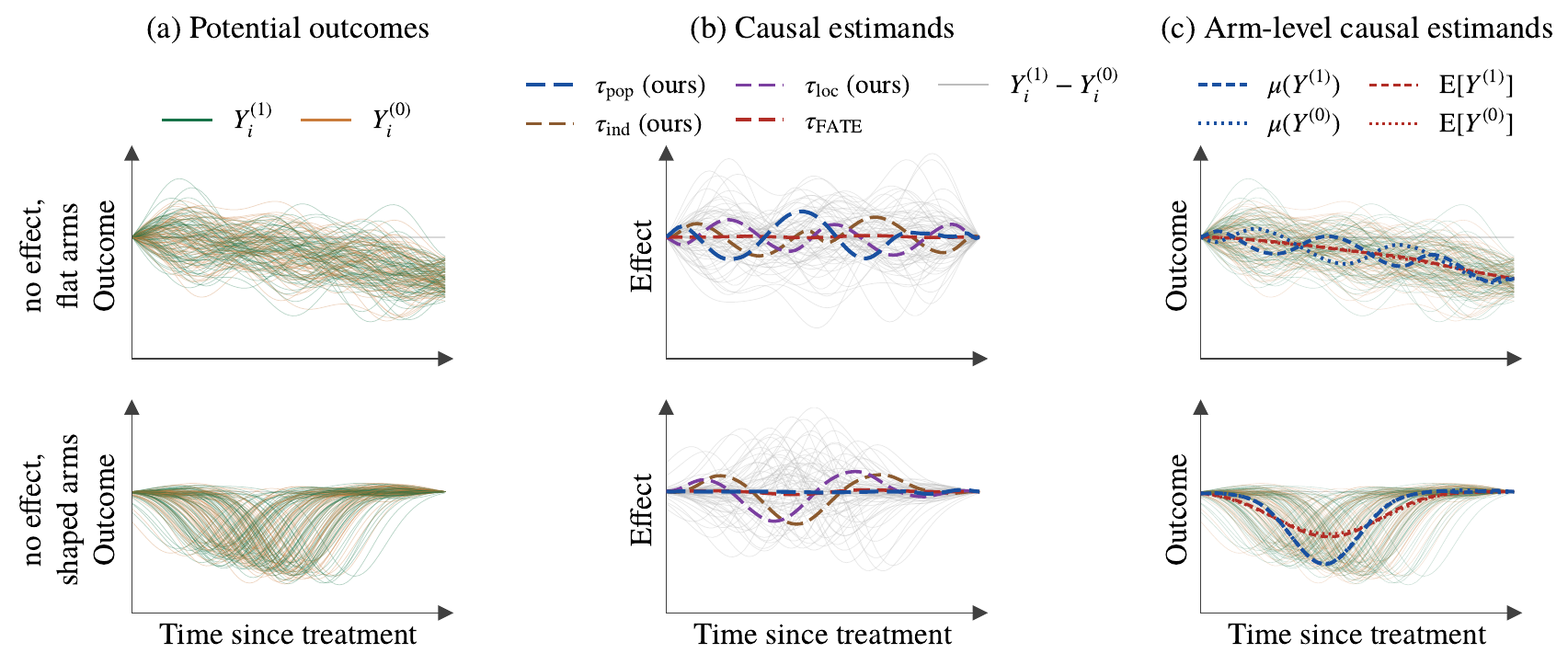}
\caption{\textbf{The causal estimands when the treatment has no effect} ($800$ draws, $\sigma_\varepsilon=0.03$). First row: both potential outcomes follow the control law (flat arms; the vertical scale is enlarged). Second row: both follow the treated law, drawn independently given $X$ (shaped arms). The noise is the same in both rows ($\sigma_\varepsilon=0.03$); it is visible
only in the first, whose vertical scale is enlarged. Columns and colours as in Figure~\ref{fig:biasgrid}.}
\label{fig:nullgrid}
\end{figure}

\paragraph{Noise level (Figure~\ref{fig:biasgrid}).}
The three rows use the same $800$ units at the noise scales $\sigma_\varepsilon\in\{0.03,0.10,0.30\}$. The noise-free versions ($\sigma_\varepsilon=0$, same units) serve as reference and are not shown. At $\sigma_\varepsilon=0.03$ (the paper's noise level), the elastic estimands barely move: they are within $0.03$ to $0.05$ in $L^2$ of their noise-free versions, which is of the order of the Monte-Carlo resolution of estimands computed on $800$ draws ($0.03$ to $0.04$), and noise increases their total variation only by about $4\%$. At $\sigma_\varepsilon=0.10$ the treated arm and the effect are still recovered ($\tau_{\mathrm{pop}}$ at about $0.08$ in $L^2$ distance of its noise-free version) while the control arm mean has become an oscillation (total variation $1.0$, almost ten times its noise-free value): this is situation~2 above. At $\sigma_\varepsilon=0.30$ every elastic estimand moves ($\tau_{\mathrm{pop}}$ and $\tau_{\mathrm{ind}}$ within $0.22$, $\tau_{\mathrm{loc}}$ within $0.43$ $L^2$ distance of their noise-free versions; total variations more than doubled; in this model $\tau_{\mathrm{ind}}=\tau_{\mathrm{loc}}$, so the gap between $0.22$ and $0.43$ is the numerical error of the two computations, which is large at this noise level), the control arm mean is a large oscillation (total variation $3.1$, nearly thirty times its noise-free value), and the elastic estimands read the noise as shape: this is situation~1. At every level the FATE is unchanged (within $0.03$ $L^2$ distance of its noise-free version): it is flattened by phase variation, but it does not react to the noise.

\paragraph{No effect (Figure~\ref{fig:nullgrid}).}
In these two simulations, each displayed in one row of Figure~\ref{fig:nullgrid}, both potential outcomes follow the same law, at the paper's noise level in both
cases. The difference between the two is that the potential outcomes have a low total variation compared to noise in the first simulation, while they have a high signal-to-noise total variation ratio in the second one. The FATE captures this absence of effect: it is zero up
to Monte-Carlo error in both cases. These two simulations illustrate the two halves of situation~3 above.
In the first, no elastic estimand escapes the template bias: $\tau_{\mathrm{ind}}$
and $\tau_{\mathrm{loc}}$ because of the absence of effect which makes the potential outcomes contrast law constant up to noise, and $\tau_{\mathrm{pop}}$
because each arm is nearly constant, which is situation~2 again. Thus, for the three targets even a small noise significantly affects the total variation of the shape captured and template bias is significant. In the second simulation, the arms carry enough total variation for their elastic means to be well determined and almost equally biased, so $\tau_{\mathrm{pop}}$ cancels, while $\tau_{\mathrm{ind}}$ and $\tau_{\mathrm{loc}}$ are still significantly affected by template bias. Note that here $\tau_{\mathrm{ind}}$ is
computed from independent potential outcomes; under the sharp null
($Y^{(1)}=Y^{(0)}$) every individual effect is the zero curve. Theorem~\ref{thm:wellposed} then does not apply: the Step-1 mean shape is the shape of the zero curve, whose only representative is $\tilde q=0$, so no admissible reference exists, and every re-timing is a best one (Appendix~\ref{app:flat}). Every choice of re-timings and of mean warp nevertheless gives zero phase-centered registered curves, so Step~3 returns $\tau_{\mathrm{ind}}\equiv0$ whatever choices are made.

\subsection{Sensitivity analyses} \label{app:sens}

The comparisons of Section~\ref{sec:sims} use the same sample size, treatment model and covariate dimension. This appendix varies each of them in Scenario~1 ($p_{\mathrm{rebound}}=100\%$, $\alpha=1$, $n=250$): Figures~\ref{fig:sensn}, \ref{fig:sensmis} and~\ref{fig:sensdim}, and Table~\ref{tab:sens} for all $L^2$ errors. Each figure reports the median $L^2$ error and the median absolute depth error of the estimators, with interquartile bands.

\paragraph{Protocol.}
Every analysis reuses the pipeline of Appendix~\ref{app:sims}. Only the treatment
model or the covariate dimension changes, so the outcome law and the estimands
stay the same. We add two reference estimators: an unweighted one, which adjusts
for nothing, and an oracle-weighted one, which uses the true propensities. The
first shows what no adjustment gives, the second what a perfect propensity model
gives.

\paragraph{Sample size.}
The $L^2$ error of the elastic estimators falls overall as the sample size grows, while
the two competitors stay at a roughly constant error: DR-FoS targets a different curve, the SRVF kernel ATE has no well-defined causal estimand, and
more data does not help (Figure~\ref{fig:sensn}). The elastic estimators are
better even at the smallest sizes, although the spread is wide there. Their error
stops falling at a floor set by the resolution of the reference estimands and by
the two computational caps, the number of alternations and the number of matched
differences. Importantly, even at small sample sizes the elastic estimators recover well the depth of the rebound template, and much better than the two competitors.

\paragraph{Misspecified propensity model (omitted interactions).}
A wrong treatment model impacts weighting only through the covariate imbalance it
leaves between the two imperfectly reweighted groups. We simulate a case where the treatment model is driven by interactions and the working model leaves them out. We use
four covariates $Z\sim N(0,I_4)$ and a severity index mixing a linear and an
interaction part,
\[
\eta_\lambda \;=\; \frac{(1-\lambda)\,\eta_{\mathrm{lin}}
  + \lambda\,\eta_{\mathrm{int}}}{\sqrt{(1-\lambda)^2+\lambda^2}},
\qquad
\eta_{\mathrm{lin}} = \tfrac{Z_1+0.5Z_2+0.25Z_3+0.1Z_4}{1.15},
\quad
\eta_{\mathrm{int}} = \tfrac{Z_1Z_2+Z_3Z_4}{\sqrt2},
\]
with $\operatorname{logit} e=2\,\eta_\lambda$, all estimators keeping a
linear-logistic propensity score model. As the interaction share $\lambda$ grows, the fitted
propensity flattens towards a constant and the populational weighted elastic estimator gets
worse, until it is no better than the unweighted one: the wrong weights remove no
confounding (Figure~\ref{fig:sensmis}). The oracle-weighted estimator does not
move, so the propensity model alone is at fault. The matched estimator, which
never models the propensity, changes much less.

\paragraph{Covariate dimension.}
We then place Scenario~1 in higher dimensions, with every covariate
carrying signal and with confounding, overlap and the estimand the same at every
dimension. Only the nuisance models get harder to fit from the same number of
observations. Both the propensity model and the matching get worse as the
dimension grows, the matched estimator more than the weighted one, while the
oracle-weighted estimator does not move (Figure~\ref{fig:sensdim}). The problem is not dimension
itself but estimating the nuisance models. In high
dimension, one can regularise or select the covariates of the propensity model, and match
on a propensity or prognostic score.

\begin{figure}[p]
\centering
\includegraphics[width=0.85\textwidth]{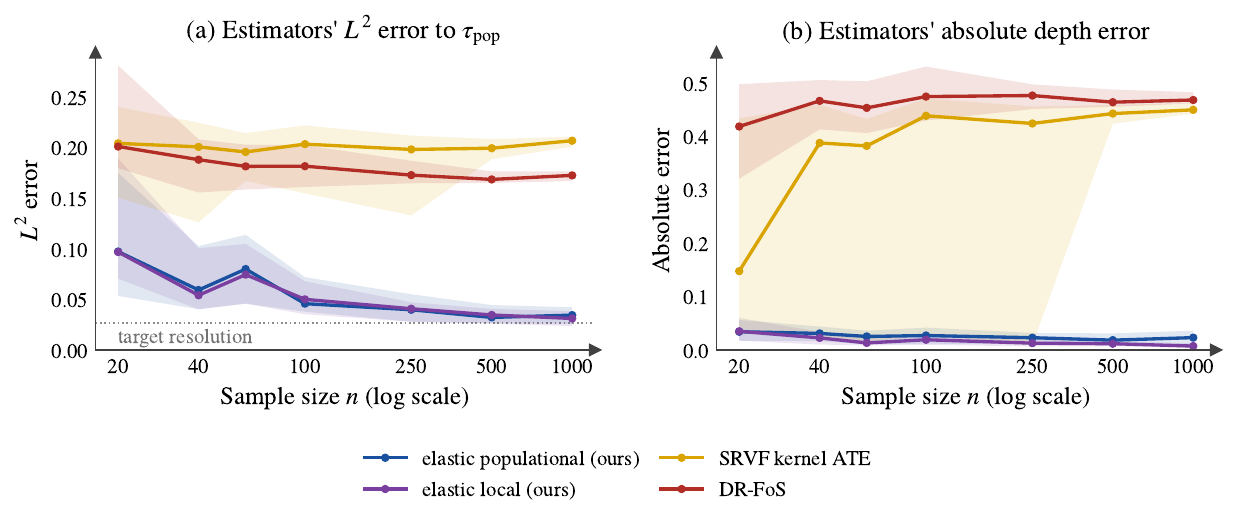}
\caption{\textbf{Sensitivity to the sample size} (Scenario~1,
$50$ replications per size, $30$ at $n=1000$). \emph{(a)} Estimators'
median $L^2$ error to $\tau_{\mathrm{pop}}$ (solid lines; bands:
interquartile range); the dotted line marks the resolution of the
Monte-Carlo target ($3200$ draws; $0.027$, Appendix~\ref{app:sims}). \emph{(b)} Estimators' median
absolute depth error.}
\label{fig:sensn}
\end{figure}

\begin{figure}[p]
\centering
\includegraphics[width=\textwidth]{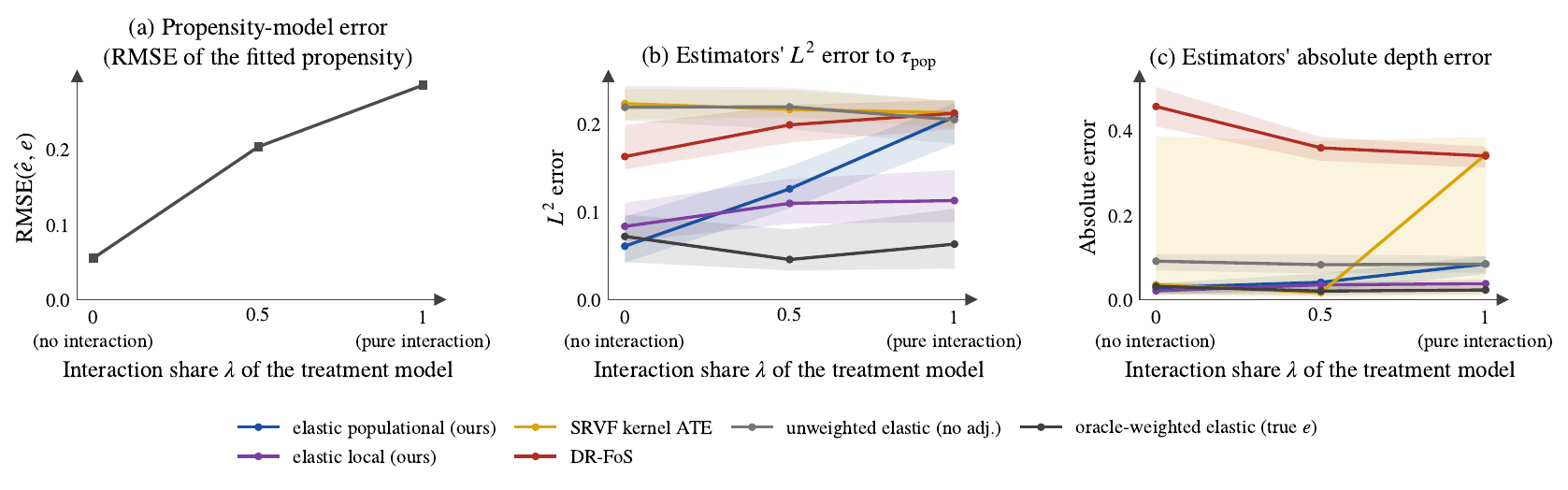}
\caption{\textbf{Sensitivity to propensity-model misspecification
(omitted interactions)} ($n=250$, $50$ replications per level; outcome
law, hence the estimand, unchanged). \emph{(a)} Root-mean-square error of the fitted
linear-logistic propensity against the true one, as a function of the
interaction share $\lambda$ (same abscissa as the other panels): as
$\lambda$ grows the fitted model flattens towards a constant while the
true propensity stays steep. \emph{(b)} Estimators' median $L^2$
error (bands: interquartile range); the unweighted (grey) and
oracle-weighted (dark grey) elastic estimators bracket the failure: at
$\lambda=1$ the fitted-weight estimator meets the unweighted one, while
the oracle stays at the well-specified level. \emph{(c)} Estimators'
median absolute depth error.}
\label{fig:sensmis}
\end{figure}

\begin{figure}[p]
\centering
\includegraphics[width=\textwidth]{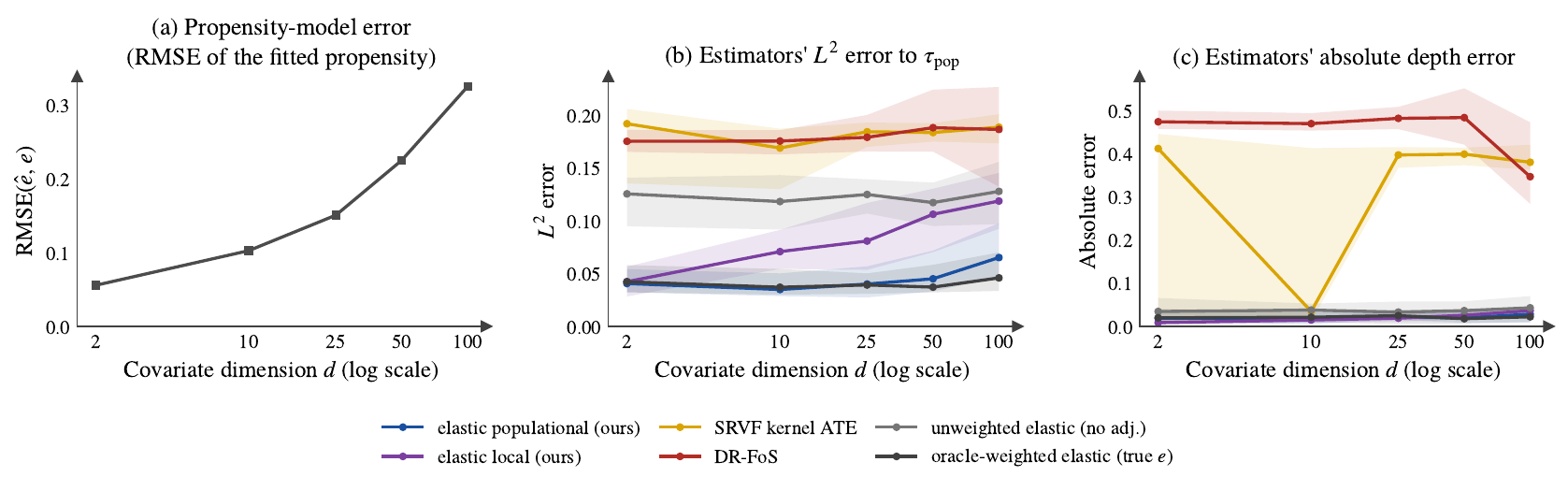}
\caption{\textbf{Sensitivity to the covariate dimension, dense design}
($n=250$, $50$ replications per dimension; all $d$ covariates carry
signal through the index $U=\Phi\big(\sum_j X_j/\sqrt d\big)$; confounding,
overlap and the estimand are identical at every $d$). \emph{(a)} Root-mean-square error of the
fitted propensity against the true one, as a function of $d$.
\emph{(b)} Estimators' median $L^2$
error (bands: interquartile range); the oracle-weighted estimator (dark
grey) is flat, so the IPW degradation is entirely nuisance-estimation
error. \emph{(c)} Estimators' median absolute depth error.}
\label{fig:sensdim}
\end{figure}

\begin{table}[p]
\centering
\caption{\textbf{Sensitivity analyses: median [interquartile range] of the
$L^2$ error to $\tau_{\mathrm{pop}}$} over $50$ replications ($30$ at
$n=1000$).  Every block varies one mechanism around the baseline
(Scenario~1, $n=250$); the misspecification and dimension
blocks leave the outcome law (hence the estimand) exactly unchanged.}
\label{tab:sens}
\vspace{4pt}
\small
\begin{tabular}{@{}lcccc@{}}
\toprule
design & elastic pop.\ (ours) & elastic local (ours) & SRVF kernel ATE & DR-FoS \\
\midrule
baseline ($n=250$) & $0.040$ [$0.028$, $0.055$] & $0.041$ [$0.028$, $0.047$] & $0.198$ [$0.133$, $0.212$] & $0.173$ [$0.165$, $0.187$] \\
\addlinespace[2pt]\multicolumn{5}{@{}l}{\emph{sample size}}\\[1pt]
\quad $n=20$ & $0.097$ [$0.053$, $0.175$] & $0.097$ [$0.071$, $0.189$] & $0.204$ [$0.151$, $0.240$] & $0.201$ [$0.180$, $0.281$] \\
\quad $n=40$ & $0.059$ [$0.040$, $0.103$] & $0.054$ [$0.040$, $0.101$] & $0.201$ [$0.126$, $0.225$] & $0.188$ [$0.156$, $0.208$] \\
\quad $n=60$ & $0.080$ [$0.046$, $0.114$] & $0.075$ [$0.046$, $0.105$] & $0.196$ [$0.167$, $0.214$] & $0.182$ [$0.159$, $0.203$] \\
\quad $n=100$ & $0.046$ [$0.038$, $0.072$] & $0.050$ [$0.035$, $0.068$] & $0.204$ [$0.155$, $0.222$] & $0.182$ [$0.161$, $0.203$] \\
\quad $n=500$ & $0.032$ [$0.027$, $0.044$] & $0.035$ [$0.026$, $0.041$] & $0.200$ [$0.189$, $0.209$] & $0.169$ [$0.165$, $0.177$] \\
\quad $n=1000$ & $0.034$ [$0.026$, $0.042$] & $0.031$ [$0.024$, $0.038$] & $0.207$ [$0.201$, $0.211$] & $0.173$ [$0.167$, $0.177$] \\
\addlinespace[2pt]\multicolumn{5}{@{}l}{\emph{propensity misspecification (omitted interaction)}}\\[1pt]
\quad $\lambda=0$ (well-spec.) & $0.061$ [$0.042$, $0.094$] & $0.083$ [$0.066$, $0.110$] & $0.223$ [$0.205$, $0.240$] & $0.163$ [$0.148$, $0.198$] \\
\quad $\lambda=0.5$ & $0.126$ [$0.104$, $0.151$] & $0.109$ [$0.086$, $0.137$] & $0.216$ [$0.207$, $0.238$] & $0.199$ [$0.178$, $0.221$] \\
\quad $\lambda=1$ & $0.208$ [$0.176$, $0.223$] & $0.113$ [$0.088$, $0.147$] & $0.212$ [$0.195$, $0.226$] & $0.212$ [$0.194$, $0.227$] \\
\addlinespace[2pt]\multicolumn{5}{@{}l}{\emph{covariate dimension (dense design)}}\\[1pt]
\quad $d=2$ & $0.041$ [$0.032$, $0.054$] & $0.042$ [$0.028$, $0.057$] & $0.192$ [$0.136$, $0.206$] & $0.175$ [$0.165$, $0.186$] \\
\quad $d=10$ & $0.035$ [$0.029$, $0.050$] & $0.071$ [$0.054$, $0.091$] & $0.169$ [$0.130$, $0.187$] & $0.176$ [$0.163$, $0.186$] \\
\quad $d=25$ & $0.040$ [$0.027$, $0.057$] & $0.081$ [$0.054$, $0.117$] & $0.184$ [$0.170$, $0.193$] & $0.179$ [$0.166$, $0.200$] \\
\quad $d=50$ & $0.045$ [$0.033$, $0.072$] & $0.106$ [$0.070$, $0.130$] & $0.184$ [$0.175$, $0.193$] & $0.188$ [$0.166$, $0.224$] \\
\quad $d=100$ & $0.065$ [$0.047$, $0.098$] & $0.119$ [$0.092$, $0.145$] & $0.189$ [$0.173$, $0.201$] & $0.187$ [$0.132$, $0.227$] \\
\bottomrule
\end{tabular}

\end{table}
\section{Real-data application: supplementary material}
\label{app:realdata}

This appendix details the material behind Section~\ref{sec:realdata}: the cohort construction 
and attrition, the covariates, the reconstruction of the weight curves and 
the stability of the estimates.

\subsection{Cohort and attrition}
\label{app:cohort_attrition}
The application studies the effect of GLP-1 RA treatment (brands Zepbound, Wegovy and
Saxenda) on the weight trajectories of patients over a horizon of $h=24$ months.
The claims data used to build the cohort cover the period from 1 September 2020 to
30 September 2025, and the weight measurements extend to August 2026.
The treated arm follows a new-user design and contains the patients with a GLP-1 RA
prescription. Their index date $t_0$ is the first fill. We require continuous insurance
coverage over the $6$ months before $t_0$ (defining the washout period), so that no earlier
exposure to the treatment is missed: an insured patient is reimbursed for their
prescriptions, which therefore appear in the claims. Treated patients who meet this
requirement are the pseudo-initiators of Figure~\ref{fig:rd_flow_chart}.
Controls are patients who have not used a GLP-1 RA during the entire observation window.
For them, we require continuous insurance enrolment over the washout period and the
horizon ($30$ months). Any GLP-1 RA fill during follow-up would then appear in the
claims, so that controls are truly non-users over the whole horizon.
Their index date $t_0$ is drawn at random within an eligible window of their
insurance coverage. Controls who meet this requirement are the pseudo-non-users of
Figure~\ref{fig:rd_flow_chart}.
The insurance coverage before $t_0$ also allows some covariates to be measured properly
(Appendix~\ref{app:covariates}).
Eligibility further requires five weight measurements: one in the $4$ months before
$t_0$, three over the horizon $[t_0,t_0+h)$, and one in
the $4$ months after $t_0+h$. The first and the last anchor the
spline at both ends of the horizon, which stabilises the baseline weight
(Appendix~\ref{app:reconstruction_weight_curves}).
We also require valid socio-demographics: age between $18$ and $80$ years, gender male or
female, weight between $30$ and $300$\,kg, and BMI between $12$ and $80$\,kg/m$^2$.

Figure~\ref{fig:rd_flow_chart} shows the attrition flow chart. For computational
reasons, controls are randomly down-sampled to $n_0=5\,000$.

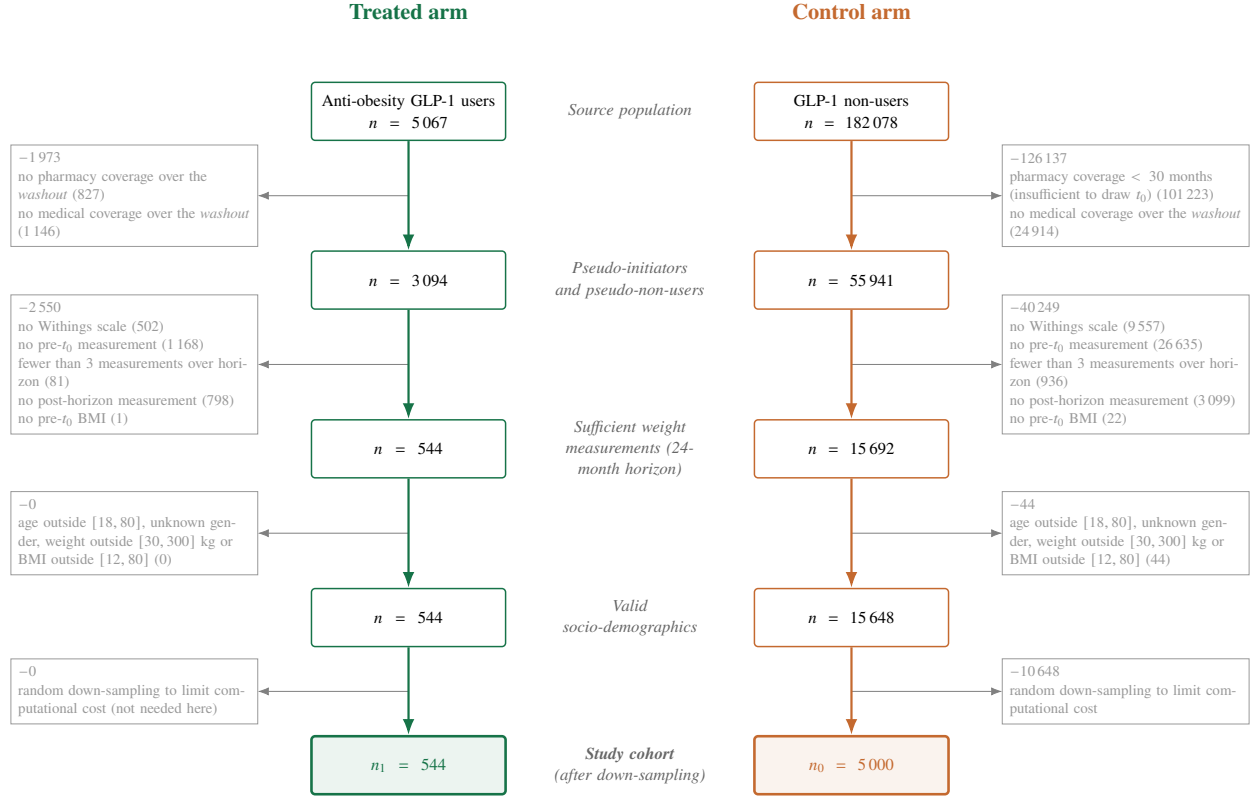
\begin{figure}[t]
\centering
\definecolor{treatedColor}{HTML}{167348}
\definecolor{controlColor}{HTML}{C4692F}
\resizebox{\textwidth}{!}{%
\begin{tikzpicture}[
  font=\footnotesize,
  arm/.style={draw, thick, rounded corners=2pt, align=center, inner sep=4pt,
              text width=34mm, minimum height=11mm},
  excl/.style={draw, black!40, align=left, inner sep=4pt,
               text width=44mm, font=\scriptsize},
  stage/.style={align=center, text width=30mm, font=\footnotesize\itshape,
                color=black!60},
  head/.style={align=center, font=\large\bfseries},
  flow/.style={-{Latex[length=3mm]}, line width=1.2pt},
  exclarrow/.style={-{Latex[length=2mm]}, black!50, thin}
]

\def\xT{-4.2}   \def\xC{4.2}
\def\xET{-9.4}  \def\xEC{9.4}
\def\yA{0}      \def\yB{-3.2}   \def\yC{-6.4}
\def\yD{-9.6}   \def\yE{-12.4}

\node[head, color=treatedColor] at (\xT, 1.9) {Treated arm};
\node[head, color=controlColor] at (\xC, 1.9) {Control arm};

\node[arm, draw=treatedColor] (t1) at (\xT,\yA)
  {Anti-obesity GLP-1 users\\\textbf{$n = 5\,067$}};
\node[arm, draw=controlColor] (c1) at (\xC,\yA)
  {GLP-1 non-users\\\textbf{$n = 182\,078$}};
\node[stage] at (0,\yA) {Source population};

\node[arm, draw=treatedColor] (t2) at (\xT,\yB) {\textbf{$n = 3\,094$}};
\node[arm, draw=controlColor] (c2) at (\xC,\yB) {\textbf{$n = 55\,941$}};
\node[stage] at (0,\yB) {Pseudo-initiators\\and pseudo-non-users};

\coordinate (mt1) at (\xT,{(\yA+\yB)/2});
\coordinate (mc1) at (\xC,{(\yA+\yB)/2});
\node[excl] (et1) at (\xET,{(\yA+\yB)/2})
  {\textbf{$-1\,973$}\\
   no pharmacy coverage over the \textit{washout} (827)\\
   no medical coverage over the \textit{washout} (1\,146)};
\node[excl] (ec1) at (\xEC,{(\yA+\yB)/2})
  {\textbf{$-126\,137$}\\
   pharmacy coverage $< 30$ months (insufficient to draw $t_0$)
   (101\,223)\\
   no medical coverage over the \textit{washout} (24\,914)};
\draw[exclarrow] (mt1) -- (et1.east);
\draw[exclarrow] (mc1) -- (ec1.west);

\node[arm, draw=treatedColor] (t3) at (\xT,\yC) {\textbf{$n = 544$}};
\node[arm, draw=controlColor] (c3) at (\xC,\yC) {\textbf{$n = 15\,692$}};
\node[stage] at (0,\yC) {Sufficient weight\\measurements (24-month horizon)};

\coordinate (mt2) at (\xT,{(\yB+\yC)/2});
\coordinate (mc2) at (\xC,{(\yB+\yC)/2});
\node[excl] (et2) at (\xET,{(\yB+\yC)/2})
  {\textbf{$-2\,550$}\\
   no Withings scale (502)\\
   no pre-$t_0$ measurement (1\,168)\\
   fewer than 3 measurements over horizon (81)\\
   no post-horizon measurement (798)\\
   no pre-$t_0$ BMI (1)};
\node[excl] (ec2) at (\xEC,{(\yB+\yC)/2})
  {\textbf{$-40\,249$}\\
   no Withings scale (9\,557)\\
   no pre-$t_0$ measurement (26\,635)\\
   fewer than 3 measurements over horizon (936)\\
   no post-horizon measurement (3\,099)\\
   no pre-$t_0$ BMI (22)};
\draw[exclarrow] (mt2) -- (et2.east);
\draw[exclarrow] (mc2) -- (ec2.west);

\node[arm, draw=treatedColor] (t4) at (\xT,\yD) {\textbf{$n = 544$}};
\node[arm, draw=controlColor] (c4) at (\xC,\yD) {\textbf{$n = 15\,648$}};
\node[stage] at (0,\yD) {Valid\\socio-demographics};

\coordinate (mt3) at (\xT,{(\yC+\yD)/2});
\coordinate (mc3) at (\xC,{(\yC+\yD)/2});
\node[excl] (et3) at (\xET,{(\yC+\yD)/2})
  {\textbf{$-0$}\\
   age outside $[18,80]$, unknown gender, weight outside $[30,300]$ kg or BMI outside $[12,80]$ (0)};
\node[excl] (ec3) at (\xEC,{(\yC+\yD)/2})
  {\textbf{$-44$}\\
   age outside $[18,80]$, unknown gender, weight outside $[30,300]$ kg or BMI outside $[12,80]$ (44)};
\draw[exclarrow] (mt3) -- (et3.east);
\draw[exclarrow] (mc3) -- (ec3.west);

\node[arm, draw=treatedColor, line width=1.4pt, fill=treatedColor!8] (t5) at (\xT,\yE)
  {\textcolor{treatedColor}{\textbf{$n_1 = 544$}}};
\node[arm, draw=controlColor, line width=1.4pt, fill=controlColor!8] (c5) at (\xC,\yE)
  {\textcolor{controlColor}{\textbf{$n_0 = 5\,000$}}};
\node[stage] at (0,\yE) {\textbf{Study cohort}\\(after down-sampling)};

\coordinate (mt4) at (\xT,{(\yD+\yE)/2});
\coordinate (mc4) at (\xC,{(\yD+\yE)/2});
\node[excl] (et4) at (\xET,{(\yD+\yE)/2})
  {\textbf{$-0$}\\
   random down-sampling to limit computational
   cost (not needed here)};
\node[excl] (ec4) at (\xEC,{(\yD+\yE)/2})
  {\textbf{$-10\,648$}\\
   random down-sampling to limit computational
   cost};
\draw[exclarrow] (mt4) -- (et4.east);
\draw[exclarrow] (mc4) -- (ec4.west);

\foreach \a/\b in {t1/t2, t2/t3, t3/t4, t4/t5}
  \draw[flow, treatedColor] (\a) -- (\b);
\foreach \a/\b in {c1/c2, c2/c3, c3/c4, c4/c5}
  \draw[flow, controlColor] (\a) -- (\b);

\end{tikzpicture}}
\caption{\textbf{Attrition flow chart of the study cohort, for a horizon of $h = 24$ months}.
Exclusion reasons and their counts are shown on either side of the diagram.}
\label{fig:rd_flow_chart}
\end{figure}

\subsection{Covariates}
\label{app:covariates}
The retained set has $p=9$ covariates ($14$ one-hot columns) in the three families
of Section~\ref{sec:cohort_covariates_outcome}: socio-demographics (age, gender and
race), health status (baseline BMI and pre-$t_0$ counts of Elixhauser comorbidities
\citep{elixhauser1998,quan2005}, medication categories and medical visits), and
motivation (pre-$t_0$ counts of weight measures and preventive care categories). Baseline BMI 
is the weighted mean of the
measurements taken before $t_0$, with Gaussian kernel weights centred at $t_0$ and a standard
deviation of $30$ days. The counts before $t_0$ are computed over the $6$-month washout
period. Age is taken at the index date $t_0$.

This set reflects the trade-off between two of the three identifying conditions.
On the one hand, conditional exchangeability pushes towards more covariates, to block as many
confounding paths as possible.
On the other hand, positivity pushes towards fewer, since each added covariate makes overlap
between the arms harder to achieve.
This is why baseline health status and motivation are summarised by event counts rather
than by one indicator per category of medical event.
Socio-economic status remains an unmeasured confounder, since no proxy for it is available in
the data. All patients of both arms own a Withings connected scale, which suggests a relatively
homogeneous socio-economic status and limits this source of confounding, without removing it.
Figure~\ref{fig:rd_covariates} shows the distributions of the nine covariates in each arm. As
expected, controls are healthier than the treated: they have a lower baseline BMI and fewer
comorbidities, medications and medical visits before $t_0$.

\begin{figure}[h]
\centering
\includegraphics[width=1\textwidth]{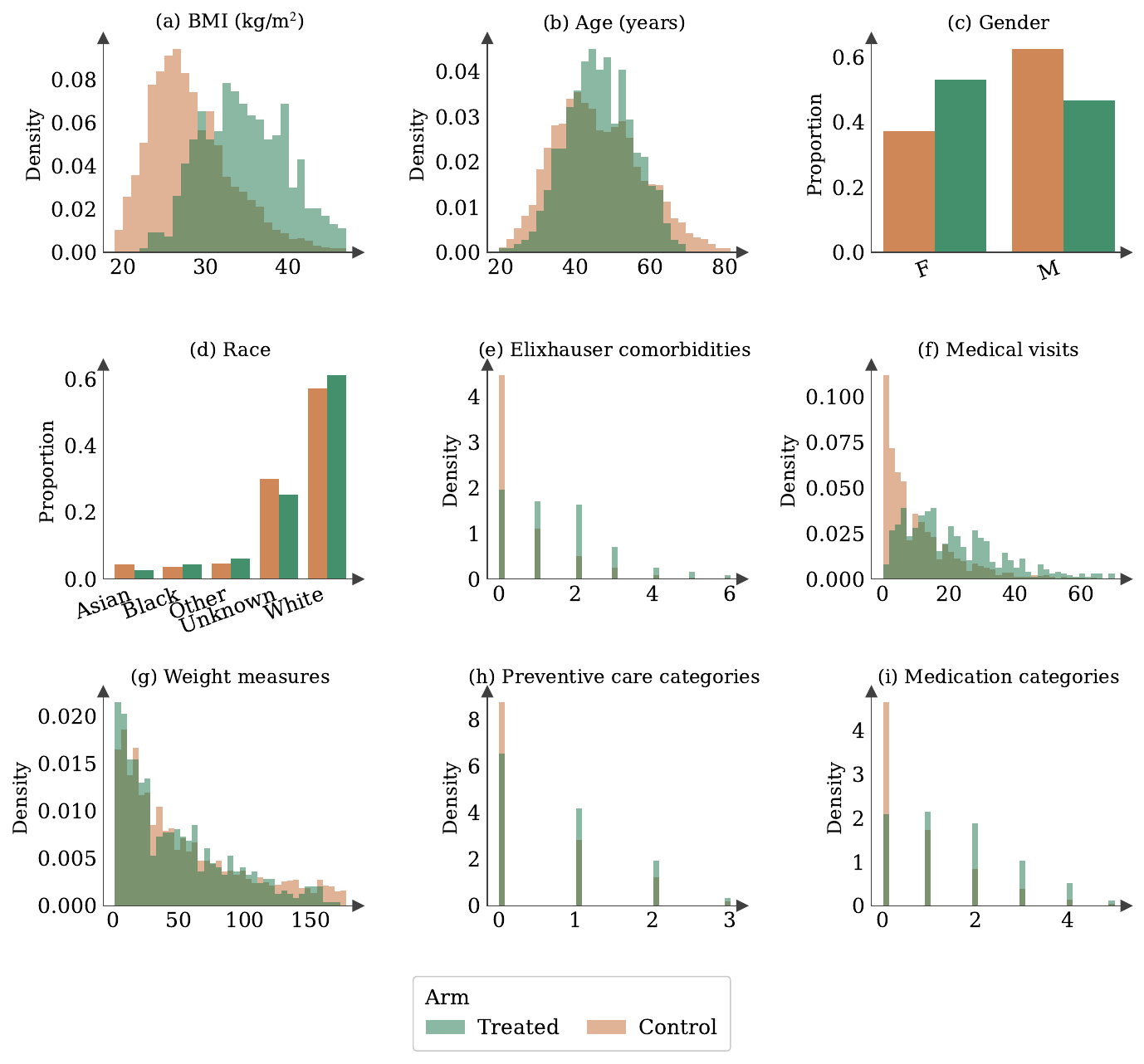}
\caption{\textbf{Treated patients have a higher BMI and more medical events
before $t_0$ than controls.} Unweighted distributions of the nine covariates in
each arm. \emph{(a)}, \emph{(b)}: baseline BMI and age at $t_0$.
\emph{(c)}, \emph{(d)}: proportions per category. \emph{(e)}--\emph{(i)}:
counts over the $6$-month washout period before $t_0$. Continuous variables and
counts are shown as normalised histograms.}
\label{fig:rd_covariates}
\end{figure}

The elastic populational estimator adjusts for these imbalances by reweighting the controls
to match the treated on the observed covariates (see Section~\ref{sec:popest}). 
The weights target the effect on the
treated, and the propensity scores they use are estimated with an
eXtreme Gradient Boosting model (XGBoost, \citealp{chen_xgboost_2016}).
We assess the quality of the adjustment with the absolute standardised mean difference (ASMD)
of each covariate:
\[
\text{ASMD}_k = \frac{|\bar{X}_{k,\text{treated}} - \bar{X}_{k,\text{control}}|}{\sqrt{\frac{1}{2}(\sigma^2_{k,\text{treated}} + \sigma^2_{k,\text{control}})}},
\]
where $\bar{X}_{k,\text{treated}}$ and $\bar{X}_{k,\text{control}}$ are the means of covariate $k$ 
in the treated and control groups, respectively, and $\sigma^2_{k,\text{treated}}$ and 
$\sigma^2_{k,\text{control}}$ are the corresponding variances.
The love plot of Figure~\ref{fig:rd_love_plot} shows the ASMDs before and after reweighting,
and the improvement in balance that the reweighting brings.
The usual threshold for acceptable balance is an ASMD below $0.1$ \citep{normand2001,austin2009}.
After reweighting, only the baseline BMI remains above it, and it stays below the looser
threshold of $0.25$ also used in practice, in line with the balance guidelines of \citet{rubin2001}.

\begin{figure}[h]
\centering
\includegraphics[width=1\textwidth]{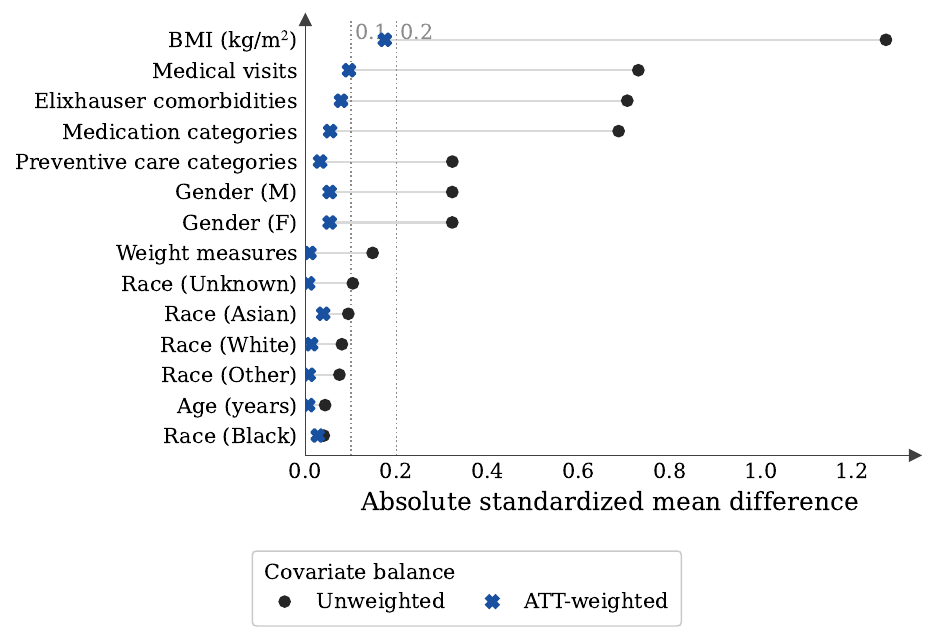}
\caption{\textbf{Reweighting balances all covariates except baseline BMI}. 
ASMD of each covariate between the treated and the controls, before reweighting (black
dots) and after ATT reweighting (blue crosses).}
\label{fig:rd_love_plot}
\end{figure}

\subsection{Reconstruction of the weight curves}
\label{app:reconstruction_weight_curves}
{
Let $W_{ij}$ be the $j$-th weight measurement of patient $i$, taken at time $t_{ij}$.
The raw weights of each patient are first reconstructed independently by a cubic
smoothing spline $\hat f_i$ minimising
\begin{equation*}
\sum_{j=1}^{m_i} \bigl(W_{ij}-f(t_{ij})\bigr)^2
   + \lambda \int_{-120}^{h+120} f''(t)^2\,dt,
\end{equation*}
where $\lambda>0$ is the smoothing parameter,
$t$ is measured in days since $t_0$ (so that $h=728$) and $m_i$ is the number of
measurements of patient $i$ on $[-120, h+120]$, that is, from $4$ months before $t_0$
to $4$ months after the end of the horizon. The baseline weight is the value of the
spline at the index date, $W_{i,\mathrm{baseline}}=\hat f_i(0)$, and the reconstructed
trajectory is expressed as a relative change from it, in percent in the figures:
\[
Y_i(t) = \frac{\hat f_i(t) - W_{i,\mathrm{baseline}}}{W_{i,\mathrm{baseline}}}.
\]
Evaluated on a weekly grid of $[0,h)$, it gives the outcome curve
$Y_i\in\mathbb{R}^{104}$, which starts at $0$ at $t_0$ for every patient. Fitting
beyond both ends of the horizon avoids the boundary effects of the spline at $t=0$
and $t=h$. Taking the baseline from the spline, rather than from the measurements
alone, also makes it less sensitive to their noise.}

The smoothing level is crucial, as it shapes the curves and hence
the estimates. It trades off two conflicting requirements. Strong smoothing removes the
measurement noise of the weights, due for instance to the time of day, clothing or the
accuracy of the scale, but flattens the shapes the elastic estimators aim to capture,
such as rebounds and plateaus. Weak smoothing preserves these features but lets noise
through, which a shape-sensitive method will then read as signal.

Calibrating $\lambda$ for each patient is especially difficult: automatic
criteria such as generalised cross-validation \citep{craven1978} optimise pointwise prediction
rather than shape fidelity, and no reference shape is available to tune
against. We therefore select one level of smoothing per arm, common to
all patients of that arm. The treated curves carry the shapes of interest, so we keep a
moderate smoothing ($\lambda=10^6$) that preserves their rebounds and plateaus. The control
curves are nearly flat on average, which is the situation where the template bias of the
elastic mean is largest relative to the signal (Section~\ref{sec:noise}), so we smooth them
strongly ($\lambda=10^8$). This reduces the part of the noise read as signal, without
removing it (Appendix~\ref{app:stability_estimations}). Both levels are chosen by visual
inspection, so that the reconstructed trajectories remain consistent with the raw
measurements. Figure~\ref{fig:rd_ex_lambda} shows reconstructed trajectories of a sample
of patients for a grid of smoothing levels.

\begin{figure}[h]
\centering
\includegraphics[width=1\textwidth]{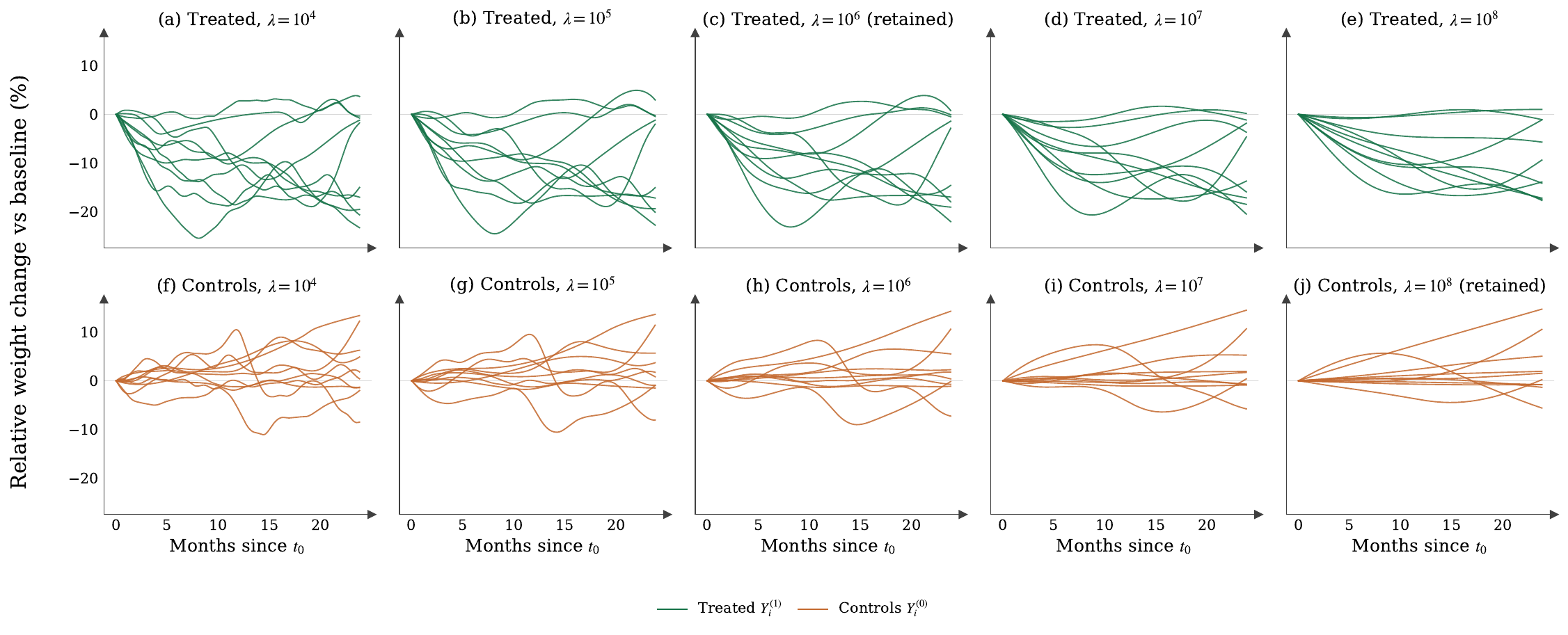}
\caption{\textbf{Examples of reconstructed weight trajectories for different smoothing levels $\lambda$}. 
Trajectories of $10$ treated and $10$ controls for $\lambda \in \lbrace 10^4, 10^5, 10^6, 10^7, 10^8\rbrace$.
The retained levels are $\lambda=10^6$ for the treated and
$\lambda=10^8$ for the controls.}
\label{fig:rd_ex_lambda}
\end{figure}

\subsection{Stability of the estimates}
\label{app:stability_estimations}
To assess the stability of the estimated effects, we compute each estimator on
$B=50$ bootstrap samples, obtained by resampling patients with replacement, 
stratified by treatment arm.
Figure~\ref{fig:rd_bootstrap} shows the results.

\begin{figure}[h]
\centering
\includegraphics[width=1\textwidth]{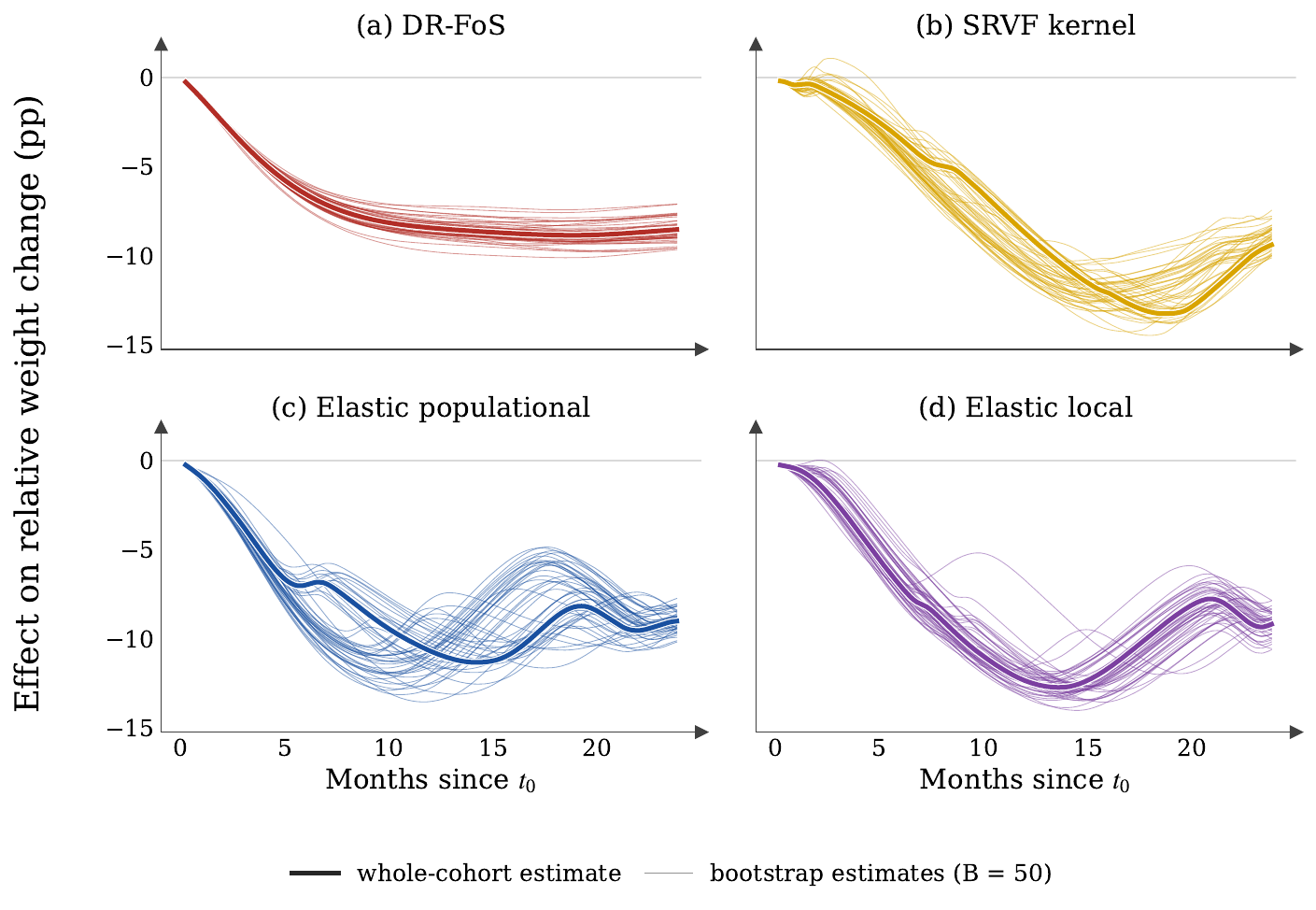}
\caption{\textbf{The elastic local estimate keeps its shape
across bootstrap samples, the populational one less so.} Estimated effect on the
treated for each estimator, on the full cohort (thick line) and on $B=50$ bootstrap
samples (thin lines).}
\label{fig:rd_bootstrap}
\end{figure}

{
DR-FoS is the most stable estimator: all bootstrap estimates share its monotone shape
and differ only in level, by about $\pm 1$ pp. The SRVF kernel estimator
and the elastic local estimator also keep their shape across bootstrap samples, with 
moderate variability in phase and amplitude. The local estimate reaches its lowest point
between about $12$ and $16$ months and its rebound peaks around $20$ months in almost
all samples.

The elastic populational estimator is less stable. The timing of its lowest point
varies from about $9$ to $15$ months, and its rebound varies in both timing and
height. The small plateau of the full-cohort estimate around $6$--$7$ months appears
in only part of the samples. This instability is likely due to the
template bias of the nearly flat control arm
(Appendix~\ref{app:reconstruction_weight_curves}): the oscillation it adds to the elastic
control mean changes in amplitude and phase from one sample to another, and the
populational estimate inherits it.

\clearpage
\bibliographystyle{plainnat}
\bibliography{refs}

\begin{thebibliography}{62}
\providecommand{\natexlab}[1]{#1}
\providecommand{\url}[1]{\texttt{#1}}
\expandafter\ifx\csname urlstyle\endcsname\relax
  \providecommand{\doi}[1]{doi: #1}\else
  \providecommand{\doi}{doi: \begingroup \urlstyle{rm}\Url}\fi

\bibitem[Afsari(2011)]{afsari_riemannian_2011}
Bijan Afsari.
\newblock Riemannian {$L^p$} center of mass: existence, uniqueness, and
  convexity.
\newblock \emph{Proceedings of the American Mathematical Society}, 139\penalty0
  (2):\penalty0 655--673, 2011.
\newblock \doi{10.1090/S0002-9939-2010-10541-5}.

\bibitem[Aliprantis and Border(2006)]{aliprantis_infinite_2006}
Charalambos~D. Aliprantis and Kim~C. Border.
\newblock \emph{Infinite Dimensional Analysis: A Hitchhiker's Guide}.
\newblock Springer, Berlin, 3rd edition, 2006.

\bibitem[Austin(2009)]{austin2009}
Peter~C. Austin.
\newblock Balance diagnostics for comparing the distribution of baseline
  covariates between treatment groups in propensity-score matched samples.
\newblock \emph{Statistics in Medicine}, 28\penalty0 (25):\penalty0 3083--3107,
  2009.

\bibitem[Bhattacharya and Patrangenaru(2003)]{bhattacharya_large_2003}
Rabi Bhattacharya and Vic Patrangenaru.
\newblock Large sample theory of intrinsic and extrinsic sample means on
  manifolds. {I}.
\newblock \emph{The Annals of Statistics}, 31\penalty0 (1):\penalty0 1--29,
  2003.

\bibitem[Bigot and Charlier(2011)]{bigot_consistency_2011}
J{\'e}r{\'e}mie Bigot and Benjamin Charlier.
\newblock On the consistency of {F}r{\'e}chet means in deformable models for
  curve and image analysis.
\newblock \emph{Electronic Journal of Statistics}, 5:\penalty0 1054--1089,
  2011.
\newblock arXiv:1010.0427.

\bibitem[Brezis(2011)]{brezis_functional_2011}
Ha\"im Brezis.
\newblock \emph{Functional Analysis, {Sobolev} Spaces and Partial Differential
  Equations}.
\newblock Universitext. Springer, New York, 2011.

\bibitem[Bridson and Haefliger(1999)]{bridson_haefliger_1999}
Martin~R. Bridson and Andr\'e Haefliger.
\newblock \emph{Metric Spaces of Non-Positive Curvature}, volume 319 of
  \emph{Grundlehren der mathematischen Wissenschaften}.
\newblock Springer, Berlin, 1999.

\bibitem[Bruveris(2016)]{bruveris_optimal_2016}
Martins Bruveris.
\newblock Optimal reparametrizations in the square root velocity framework.
\newblock \emph{SIAM Journal on Mathematical Analysis}, 48\penalty0
  (6):\penalty0 4335--4354, 2016.
\newblock \doi{10.1137/15M1014693}.
\newblock arXiv:1507.02728.

\bibitem[Burago et~al.(2001)Burago, Burago, and Ivanov]{burago_2001}
Dmitri Burago, Yuri Burago, and Sergei Ivanov.
\newblock \emph{A Course in Metric Geometry}, volume~33 of \emph{Graduate
  Studies in Mathematics}.
\newblock American Mathematical Society, 2001.

\bibitem[Chabassier et~al.(2026)Chabassier, Vittrant, Binon, Scornet, and
  Josse]{chabassier_glp1_2026}
Agathe Chabassier, Benjamin Vittrant, Pierre Binon, Erwan Scornet, and Julie
  Josse.
\newblock Real-world body composition and blood pressure changes following
  glucagon-like peptide 1 receptor agonist initiation: A causal inference study
  of connected device data.
\newblock \emph{Mayo Clinic Proceedings: Digital Health}, page 100375, 2026.
\newblock \doi{10.1016/j.mcpdig.2026.100375}.

\bibitem[Chakraborty and Panaretos(2021)]{chakraborty_functional_2021}
Anirvan Chakraborty and Victor~M. Panaretos.
\newblock Functional registration and local variations: Identifiability, rank,
  and tuning.
\newblock \emph{Bernoulli}, 27\penalty0 (2):\penalty0 1103--1130, 2021.

\bibitem[Chen and Guestrin(2016)]{chen_xgboost_2016}
Tianqi Chen and Carlos Guestrin.
\newblock {XGBoost}: A scalable tree boosting system.
\newblock In \emph{Proceedings of the 22nd ACM SIGKDD International Conference
  on Knowledge Discovery and Data Mining}, pages 785--794, 2016.
\newblock \doi{10.1145/2939672.2939785}.

\bibitem[Craven and Wahba(1978)]{craven1978}
Peter Craven and Grace Wahba.
\newblock Smoothing noisy data with spline functions: Estimating the correct
  degree of smoothing by the method of generalized cross-validation.
\newblock \emph{Numerische Mathematik}, 31\penalty0 (4):\penalty0 377--403,
  1978.
\newblock \doi{10.1007/BF01404567}.

\bibitem[Cuturi and Blondel(2017)]{cuturi_soft-dtw_2017}
Marco Cuturi and Mathieu Blondel.
\newblock Soft-{DTW}: a {Differentiable} {Loss} {Function} for {Time}-{Series}.
\newblock In \emph{Proceedings of the 34th {International} {Conference} on
  {Machine} {Learning}}, volume~70 of \emph{Proceedings of Machine Learning
  Research}, pages 894--903. PMLR, July 2017.

\bibitem[Dacorogna(2008)]{dacorogna_direct_2008}
Bernard Dacorogna.
\newblock \emph{Direct Methods in the Calculus of Variations}, volume~78 of
  \emph{Applied Mathematical Sciences}.
\newblock Springer, New York, 2nd edition, 2008.

\bibitem[Devilliers et~al.(2017{\natexlab{a}})Devilliers, Allassonni{\`e}re,
  Trouv{\'e}, and Pennec]{devilliers_inconsistency_2017}
Lo{\"i}c Devilliers, St{\'e}phanie Allassonni{\`e}re, Alain Trouv{\'e}, and
  Xavier Pennec.
\newblock Inconsistency of template estimation by minimizing of the
  variance/pre-variance in the quotient space.
\newblock \emph{Entropy}, 19\penalty0 (6):\penalty0 288, 2017{\natexlab{a}}.
\newblock \doi{10.3390/e19060288}.

\bibitem[Devilliers et~al.(2017{\natexlab{b}})Devilliers, Allassonni{\`e}re,
  Trouv{\'e}, and Pennec]{devilliers_template_2017}
Lo{\"i}c Devilliers, St{\'e}phanie Allassonni{\`e}re, Alain Trouv{\'e}, and
  Xavier Pennec.
\newblock Template estimation in computational anatomy: {F}r{\'e}chet means in
  top and quotient spaces are not consistent.
\newblock \emph{SIAM Journal on Imaging Sciences}, 10\penalty0 (3):\penalty0
  1139--1169, 2017{\natexlab{b}}.
\newblock \doi{10.1137/16M1083931}.
\newblock arXiv:1608.03703.

\bibitem[Dupuy et~al.(2011)Dupuy, Loubes, and Maza]{dupuy_structural_2011}
Jean-Fran\c{c}ois Dupuy, Jean-Michel Loubes, and Elie Maza.
\newblock Non parametric estimation of the structural expectation of a
  stochastic increasing function.
\newblock \emph{Statistics and Computing}, 21\penalty0 (1):\penalty0 121--136,
  2011.

\bibitem[Ecker et~al.(2024)Ecker, de~Luna, and Schelin]{ecker_causal_2024}
Kreske Ecker, Xavier de~Luna, and Lina Schelin.
\newblock Causal inference with a functional outcome.
\newblock \emph{Journal of the Royal Statistical Society Series C}, 73\penalty0
  (1):\penalty0 221--240, 2024.
\newblock \doi{10.1093/jrsssc/qlad092}.

\bibitem[Elixhauser et~al.(1998)Elixhauser, Steiner, Harris, and
  Coffey]{elixhauser1998}
Anne Elixhauser, Claudia Steiner, D.~Robert Harris, and Rosanna~M. Coffey.
\newblock Comorbidity measures for use with administrative data.
\newblock \emph{Medical Care}, 36\penalty0 (1):\penalty0 8--27, 1998.

\bibitem[Evans and Jaffe(2024)]{evans_jaffe_2024}
Steven~N. Evans and Adam~Q. Jaffe.
\newblock Limit theorems for {Fr\'echet} mean sets.
\newblock \emph{Bernoulli}, 30\penalty0 (1):\penalty0 419--447, 2024.
\newblock \doi{10.3150/23-BEJ1603}.
\newblock arXiv:2012.12859.

\bibitem[Even and Josse(2025)]{even_rethinking_2025}
Mathieu Even and Julie Josse.
\newblock Rethinking the win ratio: A causal framework for hierarchical outcome
  analysis, 2025.
\newblock arXiv:2501.16933 [stat.ME].

\bibitem[Fr\'echet(1948)]{frechet_elements_1948}
Maurice Fr\'echet.
\newblock Les \'el\'ements al\'eatoires de nature quelconque dans un espace
  distanci\'e.
\newblock \emph{Annales de l'Institut Henri Poincar\'e}, 10\penalty0
  (4):\penalty0 215--310, 1948.

\bibitem[Gasoyan et~al.(2025)Gasoyan, Butsch, Schulte, Casacchia, Le, Boyer,
  Griebeler, Burguera, and Rothberg]{gasoyan_changes_2025}
Hamlet Gasoyan, W.~Scott Butsch, Rebecca Schulte, Nicholas~J. Casacchia, Phuc
  Le, Christopher~B. Boyer, Marcio~L. Griebeler, Bartolome Burguera, and
  Michael~B. Rothberg.
\newblock Changes in weight and glycemic control following obesity treatment
  with semaglutide or tirzepatide by discontinuation status.
\newblock \emph{Obesity}, 33\penalty0 (9):\penalty0 1657--1667, 2025.
\newblock \doi{10.1002/oby.24331}.

\bibitem[H{\'a}jek(1971)]{hajek_comment_1971}
Jaroslav H{\'a}jek.
\newblock Comment on ``{A}n essay on the logical foundations of survey
  sampling, part one'' by {D}. {B}asu.
\newblock In V.~P. Godambe and D.~A. Sprott, editors, \emph{Foundations of
  Statistical Inference}, page 236. Holt, Rinehart and Winston, Toronto, 1971.

\bibitem[Hern\'an and Robins(2016)]{hernan_target_2016}
Miguel~A. Hern\'an and James~M. Robins.
\newblock Using big data to emulate a target trial when a randomized trial is
  not available.
\newblock \emph{American Journal of Epidemiology}, 183\penalty0 (8):\penalty0
  758--764, 2016.

\bibitem[Hern\'an and Robins(2020)]{hernan_causal_2020}
Miguel~A. Hern\'an and James~M. Robins.
\newblock \emph{Causal Inference: What If}.
\newblock Chapman \& Hall/CRC, Boca Raton, 2020.

\bibitem[Horvitz and Thompson(1952)]{horvitz_generalization_1952}
Daniel~G. Horvitz and Donovan~J. Thompson.
\newblock A generalization of sampling without replacement from a finite
  universe.
\newblock \emph{Journal of the American Statistical Association}, 47\penalty0
  (260):\penalty0 663--685, 1952.

\bibitem[Imbens and Rubin(2015)]{imbens_rubin_2015}
Guido~W. Imbens and Donald~B. Rubin.
\newblock \emph{Causal Inference for Statistics, Social, and Biomedical
  Sciences: An Introduction}.
\newblock Cambridge University Press, New York, 2015.

\bibitem[Karcher(1977)]{karcher_riemannian_1977}
Hermann Karcher.
\newblock Riemannian center of mass and mollifier smoothing.
\newblock \emph{Communications on Pure and Applied Mathematics}, 30\penalty0
  (5):\penalty0 509--541, 1977.

\bibitem[Katta et~al.(2024)Katta, Parikh, Rudin, and
  Volfovsky]{katta_interpretable_2024}
Srikar Katta, Harsh Parikh, Cynthia Rudin, and Alexander Volfovsky.
\newblock Interpretable causal inference for analyzing wearable, sensor, and
  distributional data.
\newblock In \emph{Proceedings of the 27th International Conference on
  Artificial Intelligence and Statistics (AISTATS)}, volume 238 of
  \emph{Proceedings of Machine Learning Research}, pages 3340--3348. PMLR,
  2024.
\newblock arXiv:2312.10569.

\bibitem[Kneip and Gasser(1992)]{kneip_statistical_1992}
Alois Kneip and Theo Gasser.
\newblock Statistical {Tools} to {Analyze} {Data} {Representing} a {Sample} of
  {Curves}.
\newblock \emph{The Annals of Statistics}, 20\penalty0 (3):\penalty0
  1266--1305, September 1992.
\newblock \doi{10.1214/aos/1176348769}.

\bibitem[Kurisu et~al.(2024)Kurisu, Zhou, Otsu, and
  M{\"u}ller]{kurisu_geodesic_2024}
Daisuke Kurisu, Yidong Zhou, Taisuke Otsu, and Hans-Georg M{\"u}ller.
\newblock Geodesic causal inference, 2024.
\newblock arXiv:2406.19604 [stat.ME].

\bibitem[Kurtek et~al.(2011)Kurtek, Srivastava, and Wu]{kurtek_signal_2011}
Sebastian~A. Kurtek, Anuj Srivastava, and Wei Wu.
\newblock Signal estimation under random time-warpings and nonlinear signal
  alignment.
\newblock In \emph{Advances in Neural Information Processing Systems 24}, pages
  675--683. Curran Associates, 2011.

\bibitem[Leoni(2017)]{leoni_first_2017}
Giovanni Leoni.
\newblock \emph{A First Course in {S}obolev Spaces}, volume 181 of
  \emph{Graduate Studies in Mathematics}.
\newblock American Mathematical Society, Providence, RI, 2nd edition, 2017.

\bibitem[Lin et~al.(2023)Lin, Kong, and Wang]{lin_causal_2023}
Zhenhua Lin, Dehan Kong, and Linbo Wang.
\newblock Causal inference on distribution functions.
\newblock \emph{Journal of the Royal Statistical Society Series B: Statistical
  Methodology}, 85\penalty0 (2):\penalty0 378--398, 2023.
\newblock \doi{10.1093/jrsssb/qkad008}.

\bibitem[Liu et~al.(2024)Liu, Ecker, Schelin, and de~Luna]{liu_doubly_2024}
Xijia Liu, Kreske Ecker, Lina Schelin, and Xavier de~Luna.
\newblock Doubly robust estimation with functional outcomes missing at random,
  2024.
\newblock arXiv:2411.17224 [stat.ME].

\bibitem[Lunceford and Davidian(2004)]{lunceford_stratification_2004}
Jared~K. Lunceford and Marie Davidian.
\newblock Stratification and weighting via the propensity score in estimation
  of causal treatment effects: a comparative study.
\newblock \emph{Statistics in Medicine}, 23\penalty0 (19):\penalty0 2937--2960,
  2004.

\bibitem[Marron et~al.(2015)Marron, Ramsay, Sangalli, and
  Srivastava]{marron_functional_2015}
J.~S. Marron, James~O. Ramsay, Laura~M. Sangalli, and Anuj Srivastava.
\newblock Functional data analysis of amplitude and phase variation.
\newblock \emph{Statistical Science}, 30\penalty0 (4):\penalty0 468--484, 2015.
\newblock \doi{10.1214/15-STS524}.

\bibitem[Miolane et~al.(2017)Miolane, Holmes, and
  Pennec]{miolane_template_2017}
Nina Miolane, Susan Holmes, and Xavier Pennec.
\newblock Template shape estimation: Correcting an asymptotic bias.
\newblock \emph{SIAM Journal on Imaging Sciences}, 10\penalty0 (2):\penalty0
  808--844, 2017.
\newblock \doi{10.1137/16M1084493}.
\newblock arXiv:1610.01502.

\bibitem[Normand et~al.(2001)Normand, Landrum, Guadagnoli, et~al.]{normand2001}
Sharon-Lise~T. Normand, Mary~Beth Landrum, Edward Guadagnoli, et~al.
\newblock Validating recommendations for coronary angiography following acute
  myocardial infarction in the elderly: A matched analysis using propensity
  scores.
\newblock \emph{Journal of Clinical Epidemiology}, 54\penalty0 (4):\penalty0
  387--398, 2001.

\bibitem[Panaretos and Zemel(2016)]{panaretos_amplitude_2016}
Victor~M. Panaretos and Yoav Zemel.
\newblock Amplitude and phase variation of point processes.
\newblock \emph{The Annals of Statistics}, 44\penalty0 (2):\penalty0 771--812,
  2016.
\newblock arXiv:1603.08691.

\bibitem[Petitjean et~al.(2011)Petitjean, Ketterlin, and
  Gançarski]{petitjean_global_2011}
François Petitjean, Alain Ketterlin, and Pierre Gançarski.
\newblock A global averaging method for dynamic time warping, with applications
  to clustering.
\newblock \emph{Pattern Recognition}, 44\penalty0 (3):\penalty0 678--693, March
  2011.
\newblock \doi{10.1016/j.patcog.2010.09.013}.

\bibitem[Quan et~al.(2005)Quan, Sundararajan, Halfon, et~al.]{quan2005}
Hude Quan, Vijaya Sundararajan, Patricia Halfon, et~al.
\newblock Coding algorithms for defining comorbidities in {ICD-9-CM} and
  {ICD-10} administrative data.
\newblock \emph{Medical Care}, 43\penalty0 (11):\penalty0 1130--1139, 2005.

\bibitem[Ramos-Carre{\~n}o et~al.(2024)Ramos-Carre{\~n}o, Torrecilla,
  Carbajo-Berrocal, Marcos, and Su{\'a}rez]{ramos_carreno_scikitfda_2024}
Carlos Ramos-Carre{\~n}o, Jos{\'e}~Luis Torrecilla, Miguel Carbajo-Berrocal,
  Pablo Marcos, and Alberto Su{\'a}rez.
\newblock scikit-fda: A {P}ython package for functional data analysis.
\newblock \emph{Journal of Statistical Software}, 109\penalty0 (2), 2024.

\bibitem[Ramsay and Li(1998)]{ramsay_curve_1998}
J.~O. Ramsay and Xiaochun Li.
\newblock Curve registration.
\newblock \emph{Journal of the Royal Statistical Society: Series B (Statistical
  Methodology)}, 60\penalty0 (2):\penalty0 351--363, 1998.
\newblock \doi{10.1111/1467-9868.00129}.

\bibitem[Ramsay and Silverman(2005)]{ramsay_silverman_2005}
James~O. Ramsay and Bernard~W. Silverman.
\newblock \emph{Functional Data Analysis}.
\newblock Springer, New York, 2nd edition, 2005.

\bibitem[Raykov et~al.(2025)Raykov, Luo, Strait, and
  KhudaBukhsh]{raykov_kernel_2025}
Yordan~P. Raykov, Hengrui Luo, Justin~D. Strait, and Wasiur~R. KhudaBukhsh.
\newblock Kernel-based estimators for functional causal effects, 2025.
\newblock arXiv:2503.05024 [stat.ME].

\bibitem[Rodriguez et~al.(2025)Rodriguez, Zhang, Gratzl, Do,
  Goodwin~Cartwright, Baker, Gluckman, Stucky, and
  Emanuel]{rodriguez_discontinuation_2025}
Patricia~J. Rodriguez, Vincent Zhang, Samuel Gratzl, Duy Do, Brianna
  Goodwin~Cartwright, Charlotte Baker, Ty~J. Gluckman, Nicholas Stucky, and
  Ezekiel~J. Emanuel.
\newblock Discontinuation and reinitiation of dual-labeled {GLP-1} receptor
  agonists among {US} adults with overweight or obesity.
\newblock \emph{JAMA Network Open}, 8\penalty0 (1):\penalty0 e2457349, 2025.
\newblock \doi{10.1001/jamanetworkopen.2024.57349}.

\bibitem[Rubin(2001)]{rubin2001}
Donald~B. Rubin.
\newblock Using propensity scores to help design observational studies:
  Application to the tobacco litigation.
\newblock \emph{Health Services and Outcomes Research Methodology}, 2\penalty0
  (3--4):\penalty0 169--188, 2001.

\bibitem[Rubino et~al.(2021)Rubino, Abrahamsson, Davies,
  et~al.]{rubino_effect_2021}
Domenica Rubino, Niclas Abrahamsson, Melanie Davies, et~al.
\newblock Effect of continued weekly subcutaneous semaglutide vs placebo on
  weight loss maintenance in adults with overweight or obesity: the {STEP} 4
  randomized clinical trial.
\newblock \emph{JAMA}, 325\penalty0 (14):\penalty0 1414--1425, 2021.
\newblock \doi{10.1001/jama.2021.3224}.

\bibitem[Rudin(1976)]{rudin_principles_1976}
Walter Rudin.
\newblock \emph{Principles of Mathematical Analysis}.
\newblock McGraw-Hill, New York, 3rd edition, 1976.

\bibitem[Serrin and Varberg(1969)]{serrin_varberg_1969}
James Serrin and Dale~E. Varberg.
\newblock A general chain rule for derivatives and the change of variables
  formula for the {L}ebesgue integral.
\newblock \emph{The American Mathematical Monthly}, 76\penalty0 (5):\penalty0
  514--520, 1969.

\bibitem[Srivastava and Klassen(2016)]{srivastava_klassen_2016}
Anuj Srivastava and Eric~P. Klassen.
\newblock \emph{Functional and Shape Data Analysis}.
\newblock Springer Series in Statistics. Springer, New York, NY, 2016.
\newblock \doi{10.1007/978-1-4939-4020-2}.

\bibitem[Srivastava et~al.(2011)Srivastava, Wu, Kurtek, Klassen, and
  Marron]{srivastava_registration_2011}
Anuj Srivastava, Wei Wu, Sebastian Kurtek, Eric Klassen, and J.~S. Marron.
\newblock Registration of functional data using {Fisher--Rao} metric, 2011.
\newblock arXiv:1103.3817 [math.ST].

\bibitem[Sturm(2003)]{sturm_probability_2003}
Karl-Theodor Sturm.
\newblock Probability measures on metric spaces of nonpositive curvature.
\newblock In \emph{Heat Kernels and Analysis on Manifolds, Graphs, and Metric
  Spaces}, volume 338 of \emph{Contemporary Mathematics}, pages 357--390.
  American Mathematical Society, 2003.

\bibitem[Testa et~al.(2026)Testa, Boschi, Chiaromonte, Kennedy, and
  Reimherr]{testa_doubly_2025}
Lorenzo Testa, Tobia Boschi, Francesca Chiaromonte, Edward~H. Kennedy, and
  Matthew Reimherr.
\newblock Doubly-robust functional average treatment effect estimation.
\newblock \emph{Journal of Causal Inference}, 14\penalty0 (1):\penalty0
  20250045, 2026.
\newblock \doi{10.1515/jci-2025-0045}.
\newblock arXiv:2501.06024.

\bibitem[Tucker(2020)]{tucker_fdasrsf_2020}
J.~Derek Tucker.
\newblock fdasrsf: Functional data analysis using the square root slope
  framework.
\newblock Python package, \url{https://github.com/jdtuck/fdasrsf_python}, 2020.

\bibitem[Tucker et~al.(2013)Tucker, Wu, and Srivastava]{tucker_generative_2013}
J.~Derek Tucker, Wei Wu, and Anuj Srivastava.
\newblock Generative models for functional data using phase and amplitude
  separation.
\newblock \emph{Computational Statistics \& Data Analysis}, 61:\penalty0
  50--66, 2013.

\bibitem[Wilding et~al.(2021)Wilding, Batterham, Calanna, Davies, Van~Gaal,
  Lingvay, et~al.]{wilding_once-weekly_2021}
John P.~H. Wilding, Rachel~L. Batterham, Salvatore Calanna, Melanie Davies,
  Luc~F. Van~Gaal, Ildiko Lingvay, et~al.
\newblock Once-weekly semaglutide in adults with overweight or obesity.
\newblock \emph{The New England Journal of Medicine}, 384\penalty0
  (11):\penalty0 989--1002, 2021.
\newblock \doi{10.1056/NEJMoa2032183}.

\bibitem[Wilding et~al.(2022)Wilding, Batterham, Davies, Van~Gaal, Kandler,
  Konakli, Lingvay, McGowan, Oral, Rosenstock, Wadden, Wharton, Yokote, and
  Kushner]{wilding_weight_2022}
John P.~H. Wilding, Rachel~L. Batterham, Melanie Davies, Luc~F. Van~Gaal,
  Kristian Kandler, Katerina Konakli, Ildiko Lingvay, Barbara~M. McGowan,
  Tugce~Kalayci Oral, Julio Rosenstock, Thomas~A. Wadden, Sean Wharton, Koutaro
  Yokote, and Robert~F. Kushner.
\newblock Weight regain and cardiometabolic effects after withdrawal of
  semaglutide: The {STEP} 1 trial extension.
\newblock \emph{Diabetes, Obesity and Metabolism}, 24\penalty0 (8):\penalty0
  1553--1564, 2022.
\newblock \doi{10.1111/dom.14725}.

\bibitem[Ziezold(1977)]{ziezold_expected_1977}
Herbert Ziezold.
\newblock On expected figures and a strong law of large numbers for random
  elements in quasi-metric spaces.
\newblock In J.~Ko\v{z}e\v{s}nik, editor, \emph{Transactions of the Seventh
  Prague Conference on Information Theory, Statistical Decision Functions,
  Random Processes and of the 1974 European Meeting of Statisticians},
  volume~A, pages 591--602. Reidel, Dordrecht, 1977.

\end{thebibliography}

\end{document}